\documentclass[11pt]{article}
\usepackage[left=1in,top=1in,right=1in,bottom=1in,head=.1in,nofoot]{geometry}

\usepackage{setspace,url,bm,amsmath} %
\usepackage{tabularx}

\newcommand\clearrow{\global\let\rowmac\relax}
\usepackage{titlesec} %
\titlelabel{\thetitle.\quad} %
\titleformat*{\section}{\bf\Large\center}

\usepackage{graphicx} %
\usepackage{bbm}
\usepackage{latexsym}
\usepackage{caption}
\usepackage[margin=20pt]{subcaption}
\usepackage{hyperref}
\usepackage{booktabs}

\def\undertilde#1{\mathord{\vtop{\ialign{##\crcr
				$\hfil\displaystyle{#1}\hfil$\crcr\noalign{\kern 1.5pt\nointerlineskip}
				$\hfil\tilde{}\hfil$\crcr\noalign{\kern 0pt}}}}}

\usepackage[table]{xcolor}

\newcommand{\GG}[1]{}

\usepackage{amsthm}
\usepackage{amssymb}
\usepackage{amsmath}
\usepackage{color}

\usepackage{algorithm}

\usepackage{comment}
\theoremstyle{definition}

\newtheorem*{theorem*}{Theorem}
\newtheorem{theorem}{Theorem}
\newtheorem*{rmk*}{Remark}
\newtheorem{proposition}{Proposition}
\newtheorem{lemma}{Lemma}

\newtheorem{remark}{Remark}
\newtheorem{corollary}{Corollary}
\newtheorem*{corollary*}{Corollary}

\newcommand{\PPT}{partial permutation test }
\usepackage[osf]{mathpazo}

\usepackage{natbib} %
\bibpunct{(}{)}{;}{a}{}{,} %

\usepackage{etoolbox} %
\apptocmd{\sloppy}{\hbadness 10000\relax}{}{} %

\usepackage{color}
\usepackage{listings}

\def\ind{\begin{picture}(9,8)
         \put(0,0){\line(1,0){9}}
         \put(3,0){\line(0,1){8}}
         \put(6,0){\line(0,1){8}}
         \end{picture}
        }

\def\converged{\stackrel{d}{\longrightarrow}}
\def\convergep{\stackrel{\Pr}{\longrightarrow}}

\newcommand{\bs}{\boldsymbol}
\newcommand{\bXmatrix}{\bs{\mathrm{X}}}

\def\E{\mathbb{E}}

\def\perm{\text{p}}
\def\TV{\text{TV}}

\def\basis{e}

\def\linear{\textrm{Linear}}
\def\poly{\textrm{Poly}}
\def\Pr{\text{Pr}}
\def\remainder{\text{r}}

\def\iid{\text{i.i.d.}}
\def\correct{\textrm{c}}

\usepackage{bbm}
\def\I{\mathbbm{1}}

\usepackage[textsize=scriptsize, textwidth = 2cm, shadow]{todonotes}

\usepackage{multirow}
\usepackage{float}
\newfloat{algorithm}{t}{lop}

\allowdisplaybreaks

\begin{document}
\onehalfspacing

\title{\bf Kernel-based Partial Permutation Test for Detecting Heterogeneous Functional Relationship
}
\author{
Xinran Li, Bo Jiang and Jun S. Liu
\footnote{
Xinran Li, Department of Statistics, University of Illinois at Urbana-Champaign, Champaign, IL 61820 (E-mail: \href{mailto:xinranli@illinois.edu}{xinranli@illinois.edu}). 
Bo Jiang, Managing Director, Two Sigma Investments LP, New York, NY 10013 
(E-mail: \href{mailto:bojiang83@gmail.com}{bojiang83@gmail.com}). 
Jun S. Liu, Department of Statistics, Harvard University, Cambridge, MA 02138
(E-mail: \href{mailto:jliu@stat.harvard.edu}{jliu@stat.harvard.edu}). The views expressed herein are the authors’ alone and are not necessarily the views of Two Sigma Investments LP, or any of its affiliates. This research is partly supported by the NSF grant DMS-1712714.}
}
\date{}
\maketitle
\begin{abstract}
We propose a kernel-based partial permutation test for checking the equality of functional relationship between response and covariates among different groups. The main idea, which is intuitive and easy to implement, is to keep the projections of the response vector $\bs{Y}$ on leading principle components of a kernel matrix fixed and permute $\bs{Y}$'s projections on the remaining principle components. 
The proposed test allows for different choices of kernels, corresponding to different classes of functions under the null hypothesis. 
First, using linear or polynomial kernels, our partial permutation tests are exactly valid in finite samples for linear or polynomial regression models with Gaussian noise; similar results straightforwardly extend to kernels with finite feature spaces. 
Second, by allowing the kernel feature space to diverge with the sample size, the test can be large-sample valid for a wider class of functions. 
Third, for general kernels with possibly infinite-dimensional feature space, 
the partial permutation test is exactly valid when the covariates are exactly balanced across all groups, 
or asymptotically valid when the underlying function follows certain regularized Gaussian processes. 
We further suggest test statistics using likelihood ratio between two (nested) GPR
models, 
and propose computationally efficient algorithms utilizing the EM algorithm and Newton's method, 
where the latter also involves Fisher scoring and quadratic programming and 
is particularly useful when EM suffers from slow convergence. 
Extensions to correlated and non-Gaussian noises have also been investigated theoretically or numerically.  
Furthermore, 
the test can be extended to 
use multiple kernels together 
and can thus 
enjoy properties from each kernel. 
Both simulation study and application illustrate the properties of the proposed test. 
\end{abstract}


\section{Introduction}
Testing whether the same functional relationship between response and covariates holds across different groups is a challenging and important problem.  %
For example, in clinical trial studies, people want to compare effects of several treatments conditional on some important covariates of patients such as age, gender and genetic information.
Traditional methods assume parametric forms of these functional relationships,  such as linear or quadratic with unknown coefficients.
When such assumptions cannot be supported by prior knowledge, nonparametric tests for the equality of functional relationships were recommended, especially in the exploratory stage of data analysis. Most methods in the nonparametric setting focus on univariate 
functions 
and use kernel estimator for estimating regression curves. For example,  \citet{pardo2007testing} proposed empirical process based procedures for testing the equality of multiple regression curves. 
A comprehensive review on this topic can be found in \citet{neumeyer2003nonparametric}. 

Testing the equality of functions has also been studied from a Bayesian perspective. \citet{behseta2005testing} proposed two methods for testing the equality of two univariate functions using the Bayesian adaptive regression splines. \citet{benavoli2015gaussian} used Gaussian processes for Bayesian hypothesis testing on the equality of two functions, as well as the monotonicity and periodicity of a function. 
\citet{behseta2005hierarchical} applied hierarchical Gaussian processes to study the variability among multiple functions, where they assumed an independent Gaussian process prior for each function and focused on the estimation of the variance component.

A closely related study for comparing two regression functions, but with slightly different focus, is the regression discontinuity design \citep{thistlethwaite1960regression}, 
under which there can be no overlap between covariate distributions for the two groups in comparison. 
In this case, testing equality of two functions essentially reduces to testing whether two functions can be smoothly connected at the boundary. 
Various frequentist approaches have been proposed, including nonparametric kernel regression methods and local linear regression \citep{Hahn2001}; for a comprehensive review, see \citet{Imbensrdd2008}. 
Most existing methods focus mainly on the case with univariate covariates. 
Recently, \citet{BRANSON201914} and \citet{rischard2018bayesian} proposed Bayesian approaches utilizing Gaussian processes, and extended the regression discontinuity design to multivariate settings with spatial covariates. 

In this paper, we first propose a partial permutation test 
for linear functional relationships, and then  generalize it to handle non-linear relationships via kernel methods. 
We demonstrate the exact validity of the \PPT when the kernel corresponds to a finite-dimensional feature mapping whose linear span contains the underlying true function, or the covariates are exactly balanced across all groups.  
We further establish the asymptotic validity of the \PPT for a general smooth functional relationship when we choose the kernel adaptively with the sample size, 
or when the underlying function is from some Gaussian process.  
Note that 
the 
Gaussian process regression
model has received much attention recently for modeling functional relationships \citep[see, e.g.,][]{rasmussen2006gaussian, shi2011gaussian}
and is closely related to the kernel regression, which minimizes a squared loss with penalization on the functional norm in a reproducing kernel Hilbert space (RKHS) characterized by a kernel. 
Intuitively, 
we can understand 
$p$-values under the Gaussian process regression model 
in an averaging sense over the Gaussian process prior on the underlying function. As \citet*{meng1994posterior} suggested, uniformity under parameters following prior is a useful criterion for the evaluation of any proposed $p$-value. 
We also investigate the power of the test when there exists functional heterogeneity across different groups, and extend the test to cases with correlated errors across individuals.

The paper proceeds as follows. 
Section \ref{sec:notation} introduces notations, model assumptions and the \PPT  based on the linear or polynomial kernel, 
and proves its finite-sample validity when the underlying function is linear or polynomial. 
Section \ref{sec:general} first studies the \PPT using general kernels for general underlying functions under the null hypothesis with homogeneous functional relationship across all groups, 
and then 
shows its finite-sample or asymptotic validity under additional conditions on the kernel, the underlying function and the covariate distribution.  
Section \ref{sec:ppt_alter} studies the power of the \PPT under the alternative hypotheses with heterogeneous function relations across all groups. 
Section \ref{sec:implement} discusses practical implementation of the partial permutation test. 
Section \ref{sec:corr} extends the \PPT to correlated noises. 
Section \ref{sec:simu} conducts simulation study, and Section \ref{sec:app} applies the proposed test to a real data set. 
Section \ref{sec:dicuss} concludes with a short discussion.

\section{Notations, Hypotheses, Kernels, and Permutation Tests}\label{sec:notation}

\subsection{Notations and Problem Formulation}
Let $Y_i \in \mathbb{R}$, $\bs{X}_i \in \mathbb{R}^d$ and $Z_i\in \{1,2,\cdots, H\}$ denote the response variable, covariates of dimension $d$, and the group indicator for the $i$th ($1\le i \le n$) observation, respectively,  
and let $\bs{Y} = (Y_1, Y_2, \ldots, Y_n)^\top$, $\bXmatrix = (\bs{X}_1, \bs{X}_2, \ldots, \bs{X}_n)^\top$
and 
$\bs{Z} = (Z_1, Z_2, \ldots, Z_n)^\top$ be the %
corresponding vectors of all the $n$ units. 
Given observations from multiple groups, we want to test whether they share the same (unknown) functional relationship.  
Specifically, given a response variable $Y$ and a vector of covariates $\bs{X}$, the null hypothesis assumes that individuals from $H$ ($H \geq 2$) groups have the same relationship $\E(Y \mid \bs{X}) = f_0 ( \bs{X} )$ plus a Gaussian noise with constant variance, where $f_0$ is an unknown function in a given class (e.g., linear or polynomial functions), i.e.,
\begin{align}\label{eq:H_0_fixed}
    H_0:  Y_i = f_0( \bs{X}_i ) + \varepsilon_i, \quad \varepsilon_i \mid \bXmatrix, \bs{Z} \ \overset{\iid}{\sim} \  \mathcal{N}(0, \sigma_0^2), \qquad (1\le i \le n)
\end{align}
where $\bXmatrix$ and $\bs{Z}$ can be either fixed or random,
and noises $\varepsilon_i$'s are independent and identically distributed ($\iid$) conditional on $\bXmatrix$ and $\bs{Z}$. 
The alternative hypothesis allows different groups to have different (unknown) functions $f_1, \ldots, f_H$:
\begin{align}\label{eq:H_1_equal}
    H_{1}: Y_i =  f_{Z_i}( \bs{X}_{i} )+\varepsilon_i, \quad \varepsilon_i \mid \bXmatrix, \bs{Z} \ \overset{\iid}{\sim} \  \mathcal{N}(0,\sigma_0^2), \qquad (1\le i \le n)
\end{align}
or even different noise variances in different groups:
\begin{align}\label{eq:H_1_unequal}
H'_{1}: Y_i =  f_{Z_i}(\bs{X}_{i})+\varepsilon_{i}, \quad \varepsilon_i \mid \bXmatrix, \bs{Z} \ \overset{\iid}{\sim} \  \mathcal{N}(0,\sigma_{Z_i}^2), \qquad (1\le i \le n). 
\end{align}

\subsection{Partial Permutation Test for Linear Functional Relationship}\label{sec:linear}
We first consider a special case in which the relationship between the response and covariates under $H_0$ is linear,  i.e.,
$f_0(\bs{x}) = \beta_0 + \sum_{k=1}^d \beta_k x_k$  in \eqref{eq:H_0_fixed}
for some $\bs{\beta} = (\beta_0, \beta_1, \ldots, \beta_d)^\top \in \mathbb{R}^{d+1}$. Let $K_\linear(\bs{x},\bs{x}')= 1+ \bs{x}^\top\bs{x}'$ denote the linear kernel function with $\bs{x}, \bs{x}'\in \mathbb{R}^d$. We write the corresponding sample kernel matrix as
$\bs{K}_n \in \mathbb{R}^{n\times n}$, with its $(i,j)$-th element being $[\bs{K}_n]_{ij} = K_\linear(\bs{X}_i, \bs{X}_j)$ and its eigen-decomposition denoted as  $\bs{\Gamma} \bs{C} \bs{\Gamma}^\top$. Here, $\bs{\Gamma} =(\bs{\gamma}_1, \cdots, \bs{\gamma}_n) \in \mathbb{R}^{n\times n}$ is an orthogonal matrix and  $\bs{C} = \text{diag}(c_1, \ldots, c_n)
\in \mathbb{R}^{n\times n}$ has non-negative diagonal elements in a descending order.

The linear kernel can be equivalently written as 
$K_\linear(\bs{x},\bs{x}') = \phi(\bs{x})^\top \phi(\bs{x}')$, an inner product in a feature space defined by the feature mapping $\phi: \bs{x} \rightarrow (1, \bs{x}^\top)^\top \in \mathbb{R}^{d+1}$. 
Let $\bs{\Phi} \equiv (\phi(\bs{X}_1), \ldots, \phi(\bs{X}_n))^\top \in \mathbb{R}^{n \times (d+1)}$ be the matrix 
of all the observed covariates mapped into the feature space,
and let $\bs{f}_0 \equiv (f_0(\bs{X}_1), \cdots, f_0(\bs{X}_n))^\top$ be the vector 
of function values evaluated at these covariates. Under 
the null model  \eqref{eq:H_0_fixed}, 
we can verify that 
$\bs{f}_0
= \bs{\Phi} \bs{\beta}$ lies in the column space of $\bs{\Phi}$, or equivalently the column space of kernel matrix $\bs{K}_n = \bs{\Phi} \bs{\Phi}^\top$. 
Because  $\bs{K}_n$ has at most rank $d+1$, 
the eigenvectors $(\bs{\gamma}_{d+2}, \ldots, \bs{\gamma}_{n})$ must be orthogonal to the column space of $\bs{K}_n$, as well as the vector $\bs{f}_0$ in this column space.  
As a result, under $H_0$,
we have 
\begin{align*}
\bs{\Gamma}^\top \bs{Y} =
\left(
\bs{\gamma}_1^\top \bs{f}_0, \bs{\gamma}_2^\top \bs{f}_0, \cdots, \bs{\gamma}_{d+1}^\top \bs{f}_0, 0, \cdots, 0
\right)^\top +
\left(
\bs{\gamma}_1^\top \bs{\varepsilon},  \bs{\gamma}_2^\top \bs{\varepsilon}, \cdots, \bs{\gamma}_n^\top \bs{\varepsilon}
\right)^\top,
\end{align*}
where $\bs{\varepsilon} = (\varepsilon_1, \ldots, \varepsilon_n)^\top \in \mathbb{R}^n$. 
Therefore, $\bs{\gamma}_i^\top \bs{Y}=\bs{\gamma}_i^\top \bs{\varepsilon}$ for $i=d+2, \ldots, n$,
and are $\iid$ conditional on $\bXmatrix$ and $\bs{Z}$. 
Consequently, 
given any test statistic,
we can perform permutation tests by randomly permuting $\bs{\gamma}_i^\top \bs{Y}$  for $i=d+2,\ldots,n$. 
Note that this procedure 
takes advantage of the fact that  projections of $\bs{Y}$ onto the eigenvectors corresponding to zero eigenvalues are just random noises. 
Intuitively, this observation may be generalized so that one can treat  projections of $\bs{Y}$ onto eigenvectors with small eigenvalues as Gaussian noises (i.e., $\bs{\varepsilon}$) instead of signals (i.e., $\bs{f}_0$), which are then exchangeable and permit permutation tests.

For the convenience of presentation, we summarize in Algorithm \ref{alg:partial_permu}  a general discrete or continuous partial permutation test procedure with a given kernel function $K$, permutation size $b_n$, and test statistic $T$. 
For linear functional relationships, we have the following theorem on the validity of the $p$-value from either the discrete or continuous partial permutation test using the linear kernel. 

\begin{algorithm}
\small
\begin{itemize}
		\item[1)] Perform eigen-decomposition for kernel matrix $\bs{K}_n = \bs{\Gamma} \bs{C} \bs{\Gamma}^\top$, where $[\bs{K}_n]_{ij}=K(\bs{X}_i, \bs{X}_j)$, $\bs{\Gamma}$ is an orthogonal matrix and $\bs{C}$ is a diagonal matrix with diagonal elements in descending order.
		\item[2)] Let $\bs{W} = \bs{\Gamma}^\top \bs{Y}\equiv (W_1,\ldots, W_n)^\top.$
		\begin{itemize}
		    \item[(a)] For the {\it discrete} partial permutation test, we define the permutation set $\mathcal{S}_y$ as follows:
		\begin{align*}
		\mathcal{S}_y =& \{ \bs{Y}_\psi: \bs{Y}_\psi=\bs{\Gamma} \bs{W}_\psi, \bs{W}_\psi \in \mathcal{S}_w \}, \ \mbox{with} \\
		\mathcal{S}_w =& \{ \bs{W}_\psi: \bs{W}_\psi = (W_{\psi(1)}, W_{\psi(2)}, \cdots, W_{\psi(n)}), \psi \in \mathcal{M}(n, b_n) \},
		\end{align*}
		where $\mathcal{M}(n, b_n)$ is defined to be the set of permutations of $\{1,2,\cdots,n\}$ that keep the first $n-b_n$ elements invariant, i.e.,
		\begin{align*}
		\mathcal{M}(n,b_n) = & \big\{\psi:  \ \psi(i)=i, \text{ for } i = 1,2,\cdots, n-b_n,\\
		& \ \ \text{ and } \{ \psi(n-b_n+1 ), \cdots, \psi(n) \} \text{ is a permutation of } \{ n-b_n+1,\cdots, n \} \big\}.
		\end{align*}
        Note that we allow the  sets $\mathcal{S}_y$ and $\mathcal{S}_w$ to have members of identical value. 
        For example, if  $\psi \neq \psi' \in \mathcal{M}(n, b_n)$ but $\bs{W}_{\psi} = \bs{W}_{\psi'}$ (which may happen if some $W_j$'s take on the same value), then they are treated as two elements in $\mathcal{S}_w$.
		\item[(b)] For {\it continuous} partial permutation test, define the permutation set $\mathcal{S}_y$ as follows:
		\begin{align*}
		\mathcal{S}_y =& \{ \bs{Y}^*: \bs{Y}^*=\bs{\Gamma} \bs{W}^*,\bs{W}^* \in \mathcal{S}_w \}, \ \mbox{with}\\
		\mathcal{S}_w =& \left\{ \bs{W}^*: W^*_i = W_i, i=1,2,...,n-b_n, \sum_{i=n-b_n+1}^n (W_i^*)^2 = \sum_{i=n-b_n+1}^n W_i^2 \right\}.
		\end{align*}
		\end{itemize}
		\item[3)] Draw $\bs{W}^{\perm} \in \mathcal{S}_w$ uniformly, and let  $\bs{Y}^{\perm} = \bs{\Gamma} \bs{W}^{\perm}$. Naturally,
	    $\bs{Y}^{\perm}$ is uniformly distributed on $\mathcal{S}_y$,  
        where both $\mathcal{S}_w$ and $\mathcal{S}_y$ can be viewed as a function of $\bs{X}$ and $\bs{Y}$.
		\item[4)] The resulting partial permutation $p$-value with test statistic $T$ is then defined as
		\begin{align*}
		p( \bXmatrix, \bs{Y}, \bs{Z} ) = 
		\Pr \{ T( \bXmatrix, \bs{Y}^{\perm}, \bs{Z} ) \geq T(\bXmatrix, \bs{Y}, \bs{Z} ) \mid \bXmatrix, 
        \bs{Y}, 
        \bs{Z} 
		 \}.
		\end{align*}
\end{itemize}
\caption{
{ Discrete and continuous partial permutation tests with kernel function $K$, permutation size $b_n$ and test statistic $T$ for $\{ \bXmatrix, \bs{Y}, \bs{Z} \}$}}
\label{alg:partial_permu}
\end{algorithm}

\begin{theorem}\label{linear_permutation_pval_valid}
Let $\{(\bs{X}_i,Y_i,Z_i)\}_{1\leq i \leq n}$ denote samples from 
the model under $H_0$ in (\ref{eq:H_0_fixed}), where the functional relationship $f_0(\bs{x})$ is linear in $\bs{x}$. 
Then, the $p$-value obtained by either the discrete or continuous partial permutation test described in Algorithm \ref{alg:partial_permu} with kernel $K_\linear$, permutation size $b_n \le n-(d+1)$, and any test statistic $T$ is valid, i.e., $\forall \alpha \in (0,1)$, 
$
\Pr_{H_0} \{ p(\bXmatrix, \bs{Y}, \bs{Z} ) \leq \alpha 
\mid \bXmatrix, \bs{Z}
\} \leq \alpha. 
$
\end{theorem}

\begin{remark}
When the matrix $\bs{\Phi}$ consisting of the covariates mapped into the feature space is not of full rank, we can relax the constraint to be $b_n \le n - \text{rank}(\bs{\Phi})$. Similar relaxations also hold for Theorem \ref{poly_permutation_pval_valid} and Corollaries \ref{cor:kernel_finite_dim_feature_space} and \ref{cor:diverg_kernel}.
\end{remark}

Theorem 
\ref{linear_permutation_pval_valid}
suggests that we can use any test statistic to conduct a valid permutation test as long as the underlying functional relationship 
between the response and the covariates is linear.
To achieve a high power when the null hypothesis is false, we suggest to use the likelihood ratio statistics with respect to alternative hypotheses that are of particular interest. 
For example, we may choose either \eqref{eq:H_1_equal} or  \eqref{eq:H_1_unequal} as the alternative hypothesis, where we assume that the functions 
$f_1, \ldots, f_H$
are still linear in the covariates but can have different coefficients across the $H$ groups. 

Under the Gaussian linear regression model,
Algorithm \ref{alg:partial_permu} is able to generate permutation samples that change only the responses but keep both the covariates and group indicators fixed, which means that our partial permutation test is an exact conditional test -- conditioning on the the covariates and group indicators $(\bXmatrix, \bs{Z})$.
This is important since it 
avoids imposing 
any distributional assumption on $(\bXmatrix, \bs{Z})$. 
As a side note, simply permuting the group indicators $Z_i$'s may not lead to a valid permutation test since such a permutation does not maintain the joint distribution of the covariates and the group indicator.
An analogous approach is the classic bootstrap procedure based on residual  resampling \citep{Freedman1984, Hinkley1988}, which generates new data similar to the observed ones but keeps $(\bXmatrix, \bs{Z})$ fixed. 
In general, the residual bootstrap can help relax the Gaussianity assumption on the noises, but loses the finite-sample exact validity. 
Moreover,  the parametric F-test, whose test statistic is equivalent to the likelihood ratio statistic, is finite-sample valid and also regarded as most powerful under the linear model with Gaussian noises. 
As demonstrated both theoretically and empirically in Sections \ref{sec:ppt_alter} and \ref{sec:simu}, 
our partial permutation test can have almost the same power as the F-test.

\subsection{Partial Permutation Test for Polynomial Functional Relationship}\label{sec:poly}

We consider here a more general case in which the relationship between the response and covariates under $H_0$ is a polynomial of degree 
$p$ (or smaller), where $p$ is a positive integer. 
Specifically,  under $H_0$,
we assume that  $f_0(\bs{x}) = \sum_{j_1 + j_2 + \ldots + j_d \le p} \beta_{j_1j_2\ldots j_d} x_{1}^{j_1} x_{2}^{j_2} \cdots x_{d}^{j_d}$.
Let $K_\poly(\bs{x},\bs{x}')= (1+ \bs{x}^\top\bs{x}')^p$ denote a degree-$p$ polynomial kernel function.
We again let $\bs{K}_n$ denote the corresponding sample kernel matrix with entries $[\bs{K}_n]_{ij} = K_\poly(\bs{X}_i, \bs{X}_j)$, 
and let $\bs{\Gamma} \bs{C} \bs{\Gamma}^\top$ be the eigen-decomposition of  $\bs{K}_n$, where $\bs{C} = \text{diag}(c_1, \ldots, c_n)$ is the diagonal matrix with   non-negative eigenvalues $c_i$  in descending order,
and  $\bs{\Gamma} =(\bs{\gamma}_1, \cdots, \bs{\gamma}_n) \in \mathbb{R}^{n\times n}$ is an orthogonal matrix.

As with the linear case, we can rewrite the kernel matrix's entry  as an inner product in a feature space defined by the feature mapping $\phi$, i.e., $K_\poly(\bs{x},\bs{x}') = \phi(\bs{x})^\top \phi(\bs{x}')$, where $\phi(\bs{x})$ consists of all the monomials $x_1^{j_1}x_2^{j_2}\ldots x_d^{j_d}$ with $j_1+j_2 + \ldots +j_d \le p$ up to some positive coefficients. 
Let $\bs{\Phi} =  (\phi(\bs{X}_1), \ldots, \phi(\bs{X}_n))^\top$ be the $n \times \binom{d+p}{d}$ matrix consisting of the observed covariates mapped into the feature space, and  let 
$\bs{f}_0 = (f_0(\bs{X}_1), \cdots, f_0(\bs{X}_n))^\top$ denote the vector of  function values evaluated at these covariates. 
Under the null model  \eqref{eq:H_0_fixed}, 
we can verify that 
$\bs{f}_0$ must lie in the column space of $\bs{\Phi}$ or equivalently the column space of $\bs{K}_n = \bs{\Phi} \bs{\Phi}^\top$, i.e., $\bs{f}_0 = \bs{\Phi} \bs{\beta}$ for some $\bs{\beta}\in \mathbb{R}^{\binom{d+p}{d}}$.  
Because the rank of $\bs{K}_n$ is at most $\binom{d+p}{d}$ provided that $\binom{d+p}{d}<n$,  
$(\bs{\gamma}_{\binom{d+p}{d}+1}, \ldots, \bs{\gamma}_{n})$ must be orthogonal to the column space of $\bs{K}_n$, as well as $\bs{f}_0$ in the column space.  
Therefore, under $H_0$, we have
\begin{align*}
\bs{\Gamma}^\top \bs{Y} =
\left(
\bs{\gamma}_1^\top \bs{f}_0, \cdots, \bs{\gamma}_{\binom{d+p}{d}}^\top \bs{f}_0, 0, \cdots, 0
\right)^\top +
\left(
\bs{\gamma}_1^\top \bs{\varepsilon}, \cdots, \bs{\gamma}_n^\top \bs{\varepsilon}
\right)^\top, 
\end{align*}
recalling that $\bs{Y}=(Y_1, \ldots, Y_n)^\top$ and $\bs{\varepsilon} = (\varepsilon_1, \ldots, \varepsilon_n)^\top$. 
So
$\{\bs{\gamma}_i^\top \bs{Y}, \ i = \binom{d+p}{d}+1, \ldots,  n\} = \{\bs{\gamma}_i^\top \bs{\varepsilon}, \ i = \binom{d+p}{d}+1, \ldots, n\}$ are $\iid$ conditional on $(\bXmatrix, \bs{Z})$, and
we can perform permutation test by permuting $\{\bs{\gamma}_i^\top \bs{Y}, \  i=\binom{d+p}{d}+1, \ldots, n\}$.

\begin{theorem}\label{poly_permutation_pval_valid}
	Let $\{(\bs{X}_i,Y_i,Z_i)\}_{1\leq i \leq n}$ denote samples from 
	the model under $H_0$ in (\ref{eq:H_0_fixed}), where the functional relationship $f_0(\bs{x})$ is polynomial in $\bs{x}$ with degree at most $p$. 
	Then, the $p$-value obtained by either the discrete or continuous partial permutation test with kernel $K_\poly$, permutation size $b_n \le n-\binom{d+p}{d}$, and any test statistic $T$ is valid, i.e., $\forall \alpha \in (0, 1)$, 
	$
	\Pr_{H_0}\{p(\bXmatrix,\bs{Y}, \bs{Z} ) \leq \alpha 
	\mid \bXmatrix, \bs{Z}
	\} \leq \alpha. 
	$
\end{theorem}

Similar to that for Theorem \ref{linear_permutation_pval_valid}, 
we recommend to use a likelihood ratio statistic with carefully chosen alternative hypothesis of interest in order  to achieve a good power. 
For example,  
we may hypothesize that under the alternative hypothesis \eqref{eq:H_1_equal} or \eqref{eq:H_1_unequal} the functions $f_1, \ldots, f_H$ are still polynomial up to degree $p$
but can have different coefficients across different groups. 

\section{Partial Permutation Test for General Functional   Relationship}\label{sec:general}

\subsection{Partial Permutation Test under the Null Hypothesis}\label{sec:general_fix_null}

Inspired by the partial permutation test
based on linear and polynomial kernels, 
we generalize 
it to arbitrary kernels. 
Let $K$ be any kernel that is symmetric, positive definite and continuous, and let $\bs{K}_n \in \mathbb{R}^{n \times n}$ be the corresponding sample kernel matrix with $[\bs{K}_n]_{ij} = K(\bs{X}_i, \bs{X}_j)$.
Similar to the previous sections, we define eigen-decomposition on the kernel matrix $\bs{K}_n = \bs{\Gamma} \bs{C} \bs{\Gamma}^\top$, where $\bs{\Gamma} =(\bs{\gamma}_1, \cdots, \bs{\gamma}_n) \in \mathbb{R}^{n\times n}$ is an orthogonal matrix and $\bs{C} = \text{diag}(c_1, \cdots, c_n)
\in \mathbb{R}^{n\times n}$ has non-negative diagonal elements in descending order. 
Recall that $\bs{Y}=(Y_1, \ldots, Y_n)^\top$, $\bs{\varepsilon} = (\varepsilon_1, \ldots, \varepsilon_n)^\top$, and $\bs{f}_0 = (f_0(\bs{X}_1),  \cdots, f_0(\bs{X}_n))^\top$. 
Then, we have
\begin{align}\label{eq:general_kernel_Gamma_Y}
\bs{\Gamma}^\top \bs{Y} =
\left(
\bs{\gamma}_1^\top \bs{f}_0,  \bs{\gamma}_2^\top \bs{f}_0, \cdots, \bs{\gamma}_n^\top \bs{f}_0
\right)^\top +
\left(
\bs{\gamma}_1^\top \bs{\varepsilon}, \bs{\gamma}_2^\top \bs{\varepsilon}, \cdots, \bs{\gamma}_n^\top \bs{\varepsilon}
\right)^\top.
\end{align}
Different from 
linear or polynomial kernel for linear or polynomial functions, 
the kernel matrix $\bs{K}_n$ can be of full rank, 
and
$\bs{\gamma}^\top_i \bs{f}_0 = 0$ may not hold exactly for any $i$. 
However, 
the kernel matrix $\bs{K}_n$ often has its eigenvalues decreasing quickly and is effectively rank-deficient
\citep[see, e.g.,][]{penalty2006trevor}, 
and 
$\bs{\gamma}^\top_i \bs{f}_0$ is often very close to $0$ for sufficiently large $i$ when $f_0$ is relatively smooth with respect to  kernel $K$.
Below we give some intuition for this. 

Assume that covariates $\bs{X}_i$'s are $\iid$ with respect to probability measure $\mu$.  
By  Mercer's theorem, the kernel function has the following eigen-decomposition: 
$K(\bs{x}, \bs{x}') = \sum_{i=1}^{\infty} \lambda_i \psi_i(\bs{x}) \psi_i(\bs{x}')$, 
where $\lambda_1 \ge \lambda_2 \ge \ldots$ are the eigenvalues, 
and the eigenfunctions $\psi_i$'s are orthonormal bases for the class of square-integrable functions.  
The cross-product 
$\bs{\gamma}_i^\top \bs{f}_0$ can be intuitively understood as an approximation (or sample analog) of the inner product $\int f_0 \psi_j \text{d}\mu$ between the function $f_0$ and the $i$the eigenfunction $\psi_i$  after proper scaling \citep[see, e.g.,][]{braun2008relevant}. 
Note that $\psi_i$ becomes more and more non-smooth with respect to kernel $K$ as $i$ increases. 
When the underlying function $f_0$ is relatively smooth with respect to $K$, the inner product between $f_0$ and $\psi_i$, and thus the sample version  $\bs{\gamma}_i^\top \bs{f}_0$, diminishes quickly as $i$ increases.
Consequently,
the projection of $\bs{Y}$ onto the space spanned by the $\bs{\gamma}_i$'s  for large $i$ 
is mostly dominated by the Gaussian noise, based on which we can then conduct permutation tests. 

Unlike Theorems \ref{linear_permutation_pval_valid} and \ref{poly_permutation_pval_valid}, 
with a general kernel $K$ and a general function $f_0$, 
the partial permutation test is  not finite-sample valid. 
This motivates us to investigate how
to adjust 
the partial permutation test. It turns out that the correction needed for the partial permutation $p$-value depends crucially on 
\begin{equation}\label{eq:left-over}
    \omega(b_n,\sigma_0^{-1}f_0) = \sum_{i=n-b_n+1}^n ( \sigma_0^{-1}\bs{\gamma}_{i}^\top \bs{f}_0)^2
    = \sigma_0^{-2} \sum_{i=n-b_n+1}^n ( \bs{\gamma}_{i}^\top \bs{f}_0)^2, 
\end{equation} 
which can be intuitively understood as the {\it left-over signal-proportion} (LOSP) among the components used for the partial permutation test of size $b_n$. 
Let 
$Q_{b_n}$ denote the quantile function of the $\chi^2$-distribution with degrees of freedom $b_n$, 
and for any $0<\alpha_0<1$, define 
\begin{align}\label{eq:v_fixed_H0}
v(b_n,\sigma_0^{-1}f_0,\alpha_0) = \frac{1}{2}\exp\left\{
2\sqrt{2\omega(b_n,\sigma_0^{-1}f_0) } \sqrt{Q_{b_n}(1-\alpha_0)+ \omega(b_n,\sigma_0^{-1}f_0)} 
\right\}-\frac{1}{2}
\end{align}
as a function of 
permutation size $b_n$,
standardized function $\sigma_0^{-1} f_0$ and $\alpha_0\in (0,1)$.
The following theorem shows that, by adding the correction term \eqref{eq:v_fixed_H0} to the $p$-value, the partial permutation test becomes valid under $H_0$.

\begin{theorem}\label{thm:fixed_property_theorem}
    Let $\{(\bs{X}_i,Y_i,Z_i)\}_{1\leq i \leq n}$ denote samples from model ${H}_0$ in \eqref{eq:H_0_fixed}. Given $1\leq b_n \leq n$ and $\alpha_0 \in (0,1)$, we define the corrected partial permutation $p$-value as
    \begin{align*}
    p_{\correct}( \bXmatrix, \bs{Y}, \bs{Z} ) = p(\bXmatrix, \bs{Y}, \bs{Z}  ) + v(b_n,\sigma_0^{-1}f_0,\alpha_0) +\alpha_0,
    \end{align*}
    where $p(\bXmatrix, \bs{Y}, \bs{Z})$ is the $p$-value from either the discrete or continuous partial permutation test (as in Algorithm~\ref{alg:partial_permu}) with kernel $K$, permutation size $b_n$, and any test statistic $T$; and $v(b_n,\sigma_0^{-1}f_0,\alpha_0)$ is defined as in \eqref{eq:v_fixed_H0}. Then the corrected partial permutation $p$-value is valid under model ${H}_0$, i.e., 
    $\forall \alpha \in (0,1)$, 
    $
    \Pr_{{H}_0}\{{p}_{\correct}( \bXmatrix, \bs{Y}, \bs{Z} )\leq\alpha \mid  \bXmatrix, \bs{Z} \}\leq\alpha. 
    $
\end{theorem}

In Theorem \ref{thm:fixed_property_theorem}, 
the correction term $v(b_n,\sigma_0^{-1}f_0,\alpha_0)$ is increasing in $b_n$, i.e., the larger the permutation size,
the larger the correction for the  $p$-value will be.
Note that the corrected permutation $p$-value ${p}_{\correct}( \bXmatrix, \bs{Y}, \bs{Z} )$ depends on the unknown true function $f_0$ and cannot be calculated directly. 
Besides, the asymptotic validity of the uncorrected partial permutation $p$-value, $p( \bXmatrix, \bs{Y}, \bs{Z} )$, requires that $b_n \cdot \omega(b_n,\sigma_0^{-1}f_0)$ converges to zero in probability as $n\rightarrow\infty$, which may or may not hold depending on the complexity of  function $f_0$ as well as the choice of the permutation size. Nevertheless, 
Theorem \ref{thm:fixed_property_theorem} helps us understand the bias of this $p$-value
for finite samples and provides insights on how to correct for it.  
In  Sections \ref{sec:fix_special_finite}--\ref{sec:fix_special_balance}
we will consider special cases under which the LOSP $\omega(b_n,\sigma_0^{-1}f_0)$ defined in \eqref{eq:left-over} can be exactly or asymptotically zero and the partial permutation test can itself be finite- or large-sample valid without requiring any correction. 
Moreover, 
in Section \ref{sec:GP}
we consider Gaussian process regression models, under which the LOSP can be bounded stochastically and the permutation test becomes asymptotically valid.

\subsection{Special Case: Kernels with Finite-Dimensional Feature Space}\label{sec:fix_special_finite}

We first consider the case in which  kernel $K$ has only a finite number of nonzero eigenvalues, or equivalently, the corresponding feature space is finite-dimensional, i.e., $K(\bs{x}, \bs{x}') = 
\phi(\bs{x})^\top \phi(\bs{x}')$  
with $\phi(\bs{x})\in \mathbb{R}^q$
for some $q < \infty$.
Following the same arguments as with
the linear and polynomial kernels discussed in Sections \ref{sec:linear} and \ref{sec:poly}, which are special cases of the current setting, we
decompose  kernel matrix $\bs{K}_n$ as $\bs{\Gamma} \bs{C} \bs{\Gamma}^\top$, with $\bs{\Gamma}$ being the orthogonal matrix of eigenvectors and $\bs{C}$
the diagonal matrix of eigenvalues. 
If  function 
$f_0(\bs{x})$ is linear in  $\phi(\bs{x})$,  then 
$\bs{\gamma}_i^\top \bs{f}_0 = 0$ for $i > q$ and thus the partial permutation test is finite-sample valid when the permutation size $b_n$ is no larger than $n-q$. 
We summarize the results in the following corollary. 
Although it is a straightforward extension of Theorem \ref{poly_permutation_pval_valid},
this result provides us some intuitions and a bridge to general kernels.

\begin{corollary}\label{cor:kernel_finite_dim_feature_space}
    Let $\{(\bs{X}_i,Y_i,Z_i)\}_{1\leq i \leq n}$ denote samples from model $H_0$ in (\ref{eq:H_0_fixed}). 
    Suppose that the kernel function $K$ has the decomposition 
	$K(\bs{x}, \bs{x}') = \phi(\bs{x})^\top \phi(\bs{x}')$ with $ \phi(\bs{x})\in \mathbb{R}^q$ for some $q< \infty$, 
	and the underlying function $f_0(\bs{x})$ is linear in $\phi(\bs{x})$. 
	Then, the $p$-value obtained by either the discrete or continuous partial permutation test with kernel $K$, permutation size $b_n \le n-q$, and any test statistic $T$ is valid, i.e., $\forall \alpha \in (0,1)$, 
	$
	\Pr_{H_0}\{p(\bXmatrix,\bs{Y}, \bs{Z} ) \leq \alpha 
	\mid \bXmatrix, \bs{Z}
	\} \le \alpha. 
	$
\end{corollary}

\subsection{Special Case: Kernels with Diverging-Dimensional Feature Space}\label{sec:null_diverg}

We extend Section \ref{sec:fix_special_finite} to consider kernels whose feature space dimensions can increase with the sample size, under which the partial permutation test can be (asymptotically) valid for a wider class of underlying functional relationships.  
Specifically, let $\{\basis_j: j = 1,2 , \ldots\}$ be a given series of basis functions of the covariate. 
For each integer $q>0$, we define kernel $K_q(\bs{x}, \bs{x}') \equiv \phi_q(\bs{x})^\top \phi_q(\bs{x}') \equiv \sum_{j=1}^q e_j(\bs{x}) e_j(\bs{x}')$, where $\phi_q(\bs{x}) = (e_1(\bs{x}), e_2(\bs{x}), \ldots, e_q(\bs{x}))^\top$ denotes the corresponding feature mapping based on the first $q$ basis functions. 
Motivated by Corollary \ref{cor:kernel_finite_dim_feature_space}, intuitively, the partial permutation test using kernel $K_q$ is approximately valid 
if the underlying function $f_0(\cdot)$ can be approximated well by a linear combination of the first $q$ basis functions. Moreover, we can increase the feature space dimension $q$ at a proper rate as the sample size increases, and render the partial permutation test asymptotically valid provided that $f_0(\cdot)$ lies in the space generated  %
by the infinite series of basis functions %
$\{e_1(\bs{x}), e_2(\bs{x}),\ldots\}$, 
as characterized more precisely in the following corollary. 
For any function $f$ of the covariate, we introduce 
\begin{align*}
    \remainder(f; q) = \min_{\bs{b}\in \mathbb{R}^q} \int \big(f - \bs{b}^\top \phi_q\big)^2 \text{d}\mu = 
\min_{b_1, \ldots, b_q \in \mathbb{R}} \int \big( f - \sum_{j=1}^q b_j e_j \big)^2 \text{d}\mu
\end{align*}
to 
denote the squared distance between $f$ and its best linear approximation using the first $q$ basis functions. 
The limiting behavior of $\remainder(f; q)$ as $q$ goes to infinity then characterizes how well $f$ can be linearly approximated by this infinite series of basis functions. 
Note that here we implicitly assume that both $f$ and $e_j$'s are square-integrable. 

\begin{corollary}\label{cor:diverg_kernel}
Let $\{(\bs{X}_i,Y_i,Z_i)\}_{1\leq i \leq n}$ denote samples from model $H_0$ in (\ref{eq:H_0_fixed}), 
and assume that $\bs{X}_i$'s are identically distributed from some probability measure $\mu$. 
Suppose that kernel function $K_q$ has the form 
$K_q(\bs{x}, \bs{x}') \equiv \phi_q(\bs{x})^\top \phi_q(\bs{x}') \equiv \sum_{j=1}^q \basis_j(\bs{x}) \basis_j(\bs{x}')$ for $q\ge 1$ and some series of basis functions $\{\basis_j\}_{j=1}^{\infty}$. 
If 
there exists a sequence $\{q_n\}_{n=1}^\infty$ such that $q_n < n$ for all $n$ and $n (n-q_n) \remainder(f_0; q_n) \rightarrow 0$ as $n\rightarrow \infty$, 
then the resulting $p$-value obtained by either the discrete or continuous partial permutation test (as in Algorithm~\ref{alg:partial_permu}) with kernel $K_{q_n}$, permutation size $b_n \le n-q_n$, and any test statistic $T$ is asymptotically valid, i.e., $\forall \alpha \in (0,1)$, 
$
\limsup_{n\rightarrow \infty}
\Pr_{H_0}\{p(\bXmatrix,\bs{Y}, \bs{Z} ) \leq \alpha 
\} \le \alpha. 
$
\end{corollary}

From Corollary \ref{cor:diverg_kernel}, the existence and construction of a valid large-sample partial permutation test depends crucially on how well the underlying functional relationship $f_0$ %
can be linearly approximated by the basis functions $\{\basis_j\}_{j=1}^\infty$.
Below we consider constructing $\{\basis_j\}_{j=1}^\infty$ based on a general kernel $K$ with infinite-dimensional feature space. 
Recall the discussion in Section \ref{sec:general_fix_null}. 
Suppose we have the eigen-decomposition of the kernel $K(\bs{x}, \bs{x}') = \sum_{j=1}^{\infty} \lambda_j \psi_j(\bs{x}) \psi_j(\bs{x}')$, 
where $\lambda_1 \ge \lambda_2 \ge \ldots$ are the eigenvalues, and $\psi_1, \psi_2, \ldots$ are the eigenfunctions and form an orthonormal basis for the space of square-integrable functions. 
If we choose $e_j = \lambda_j^{1/2} \psi_j$ for all $j\ge 1$, then  kernel $K_q$ based on the first $q$ basis functions converges to $K$ as $q$ goes to infinity. 
Moreover, 
if the underlying function $f_0$ belongs to the RKHS $\mathcal{H}_K$ corresponding to kernel $K$, 
then we have $f_0 = \sum_{j=1}^\infty \alpha_j \lambda_j^{1/2} \psi_j = \sum_{j=1}^\infty \alpha_j e_j$ for some coefficients $\alpha_j$'s with $\sum_{j=1}^{\infty} \alpha_j^2 < \infty$, 
under which 
$\remainder(f_0; q) = \sum_{j>q} \alpha_j^2 \lambda_j \le \lambda_q  \sum_{j>q} \alpha_j^2 = o(\lambda_q)$ as $q\rightarrow \infty$. 
As discussed later in Section \ref{sec:kernel}, the eigenvalues $\lambda_j$'s often decays at a polynomial rate with power greater than $1$, in the sense that $\lambda_q = O(q^{-\kappa})$ for some $\kappa>1$. 
This then implies that $\remainder(f_0; q) = o(\lambda_q) = o(q^{-\kappa})$ . 

Intuitively, we prefer a larger permutation size and thus a smaller $q_n$, which can generally lead to a 
more powerful test.
However, conditions in Corollary \ref{cor:diverg_kernel} require us to be more considerate in selecting $q_n$.
Let us focus on the case where the approximation error for $f_0$ decays polynomially, i.e., $\remainder(f_0; q) = o(\lambda_q) = o(q^{-\kappa})$ for some $\kappa > 1$. 
When $\kappa \in (1, 2)$, we can choose $q_n = n - c_n n^{\kappa-1}$ with $c_n$ being of constant order; 
in this case, the permutation size can be $n-q_n \asymp n^{\kappa-1}$ 
and $n(n-q_n) \remainder(f_0; q_n)$ must be of order $o(1)$. 
When $\kappa \ge 2$, we can choose $q_n = c_n n^{2/\kappa}$ with $c_n$ being of constant order; 
in this case, the permutation size can be $n-q_n = n - c_n n^{2/\kappa} \asymp n$ and $n(n-q_n) \remainder(f_0; q_n)$ must be of order $o(1)$.

\subsection{Special Case: Exactly Balanced Covariates across All Groups}\label{sec:fix_special_balance}
Assume that 
the design matrix $\bXmatrix$ enjoys a balancing property that the empirical distributions of covariates are exactly the same across all $H$ groups, i.e.,
\begin{align}\label{eq:balanced}
    \{\bs{X}_i: Z_i = 1, 1\le i \le n\} = \{\bs{X}_i: Z_i = 2, 1\le i \le n\} = \ldots = \{\bs{X}_i: Z_i = H, 1\le i \le n\}. 
\end{align}
Let $r$ denote the number of distinct covariate values. Obviously, $r \le n/H$, and the equality holds if and only if the covariates within each group are all distinct. 
We can verify that the rank of kernel matrix $\bs{K}_n$ for all units is the same as that for the $r$ distinct covariate values. 
Thus, $\text{rank}(\bs{K}_n) \le r$; moreover, the equality generally holds when 
 kernel function $K$ 
corresponds to 
an infinite-dimensional feature space, e.g., the Gaussian kernel. 
When $\bs{K}_n$ is indeed of rank $r$,
as demonstrated in the Supplementary Material, for any underlying function $f_0$, 
the LOSP  $\omega(b_n, \sigma_0^{-1}f_0)$ 
in \eqref{eq:left-over} is exactly zero as long as the permutation size $b_n$ is no larger than $n-r$, under which the correction term $v(b_n,\sigma_0^{-1}f_0,\alpha_0)$ in \eqref{eq:v_fixed_H0} also reduces to zero. 
Consequently, 
the partial permutation $p$-value $p(\bXmatrix,\bs{Y}, \bs{Z} )$ 
must be 
valid under model $H_0$, 
as summarized in the following corollary. 

\begin{corollary}\label{cor:fixed_property_theorem_balanced_design}
Let $\{(\bs{X}_i,Y_i,Z_i)\}_{1\leq i \leq n}$ denote samples from model ${H}_0$ in (\ref{eq:H_0_fixed}). 
If the design matrix is exactly balanced in the sense that \eqref{eq:balanced} holds 
and the kernel matrix for all the $r \le n/H$ distinct covariate values in each group is of full rank (or equivalently $\text{rank}(\bs{K}_n) = r$), 
then
the partial permutation $p$-value from either the discrete or continuous partial permutation test with kernel $K$, permutation size $b_n \le n-r$, and any test statistic $T$ is valid under model $H_0$, i.e., $\forall$ $\alpha \in (0,1)$,  
$
\Pr_{H_0}\{p(\bXmatrix,\bs{Y}, \bs{Z} ) \leq \alpha 
\mid \bXmatrix, \bs{Z}
\} \le \alpha. 
$
\end{corollary}

The validity of the  test in Corollary \ref{cor:fixed_property_theorem_balanced_design} is closely related to that of the usual permutation test, which permutes the group indicators of samples with the same covariate value. 
Corollary \ref{cor:fixed_property_theorem_balanced_design} is more general in the sense that it allows for more general rotations (instead of purely switching) of the responses, utilizing the Gaussianity of the noises.

\subsection{Special Case: Gaussian Process Regression Model}\label{sec:GP}

In this subsection, 
instead of treating $f_0$ in \eqref{eq:H_0_fixed} under  $H_0$ as a fixed unknown function  as in previous sections, %
we here assume that the function follows a Gaussian process and 
show that $p( \bXmatrix, \bs{Y}, \bs{Z} )$ is asymptotically valid under such a Gaussian process regression (GPR) model. We note that the GPR model has been widely used in functional analysis.

\subsubsection{The model formulation}

Given a symmetric, positive definite, and continuous kernel $K$, 
the GPR model assumes that %
\begin{align}\label{eq:h0_gp}
    \tilde{H}_0: & Y_i =  f (\bs{X}_i) + \varepsilon_i, \quad \varepsilon_i \mid \bXmatrix, \bs{Z} \ \overset{\iid}{\sim} \ \mathcal{N}(0,\sigma_0^2) , \quad f \sim  \text{GP}\left( 0, \ \frac{\delta_0^2}{n^{1-\gamma}} K \right),
\end{align}
where $f$ is independent of $\bXmatrix, \bs{Z}$ and the $\varepsilon_i$'s, 
and %
$\delta_0^2/{n^{1-\gamma}}$, which depends on the sample size $n$, 
represents our belief on the smoothness of the underlying function. 
The GPR model is closely related to kernel regression, which minimizes a penalized mean squared loss over a RKHS $\mathcal{H}_K$ corresponding to kernel $K$. 
Specifically, the kernel regression estimator $\hat{f}_{n,\tau_n}$
is given by
\begin{align}\label{eq:kernel_penal}
	\hat{f}_{n,\tau_n} = \arg\min_{f \in \mathcal{H}_K}\frac{1}{n}\sum_{i=1}^{n}|Y_i-f(X_i)|^2 + \tau_n \|f\|_{\mathcal{H}_K}^2,
\end{align}
where $\tau_n$ is a regularization parameter penalizing the $\mathcal{H}_K$ norm of $f$. 
This estimator is identical to the posterior mean of $f$ under the GPR model %
in \eqref{eq:h0_gp} when $\tau_n = \sigma_0^2/(n^\gamma \delta_{0}^2)$. 
\citet{christmann2007consistency} studied  sufficient conditions on $\tau_n$ to guarantee the consistency of  kernel regression estimator $\hat{f}_{n,\tau_n}$, 
which then provides us some guidance on the choice of 
the smoothness parameter $\delta_{0}^2/n^{1-\gamma}$. 
The following proposition is a direct corollary of \citet[][Theroem 12]{christmann2007consistency}. 

\begin{proposition}\label{thm:general_consistency}
Let $\{(\bs{X}_i, Y_i, Z_i)\}_{i=1}^n$ be random samples from model $H_0$ in (\ref{eq:H_0_fixed}). If the $\bs{X}_i$'s are $\iid$ with a compact support $\mathcal{X}$, 
$f_0$ is a measurable function, $\mathbb{E}_{H_0}(Y^2)<\infty$, $K$ is a universal kernel, and  $0<\gamma<1/4$, 
then the posterior mean $\tilde{f}$ %
induced by model $\tilde{H}_0$ in \eqref{eq:h0_gp} is consistent for the underlying true $f_0$ in \eqref{eq:H_0_fixed},
i.e.,
$
\mathbb{E}_{H_0}|\tilde{f}(\bs{X})-f_0(\bs{X})|^{2}\stackrel{\Pr}{\longrightarrow} 0$ as $n \rightarrow \infty$.
\end{proposition}

In Proposition \ref{thm:general_consistency}, 
the universal kernel was introduced by \citet{micchelli2006universal}, which has the property that the corresponding RKHS is dense in $\mathcal{C}(\mathcal{X})$, the space consisting of all continuous functions on $\mathcal{X}$ with the infinity norm. 
We can intuitively summarize conditions on 
the Gaussian process prior of $f$ and the relation between its variance parameter and sample size as follows. First, $f$ should be almost surely continuous. Second, if two realizations of $f$ fit observations equally well, it is preferable to give the smoother one more weight. Third, as the sample size increases, the posterior mean and mode of $f$ should increasingly concentrate around the true functional relationship.
All the requirements above can be satisfied by the GPR model with an appropriate choice of the kernel function and by letting the variance parameter decrease at a proper rate as the sample size increases.

\subsubsection{Large-sample valid partial permutation test}

The following theorem shows that 
the partial permutation $p$-value is asymptotically valid under the GPR model $\tilde{H}_0$ in \eqref{eq:h0_gp} under certain conditions.

\begin{theorem}\label{thm:general_ppt_valid}
	Let $\{(\bs{X}_i,Y_i,Z_i)\}_{1\leq i \leq n}$ denote samples from model $\tilde{H}_0$ in (\ref{eq:h0_gp}).
	If 
    the  $\bs{X}_i$'s are $\iid$ from a compact support $\mathcal{X}$ with some probability measure $\mu$, 
    the eigenvalues $\{\lambda_k\}$ of kernel $K$ on $(\mathcal{X},\mu)$ satisfy $\lambda_k = O(k^{-\rho})$ for some $\rho >1$, and 
    $\gamma< 1-\rho^{-1}$, then for sequence $\{b_n\}$ satisfying 
    $b_n=   O(n^\kappa)$ with $0<\kappa <1-\rho^{-1}-\gamma$, 
    the partial permutation $p$-value from either the discrete or continuous partial permutation test with kernel $K$, permutation size $b_n$, and any test statistic $T$ is asymptotically valid under $\tilde{H}_0$, 
    i.e., $\forall \alpha\in (0,1)$, 
    $
    \limsup_{n \rightarrow \infty} \Pr_{\tilde{H}_0} \{p(\bs{X}, \bs{Y}, \bs{Z} ) \leq \alpha \} \leq \alpha. 
    $
\end{theorem}

In Section \ref{sec:kernel} we will show that there exist universal kernels with polynomially decaying eigenvalues. Coupled with a choice of the smoothness parameter $\gamma$ that satisfies the conditions in Proposition \ref{thm:general_consistency} and Theorem \ref{thm:general_ppt_valid}, the partial permutation test is asymptotically valid under a GPR model that imposes a reasonable amount of regularization on the underlying function.
Furthermore, we emphasize that the asymptotic validity of the partial permutation test essentially requires that the ratio between the variance parameter for the Gaussian process prior and the variance of observation noises is of order $n^{-(1-\gamma)}$ for some $\gamma < 1 - \rho^{-1}$. 
Thus, even if the Gaussian process prior on $f$ does not follow the regularized form as in \eqref{eq:h0_gp}, we 
can still 
perform asymptotically valid partial permutation test by adding noises to the responses.

Theorem \ref{thm:general_ppt_valid} proves the large-sample validity of the partial permutation test.  
Below we investigate its finite-sample performance in analogous to Section \ref{sec:general_fix_null}. 
Let $\xi_n=(\delta_0^2/n^{1-\gamma})/\sigma_0^2$ denote the variance ratio for the function and noise. 
For any given $b_n$, 
we define
\begin{align*}
    \tilde{\omega}\left(b_n, 
    \xi_n
    \right)
    = 
    \frac{\delta_0^2/n^{1-\gamma}}{\sigma_0^2} \cdot c_{n-b_n+1} 
    = \xi_n \cdot c_{n-b_n+1} 
\end{align*}
to denote the LOSP for the components of $\bs{Y}$ used for the partial permutation test of size $b_n$, 
recalling that $c_{n-b_n+1}$ is the $(n-b_n+1)$th largest eigenvalue of $\bs{K}_n$. 
Note that here the LOSP $\omega(b_n, \sigma_0^{-1}f)$ defined in Section \ref{sec:general_fix_null} can be bounded by $b_n \cdot \tilde{\omega}(b_n, \xi_n)$ in expectation under the 
GPR in
\eqref{eq:h0_gp}. 
Recall that $Q_{b_n}$ is the quantile function of the $\chi^2$-distribution with degrees of freedom $b_n$. 
For $1\le b_n \le n$ and $\alpha_0 \in (0,1)$, 
we define 
\begin{align}\label{eq:v_correct_gaussian}
\tilde{v}(b_n, \xi_n, \alpha_0)=\frac{1}{2} \exp\left[
\frac{1}{2}
\tilde{\omega}(b_n, \xi_n)
\cdot Q_{b_n}(1-\alpha_0)
\right]-\frac{1}{2}. 
\end{align}
The following theorem shows that, by adding a correction term, the partial permutation $p$-value becomes finite-sample valid under $\tilde{H}_0$.  

\begin{theorem}\label{thm:gp_finite_property_theorem}
	Let $\{(\bs{X}_i,Y_i,Z_i)\}_{1\leq i \leq n}$ denote samples from model $\tilde{H}_0$ in \eqref{eq:h0_gp}.
	Given $1\leq b_n \leq n$ and $0<\alpha_0<1$, we define the corrected partial permutation $p$-value as follows:
	\begin{align*}
	\tilde{p}_{\correct}(\bXmatrix, \bs{Y}, \bs{Z}) = p(\bXmatrix, \bs{Y}, \bs{Z} ) + \tilde{v}(b_n, \xi_n, \alpha_0) + \alpha_0,
	\end{align*}
	where $p(\bXmatrix, \bs{Y}, \bs{Z})$ is the $p$-value from either the discrete or continuous partial permutation test with kernel $K$, permutation size $b_n$, and any test statistic $T$, and $\tilde{v}(b_n, \xi_n, \alpha_0)$ is as defined in \eqref{eq:v_correct_gaussian}. Then the corrected partial permutation $p$-value is valid under model $\tilde{H}_0$, i.e., 
	$\forall \alpha \in (0,1)$, 
	$
	\Pr_{\tilde{H}_0}
	\{
	\tilde{p}_c(\bXmatrix, \bs{Y}, \bs{Z} )\leq\alpha
	\mid \bXmatrix,\bs{Z}
	\}
	\leq\alpha. 
	$
\end{theorem}

Note that in Theorem \ref{thm:gp_finite_property_theorem}, 
the correction term $\tilde{v}(b_n, \xi_n, \alpha_0)$ is 
monotone increasing in 
the permutation size $b_n$. 
This is intuitive since the larger the permutation size, the 
larger the correction for the partial permutation $p$-value
is needed. 
As discussed shortly in Section \ref{sec:choice_size},
Theorem \ref{thm:gp_finite_property_theorem}  provides us some guidance on the choice of permutation size in finite samples.

\section{Partial Permutation Test under Alternative Hypotheses}\label{sec:ppt_alter}

\subsection{Kernels with Finite-Dimensional Feature Space}\label{sec:power_finite}

While previous discussions focused  on the validity of partial permutation tests under the null hypothesis that the samples share the same functional relationship across all groups, we here investigate 
how such tests behave under alternative hypotheses. 
As the permutation test allows for a flexible choice of test statistics, which can be tailored based on the alternative hypotheses of interest,
we study a special class of test statistics that are linked to a certain form of likelihood ratio statistics under a general kernel with finite-dimensional feature space. 
That is, 
the kernel function can be decomposed as 
$K(\bs{x}, \bs{x}') =
\phi(\bs{x})^\top \phi(\bs{x}')$ 
with $\phi(\bs{x})\in \mathbb{R}^q$ for some $q < \infty$. 

As demonstrated in Corollary \ref{cor:kernel_finite_dim_feature_space}, 
the partial permutation test is exactly valid under model $H_0$ in \eqref{eq:H_0_fixed} when   $f_0(\bs{x})$ is linear in the transformed covariates $\phi(\bs{x})$. 
It is then straightforward to hypothesize that, under the alternative model specified in  \eqref{eq:H_1_equal},  
the functional relationship between the response and covariates is also linear in the transformed covariates, but the coefficients can vary across groups, i.e., 
\begin{align}\label{eq:linear_feature_alter}
    Y_i = 
    \sum_{h=1}^H \I(Z_i=h)\bs{\beta}_h^\top \phi(\bs{X}_i)
    + \varepsilon_i, 
    \quad 
    \varepsilon_i \mid \bs{\mathrm{X}}, \bs{Z} \ \overset{\iid}{\sim} \ \mathcal{N}(0, \sigma^2_0), 
    \quad 
    (1\le i \le n)
\end{align}
where $\bs{\beta}_h$ denotes the regression coefficient vector for samples in the $h$th group. 
This motivates us to use the F statistic for testing 
$\bs{\beta}_1 = \ldots = \bs{\beta}_H$
as our test statistic, 
which is equivalent to the likelihood ratio statistic up to a monotone transformation. 
Let $\bs{P}_0$ and $\bs{P}_1$ denote the projection matrices onto the column spaces of the transformed covariates for regression model \eqref{eq:linear_feature_alter} under the null model that $\bs{\beta}_1 = \ldots = \bs{\beta}_H$ and the full model without any constraint on the parameters, respectively, 
and let $p_0$ and $p_1$ denote the matrices' ranks. 
Then, the $F$ statistic for testing $\bs{\beta}_1 = \ldots = \bs{\beta}_H$ has the form 
\begin{align}\label{eq:Fstat}
    F(\bXmatrix, \bs{Y}, \bs{Z}) = \frac{\bs{Y}^\top (\bs{P}_1 - \bs{P}_0) \bs{Y}^\top/(p_1 - p_0)}{\bs{Y}^\top (\bs{I}_n - \bs{P}_1) \bs{Y}^\top/(n - p_1)}. 
\end{align}
It turns out that 
the permutation distribution of the F statistic in \eqref{eq:Fstat} under our continuous partial permutation test with kernel $K$ and permutation size $b_n = n-p_0$ is F distributed with degrees of freedom $p_1 - p_0$ and $n-p_1$, which matches the repeated sampling distribution of the F statistic when model \eqref{eq:linear_feature_alter} holds with $\bs{\beta}_1 = \ldots = \bs{\beta}_H$. 
Therefore, 
with the same choice of the test statistic (i.e., F statistic or equivalently the likelihood ratio statistic), 
the partial permutation test is equivalent to the usual F-test or likelihood ratio test for nested regression models. We summarize the results in the following theorem. 
Let $F_{d_1, d_2}$ denote the distribution function of the F distribution with degrees of freedom $d_1$ and $d_2$.

\begin{theorem}\label{thm:general_F_test}
    Consider any samples $\{(\bs{X}_i,Y_i,Z_i)\}_{1\leq i \leq n}$
    and any kernel $K$ of form $K(\bs{x}, \bs{x}') = \phi(\bs{x})^\top \phi(\bs{x}')$ with $\phi(\bs{x}) \in \mathbb{R}^q$ and $q< \infty$. 
    The permutation distribution of the F statistic in \eqref{eq:Fstat} under the continuous partial permutation test with kernel $K$ and permutation size $n-p_0$ is an F distribution with degrees of freedom $p_1 - p_0$ and $n-p_1$, 
    and the corresponding partial permutation $p$-value is 
    $p(\bs{\mathrm{X}}, \bs{Y}, \bs{Z}) = 1 - F_{p_1-p_0, n-p_1}(F(\bs{\mathrm{X}}, \bs{Y}, \bs{Z}))$. 
\end{theorem}

In Theorem \ref{thm:general_F_test}, if the transformed covariates are linearly independent within each group (i.e., the matrix whose rows consist of $\phi(\bs{X}_i)^\top$ for samples in group $h$ is of full column rank, $1\le h \le H$), 
then $p_0 = q$ and $p_1 = H q$.
The equivalence between the partial permutation test and F-test in 
Theorem \ref{thm:general_F_test} has two implications. 
First, it confirms the finite-sample validity of the partial permutation test when the null hypothesis (i.e., model \eqref{eq:linear_feature_alter} with $\bs{\beta}_1 = \ldots = \bs{\beta}_H$) holds. 
Second, it shows that the partial permutation test using the F statistic is most powerful when the the alternative hypothesis is indeed of form \eqref{eq:linear_feature_alter} with possibly unequal $\bs{\beta}_h$'s.
Furthermore, as demonstrated in Corollary \ref{cor:kernel_finite_dim_feature_space}, 
the partial permutation test allows for a more  flexible choice of test statistics, which can be tailored towards any alternative hypothesis of interest, 
since the test uses partial permutation to get the valid null distribution. 
Finally, although Theorem \ref{thm:general_F_test} considers only kernels with a finite-dimensional feature space, 
it sheds light on general kernels as well since we can always view a general kernel as the limit of kernels with finite-dimensional feature spaces. 

\subsection{Kernels with Diverging-Dimensional Feature Space}\label{sec:diverg_power}

We now extend the discussion in Section \ref{sec:power_finite} to kernels with diverging-dimensional feature spaces as the sample size increases, similar to that in Section \ref{sec:null_diverg}.
Let $\{e_j\}_{j=1}^\infty$ be a given series of basis functions of the covariate, 
and let $K_q(\bs{x}, \bs{x}') = \phi_q (\bs{x})^\top \phi_q (\bs{x}') = \sum_{j=1}^q e_j(\bs{x})e_j(\bs{x}')$ be the kernel with feature mapping $\phi_q(\bs{x}) = (e_1(\bs{x}), \ldots, e_q(\bs{x}))^\top$ consisting of the first $q$ basis functions. 
We consider partial permutation test based on kernel $K_{q_n}$ whose feature space dimension $q_n$ can vary with the sample size $n$, 
and studies its power using the F statistic as in \eqref{eq:Fstat} with $\phi$ replaced by $\phi_{q_n}$. 
Analogously, 
we let $\bs{P}_{n0}$ and $\bs{P}_{n1}$ denote the projection matrices on to the column spaces of the transformed covariates under the null and the full models, 
and let $p_{n0}$ and $p_{n1}$ denote the matrices' ranks, respectively. %
Moreover, since we will investigate the power of the test under local alternatives, we allow the  functional relationship between response and covariates under model $H_1$ in \eqref{eq:H_1_equal} to also vary with the sample size, and write them explicitly as $f_{n1}, f_{n2}, \ldots, f_{nH}$. 
Throughout this subsection, we  assume that the covariates $\bs{X}_i$'s are identically distributed from some probability measure $\mu$, 
and use $\remainder(f_{nh}; q) = \min_{\bs{b}\in \mathbb{R}^q} \int (f_{nh} - \bs{b}^\top \phi_q)^2 \text{d}\mu$ to denote the squared error for the best linear approximation of $f_{nh}$ using the first $q$ basis functions.

\begin{theorem}\label{thm:power_diverg_kernel}
Let $\{(\bs{X}_i,Y_i,Z_i)\}_{1\leq i \leq n}$ denote samples from model $H_1$ in \eqref{eq:H_1_equal}, 
and assume that the $\bs{X}_i$'s follow probability measure $\mu$. 
If, as $n\rightarrow \infty$ and for some $\theta \ge 0$,  
\begin{align}\label{eq:cond_diverg_kernel}
    p_{n1} - p_{n0} \rightarrow \infty, \ \  
    \frac{p_{n1} - p_{n0}}{n-p_{n1}} \rightarrow 0, \ \  
    \frac{n \sum_{h=1}^H \remainder(f_{nh}; q_n)}{\sqrt{p_{n1} - p_{n0}}} \rightarrow 0, \ \  
     \frac{\bs{f}^\top (\bs{I}_n - \bs{P}_{n0}) \bs{f}}{\sqrt{p_{n1} - p_{n0}}} = \text{(or $\ge$) } \theta + o_{\Pr}(1),
\end{align}
where $\bs{f} = (f_{n Z_1}(\bs{X}_1), f_{n Z_2}(\bs{X}_2), \ldots, f_{n Z_n}(\bs{X}_n))^\top$, 
then, $\forall \alpha\in (0,1)$, the $p$-value
from
the continuous partial permutation test with kernel $K_{q_n}$, permutation size $n - p_{n0}$, and F test statistic as in \eqref{eq:Fstat} 
must satisfy that 
\begin{align*}
    \Pr\left( p(\bXmatrix, \bs{Y}, \bs{Z}) \le \alpha \right) = \text{(or $\ge$) } \Phi\left( z_{\alpha} + \theta/\sqrt{2} \right) + o(1),
\end{align*}
where $\Phi(\cdot)$ denotes the distribution function of standard Gaussian distribution. 
\end{theorem}

From the discussion after Theorem \ref{thm:general_F_test},
we generally expect that $p_{n0} \asymp q_n$ 
and $p_{n1}-p_{n0} \asymp q_n$, 
under which the first two conditions in \eqref{eq:cond_diverg_kernel} reduce to that $q_n \rightarrow \infty$ and $q_n/n \rightarrow 0$ as $n\rightarrow \infty$. 
Below we assume these are true and discuss two implications from Theorem \ref{thm:power_diverg_kernel}.

First, we consider the case where the null hypothesis holds, i.e., $f_{n1}= f_{n2} = \ldots = f_{nH} = f_0$ for some $f_0$ that depends neither on the group index nor the sample size. 
We can then demonstrate that $\bs{f}^\top (\bs{I}_n - \bs{P}_{n0}) \bs{f} = n \remainder(f_0; q_n) \cdot O_{\Pr}(1)$. 
From Theorem \ref{thm:power_diverg_kernel} with $\theta=0$, the partial permutation test will be asymptotically valid when $n \remainder(f_0; q_n) = o(q_n^{1/2})$. 
Suppose the approximation error for $f_0$ decays polynomially, i.e., $\remainder(f_0; q) = o( q^{-\kappa} )$ for some $\kappa > 0$. 
Then a sufficient condition for the large-sample validity of the partial permutation test will be $n q_n^{-\kappa} = O(q_n^{1/2})$, under which we can choose $q_n \asymp n^{2/(2\kappa+1)}$. 
Compared to Corollary \ref{cor:diverg_kernel}, Theorem \ref{thm:power_diverg_kernel} imposes weaker conditions on $\{q_n\}$ for ensuring the validity of the test. 
This is not surprising since Corollary \ref{cor:diverg_kernel} allows for an arbitrary choice of test statistics while Theorem \ref{thm:power_diverg_kernel} concerns only the F statistic. 

Second, we consider the case where the alternative hypothesis holds, 
and assume that the underlying functions have the form $f_{nh} = f_0 + \delta_n \zeta_h$ ($1\le h\le H$), for some constant sequence $\delta_n = O(1)$ and some functions $f_0, \zeta_1, \ldots, \zeta_H$ that do not vary with the sample size. 
Intuitively, $\{\delta_n\}$ and $\tau_{hh'} \equiv \int (\zeta_h-\zeta_{h'})^2 \text{d}\mu$ measure the functional heterogeneity across the $H$ groups. 
For simplicity, we further assume that the covariates in the $H$ groups are exactly balanced and the covariates within each 
group are i.i.d., 
under which we can bound 
$\bs{f}^\top (\bs{I}_n - \bs{P}_{n0}) \bs{f}$ from below by $(2H)^{-1} \delta_n^2 \sum_{i=1}^n (\zeta_h(\bs{X}_i) - \zeta_{h'}(\bs{X}_i))^2 = (2H)^{-1} n \delta_n^2 \{ \tau_{hh'} + o_{\Pr}(1) \}$ for all $1\le h,h'\le H$. 
From Theorem \ref{thm:power_diverg_kernel},
if $n \remainder(f_0; q_n) = o(q_n^{1/2})$, $n \remainder(\zeta_h; q_n) = o(q_n^{1/2})$ for all $h$, and $n \delta_n^2 \ge \theta \sqrt{8H^3 q_n}$ for sufficiently large $n$ and some $\theta\ge 0$, 
then asymptotically the power of the level-$\alpha$ partial permutation test is at least $\Phi(z_{\alpha} + \theta \max_{h,h'} \tau_{hh'})$; see the Supplementary Material for details. 
Suppose that the approximation errors for functions $f_0,\zeta_1, \ldots, \zeta_H$ all decay polynomially, i.e., $\remainder(f_0; q) = o(q^{-\kappa})$ and $\remainder(\zeta_h; q) = o(q^{-\kappa})$ as $q\rightarrow \infty$, 
and that $\tau_{h,h'} > 0$ for at least one pair of $h\ne h'$. 
From the discussion before, we can then choose $q_n \asymp n^{2/(2\kappa+1)}$ to ensure type-I error control. 
Consequently, 
if $\delta_n \gg (\sqrt{q_n}/n)^{1/2} = n^{-\kappa/(2\kappa+1)}$, 
then the power of the level-$\alpha$ partial permutation test must converge to $1$ as the sample size $n$ goes to $\infty$. 
Recall the discussion in Section \ref{sec:null_diverg} and note that the $m$th Sobolev space on $[0,1]$ corresponds to a RKHS with eigenvalue $\lambda_j$ decaying polynomially at rate $j^{-2m}$ \citep{xin2020}. 
The derived rate with $\kappa=2m$ actually matches the minimax distinguishable rate $n^{-2m/(4m+1)}$ in \citet{xin2020} for testing whether two functions in the $m$th order Sobolev space are parallel; see also \citet{shang2013local}.

\section{Implementation of Partial Permutation Test}\label{sec:implement}

\subsection{Choice of the Kernel Function}\label{sec:kernel}

We first show there exist kernel functions 
with polynomially decaying eigenvalues as discussed in Sections \ref{sec:null_diverg} and \ref{sec:GP}. 
Indeed, 
as demonstrated by \citet{kuhn1987eigenvalues}, such a property holds for a general kernel as long as it is sufficiently smooth.
For any set $\mathcal{D} \subset \mathbb{R}^d$ and $0\le s\le 1$, 
define $\mathcal{C}^{s,0}(\mathcal{D}, \mathcal{D})$ as the set consisting  of all continuous functions $G:\mathcal{D}\times\mathcal{D}\rightarrow\mathbb{R}$
such that
\begin{align}\label{eq:C_s0}
 \|G\|_{C^{s,0}(\mathcal{D}, \mathcal{D})} \equiv \max\left\{ \sup_{\bs{x_1},\bs{x_2}\in\mathcal{D}}\left|G(\bs{x}_1,\bs{x}_2)\right|,
 \sup_{\bs{x}_1,\bs{x}_2,\bs{x}_3\in\mathcal{D},\bs{x}_1\neq \bs{x}_2} \frac{\left|G(\bs{x}_1,\bs{x}_3)-G(\bs{x}_2,\bs{x}_3)\right|}{\|\bs{x}_1-\bs{x}_2\|_2^{s}} \right\} <\infty.   
\end{align}
The following proposition is a direct corollary of \citet{kuhn1987eigenvalues}.

\begin{proposition}\label{prop:decay_eigenvalue}
    For any compact set $\mathcal{X} \subset \mathbb{R}^d$ with any probability measure $\mu$ and any positive definite kernel $K: \mathcal{X}\times \mathcal{X} \rightarrow \mathbb{R}$,   
    if there exists $b$ such that (i) $\mathcal{X} \subset \overline{\mathcal{X}} \equiv [-b, b]^d$, 
    (ii) the kernel function $K$ can be extended to domain $\overline{\mathcal{X}} \times \overline{\mathcal{X}}$, 
    and (iii) $K \in C^{s, 0}(\overline{\mathcal{X}}, \overline{\mathcal{X}})$, 
    then the corresponding eigenvalues of $K$,  $\{\lambda_k\}_{k\geq 1} $, satisfy that $\lambda_k = O(k^{-1-s/d}).$
\end{proposition}

From Proposition \ref{prop:decay_eigenvalue} and by the definition in \eqref{eq:C_s0}, 
if a symmetric and positive definite kernel function is continuously differentiable on $\mathbb{R} \times \mathbb{R}$ and the covariate support $\mathcal{X}$ is  compact, 
then the eigenvalue $\lambda_k$ of the kernel must decay at least in an order of $k^{-1-1/d}$, 
under which the condition in Theorem \ref{thm:general_ppt_valid} holds with $\rho = 1 +1/d$. 
Two examples of continuously differentiable kernels are the Gaussian kernel and rational quadratic kernel, which have the following forms: 
\begin{align}\label{eq:gaussian_kernel}
	K_{\text{G}}(\bs{x}, \bs{x}') = \exp\left\{-\sum_{k=1}^d \omega_k (x_{k}-x'_{k})^2 \right\}, 
	\qquad
	K_{\text{R}}(\bs{x}, \bs{x}') = \left\{1+\sum_{k=1}^d \omega_k (x_{k}-x'_{k})^2\right\}^{-\eta},
\end{align}
where $\omega_j$'s and $\eta$ are arbitrary positive numbers. 

Moreover, 
both kernels in \eqref{eq:gaussian_kernel} are also universal \citep{micchelli2006universal}. 
Thus, if we use
any of them
for  model %
 \eqref{eq:h0_gp} 
and let the smoothness parameter be  any constant between between $0$ and $\min\{1/4, 1/(d+1)\}$, 
then the conditions in both Proposition \ref{thm:general_consistency} and Theorem \ref{thm:general_ppt_valid} hold. 
Consequently, 
we are able to conduct asymptotically valid partial permutation test under the GPR model, with a certain regularized but still flexible prior for the underlying functional relationship. 
As discussed shortly in the next subsection, 
the choice of $\gamma$ is not crucial in practice. 
However, the choice of parameters for the kernel function, e.g., the $\omega_j$'s for the Gaussian kernel in \eqref{eq:gaussian_kernel}, does play an important role. 

Parameters in the kernel function play an important role in controlling the smoothness of the underlying functional relationship. 
For instance, for the Gaussian kernel in \eqref{eq:gaussian_kernel}, 
smaller $\omega_j$'s imply wider, flatter kernels and a suppression of wiggly and rough functions \citep{penalty2006trevor}. 
In contrast, larger $\omega_j$'s indicate a more wiggly functional relation and thus generally lead to a smaller permutation size. 
Theoretical investigation for the optimal choice of kernel parameters for testing is challenging, and it may differ from that for the optimal estimation \citep{shang2013local,xin2020}. 
In the literature, various approaches have been proposed to choose  kernel parameters, or more generally kernel functions, adaptively based on the data, 
such as cross validation and maximizing marginal likelihood \citep{rasmussen2006gaussian}. 
We here opt to use the maximum marginal likelihood approach, choosing the kernel parameter to be the one that maximizes the marginal likelihood of $\bs{Y}$ given $\bXmatrix$ and $\bs{Z}$ under the Gaussian process model $\tilde{H}_0$ in \eqref{eq:h0_gp}. 

When the data follow an alternative hypothesis model in which the functional relationships for different $h$ are  different,  the marginal likelihood for the null, which is based on a common model built using the pooled data,  tends to suggest kernels that can tolerate more erratic functions, e.g.,
 large values of $\omega_j$'s for the Gaussian kernel. 
This may be due to the fact that, when the data contain multiple functional relationships between the response and covariates, enforcing a common functional relationship necessarily results in an overly volatile function,
which  then reduces the partial permutation size and damages the power of the test. 
To avoid this potential power loss, 
we  also obtain kernel parameters that maximize the marginal likelihood using samples from each group separately. 
If all of them suggest smoother functional relationships (e.g., smaller $\omega_j$'s for Gaussian kernels) than the pooled data, 
we require the smoothness of the shared functional relationship to be no worse than the most non-smooth one among those obtained within each group (e.g., choosing the maximum $\omega_j$'s estimated from individual groups).

\subsection{Choice of Permutation Size}\label{sec:choice_size}

Both Theorems \ref{thm:fixed_property_theorem} and \ref{thm:gp_finite_property_theorem} provide us with guidance on the choice of permutation size $b_n$: 
 we want  $b_n$ to be large and the correction terms $v$ in \eqref{eq:v_fixed_H0} (or $\tilde{v}$ in \eqref{eq:v_correct_gaussian}) and $\alpha_0$ to be small 
 in order to have a good power for the test.
Note that 
either $v(b_n, \sigma_0^{-1} f_0, \alpha_0)$ in \eqref{eq:v_fixed_H0} or $\tilde{v}(b_n, \xi_n, \alpha_0)$ in \eqref{eq:v_correct_gaussian}
depends on unknown functional relation $f_0$ and noise level $\sigma_0$ or the unknown variance ratio $\xi_n$. 
Therefore, we first estimate $f_0$ and $\sigma_0$ (or $\xi_n$)  
and then use a plug-in approach to 
compute $v$ or $\tilde{v}$ under model $H_0$ or $\tilde{H}_0$. To be more specific, we choose $\alpha_0$ and $b_n$ in the following way:  
\begin{itemize}
\item[1)] for model \eqref{eq:H_0_fixed} of $H_0$, 
$
\alpha_0 = 10^{-4}\alpha, 
$
and 
$
b_n  = \max\{b_n : {v}(b_n, \hat{\sigma}_0^{-1}\hat{f}_0, \alpha_0)+\alpha_0 \leq 10^{-3}\alpha 
\}; 
$
\item[2)] For model \eqref{eq:h0_gp} of $\tilde{H}_0$, 
$
\alpha_0 = 10^{-4}\alpha,
$
and 
$
b_n  = \max\{b_n : {\tilde{v}}(b_n, \hat{\xi}_n, \alpha_0)+\alpha_0 \leq 10^{-3}\alpha 
\}.
$
\end{itemize}
There is a trade-off for the choice of $\alpha_0$ and $b_n$: a larger permutation size $b_n$ can lead to a larger power for detecting violation of the null hypothesis while at the same time requires a larger correction to avoid type-I error inflation. Here we consider an intuitive scheme that requires only a small correction for the partial permutation $p$-value.
For model \eqref{eq:h0_gp},
the estimate $\hat{\xi}_n$ can be obtained by using the maximum likelihood estimates for $\delta_{0}^2/n^{1-\gamma}$ and $\sigma_0^2$. 
For model \eqref{eq:H_0_fixed},
we can estimate $f_0$ based on the penalized regression of form \eqref{eq:kernel_penal} or other regularization method such as early stopping \citep{raskutti2014early,Liu2018early}.  
Here, for simplicity, 
we first obtain the posterior mean of $f$ under $\tilde{H}_0$, denoted by $\hat{f}$, after plugging in the maximum likelihood estimates, and then use $\hat{f}$ as an estimator for $f_0$ and $
n^{-1} \sum_{i=1}^{n}(Y_{i}-\hat{f}(\bs{X}_{i}))^{2}
$
as an estimator for the variance of noise.
Finally, the corrected $p$-value is simply the $p$-value from partial permutation plus the correction term $10^{-3}\alpha$.

\subsection{Choice of the Test Statistic}\label{sec:test_stat}

One advantage of the permutation test is that it allows for a flexible choice of test statistics, for which we can use permutations, instead of a complicated  and often unreliable asymptotic analysis,  to get its reference null distribution. 
Moreover, we can choose the test statistic tailored to the alternative hypothesis of interest so as to gain power.

For a general kernel function, we first consider test statistics based on kernel  regression of form \eqref{eq:kernel_penal}. 
Specifically, we  perform kernel regression both to fit a common function using all samples and to fit  group-specific functions using samples from each group separately, with, say, cross-validation or marginal likelihood maximization for choosing the regularization parameter $\tau_n$ in \eqref{eq:kernel_penal}. 
Motivated by the likelihood ratio test for nested regression models, 
we  compare the 
mean squared errors from the pooled and group-specific kernel regressions to construct test statistics. 
For example, we can consider test statistic of the following form:
\begin{align}\label{eq:test_mse}
    T(\bXmatrix, \bs{Y}, \bs{Z} )
    & = 
    n \log (\text{MSE}) - \sum_{h=1}^H n_h \log (\text{MSE}_h), 
\end{align}
where $n_1, \ldots, n_H$ are the group sizes, $\text{MSE} = n^{-1} \sum_{i=1}^n (Y_i - \hat{f}(\bs{X}_i))^2$ with $\hat{f}$ being the kernel regression estimate using all the samples, 
and 
$\text{MSE}_h = n_h^{-1} \sum_{i:Z_i=h} (Y_i - \hat{f}(\bs{X}_i))^2$ with $\hat{f}_h$ being the kernel regression estimate using only the samples in group $h$.
Due to the flexibility of the permutation method, 
we can use loss functions other than the squared loss in \eqref{eq:kernel_penal} to conduct kernel regression, such as the epsilon-intensive loss and Huber loss \citep[see, e.g.,][]{wang2005support, cavazza2016active}.

We then consider test statistics based on GPR models. 
We introduce two general alternative models and compute their likelihood ratios against  $\tilde{H}_0$. 
Specifically, we model  functions in different groups as dependent Gaussian processes under the alternative hypothesis, and decompose each function into two components, a shared component and a group-specific component, assuming that these components follow independent Gaussian processes with the same general kernel but different variances:
\begin{align}\label{eq:h1}
	\tilde{H}_1 : & Y_i =  f_{Z_i}(\bs{X}_i) + \varepsilon_i, \quad \varepsilon_i \mid \bs{X}_i, Z_i \sim \mathcal{N}(0,\sigma_0^2), \quad f_{h} =  f + \bar{f}_{h}, \\\nonumber
	& f \sim  \text{GP}\left( 0, \frac{\delta_{0}^2}{n^{1-\gamma}} K \right), \quad \bar{f}_h \sim \text{GP}\left( 0, \frac{\delta_{h}^2}{n^{1-\gamma}} K \right),
\end{align}  
where $\{(\bs{X}_i, Z_i, \varepsilon_i)\}_{i=1}^n$ are $\iid$, and $f, \bar{f}_1, \cdots, \bar{f}_H$ and $\{(\bs{X}_i, Z_i, \varepsilon_i)\}_{i=1}^n$ are jointly independent. We can further extend the above homoscedastic model to allow noises to have different conditional variances in different groups as follows:
\begin{align}\label{eq:h1prime}
	\tilde{H}'_1 : \text{same as $\tilde{H}_1$ in (\ref{eq:h1}) except that } \ \varepsilon_i \mid \bs{X}_i, Z_i \sim \mathcal{N}(0,\sigma_{Z_i}^2).
\end{align}  
We define the test statistic based on the likelihood ratio of $\tilde{H}_1$ in \eqref{eq:h1} (or $\tilde{H}_1'$ in \eqref{eq:h1prime}) versus $\tilde{H}_0$ in \eqref{eq:h0_gp}, that is,
\begin{align}\label{eq:lik_ratio}
T(\bXmatrix, \bs{Y}, \bs{Z} ) = & \frac{\max f( \bs{Y} \mid \bXmatrix, \bs{Z}, \tilde{H}_1)}{\max f( \bs{Y} \mid \bXmatrix, \bs{Z}, \tilde{H}_0)}\  \left( \text{or \ }
 \frac{\max f( \bs{Y} \mid \bXmatrix, \bs{Z}, \tilde{H}'_1)}{\max f( \bs{Y} \mid \bXmatrix, \bs{Z}, \tilde{H}_0)}\ 
\right).
\end{align}
In the Supplementary Material, we discuss different ways to compute (\ref{eq:lik_ratio}) including the EM algorithm \citep{rubinEM1977}, Newton's method, the Fisher scoring, and quadratic programming. 

Here we briefly comment on hypothesis testing of $\tilde{H}_0$ against $\tilde{H}_1$ or $\tilde{H}_1'$. 
Note that under  $\tilde{H}_0$, 
the variance parameters $\delta_h^2$'s are zero and thus are at their boundaries. 
Therefore, the classical likelihood ratio testing procedure using the chi-square approximation for the null distribution does not work here.  This also suggests the importance and nontriviality of Theorem \ref{thm:general_ppt_valid}. 
To reduce the computational cost, we further introduce the following "pseudo" alternative model, which may not contain the null model $\tilde{H}_0$ as a submodel:
\begin{align}\label{eq:h_ps}
	\tilde{H}_{\text{pseudo}} : & Y_i =  f_{Z_i}(X_i) + \varepsilon_i, \ \ \varepsilon_i \mid \bs{X}_i, Z_i \sim \mathcal{N}(0,\sigma_0^2), \ \ f_h \sim \text{GP}\left( 0, \frac{\delta_{h}^2}{n^{1-\gamma}} K \right).
\end{align}  
As discussed in the Supplementary Material, the likelihood ratio of $\tilde{H}_{\text{pseudo}}$ versus $\tilde{H}_0$ can be efficiently computed using the EM algorithm.

\section{Extension to Correlated Noises}\label{sec:corr}

In the following discussion, we assume that the noises $\varepsilon_i$'s are correlated instead of $\iid$ as in models $H_0$ in \eqref{eq:H_0_fixed} and $\tilde{H}_0$ in \eqref{eq:h0_gp}, and the
covariance matrix of  $\bs{\varepsilon} = (\varepsilon_1, \ldots, \varepsilon_n)^\top$ is known up to a certain positive scale, unless otherwise stated. 
For example, when the residuals have equal variance, we essentially require that the correlation matrix of $\bs{\varepsilon}$ is known. 
In practice, we suggest to first estimate the covariance matrix for $\bs{\varepsilon}$ based on all the structure information we have (e.g., equal correlations or block-wise independence), and then plug in the estimate to  conduct the partial permutation tests described below. 

We 
extend the regression model $H_0$ in \eqref{eq:H_0_fixed} to allow for correlated noises: 
\begin{align}\label{eq:H_0_fixed_corr}
    H_{0}^{\text{C}}: & Y_i = f_0( \bs{X}_i ) + \varepsilon_i, \quad 
    \bs{\varepsilon} = (\varepsilon_1, \varepsilon_2, \ldots, \varepsilon_n)^\top \sim \mathcal{N}(\bs{0}, \ \sigma_0^2  \bs{\Sigma}), \qquad (1\le i \le n)
\end{align}
where we use the supscript in $H_{0}^{\text{C}}$ to emphasize that the noises under model 
\eqref{eq:H_0_fixed_corr} are allowed to be correlated. 
Moreover, we assume that $\bs{\Sigma}$ is known and positive definitive but $\sigma_0^2$ can be unknown, 
i.e., the covariance matrix of $\bs{\varepsilon}$ is known up to a positive scale. 
Recall that 
$\bs{Y}=(Y_1, \ldots, Y_n)^\top$ and $\bs{f}_0 = (f_0(\bs{X}_1), \cdots, f_0(\bs{X}_n))^\top$. 
Under $H_{0}^{\text{C}}$ in \eqref{eq:H_0_fixed_corr}, we have 
$
\bs{\Sigma}^{-1/2} \bs{Y} = \bs{\Sigma}^{-1/2} \bs{f}_0 + \bs{\Sigma}^{-1/2} \bs{\varepsilon}, 
$
where $\bs{\Sigma}^{-1/2}$ is the inverse of the positive definitive square root of $\bs{\Sigma}$. 
By our model assumption, it is easy to see that 
the elements of $\bs{\Sigma}^{-1/2} \bs{\varepsilon}$ are $\iid$ Gaussian with mean zero and variance $\sigma_0^2$. 
This then motivates us to consider a partial permutation test based on response vector $\bs{Y}^{\text{C}} \equiv \bs{\Sigma}^{-1/2} \bs{Y}$ and sample ``kernel'' matrix $\bs{K}^{\text{C}}_{n} \equiv \bs{\Sigma}^{-1/2}  \bs{K}_n \bs{\Sigma}^{-1/2}$. 
More precisely, in Algorithm \ref{alg:partial_permu}, we replace $\bs{Y}$ and $\bs{K}_n$ by $\bs{Y}^{\text{C}}$ and $\bs{K}^{\text{C}}_{n}$, and denote the resulting $p$-value by $p(\bs{X}, \bs{Y},  \bs{Z}, \bs{\Sigma})$, which depends crucially on the noise covariance structure $\bs{\Sigma}$.

By the same logic as Theorem \ref{thm:fixed_property_theorem}, we can derive a finite-sample valid partial permutation test with a certain correction on the permutation $p$-value. 
Specifically,  for $1\le b_n\le n$ and $0<\alpha_0<1$,
we define $\omega_{\text{C}}(b_n,\sigma_0^{-1}f_0, \bs{\Sigma}) = \sigma_0^{-2} \sum_{i=n-b_n+1}^n (\bs{\gamma}_{i}^\top \bs{\Sigma}^{-1/2} \bs{f}_0)^2$ to denote the LOSP for the components used for partial permutation,  
and
\begin{align}\label{eq:v_fixed_H0_corr}
v_{\text{C}}(b_n,\sigma_0^{-1}f_0, \bs{\Sigma}, \alpha_0) = \frac{1}{2}\exp\left\{
2\sqrt{2\omega_{\text{C}}(b_n,\sigma_0^{-1}f_0, \bs{\Sigma})}
\sqrt{Q_{b_n}(1-\alpha_0)+\omega_{\text{C}}(b_n,\sigma_0^{-1}f_0, \bs{\Sigma})} 
\right\}-\frac{1}{2}, 
\end{align}
where $Q_{b_n}$ denotes the quantile function of the $\chi^2_{b_n}$-distribution.

\begin{theorem}\label{thm:fixed_property_theorem_corr}
    Let $\{(\bs{X}_i,Y_i,Z_i)\}_{1\leq i \leq n}$ denote samples from model ${H}_{0}^{\text{C}}$ in \eqref{eq:H_0_fixed_corr}. 
    Given $1\leq b_n \leq n$ and $\alpha_0 \in (0,1)$, we define the corrected partial permutation $p$-value as
    \begin{align*}
    p_{\correct}( \bXmatrix, \bs{Y}, \bs{Z}, \bs{\Sigma} ) = p(\bXmatrix, \bs{Y}, \bs{Z}, \bs{\Sigma}  ) + v_{\text{C}}(b_n,\sigma_0^{-1}f_0, \bs{\Sigma}, \alpha_0) +\alpha_0,
    \end{align*}
    where $p(\bXmatrix, \bs{Y}, \bs{Z}, \bs{\Sigma}  )$ is the $p$-value from either the discrete or continuous partial permutation test based on kernel $K$, permutation size $b_n$, any test statistic $T$, and covariance matrix $\bs{\Sigma}$, 
    and $v_{\text{C}}(b_n,\sigma_0^{-1}f_0, \bs{\Sigma},\alpha_0)$ is as defined in \eqref{eq:v_fixed_H0_corr}. Then the corrected partial permutation $p$-value is valid under model ${H}_{0}^{\text{C}}$, i.e., 
    $\forall \alpha \in (0,1)$, 
    $
    \Pr_{{H}_{0}^{\text{C}}}\{{p}_{\correct}( \bXmatrix, \bs{Y}, \bs{Z}, \bs{\Sigma} )\leq\alpha \mid  \bXmatrix, \bs{Z} \}\leq\alpha. 
    $
\end{theorem}

Again, it is generally difficult to show the asymptotic validity of the $p$-value $p(\bXmatrix, \bs{Y}, \bs{Z}, \bs{\Sigma}  )$ for a general kernel under general underlying function and noise covariance structure, and its correction term in \eqref{eq:v_fixed_H0_corr} depends on the unknown $f_0$ and $\sigma_0$. 
In practice, we can adopt similar strategies as discussed in Section \ref{sec:implement}. 
Below we consider four special cases, in parallel to Sections \ref{sec:fix_special_finite}--\ref{sec:GP}, 
under which we can demonstrate the exact or asymptotic validity of the partial permutation test that takes into account the covariance structure.

\subsection{Special Case: Kernels with Finite-Dimensional Feature Space}

When the kernel has a finite-dimensional feature space and the underlying function is linear in features mapped to this space, the partial permutation test is exactly valid.

\begin{corollary}\label{cor:kernel_finite_dim_feature_space_corr}
    Let $\{(\bs{X}_i,Y_i,Z_i)\}_{1\leq i \leq n}$ denote samples from model $H_{0}^{\text{C}}$ in \eqref{eq:H_0_fixed_corr}. Suppose  kernel function $K$ has the decomposition 
	$K(\bs{x}, \bs{x}') = \phi(\bs{x})^\top \phi(\bs{x}')$ with $\bs{\phi}(\bs{x}) \in \mathbb{R}^q$ for some $q< \infty$, 
	and the underlying function $f_0(\bs{x})$ is linear in $\phi(\bs{x})$. 
	Then, the $p$-value obtained by either the discrete or continuous partial permutation test with kernel $K$, permutation size $b_n \le n-q$, any test statistic $T$, and covariance matrix $\bs{\Sigma}$ is valid, i.e., $\forall \alpha \in (0,1)$, 
	$
	\Pr_{H_{0}^{\text{C}}}\{ p( \bXmatrix,\bs{Y}, \bs{Z}, \bs{\Sigma} ) \leq \alpha 
	\mid \bXmatrix, \bs{Z}
	\} \le \alpha. 
	$
\end{corollary}

\subsection{Special Case: Kernels with Diverging-Dimensional Feature Space}

Similar to Sections \ref{sec:null_diverg} and \ref{sec:diverg_power}, we consider kernels with diverging-dimensional feature space, i.e., 
$K_q(\bs{x}, \bs{x}') = \phi_{q}(\bs{x})^\top \phi_q(\bs{x}')$ for $q\ge 1$ with 
$\phi_q(\bs{x}) = (e_1(\bs{x}), e_2(\bs{x}), \ldots, e_q(\bs{x}))^\top$ and $\{e_j\}_{j=1}^\infty$ being a series of basis functions. 
We assume that the covariates are identically distributed from some probability measure $\mu$, 
and use $\remainder(f; q) = \min_{\bs{b}\in \mathbb{R}^q} \int (f - \bs{b}^\top \phi_q)^2 \text{d} \mu$ to denote the squared error for the best linear approximation of $f$ using the first $q$ basis functions. 
The following corollary shows that the partial permutation test is asymptotically valid when the underlying functional relationship can be well approximated by the basis functions and the smallest eigenvalue of the noise covariance matrix $\lambda_{\min}(\bs{\Sigma})$ decays not too fast.  

\begin{corollary}\label{cor:corr_diverg_kernel}
Let $\{(\bs{X}_i,Y_i,Z_i)\}_{1\leq i \leq n}$ denote samples from model $H_{0}^{\text{C}}$ in \eqref{eq:H_0_fixed_corr}, 
and assume that $\bs{X}_i$'s are identically distributed from some probability measure $\mu$. 
Suppose that the kernel function $K_q$ has the form 
$K_q(\bs{x}, \bs{x}') \equiv \phi_q(\bs{x})^\top \phi_q(\bs{x}') \equiv \sum_{j=1}^q \basis_j(\bs{x}) \basis_j(\bs{x}')$ for $q\ge 1$ and some series of basis functions $\{\basis_j\}_{j=1}^{\infty}$. 
If 
there exists a sequence $\{q_n\}_{n=1}^\infty$ such that $q_n < n$ for all $n$ and $n (n-q_n) \remainder(f_0; q_n) / \lambda_{\min}(\bs{\Sigma}) \rightarrow 0$ as $n\rightarrow \infty$, 
then the resulting $p$-value obtained by either the discrete or continuous partial permutation test with kernel $K_{q_n}$, permutation size $b_n \le n-q_n$, any test statistic $T$, and covariance matrix $\bs{\Sigma}$ is asymptotically valid, i.e., $\forall \alpha \in (0,1)$, 
$
\limsup_{n\rightarrow \infty}
\Pr_{H_0^{\text{C}}}\{p(\bXmatrix,\bs{Y}, \bs{Z}, \bs{\Sigma} ) \leq \alpha 
\} \le \alpha. 
$
\end{corollary}

\subsection{Special Case: Exactly Balanced Covariates across All Groups}

In the case that the covariates are exactly balanced across all groups as in \eqref{eq:balanced} and the kernel matrix for distinct covariates within each group is of full rank (which generally holds
when the kernel has an infinite-dimensional feature space, 
e.g., the Gaussian kernel),
the following corollary shows that the partial permutation test is exactly valid under a general functional relationship. 

\begin{corollary}\label{cor:fixed_property_theorem_balanced_design_corr}
Let $\{(\bs{X}_i,Y_i,Z_i)\}_{1\leq i \leq n}$ denote samples from model ${H}_{0}^{\text{C}}$ in \eqref{eq:H_0_fixed_corr}. 
If the design matrix is exactly balanced in the sense that \eqref{eq:balanced} holds 
and the kernel matrix for the $r\le n/H$ distinct covariates within each group is of full rank (or equivalently $\text{rank}(\bs{K}_n) = r$), 
then
the partial permutation $p$-value from either the discrete or continuous partial permutation test with kernel $K$, permutation size $b_n \le n-r$, any test statistic $T$, and covariance matrix $\bs{\Sigma}$ 
is valid under model $H_{0}^{\text{C}}$, i.e., $\forall \alpha \in (0,1)$,  
$
\Pr_{H_{0}^{\text{C}}}\{p(\bXmatrix,\bs{Y}, \bs{Z}, \bs{\Sigma} ) \leq \alpha 
\mid \bXmatrix, \bs{Z}
\} \le \alpha. 
$
\end{corollary}

Furthermore, if all covariates within each group are distinct and the covariance among noises enjoys the following structure: (i) the noises have equal variances, 
(ii) the noises for samples with different covariates are uncorrelated, 
and (iii) the noises for samples with the same covariates are equally correlated with correlation $\rho$, 
then the partial permutation test is always valid even if we use a correlation matrix with incorrect correlation $\tilde{\rho} \ne \rho$.
This means that, 
with this special covariance structure, we are able to conduct valid permutation tests even if the true correlation matrix is unknown. 
Such covariance structure is reasonable when the same covariate corresponds to the same individual and the response within each group corresponds to measurement at different time periods. 
As a side note, the usual permutation test that switches group indicators of samples with the same covariates is valid in more general setting, as long as the noises for samples with different covariate values are mutually independent and the noises for samples with the same covariate values are exchangeable. 
The partial permutation test allows for more general permutation or rotation, but with a stronger Gaussianity assumption on the noises.

\subsection{Special Case: Gaussian Process Regression Model}

Finally we extends the GPR model $\tilde{H}_{0}$ in \eqref{eq:h0_gp} to allow for correlated noises: 
\begin{align}\label{eq:h0_gp_corr}
    \tilde{H}_{0}^{\text{C}}: & Y_i =  f(\bs{X}_i) + \varepsilon_i, 
    \quad 
    \bs{\varepsilon} \mid \bXmatrix, \bs{Z} \ \overset{\iid}{\sim} \ \mathcal{N}( \bs{0}, \sigma_0^2 \bs{\Sigma}), 
    \quad 
    f \sim  \text{GP}\left( 0, \ \frac{\delta_0^2}{n^{1-\gamma}} K(\cdot,\cdot) \right). 
\end{align}
The following theorem extends Theorem \ref{thm:general_ppt_valid} and  demonstrates the asymptotic validity of the partial permutation test after taking into the account the noise covariance structure.

\begin{theorem}\label{thm:general_ppt_valid_corr}
	Let $\{(\bs{X}_i,Y_i,Z_i)\}_{1\leq i \leq n}$ denote samples from model $\tilde{H}_{0}^{\text{C}}$ in \eqref{eq:h0_gp_corr}. If 
    the covariates $\bs{X}_i$'s are $\iid$ from a compact support $\mathcal{X}$ with probability measure $\mu$, 
    the eigenvalues $\{\lambda_k\}$ of kernel $K$ on $(\mathcal{X},\mu)$ satisfy $\lambda_k = O(k^{-\rho})$ with $\rho >1$, 
    the smallest eigenvalue of 
    $\bs{\Sigma}$ for the noises satisfies $\lambda_{\min}(\bs{\Sigma}) \ge c n^{-\zeta}$ for some positive $c$ and $\zeta < 1 - \rho^{-1}$, 
    and $\gamma$ is a constant less than $1-\rho^{-1} - \zeta$, then for sequence $\{b_n\}$ satisfying 
    $b_n=   O(n^\kappa)$ with $0<\kappa <1-\rho^{-1}-\zeta-\gamma$, 
    the partial permutation $p$-value from either the discrete or continuous partial permutation test with kernel $K$, permutation size $b_n$, any test statistic $T$, and covariance matrix $\bs{\Sigma}$ is asymptotically valid under $\tilde{H}_0^{\text{C}}$, i.e., $\forall \alpha\in (0,1)$, 
    $
    \limsup_{n \rightarrow \infty} \Pr_{\tilde{H}_0^{\text{C}}} \{p(\bs{X}, \bs{Y}, \bs{Z},  \bs{\Sigma}) \leq \alpha \} \leq \alpha. 
    $
\end{theorem}

Theorem \ref{thm:general_ppt_valid} proves the large-sample validity of the partial permutation test.  
Below we investigate its finite-sample performance. 
Analogous to Theorem \ref{thm:gp_finite_property_theorem}, 
let $\xi_n=(\delta_0^2/n^{1-\gamma})/\sigma_0^2$ denote the variance ratio, and 
$\tilde{\omega}_{\text{C}}(b_n, \xi_n, \bs{\Sigma}) = \xi_n \cdot \zeta_{n-b_n+1}$ denote the LOSP,  
where 
$\zeta_{n-b_n+1}$ denotes the $(n-b_n+1)$th largest eigenvalue of $\bs{K}_n^{\text{C}}$. 
We then define 
\begin{align}\label{eq:v_correct_gaussian_corr}
\tilde{v}_{\text{C}}(b_n, \xi_n, \bs{\Sigma}, \alpha_0)=\frac{1}{2} \exp\left[
\frac{1}{2}
\tilde{\omega}_{\text{C}}(b_n, \xi_n, \bs{\Sigma}) \cdot Q_{b_n}(1-\alpha_0)
\right]-\frac{1}{2}, 
\end{align}
recalling that $Q_{b_n}$ is the quantile function of the $\chi^2_{b_n}$-distribution. 
The following theorem shows that the partial permutation $p$-value can be finite-sample valid under $\tilde{H}_{0}^{\text{C}}$ after an adjustment. %

\begin{theorem}\label{thm:gp_finite_property_theorem_corr}
	Let $\{(\bs{X}_i,Y_i,Z_i)\}_{1\leq i \leq n}$ denote samples from model $\tilde{H}_{0}^{\text{C}}$ in \eqref{eq:h0_gp_corr}.
	Given $1\leq b_n \leq n$ and $0<\alpha_0<1$, we define the corrected partial permutation $p$-value as follows,
	\begin{align*}
	\tilde{p}_{\correct}(\bXmatrix, \bs{Y}, \bs{Z}, \bs{\Sigma}) = p(\bXmatrix, \bs{Y}, \bs{Z}, \bs{\Sigma} ) + \tilde{v}(b_n, \xi_n, \bs{\Sigma}, \alpha_0) + \alpha_0,
	\end{align*}
	where $p(\bXmatrix, \bs{Y}, \bs{Z}, \bs{\Sigma})$ is the $p$-value from either the discrete or continuous partial permutation test with kernel $K$, permutation size $b_n$, any test statistic $T$, and covariance matrix $\bs{\Sigma}$, and $\tilde{v}(b_n, \xi_n, \bs{\Sigma}, \alpha_0)$ is as defined in \eqref{eq:v_correct_gaussian}. Then the corrected partial permutation $p$-value is valid under model $\tilde{H}_{0}^{\text{C}}$, that is, 
	$\forall \alpha \in (0,1)$, 
	$
	\Pr_{\tilde{H}_0^{\text{C}}}
	\{
	\tilde{p}_{\correct}(\bXmatrix, \bs{Y}, \bs{Z}, \bs{\Sigma})\leq\alpha
	\mid \bXmatrix,\bs{Z}
	\}
	\leq\alpha. 
	$
\end{theorem}

By the same logic as Section \ref{sec:choice_size}, we can then use Theorem \ref{thm:gp_finite_property_theorem_corr} to guide the choice of permutation size in finite samples.

\section{Simulation Study}\label{sec:simu}

In this section, we present simulation results based on various choices of the kernels and  
discrete partial permutation tests described in Algorithm \ref{alg:partial_permu}. 
Specifically, in Sections \ref{sec:scalar_simulation}--\ref{sec:simu_non_smooth} we investigate type-I error control under the null hypothesis, 
and in Sections \ref{sec:power_Ftest} and \ref{sec:power_parallel} we compare  powers of the partial permutation test and some other methods. 
We also conduct simulations with non-Gaussian or correlated noises, which are  
 relegated to  Supplementary Material.
Moreover, 
for simulation under the null hypothesis, 
we focus mainly on the Gaussian kernel 
and choose the tuning parameters,
permutation size $b_n$ and test statistic $T$, as follows:
(1) we standardize both the response and covariates, and consider Gaussian kernel $K_{\text{G}}$ in \eqref{eq:gaussian_kernel} with the
maximum marginal likelihood estimates for  parameters $\omega_k$'s as discussed in Section \ref{sec:kernel}; 
(2) we choose the permutation size $b_n$ as suggested in Section \ref{sec:choice_size} based on model $\tilde{H}_0$ 
with significance level $\alpha=0.05$;
(3) we choose the likelihood ratio of $\tilde{H}_{\text{pseudo}}$ versus $\tilde{H}_0$ as  the test statistic $T$ due to its lower computation cost, unless otherwise stated.

\subsection{Simulation under the null hypothesis with scalar covariate}\label{sec:scalar_simulation}

We first consider partial permutation test under $H_0$ with a scalar covariate and two groups. 
We generate data as $\iid$ samples from the following model: 
\begin{align}\label{eq:scalar_null}
\text{Scenario 1:} \quad & Y = f_0(X) + \varepsilon, \quad \varepsilon \mid X, Z \sim \mathcal{N}(0,\sigma_0^2), 
\nonumber
\\
& X \mid Z \sim  a_{Z} \cdot \text{Unif}(-1,0)+(1-a_{Z}) \cdot \text{Unif}(0,1), 
\nonumber
\\
& \Pr(Z=h) = p_h, \quad h=1,2,
\end{align}
where 
$\text{Unif}(-1,0)$ and $\text{Unif}(0,1)$ refer to uniform distributions on $(-1, 0)$ and $(0,1)$,  
$(a_1, a_2)$ control the mixture weights for covariate distributions in two groups, 
and 
$(p_1, p_2)$ denote the fractions of observations (in expectation) from two groups. We consider the five cases in Table \ref{tab:balance} that vary both the proportions of units and the covariate distributions in two groups.
\begin{table}[htbp]
    \centering
    \caption{Cases with varying balancedness of group sizes and covariate distributions between the two groups in comparison.}
    \label{tab:balance}
    \begin{tabular}{ccccc}
    \toprule
    Case & Groups & Covariates & $(p_1, p_2)$ & $(a_1, a_2)$  \\
    \midrule
    (a) & Balanced & Balanced & (0.5, 0.5) & (0.5, 0.5)  \\
    (b) & Unbalanced & Balanced & (0.2, 0.8)  & (0.5, 0.5) \\
    (c) & Balanced & Unbalanced & (0.5, 0.5) & (0.8, 0.2)  \\
    (d) & Unbalanced & Unbalanced & (0.2, 0.8) & (0.8, 0.2)  \\
    (e) & Balanced & Non-overlap %
     & (0.5, 0.5) & (1, 0). \\
    \bottomrule
    \end{tabular}
\end{table}
Specifically, 
in case (e), 
covariates from the two groups do not  overlap at all. 
Therefore, case (e) resembles the regression discontinuity design, under which we can interpret the null  hypothesis $H_0$ in \eqref{eq:H_0_fixed} as that the underlying functions for the two groups can be smoothly connected at the boundary.
Finally, 
we fix $\sigma_0^2=0.1$ for all cases in Table \ref{tab:balance}, and consider the following six choices of the underlying function $f_0$, all in the range of $[-1, 1]$: 
\begin{align}\label{eq:f_choice_scalar}
\begin{tabular}{lll}
$\text{(i)}~ f_0 = x,$ 
& $\text{(ii)}~ f_0 = 2x^2-1,$ 
& $\text{(iii)}~  f_0 = 4x^3/3 - x/3,$
\\
$\text{(iv)}~  f_0 = 4/(1+x^2) - 3,$ 
& 
$\text{(v)}~  f_0 = \sin(4x),$ 
& 
$\text{(vi)}~  f_0 = \sin(6x)$.  
\end{tabular}
\end{align}%
\begin{figure}[htb]
	\centering
		\includegraphics[width=0.8\linewidth]{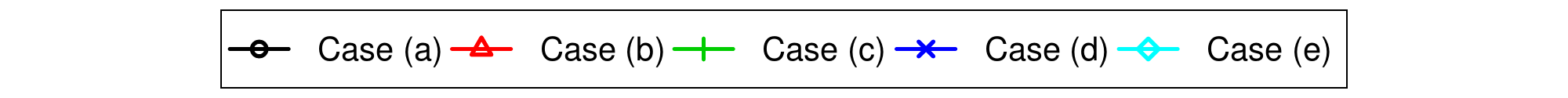}
	\begin{subfigure}{.25\textwidth}
		\centering
		\includegraphics[width=1\linewidth]{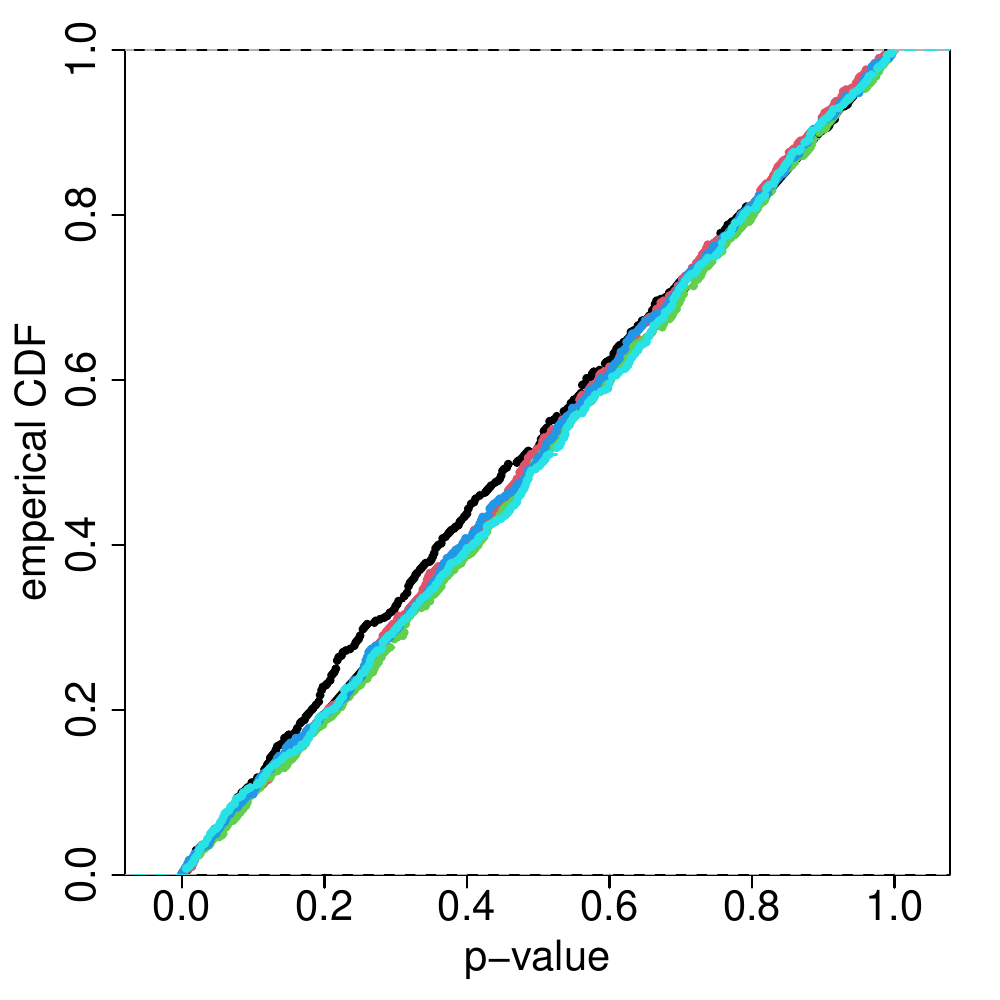}
		\caption*{$x$}
	\end{subfigure}%
	\begin{subfigure}{.25\textwidth}
		\centering
		\includegraphics[width=1\linewidth]{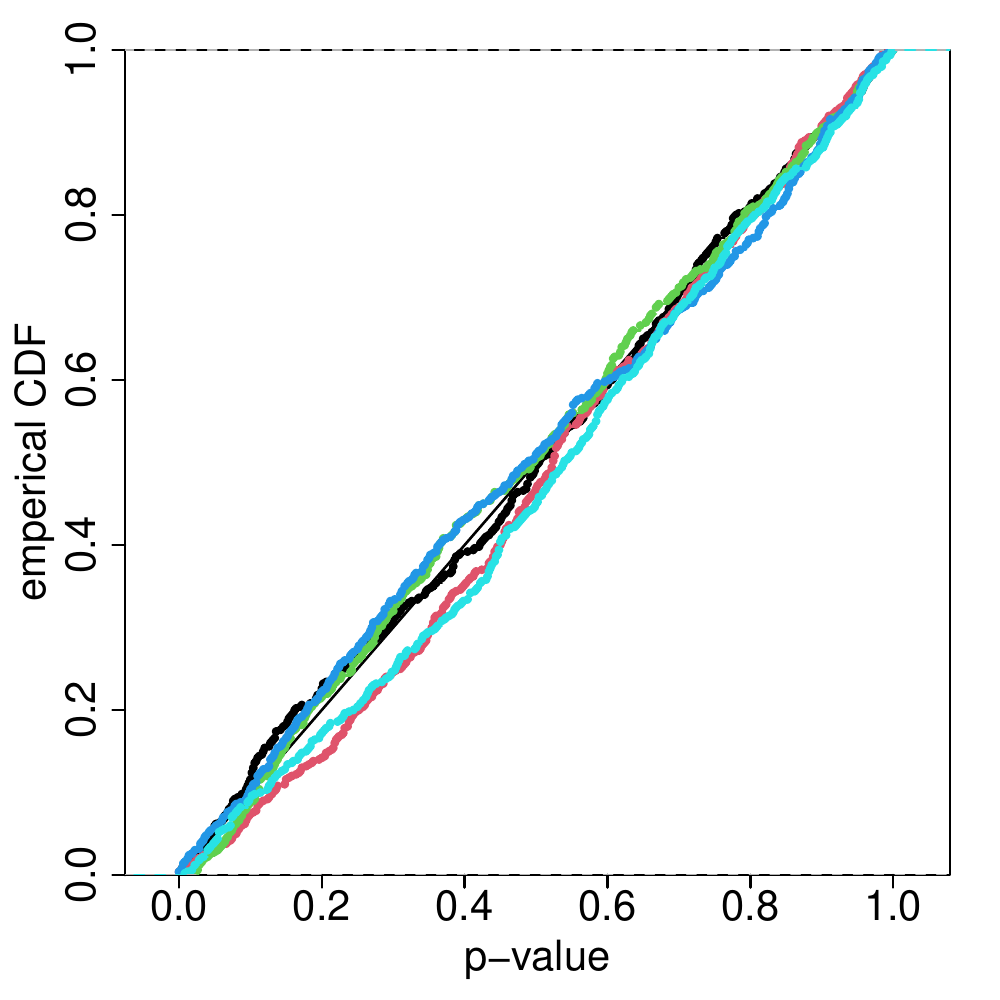}
		\caption*{$x^2$}
	\end{subfigure}%
	\begin{subfigure}{.25\textwidth}
		\centering
		\includegraphics[width=1\linewidth]{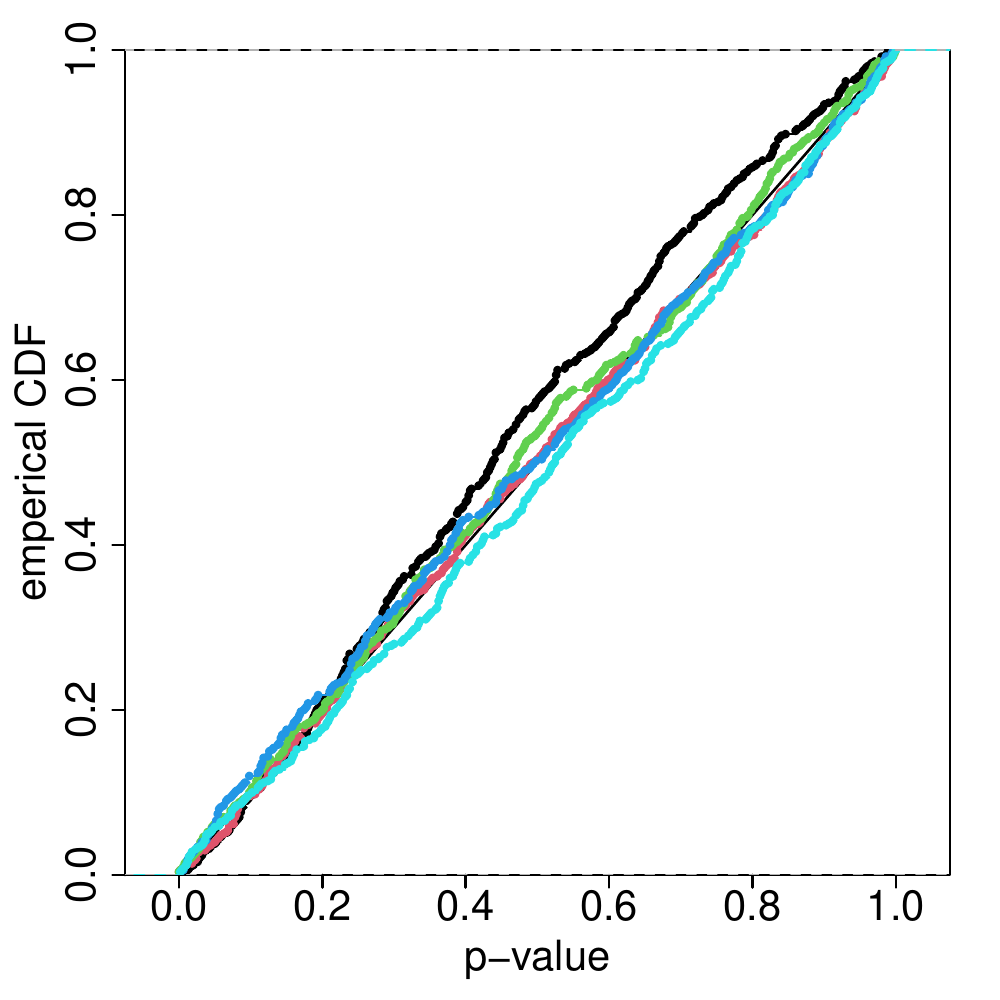}
		\caption*{$x^3 - x/4$}
	\end{subfigure}
	\begin{subfigure}{.25\textwidth}
		\centering
		\includegraphics[width=1\linewidth]{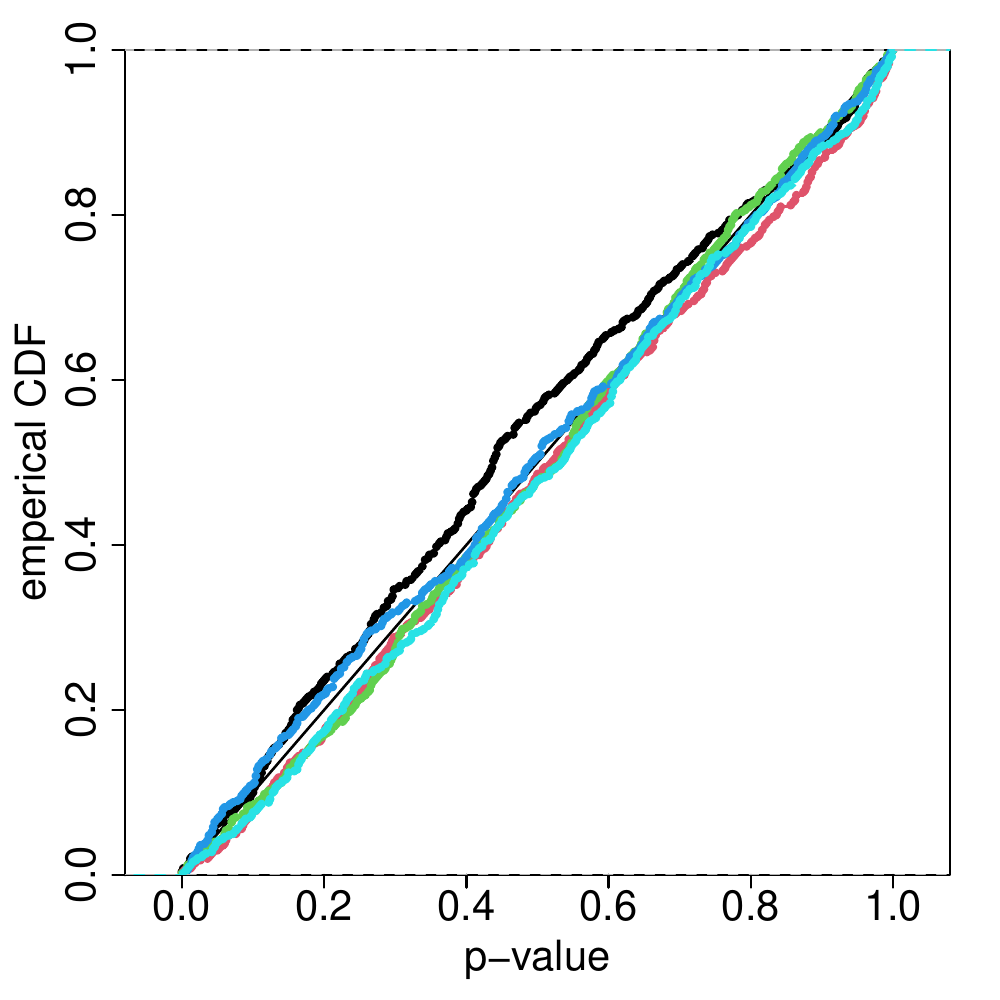}
		\caption*{$(1+x^2)^{-1}$} 
	\end{subfigure}
	\begin{subfigure}{.25\textwidth}
		\centering
		\includegraphics[width=1\linewidth]{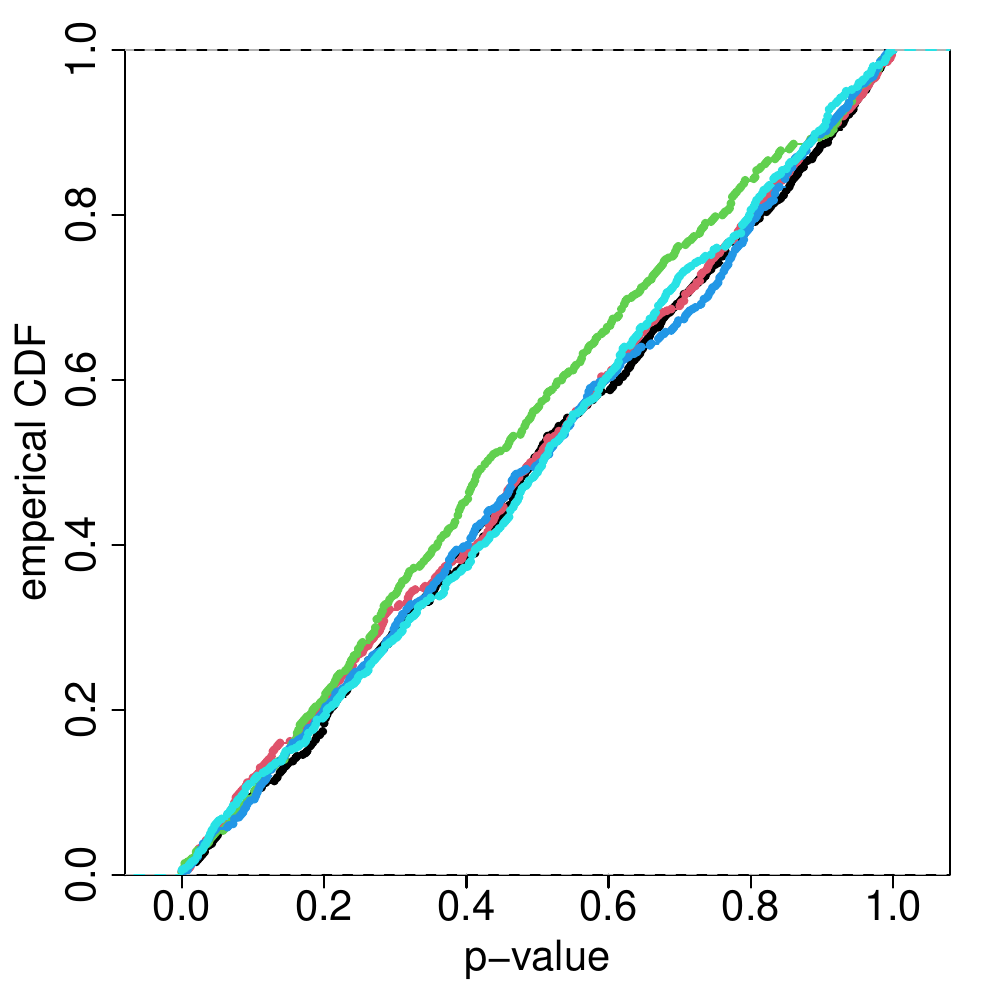}
		\caption*{$\sin(4x)$} 
	\end{subfigure}%
	\begin{subfigure}{.25\textwidth}
	\centering
	\includegraphics[width=1\linewidth]{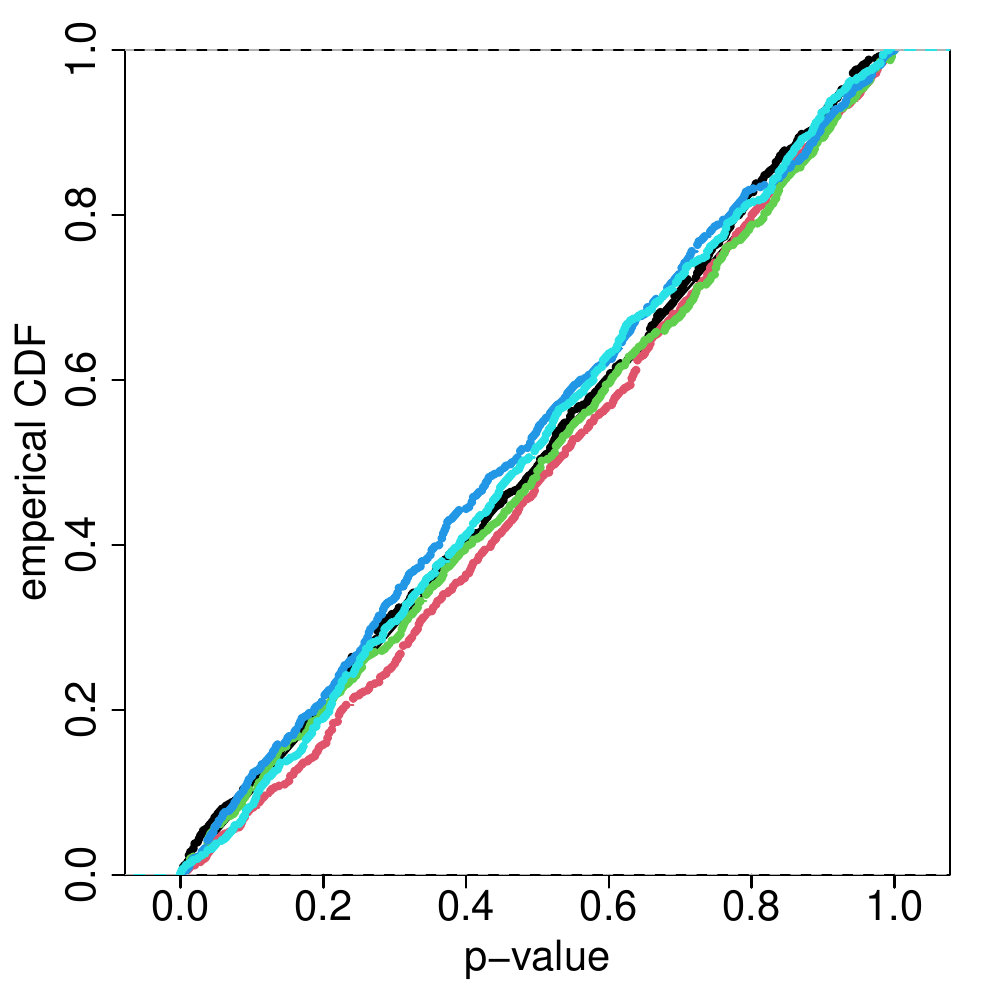}
	\caption*{$\sin(6x)$}
	\end{subfigure}
	\caption{
	Empirical distributions of the partial permutation $p$-values when data are generated from Scenario 1 in \eqref{eq:scalar_null} under all cases in Table \ref{tab:balance} with sample size $n=200$. 
	The six figures correspond to six choices of the underlying function $f_0$ in \eqref{eq:f_choice_scalar}. 
    }
      \label{S1_0_5_GP}
\end{figure}%
Figure \ref{S1_0_5_GP} shows the empirical distribution functions of the partial permutation $p$-values under all cases with sample size $n=200$,
showing that all are very close to Unif$(0,1)$ and demonstrating the 
validity of 
the partial permutation test.

\subsection{Simulation under the null hypothesis with two-dimensional covariates}\label{sec:simu_null_two_dim}
We generate data as $\iid$ samples from the following two-dimensional covariates model: 
\begin{align}\label{eq:model_two_cov}
\text{Scenario 2:} \quad & Y = f_0(X_1, X_2) + \varepsilon, \quad \varepsilon \mid \bs{X}, Z \sim \mathcal{N}(0,\sigma_0^2), \nonumber \\
& X_{1}\mid Z \sim a_{Z} \cdot \text{Unif}[-1,0] + (1-a_{Z}) \cdot \text{Unif}[0,1], \nonumber \\
& X_{2}\mid Z \sim a_{Z} \cdot \text{Unif}[-1,0] + (1-a_{Z}) \cdot \text{Unif}[0,1], \nonumber \\
& X_{1} \ind X_{2} \mid Z, \ \ P(Z=h) = p_h, \quad h=1,2,
\end{align}
where the choice of $(a_1,a_2)$ and $(p_1,p_2)$ is the same as in Table \ref{tab:balance}. 
We again fix $\sigma_0^2=0.1$ and consider the following six choices of the underlying function $f_0$, all in the range of $[-1, 1]$: 
\begin{align}\label{eq:choice_f_two_cov}
\begin{tabular}{lll}
$\text{(i)}~ f_0 = (x_1 + x_2)/2,$ 
&
$\text{(ii)}~ f_0  = x_1 x_2,$ 
&
$\text{(iii)}~ f_0 = 2(x_1+x_2)^3/15 - (x_1+x_2)/30,$
\\
$\text{(iv)}~ f_0 = 3/(1+x_1^2+x_2^2) - 2,$ 
&
$\text{(v)}~ f_0 = \sin(6x_1)+x_2,$
&
$\text{(vi)}~ f_0 = \sin(6x_1+6x_2).$
\end{tabular}
\end{align}

Figure \ref{S2_0_5_GP} shows the empirical distribution functions of the partial permutation $p$-values, which are close to Unif$(0,1)$ for all cases. 
\begin{figure}[h]
	\centering
		\includegraphics[width=0.8\linewidth]{plots/null_balance_legend_1031}
	\begin{subfigure}{.25\textwidth}
		\centering
		\includegraphics[width=1\linewidth]{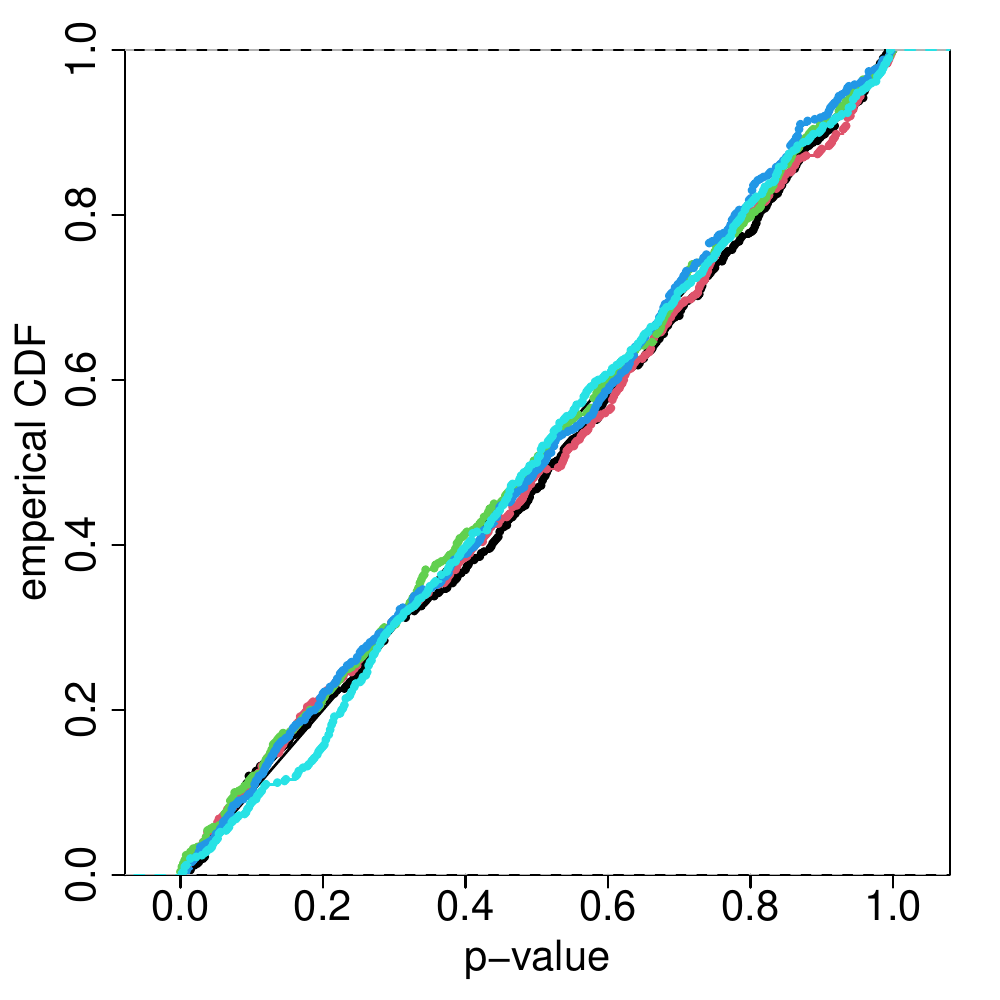}
		\caption*{$x_1 + x_2$}
	\end{subfigure}%
	\begin{subfigure}{.25\textwidth}
		\centering
		\includegraphics[width=1\linewidth]{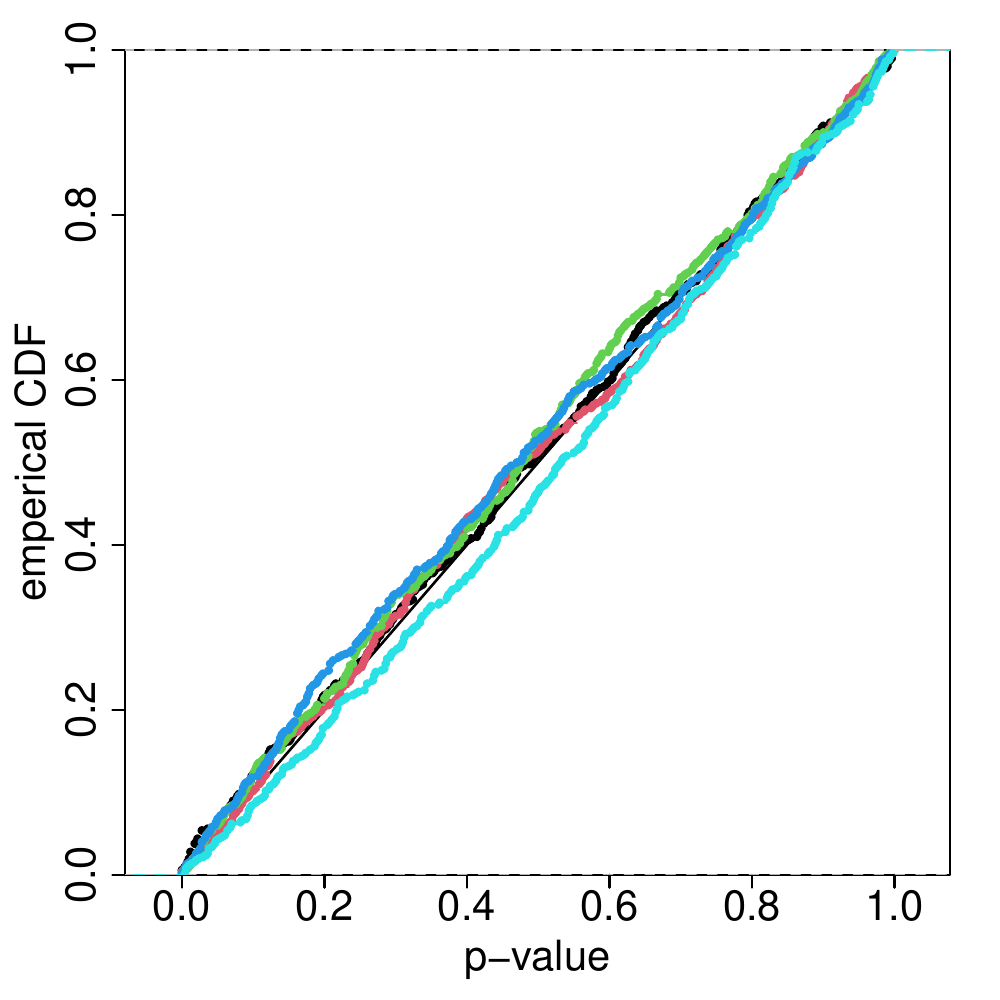}
		\caption*{$x_1 x_2$}
	\end{subfigure}%
	\begin{subfigure}{.25\textwidth}
		\centering
		\includegraphics[width=1\linewidth]{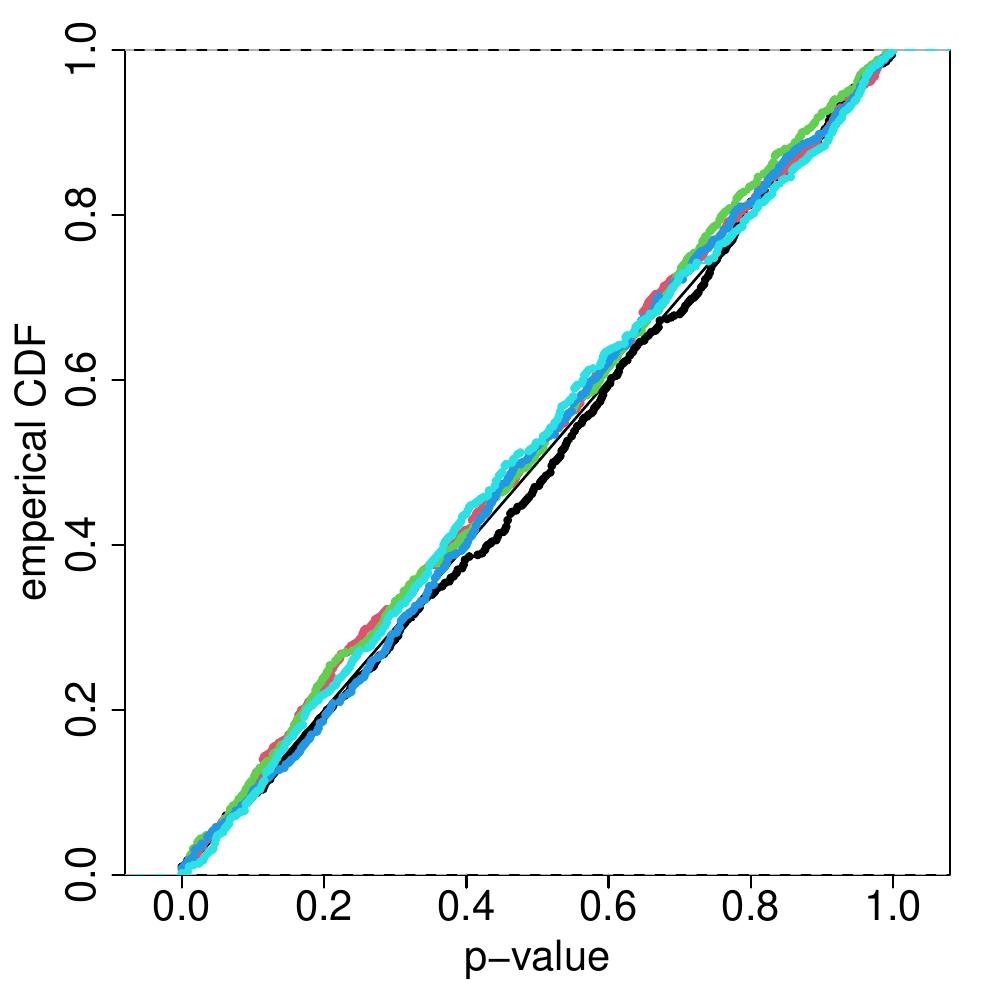}
		\caption*{\footnotesize {$(x_1+x_2)^2 - \frac{(x_1+x_2)}{4}$}}
	\end{subfigure}
	\begin{subfigure}{.25\textwidth}
		\centering
		\includegraphics[width=1\linewidth]{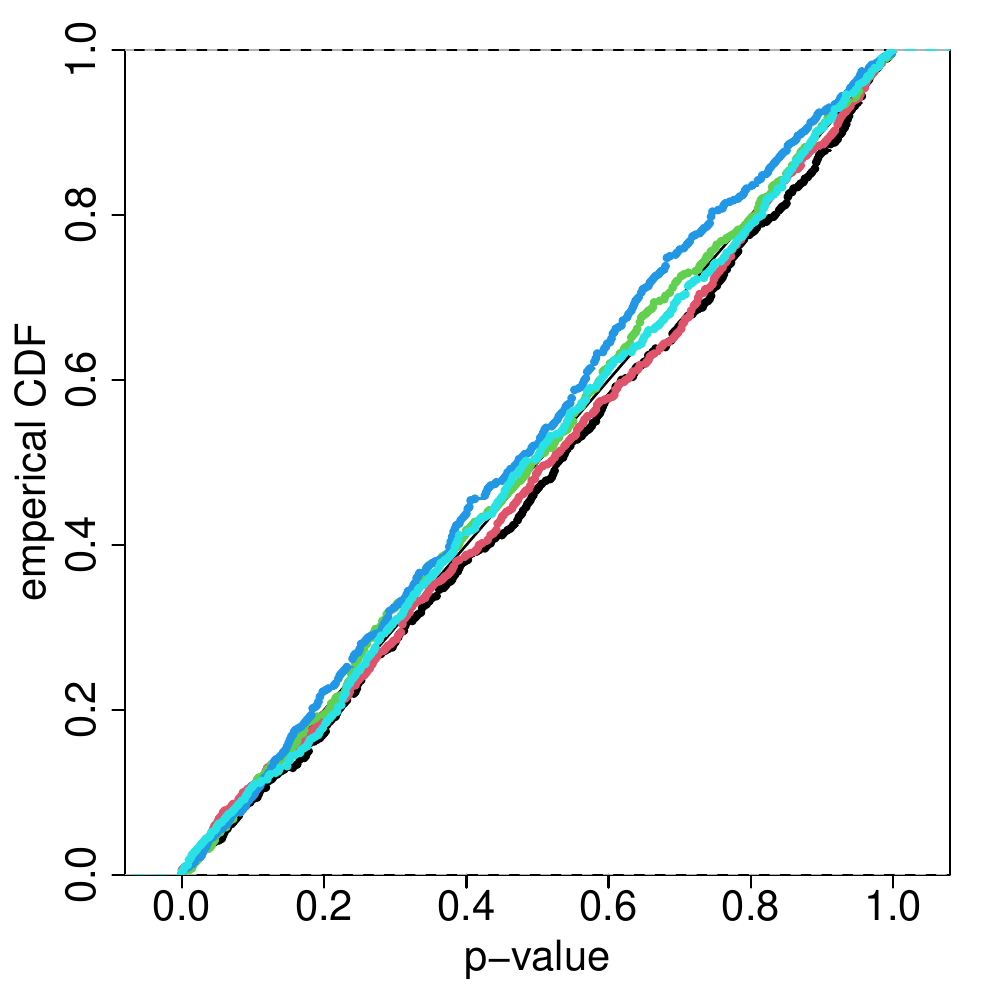}
		\caption*{{$(1+x_1^2+x_2^2)^{-1}$}} 
	\end{subfigure}
	\begin{subfigure}{.25\textwidth}
		\centering
		\includegraphics[width=1\linewidth]{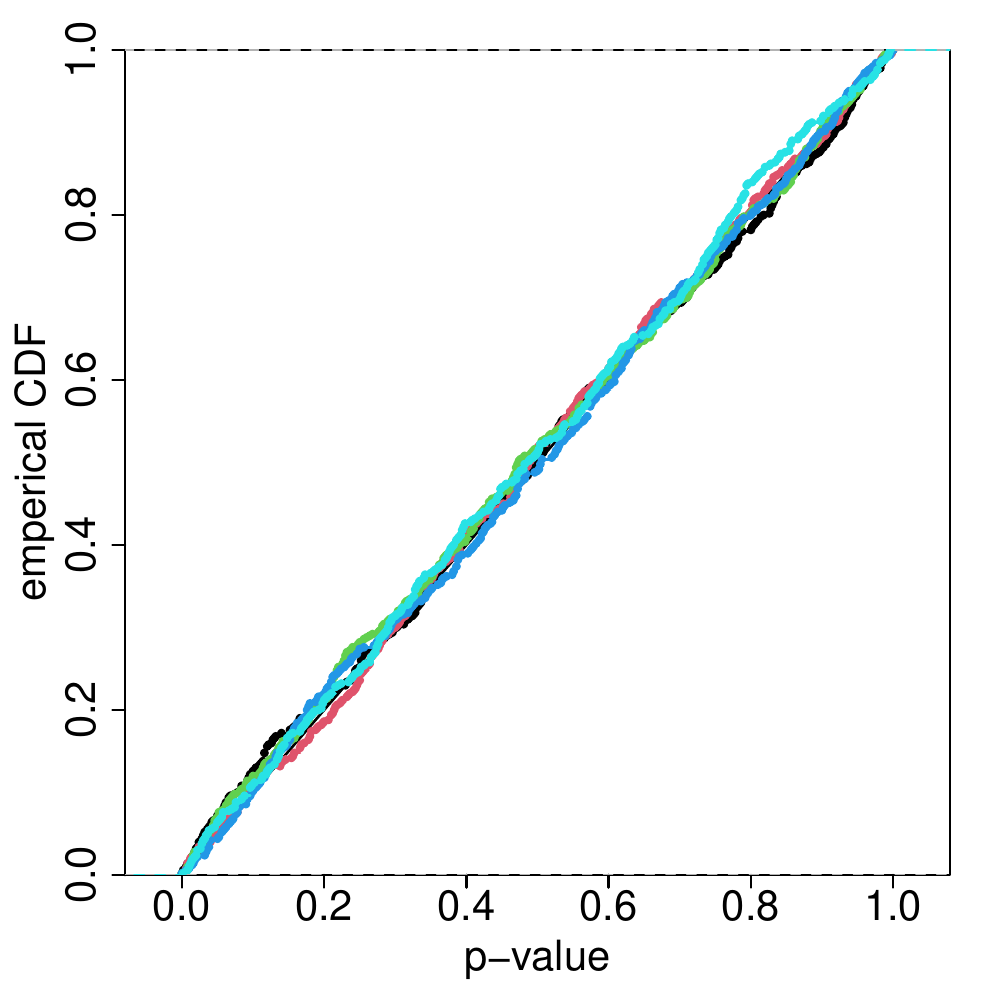}
		\caption*{$\sin(6x_1)+x_2$} 
	\end{subfigure}%
	\begin{subfigure}{.25\textwidth}
	\centering
	\includegraphics[width=1\linewidth]{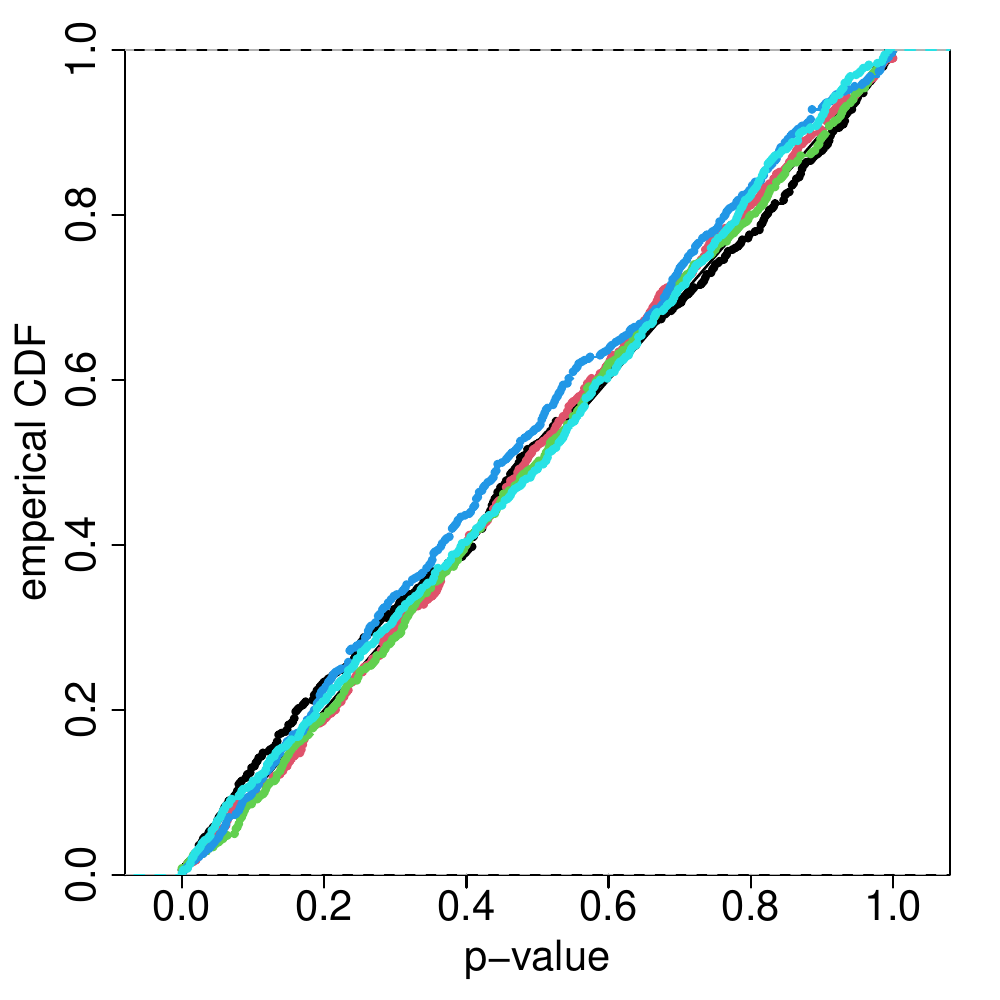}
	\caption*{ $\sin(6x_1+6x_2)$ } 
	\end{subfigure}
	\caption{
	Empirical distributions of the partial permutation $p$-values when data are generated from Scenario 2 in \eqref{eq:model_two_cov} under all cases in Table \ref{tab:balance} with sample size $n=200$. 
	The six figures correspond to six choices of the underlying function $f_0$ in \eqref{eq:choice_f_two_cov}. 
    }
      \label{S2_0_5_GP}
\end{figure}

\subsection{Simulation under the null hypothesis with non-smooth functions}\label{sec:simu_non_smooth}

In the previous two subsections we focus on null hypothesis with smooth functions. Here we consider 
the following continuous but non-differentiable univariate function: 
\begin{align}\label{eq:non_smooth_fun}
g_0(x) = 2 * \min\{|3x-\lfloor 3x \rfloor|, |3x-\lfloor 3x \rfloor-1|\} \cdot (\lfloor 3x \rfloor \text{ mod } 2 + 1) - 1, 
\end{align}
where $\lfloor 3x \rfloor$ denotes the largest integer less than or equal to $3x$ and $( \lfloor 3x \rfloor \text{ mod } 2 )$ denotes the remainder of $\lfloor 3x \rfloor$ divided by $2$. Figure \ref{fig:non_smooth_fun}(a) 
shows the shape of $g_0(x)$.

We consider simulations from model \eqref{eq:scalar_null} with a single covariate and function $f_0(x) = g_0(x)$, 
and from model \eqref{eq:model_two_cov} with two covariates and function $f_0(x_1, x_2) = g_0(x_1) g_0(x_2)$, with sample size $n=200$. 
Figures \ref{fig:non_smooth_fun}(b) and (c) show the empirical distributions of the partial permutation $p$-values, for models \eqref{eq:scalar_null} and \eqref{eq:model_two_cov} respectively, under the five cases with varying imbalance in covariate distributions and group sizes as shown in Table \ref{tab:balance},
which 
demonstrate that the type-I error is still approximately controlled. 
Note that, with two-dimensional covariates, the distributions of the partial permutation $p$-values are quite different from Unif$(0,1)$, and the $p$-values appear to be slightly conservative at significance levels higher than 0.3. 
The reason is that, over all simulations, about $25\%$ of the time the partial permutation test has permutation size 1 and thus results in $p$-value equal to 1. 
Such extreme permutation size is due to the non-smoothness of the underlying functional relationship, under which we lack enough permutation size as well as power for rejecting the null hypothesis. 
This is also intuitive as it is difficult to distinguish whether the multiple groups in comparison share the same functional relation if the underlying function is very non-smooth. 
In such cases, a conservative $p$-value is preferred so as to avoid inflating the type-I error. 

\begin{figure}[h]
	\centering
		\includegraphics[width=0.8\linewidth]{plots/null_balance_legend_1031}
	\begin{subfigure}{.25\textwidth}
		\centering
		\includegraphics[width=1\linewidth]{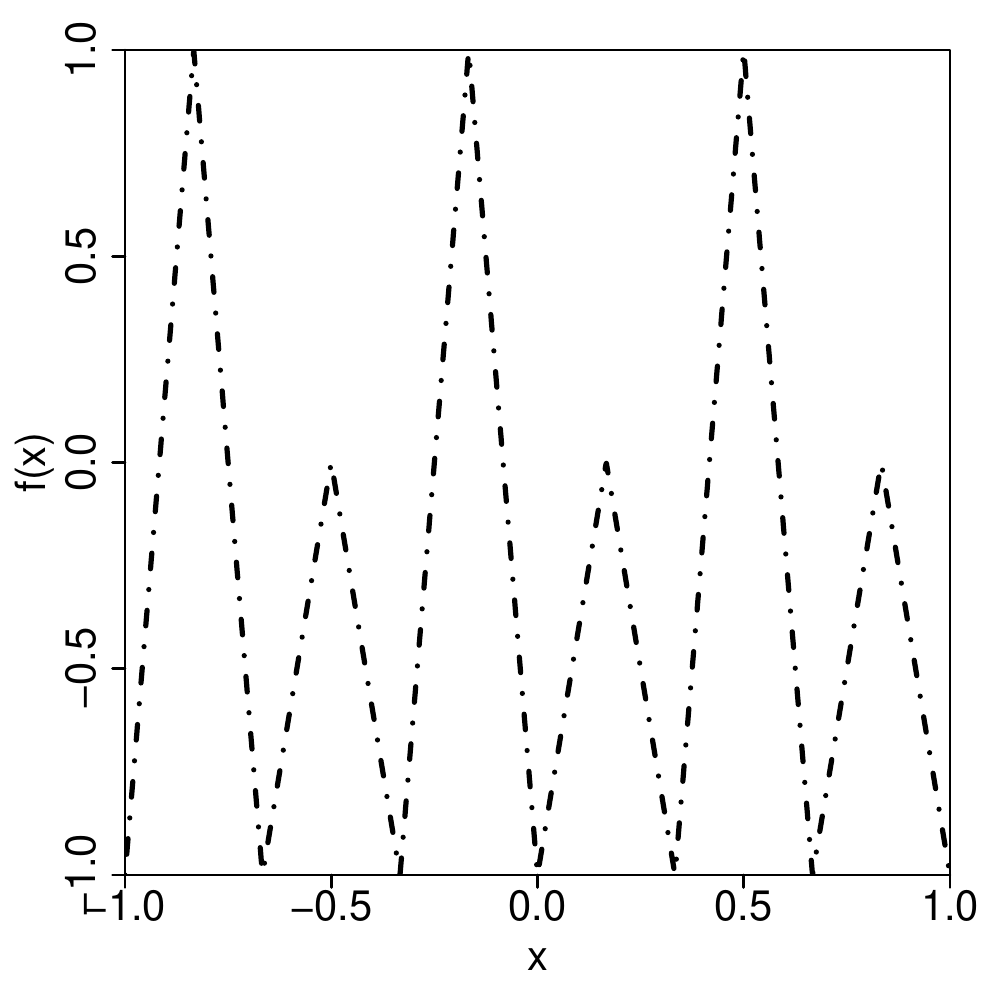}
		\caption{}
	\end{subfigure}%
	\begin{subfigure}{.25\textwidth}
		\centering
		\includegraphics[width=1\linewidth]{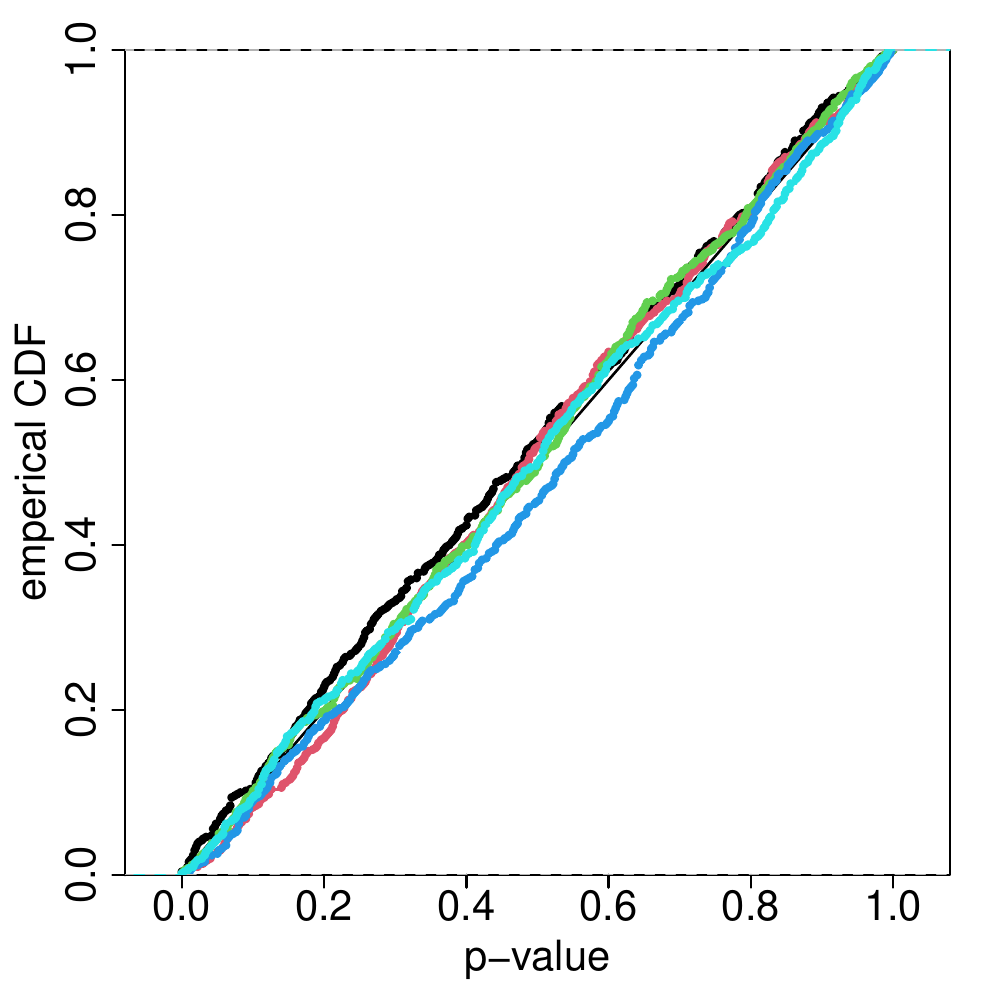}
		\caption{}
	\end{subfigure}%
	\begin{subfigure}{.25\textwidth}
		\centering
		\includegraphics[width=1\linewidth]{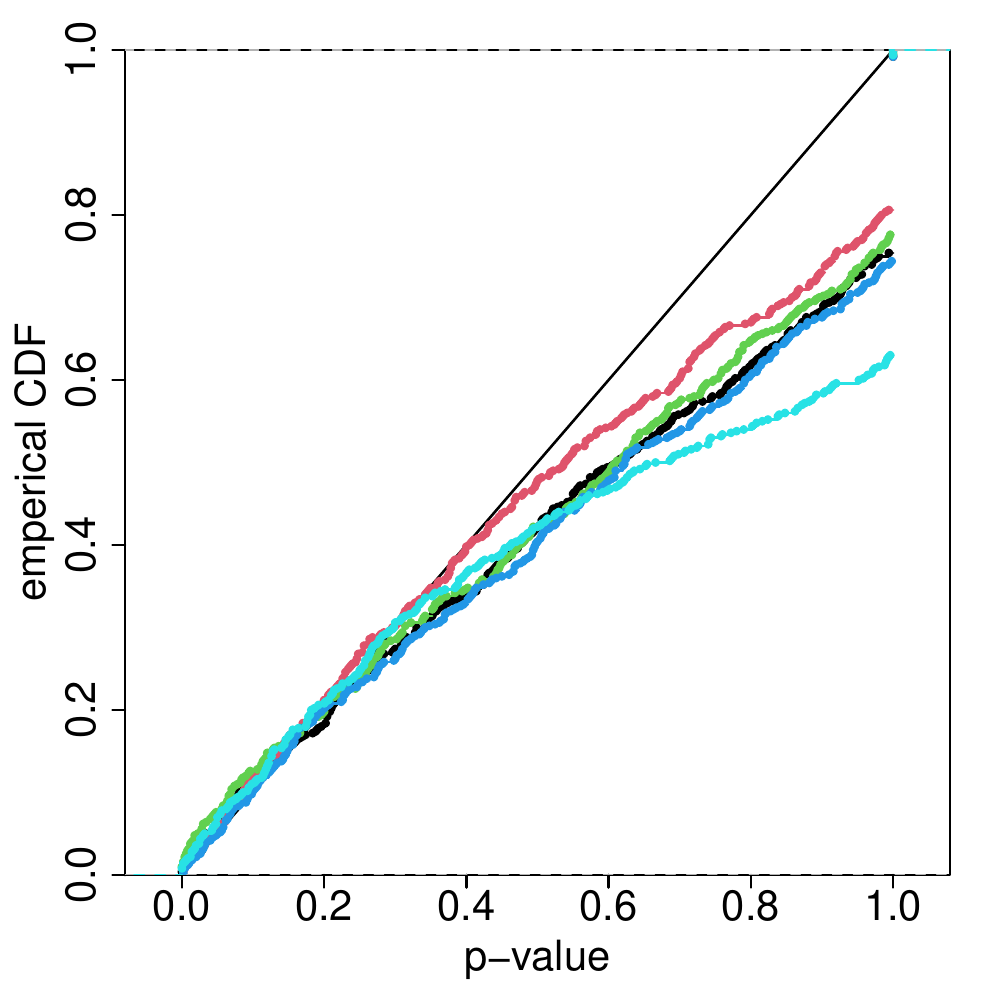}
		\caption{}
	\end{subfigure}
	\caption {(a) Plot  of function $g_0(x)$ in \eqref{eq:non_smooth_fun}.
(b) and (c): Empirical distributions of the partial permutation $p$-values when data are generated from models \eqref{eq:scalar_null} with underlying function $g_0(x)$, 
	and model \eqref{eq:model_two_cov} with the underlying function $g_0(x_1)g_0(x_2)$, respectively. 
    }\label{fig:non_smooth_fun}
\end{figure}

\subsection{Power comparison with the classical F-test under the alternative  hypotheses}\label{sec:power_Ftest}
We generate $\iid$ samples from  the following two data generating scenarios (under alternative hypotheses) with one- and two-dimensional covariates:
\begin{align}\label{eq:alternative_scalar}
    \text{Scenario 3:} \quad & Y = f_{Z}(X) + \varepsilon,\ \  \varepsilon \mid X, Z \sim \mathcal{N}(0,\sigma_0^2),
    \nonumber
    \\
    & X \mid Z \sim \text{Unif}(-1,1),
    \nonumber
    \\
    & P(Z = h)=p_{h}, \quad h=1,2,
\end{align}
and %
\begin{align}\label{eq:alternative_two_dim}
    \text{Scenario 4:} \quad & Y = f_{Z}(X_1, X_2) + \varepsilon, \ \  \varepsilon \mid \bs{X}, Z \sim \mathcal{N}(0,\sigma_0^2), \nonumber \\
& X_{k} \mid Z \sim \text{Unif}(-1,1), \quad k=1,2, \nonumber \\
& X_{1}\ind X_{2} \mid Z, \ \ P(Z=h) = p_h, \quad h=1,2.
\end{align}
For Scenario 3, we consider the following three choices of $(f_1,f_2)$:
\begin{align}\label{eq:alter_f_scalar}
\begin{tabular}{lll}
$\text{(i)}~ f_1 = 1+x,$ \qquad \quad
&
$f_2 = 2+3x,$ 
\\
$\text{(ii)}~ f_1  = 1/3 + x/2,$ \qquad \quad
&
$f_2 = (x+1)^2/4,$ 
\\
$\text{(iii)}~ f_1 = 1/3 + x/2,$ \qquad \quad
&
$f_2 = 1/5 + x/2 - x^4 + x^2;$ 
\end{tabular}
\end{align}
For Scenario 4, 
we consider the following three choices of $(f_1,f_2)$: 
\begin{align}\label{eq:alter_f_two_scalar}
\begin{tabular}{lll}
$\text{(iv)}~ f_1 = 1 + x_1 + x_2,$ \qquad 
\quad
&
$f_2 = 2 + 3x_1 + x_2,$ 
\\
$\text{(v)}~ f_1  = 1/3 + x_1/2 + x_2/2,$ \qquad \quad
&
$f_2 = (x_1+1)^2/4 + (x_2+1)^2/4 - 1/3,$ 
\\
$\text{(vi)}~ f_1 = 1/3 + x_1/2 + x_2/2,$ \qquad \quad
&
$f_2 = 1/3 + x_1/2 + x_2/2 + \sin(\pi x_1) \cdot \sin(\pi x_2).$ 
\end{tabular}
\end{align}

We conduct partial permutation test using either the Gaussian or polynomial kernels. 
For the Gaussian kernel, we 
consider three choices of test statistics, 
the likelihood ratio \eqref{eq:lik_ratio} of $\tilde{H}_1$ against $\tilde{H}_0$, the pseudo likelihood ratio of $\tilde{H}_{\text{pseudo}}$ against $\tilde{H}_0$, 
and \eqref{eq:test_mse} based on the mean squared errors from the pooled and group-specific kernel regression, 
and choose the permutation size based on $\tilde{H}_0$ as discussed in Section \ref{sec:choice_size}. 
For polynomial kernels, we consider degree $p$ of 1, 2 and 3, use the likelihood ratio of the model where the underlying functions are polynomial of degree up to $p$ and can vary across groups against that with the same polynomial function of degree up to $p$ across all groups, and choose the permutation size based on Theorem \ref{poly_permutation_pval_valid}. 
We also consider the classical F-test or equivalently the likelihood ratio test for whether the functions for different groups are the same polynomial function of degree $p$, for $p=1,2,3$. 
Here, the F-test is considered to be most powerful as long as the polynomial regression model is true within each group and does not include unnecessary higher order terms.

\begin{figure}[h]
	\centering
		\includegraphics[width=0.8\linewidth]{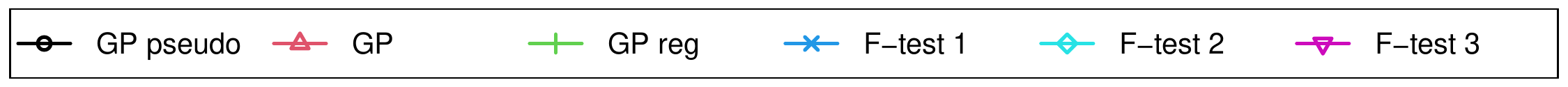}
	\begin{subfigure}{.25\textwidth}
		\centering
		\includegraphics[width=1\linewidth]{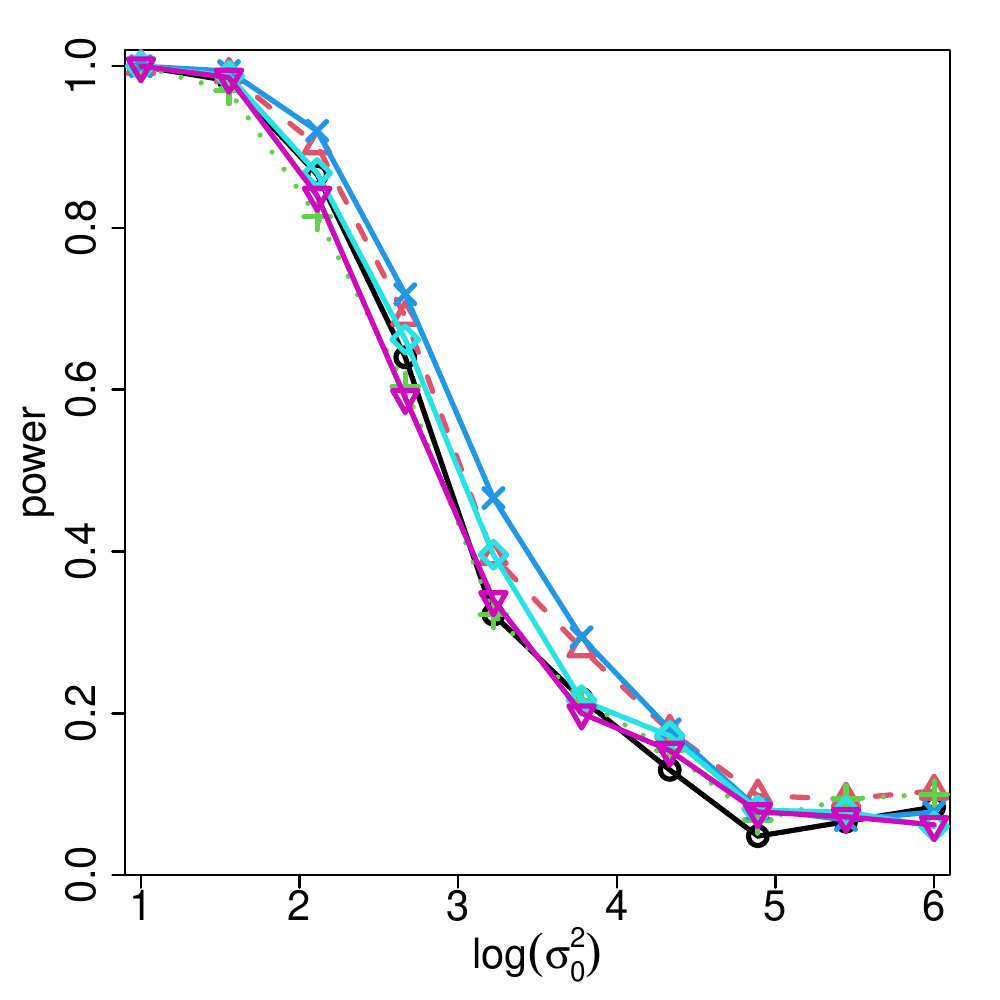}
		\caption*{(i)}
	\end{subfigure}%
	\begin{subfigure}{.25\textwidth}
		\centering
		\includegraphics[width=1\linewidth]{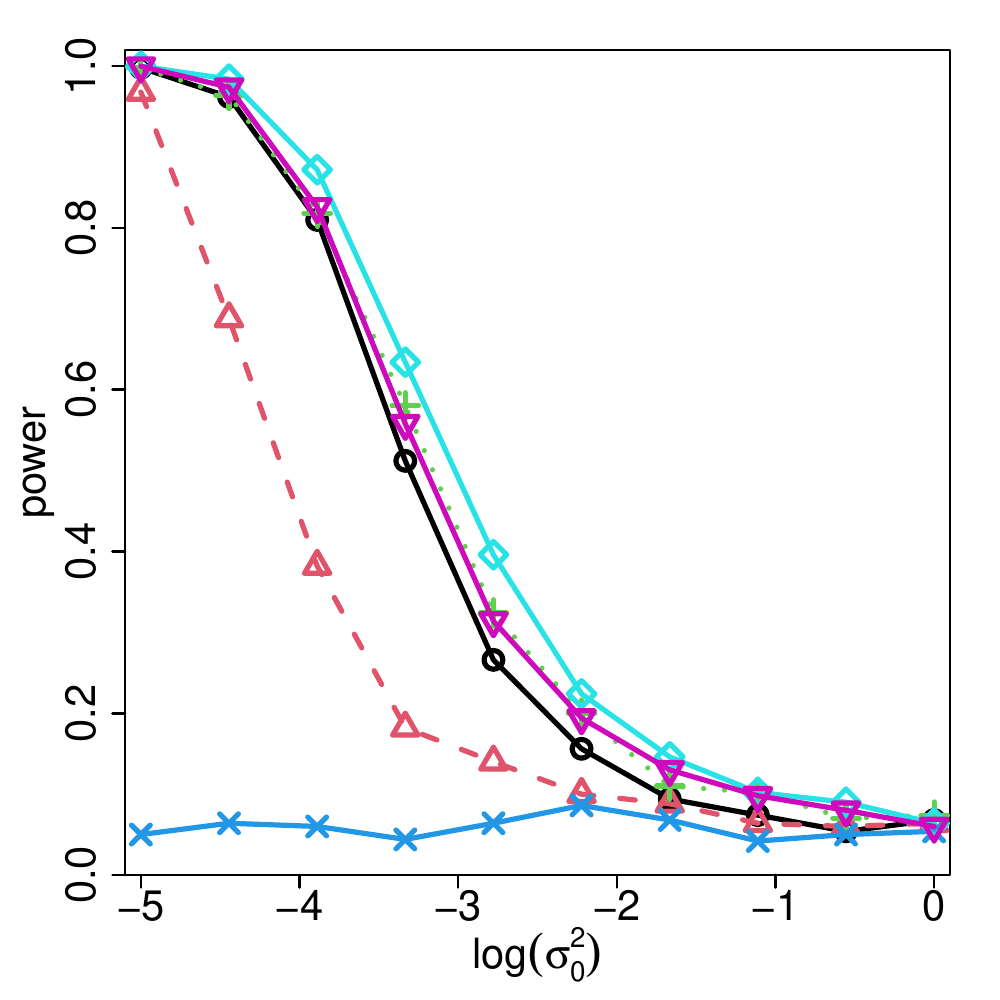}
		\caption*{(ii)}
	\end{subfigure}%
	\begin{subfigure}{.25\textwidth}
		\centering
		\includegraphics[width=1\linewidth]{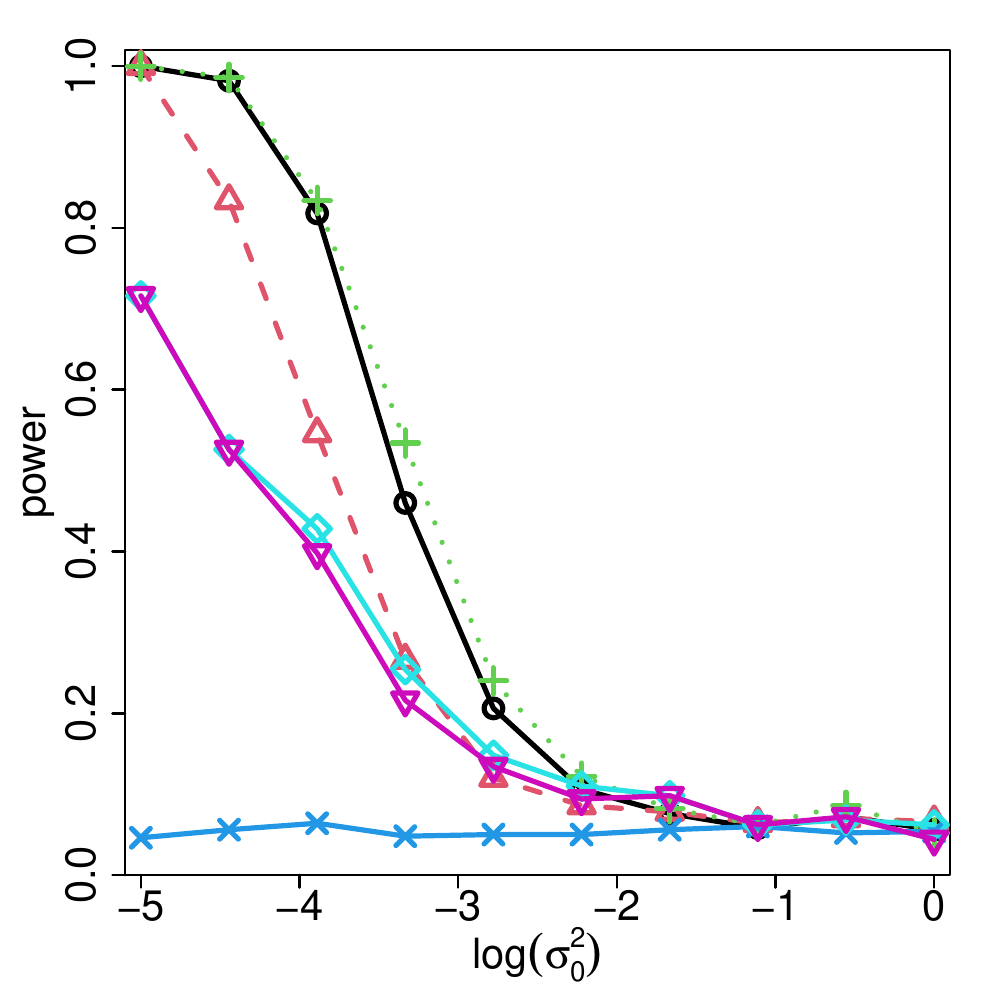}
		\caption*{(iii)}
	\end{subfigure}
	\begin{subfigure}{.25\textwidth}
		\centering
		\includegraphics[width=1\linewidth]{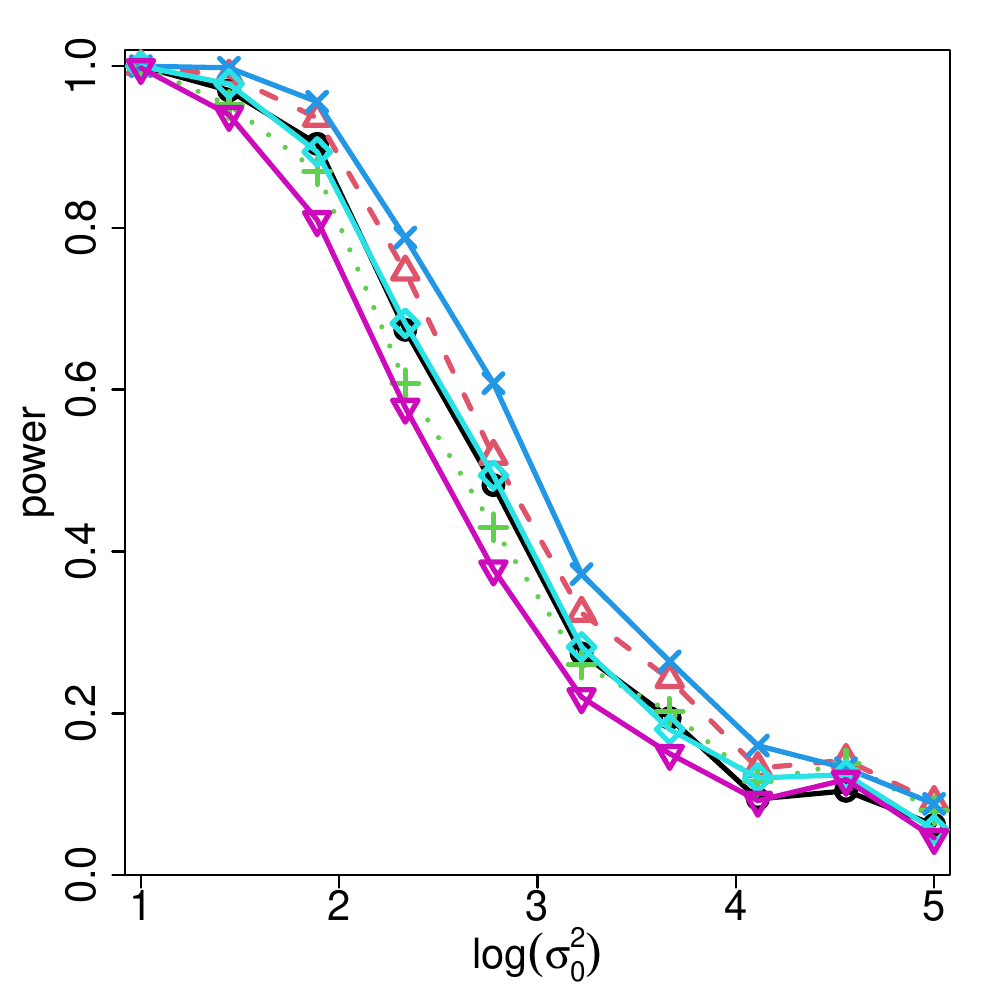}
		\caption*{(iv)} 
	\end{subfigure}
	\begin{subfigure}{.25\textwidth}
		\centering
		\includegraphics[width=1\linewidth]{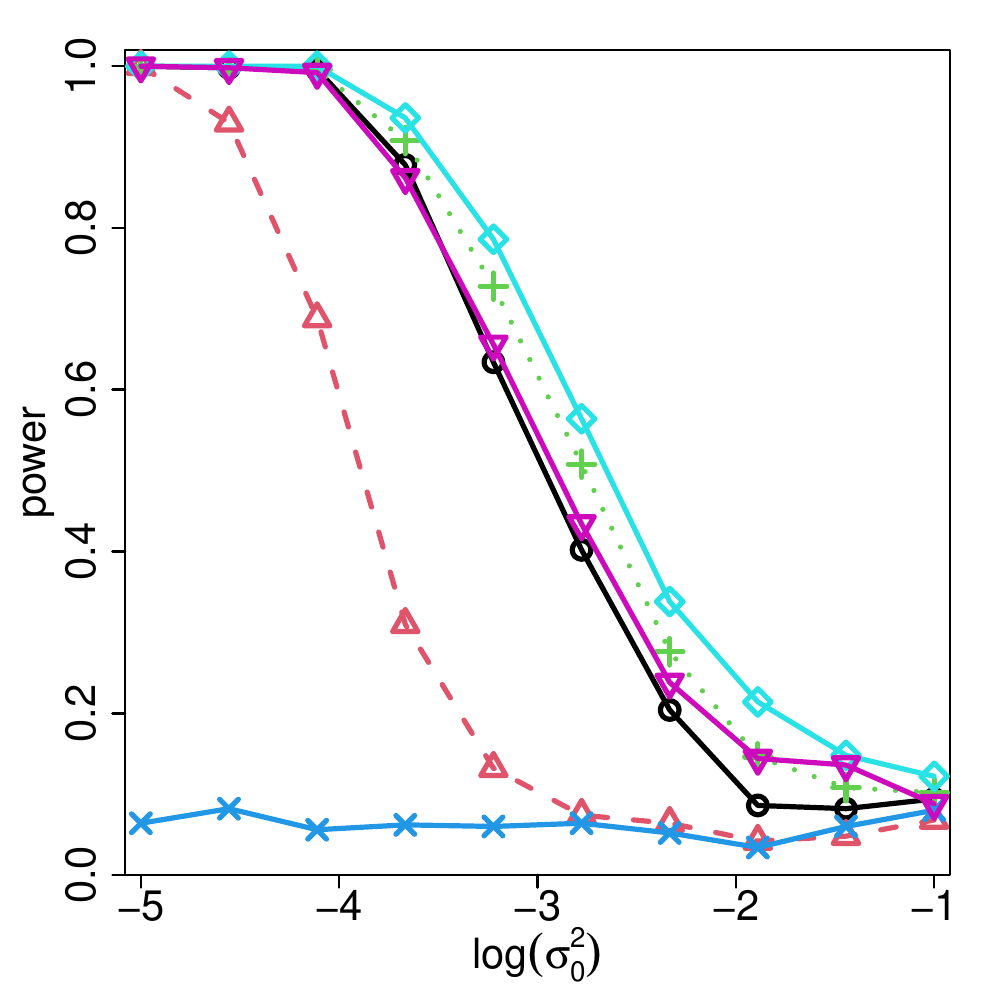}
		\caption*{(v)} 
	\end{subfigure}%
	\begin{subfigure}{.25\textwidth}
	    \centering
	    \includegraphics[width=1\linewidth]{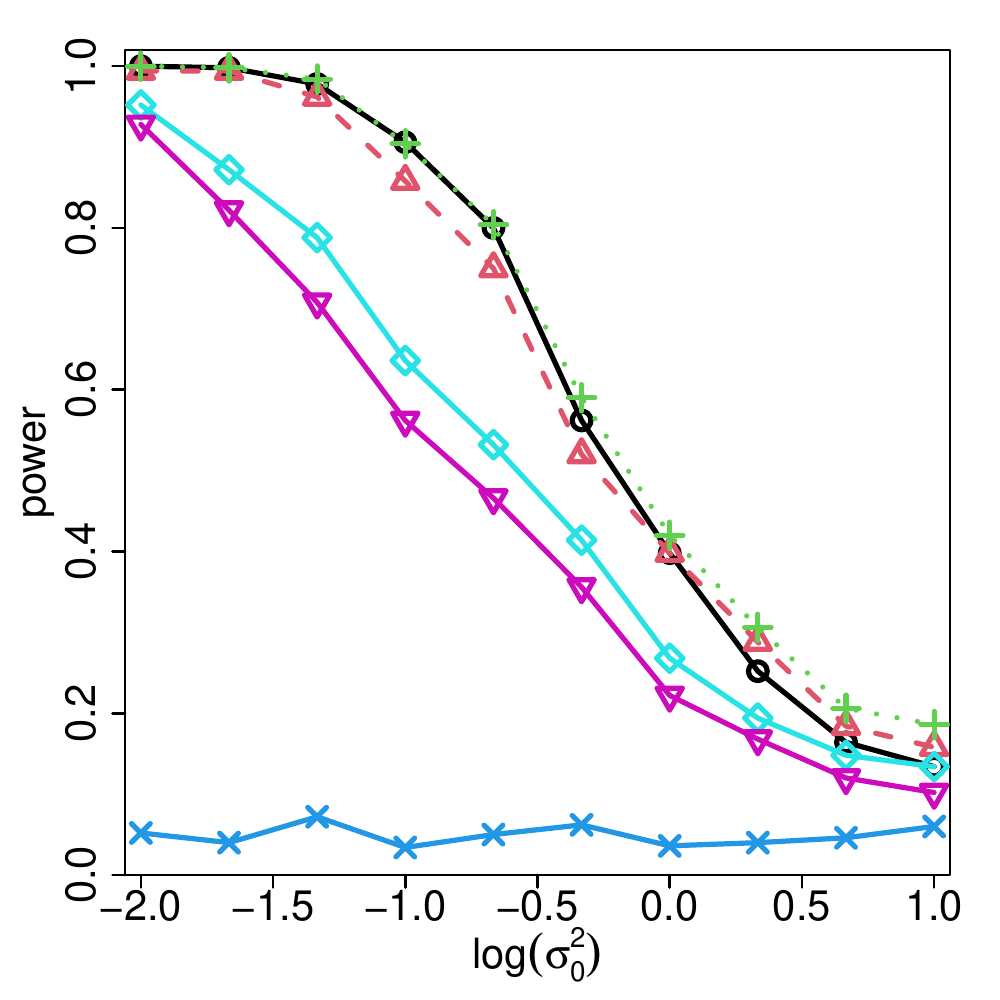}
	\caption*{(vi)}
	\end{subfigure}
	\caption{
	Power of the partial permutation tests when data are generated from Scenario 3 in \eqref{eq:alternative_scalar} and Scenario 4 in \eqref{eq:alternative_two_dim} with sample size $n=200$. 
	The six figures correspond to six choices of the underlying functions $f_1$ and $f_2$ in \eqref{eq:alter_f_scalar} and \eqref{eq:alter_f_two_scalar}. The partial permutation tests here use Gaussian kernel with three choices of test statistics, 
	the likelihood ratio \eqref{eq:lik_ratio} of $\tilde{H}_1$ against $\tilde{H}_0$ (denoted by GP), the pseudo likelihood ratio of $\tilde{H}_{\text{pseudo}}$ against $\tilde{H}_0$ (denoted by GP pseudo), 
    and \eqref{eq:test_mse} based on the mean squared errors (denoted by GP reg). 
    The F-tests test whether the functions for different groups are the same polynomial functions of degree p (denoted by F-test $p$), for $p=1,2,3$.
	}\label{S34_0_5_GP}
\end{figure}

Figure \ref{S34_0_5_GP} 
shows the power of different tests. 
Since the partial permutation tests using polynomial kernels have almost the same power as the corresponding F-tests,
which is not surprising given Theorem \ref{thm:general_F_test}, they are omitted in Figure \ref{S34_0_5_GP}.
As shown in Figures \ref{S34_0_5_GP}(i), (ii), (iv) and (v), 
when the underlying functions are indeed polynomial, 
the F-test with the correct degrees of freedom is the most powerful one. 
However, as suggested by Figures \ref{S34_0_5_GP}(ii) and (v), if we fail to include some higher order terms, it is possible that the F-tests have almost no power to detect the functional heterogeneity across two groups. 
Furthermore, the powers of the partial permutation test using  the Gaussian kernel 
with either test statistic \eqref{eq:test_mse} or the pseudo likelihood ratio statistic are similar and are also close to that of the corresponding most powerful F-test, although the gap seems to increase with the dimension of the covariates. They both performed better than that with the likelihood ratio statistic of $H_1$ versus $H_0$, 
partly because the former two consider different noise variances in different groups. 
Finally, 
as shown in Figures \ref{S34_0_5_GP}(iii) and (vi), 
when the underlying functions contain either higher-order or non-polynomial terms, 
the partial permutation test using Gaussian kernel can have a much higher power than 
the classical F-test.

\subsection{Power comparison with other nonparametric methods under balanced covariates}\label{sec:power_parallel}

Our partial permutation test focuses on whether samples from different groups share the same functional relationship.
This is closely related to the literature focusing on whether different groups share parallel functional relationship \citep{wu2012parallel, xin2020}. 
Specifically, with exactly balanced covariates as in \eqref{eq:balanced} and centered response within each group (assuming the true average function values within each group is known), 
the groups in comparison shares parallel functional relation if and only if they share the same functional relation. 
Following \citet{xin2020}, we generate data from the following model:  
\begin{align}\label{eq:parallel_model}
    \text{Scenario 5:} \quad & Y_i = f_{Z_i}(X_i) + \varepsilon_i, \quad  \varepsilon_i \mid X_i, Z_i \ \overset{\iid}{\sim} \ \mathcal{N}(0,\sigma_0^2),
    \nonumber
    \\
    & X_{i} \equiv X_{n/2+i} \ \overset{\iid}{\sim} \  \text{Unif}(0,1), \quad (1\le i \le n/2)
    \nonumber
    \\
    & Z_1 = \ldots = Z_{n/2} = 1, \quad 
    Z_{n/2+1} =\ldots = Z_{n} = 2. 
\end{align}
and consider the following two choices in which the two groups share neither the same nor parallel functional relationships: 
\begin{align}\label{eq:parallel_f_choice}
\begin{tabular}{lll}
$\text{(i)}~ f_1 = 2.5 \cdot \sin (3\pi x) \cdot (1-x) - m_1,$ \qquad \quad
&
$f_2 = 3.5 \cdot \sin (3\pi x) \cdot (1-x) - m_2,$ 
\\
$\text{(ii)}~ f_1  = 2.5 \cdot \sin (3\pi x) \cdot (1-x) - m_1,$ \qquad \quad
&
$f_2 = 2.5 \cdot \sin(3.4 \pi x) \cdot (1-x) - m_3,$ 
\end{tabular}
\end{align}
To make the comparison fairer, we choose  constants $m_1, m_2$ and $m_3$ such that each function has mean zero, i.e., $\E (f_k(X)) = 0$ with $X\sim \text{Unif}(0, 1)$, which helps avoid the partial permutation test to gain additional power by the mean shift. 
Tables \ref{tab:parallel_1} and \ref{tab:parallel_2} show the power of the partial permutation test using the pseudo likelihood ratio as the test statistic and that of the minimax nonparametric parallelism test  in \citet{xin2020}, which was shown to be superior to other %
tests in the literature under similar simulation settings. 
For Tables \ref{tab:parallel_1} and \ref{tab:parallel_2}, we let the
sample size $n=200$, 
 the noise level $\sigma_0^2$ vary in [0.01, 4.5], and the significance level be fixed at 0.05. 

Tables \ref{tab:parallel_1} and \ref{tab:parallel_2} show that, although the parallelism test has a better power, its type-I error is significantly inflated. In contrast,
the partial permutation test controls its type-I errors well at the nominal level.   
After correcting the type-I error by using the $0.05$ quantile of the null distribution (i.e., the functions in both groups are the same, as $f_1$ in \eqref{eq:parallel_f_choice}) of the $p$-value as the threshold, the power of the two tests %
becomes similar. 
We further increase the sample size to $n=500$ and $1000$. 
As shown in Table \ref{tab:n5001000_parallel},
type-I errors of  the partial permutation test are always well controlled, whereas those  of the parallelism test are still inflated but are closer to the nominal level as the sample size increases. The two tests always have similar powers after the type-I error correction.

\begin{table}[htb]
    \centering
    \caption{
    Comparison between the parallelism and partial permutation tests. Data are generated from model
    in \eqref{eq:parallel_model} with functions in \eqref{eq:parallel_f_choice} (i) and sample size $n=200$. The heading row indicates  various noise levels.
    PPT and
    PPT+Parallel refer to the partial permutation tests  using the pseudo likelihood ratio and the minus $p$-values from the parallelism test, respectively, as test statistics. 
    }\label{tab:parallel_1}
    \resizebox{\textwidth}{!}{
    \begin{tabular}{cccccccccccc}
    \toprule
    Result & Method & 
    0.01 & 0.50 & 1.00 & 1.50 & 2.00 & 2.50 & 3.00 & 3.50 & 4.00 & 4.50
    \\
    \midrule
    Power & Parallel & 
     1.000 & 0.920 & 0.624 & 0.484 & 0.378 & 0.312 & 0.300 & 0.270 & 0.244 & 0.196
    \\
    under
    &  PPT & 
    1.000 & 0.824 & 0.482 & 0.342 & 0.280 & 0.216 & 0.200 & 0.156 & 0.172 & 0.152
    \\
    $H_1$ &  PPT + Parallel & 
    1.000 & 0.836 & 0.490 & 0.352 & 0.256 & 0.188 & 0.206 & 0.168 & 0.162 & 0.114
    \\
    \midrule
    Type I  & Parallel & 
    0.146 & 0.130 & 0.138 & 0.090 & 0.112 & 0.070 & 0.110 & 0.092 & 0.112 & 0.078
    \\
    error 
    &  PPT  &  
    0.044 & 0.056 & 0.058 & 0.046 & 0.046 & 0.036 & 0.044 & 0.042 & 0.050 & 0.042
    \\
    under $H_0$ &  PPT + Parallel & 
    0.062 & 0.062 & 0.054 & 0.040 & 0.064 & 0.040 & 0.036 & 0.040 & 0.040 & 0.032
    \\
    \midrule
    Corrected & Parallel & 
    1.000 & 0.838 & 0.476 & 0.376 & 0.242 & 0.238 & 0.230 & 0.204 & 0.186 & 0.140
    \\
    Power &  PPT & 
    1.000 & 0.812 & 0.468 & 0.346 & 0.296 & 0.266 & 0.216 & 0.192 & 0.172 & 0.166
    \\
    under $H_1$ &  PPT + Parallel & 
    1.000 & 0.798 & 0.484 & 0.394 & 0.242 & 0.248 & 0.230 & 0.202 & 0.176 & 0.150
    \\
    \bottomrule
    \end{tabular}%
    }
\end{table}

Note that the partial permutation test allows for an arbitrary choice of the test statistic. As shown in Tables~\ref{tab:parallel_1} and \ref{tab:parallel_2}, we also 
use
the minus $p$-value from the parallelism test as our test statistic. 
From Tables \ref{tab:parallel_1} and \ref{tab:parallel_2}, the resulting type-I error is well controlled and the power is similar to the original parallelism test after correcting the inflated type-I error.  In practice, however, 
such 
type-I error corrections cannot be easily achieved since  the underlying true functions are unknown.  We may use the distribution from the partial permutation as a reference null distribution to calibrate the $p$-value from the parallelism test. 

Similar to other permutation-based method, our partial permutation test relies on permutations to generate the reference distribution instead of a closed-form asymptotic approximation, and thus requires more computation. 
Averaging over all simulations for Tables \ref{tab:parallel_1} and \ref{tab:parallel_2} with $n=200$, 
the parallelism test, ``PPT'', and ``PPT+Parellel'' took   $0.39$, $34.57$, and  $269.17$ seconds, respectively.
For Table \ref{tab:n5001000_parallel} with sample size $n=500$ and $1000$, on average, the parallelism test took $3.05$ and $21.29$ seconds, while 
the ``PPT''
took $61.23$ and $404.22$ seconds. 
The issue of computational cost for the permutation method can be mitigated by parallelizing the calculation of the test statistic over permutations.

\begin{table}[htb]
    \centering
    \caption{Comparison between the parallelism  and partial permutation tests. Data are generated from  \eqref{eq:parallel_model} with functions in \eqref{eq:parallel_f_choice}(ii) and sample size $n=200$. 
    The description of the table is the same as that of Table \ref{tab:parallel_1}. 
    }\label{tab:parallel_2}
    \resizebox{\textwidth}{!}{
    \begin{tabular}{cccccccccccc}
    \toprule
    Result & Method & 
    0.01 & 0.50 & 1.00 & 1.50 & 2.00 & 2.50 & 3.00 & 3.50 & 4.00 & 4.50
    \\
    \midrule
    Power & Parallel & 
    1.000 & 0.886 & 0.642 & 0.440 & 0.354 & 0.286 & 0.262 & 0.218 & 0.208 & 0.188
    \\
    under
    &  PPT & 
    1.000 & 0.764 & 0.482 & 0.288 & 0.202 & 0.160 & 0.178 & 0.098 & 0.122 & 0.092
    \\
    $H_1$ &  PPT + Parallel & 
    1.000 & 0.764 & 0.486 & 0.290 & 0.232 & 0.170 & 0.156 & 0.122 & 0.126 & 0.102
    \\
    \midrule
    Type I  & Parallel & 
    0.144 & 0.128 & 0.140 & 0.092 & 0.108 & 0.076 & 0.112 & 0.096 & 0.110 & 0.082
    \\
    error 
    &  PPT & 
    0.046 & 0.056 & 0.050 & 0.050 & 0.044 & 0.034 & 0.044 & 0.036 & 0.050 & 0.048
    \\
    under $H_0$ &  PPT + Parallel & 
    0.062 & 0.062 & 0.056 & 0.042 & 0.058 & 0.042 & 0.036 & 0.040 & 0.038 & 0.038
    \\
    \midrule
    Corrected & Parallel & 
    1.000 & 0.764 & 0.480 & 0.308 & 0.200 & 0.198 & 0.172 & 0.144 & 0.146 & 0.138
    \\
    Power &  PPT & 
    1.000 & 0.752 & 0.482 & 0.288 & 0.224 & 0.206 & 0.190 & 0.126 & 0.122 & 0.100
    \\
    under $H_1$ &  PPT + Parallel & 
    1.000 & 0.730 & 0.474 & 0.320 & 0.192 & 0.210 & 0.176 & 0.136 & 0.142 & 0.156
    \\
    \bottomrule
    \end{tabular}}%
\end{table}

\begin{table}[htb]
    \centering
    \caption{Comparison between the parallelism and partial permutation tests with sample sizes $n=500$ and $1000$. Data are generated from \eqref{eq:parallel_model} with functions  in \eqref{eq:parallel_f_choice}(i).
    The description of the table is the same as that of Table \ref{tab:parallel_1}, except that here we do not consider ``PPT+Parallel''. %
    }\label{tab:n5001000_parallel}
    \resizebox{\textwidth}{!}{
    \begin{tabular}{ccccccccccccc}
    \toprule
    $n$ & Result & Method & 
    0.01 & 0.50 & 1.00 & 1.50 & 2.00 & 2.50 & 3.00 & 3.50 & 4.00 & 4.50
    \\
    \midrule[0.8pt]
    $500$  & Power & Parallel & 1.000 & 1.000 &0.958 &0.838 &0.688 &0.624 &0.534 &0.452 &0.392 &0.384
    \\
    & under $H_1$
    &  PPT & 1.000 & 1.000 &0.892 &0.718 &0.580 &0.492 &0.424 &0.350 &0.300 &0.262
    \\
    \midrule
    & Type I error  & Parallel & 0.078 &0.086 &0.098 &0.082 &0.086 &0.082 &0.096 &0.068 &0.092 &0.09
    \\
    & under $H_0$
    &  PPT & 0.044 &0.054 &0.042 &0.046 &0.058 &0.070 &0.064 &0.050 &0.038 &0.05
    \\
    \midrule
    & Corrected Power  & Parallel & 1.000 &0.998 &0.936 &0.786 &0.612 &0.500 &0.462 &0.35 &0.290 &0.322
    \\
    & under $H_1$ &  PPT & 1.000 &0.998 &0.898 &0.732 &0.542 &0.442 &0.402 &0.35 &0.336 &0.282
    \\
    \midrule[0.8pt]

    $1000$ & Power & Parallel & 1.000 & 1.000 & 1.000 & 0.988 & 0.946 & 0.886 & 0.836 & 0.724 & 0.716 & 0.634
    \\
    & under $H_1$
    &  PPT & 1.000 & 1.000 & 1.000 & 0.964 & 0.892 & 0.812 & 0.752 & 0.652 & 0.604 & 0.494
    \\
    \midrule
    & Type I error  & Parallel & 0.064 & 0.078 & 0.074 & 0.072 & 0.080 & 0.078 & 0.074 & 0.088 & 0.074 & 0.066
    \\
    & under $H_0$
    &  PPT & 0.034 & 0.036 & 0.066 & 0.064 & 0.046 & 0.050 & 0.042 & 0.062 & 0.036 & 0.044
    \\
    \midrule
    & Corrected Power  & Parallel & 1.000 & 1.000 & 1.000 & 0.982 & 0.918 & 0.858 & 0.820 & 0.652 & 0.674 & 0.586
    \\
    & under $H_1$ &  PPT & 1.000 & 1.000 & 1.000 & 0.962 & 0.894 & 0.812 & 0.786 & 0.604 & 0.648 & 0.520
    \\
    \bottomrule
    \end{tabular}}%
\end{table}

\section{Application}\label{sec:app}
We apply the partial permutation test to a data set analyzed in \citet{pardo2007testing}, 
which consists of monthly expenditures of several Dutch households and the numbers of members in each households. The data set includes accumulated expenditures on food and total expenditures over the year (October 1986 to September 1987) for households with two members (159 in total), three members (45 in total) and four members (73 in total).

\begin{table}[htbp]
    \centering
    \caption{Partial permutation test $p$-values for comparing relationships between the expenditure on food and the total expenditure among households with different numbers of members.}
    \label{tab:real_data}
    \begin{tabular}{lccccccccc}
        \toprule
        Test statistic & \multicolumn{4}{c}{Comparison} & & \multicolumn{4}{c}{Comparison after truncation}
        \\
        \cline{2-5} \cline{7-10}
        & $(2, 3, 4)$ & $(2, 3)$ & $(3,4)$ & $(2,4)$ & &  $(2, 3, 4)$ & $(2, 3)$ & $(3,4)$ & $(2,4)$ \\
        \midrule
        \eqref{eq:lik_ratio} with $\tilde{H}_1$ vs $\tilde{H}_0$ & 0.002 & 0.050 & 0.908 & 0 & & 0.006 & 0.048 & 0.780 & 0\\
     \eqref{eq:lik_ratio} with $\tilde{H}_1'$ vs $\tilde{H}_0$ & 0 & 0.030 & 0.664 & 0.002 & & 0.004 & 0.064 & 0.860 & 0\\
        \eqref{eq:lik_ratio} with $\tilde{H}_{\text{pseudo}}$ vs $\tilde{H}_0$ & 0.002 & 0.006 & 0.498 & 0.002 & & 0 & 0.006 & 0.532 & 0\\
        \eqref{eq:test_mse}  & 0 & 0.002 & 0.346 & 0 & & 0 & 0 & 0.440 & 0
        \\
        \bottomrule
    \end{tabular}
\end{table}

Let $Y$ be the logarithm of the expenditure on food, 
 $X$ be the logarithm of the total expenditure, 
and $Z$ be the number of house members minus one (indicating the size of a family). 
To compare the relationship between $Y$ and $X$ among the three groups  defined by $Z$, 
we use the partial permutation test with Gaussian kernel 
after standardizing both covariates and outcomes, 
and choose the permutation size based on model $\tilde{H}_0$
as suggested in Section \ref{sec:choice_size} 
at significance level $\alpha=0.05$. 
We first test whether the same functional relationship between $Y$ and $X$ holds across all three groups, and then perform pairwise comparisons. 
Table \ref{tab:real_data} shows the resulting $p$-values using different test statistics, 
including the likelihood ratio statistics in \eqref{eq:lik_ratio} of $\tilde{H}_1$, $\tilde{H}_1'$ and $\tilde{H}_{\text{pseudo}}$ against $\tilde{H}_0$, and the test statistic \eqref{eq:test_mse} based on mean squared errors from pooled and group-specific kernel regression. 
It is very interesting to observe from
Table \ref{tab:real_data}  that
the relationship between $X$ and $Y$  differs significantly
between ``no-kid'' households (size=2) and larger-sized ones.
However, between the households of size  3 and those of size 4, the relationships between $X$ and $Y$ are not significantly different. 
To avoid potential sensitivity to heavy-tailed errors in the data, we also conducted the tests after truncating extreme fitted residuals; see the Supplementary Material for details.
 Table \ref{tab:real_data} shows that the conclusions are consistent across different test statistics, and are robust to the use of truncation. 
Our results confirm the findings in \citet{pardo2007testing}.

\section{Discussion}\label{sec:dicuss}

We developed a partial permutation test for comparing across different groups the functional relationship between a response variable and some covariates, 
and studied its properties under null models \eqref{eq:H_0_fixed} and \eqref{eq:h0_gp} when the underlying function either is fixed or follows a Gaussian process. 
The key idea of the proposed tests is to keep invariant the projection of the response vector onto the space spanned by leading principle components of the kernel matrix, and permute the remaining (residual) part.
Practically, we can also accommodate
multiple kernels by
conducting a partial permutation test that retains the projections of the response vector on the leading principle components of multiple kernel matrices. 
For example, 
if we use both the polynomial kernel of degree $p$ and the Gaussian kernel, 
then the partial permutation test is exactly valid when the underlying function is polynomial up to degree $p$ as implied by Theorem \ref{poly_permutation_pval_valid}, 
and  also has nice properties with flexible underlying functions as implied by Theorems \ref{thm:fixed_property_theorem}, \ref{thm:general_ppt_valid} and \ref{thm:gp_finite_property_theorem}. 
Furthermore, based on the simulation studies, we suggest to use test statistics based on a comparison between 
the null GPR model and its pseudo alternative as in \eqref{eq:h_ps}, 
or a comparison between mean squared errors from pooled and group-specific kernel regressions. 
These test statistics are easy to calculate and have a superior power.

Our testing procedure is also related to Bayesian model checking, especially the conditional predictive $p$-value proposed by
\citet{bayarri1997measures, bayarri1999quantifying, bayarri2000p}. 
The authors generated predictive samples from the model with parameters following the prior distribution, but only kept those samples that have the same value of a summary statistic $U$ as the observed data. Then, they compared the test statistic of predictive samples with that of the observed samples. As pointed by \citet{bayarri1999quantifying}, the intuition behind a suitable choice of $U$ is that $U$ should contain as much information about the unknown parameters as possible. In the extreme case where $U$ is chosen to be the sufficient statistic for all parameters, the conditional predictive $p$-value is valid under the model where the parameters are fixed and unknown. 

In our case, although we are considering a nonparametric model \eqref{eq:H_0_fixed}, the idea of conditional predictive $p$-value can still be applied. Suppose the variance of residuals is fixed and known and the underlying function follows a Gaussian process prior. We can 
perform the conditional predictive checking by choosing $U$ to be $(\bXmatrix,\mathcal{S}_y, \bs{Z} )$, where $\mathcal{S}_y$ is from either the discrete or the continuous partial permutation test in Algorithm \ref{alg:partial_permu}. Such choice of $U$ contains information about the smooth components of the underlying function. However, 
it is generally computationally challenging to generate the predictive samples. From Theorem \ref{thm:general_ppt_valid}, under some regularity conditions on the Gaussian process prior, the predictive samples can be asymptotically equivalent to the ones from  partial  permutation, 
and 
the conditional predictive $p$-value can be approximated by the partial permutation $p$-value given the same choice of the test statistic.

In practice, we may face high-dimensional covariates, under which the comparison of functional relation among various groups becomes much more challenging. 
Generally, the permutation size of our partial permutation test decreases as the dimension of covariates increases, and will eventually lose power due to the lack of permutation size. 
This is intuitive due to the nature of the problem: with high-dimensional covariates, the underlying function can have a complex structure making it hard to distinguish whether multiple groups share the same functional relation or not, especially when there are limited sample size and limited overlaps of covariates from different groups. 
The issue may be mitigated by imposing additional structural assumptions, such as sparsity, and we leave it for future work.

\section*{Supplementary Material}
The Supplementary Material contains computation details for 
maximizing the likelihood under GPR models, 
additional simulations for non-Gaussian or correlated noises and the choice of kernel parameters, 
and the proofs of all theorems, corollaries and propositions.

\section*{Acknowledgments}

We thank the Associate Editor and two reviewers for constructive comments. 

\bibliographystyle{plainnat}
\bibliography{bibliography}

\begin{thebibliography}{33}
\providecommand{\natexlab}[1]{#1}
\providecommand{\url}[1]{\texttt{#1}}
\expandafter\ifx\csname urlstyle\endcsname\relax
  \providecommand{\doi}[1]{doi: #1}\else
  \providecommand{\doi}{doi: \begingroup \urlstyle{rm}\Url}\fi

\bibitem[Bayarri and Berger(1997)]{bayarri1997measures}
M.~J. Bayarri and J.~O. Berger.
\newblock Measures of surprise in bayesian analysis.
\newblock \emph{Duke university institute of statistics and decision sciences
  working paper}, \penalty0 (97-46), 1997.

\bibitem[Bayarri and Berger(1999)]{bayarri1999quantifying}
M.~J. Bayarri and J.~O. Berger.
\newblock Quantifying surprise in the data and model verification.
\newblock \emph{Bayesian statistics}, 6:\penalty0 53--82, 1999.

\bibitem[Bayarri and Berger(2000)]{bayarri2000p}
M.~J. Bayarri and J.~O. Berger.
\newblock P values for composite null models.
\newblock \emph{Journal of the American Statistical Association}, 95\penalty0
  (452):\penalty0 1127--1142, 2000.

\bibitem[Behseta and Kass(2005)]{behseta2005testing}
S.~Behseta and R.~E. Kass.
\newblock Testing equality of two functions using bars.
\newblock \emph{Statistics in medicine}, 24\penalty0 (22):\penalty0 3523--3534,
  2005.

\bibitem[Behseta et~al.(2005)Behseta, Kass, and
  Wallstrom]{behseta2005hierarchical}
S.~Behseta, R.~E. Kass, and G.~L. Wallstrom.
\newblock Hierarchical models for assessing variability among functions.
\newblock \emph{Biometrika}, 92:\penalty0 419--434, 2005.

\bibitem[Benavoli and Mangili(2015)]{benavoli2015gaussian}
A.~Benavoli and F.~Mangili.
\newblock Gaussian processes for bayesian hypothesis tests on regression
  functions.
\newblock In \emph{Proceedings of the Eighteenth International Conference on
  Artificial Intelligence and Statistics}, pages 74--82, 2015.

\bibitem[Branson et~al.(2019)Branson, Rischard, Bornn, and
  Miratrix]{BRANSON201914}
Z.~Branson, M.~Rischard, L.~Bornn, and L.~W. Miratrix.
\newblock A nonparametric bayesian methodology for regression discontinuity
  designs.
\newblock \emph{Journal of Statistical Planning and Inference}, 202:\penalty0
  14 -- 30, 2019.

\bibitem[Braun(2006)]{braun2006accurate}
M.~L. Braun.
\newblock Accurate error bounds for the eigenvalues of the kernel matrix.
\newblock \emph{The Journal of Machine Learning Research}, 7:\penalty0
  2303--2328, 2006.

\bibitem[Braun et~al.(2008)Braun, Buhmann, and
  M{\~A}{\v{z}}ller]{braun2008relevant}
M.~L. Braun, J.~M. Buhmann, and K.-R. M{\~A}{\v{z}}ller.
\newblock On relevant dimensions in kernel feature spaces.
\newblock \emph{Journal of Machine Learning Research}, 9:\penalty0 1875--1908,
  2008.

\bibitem[Cavazza and Murino(2016)]{cavazza2016active}
J.~Cavazza and V.~Murino.
\newblock Active regression with adaptive huber loss.
\newblock \emph{arXiv preprint arXiv:1606.01568}, 2016.

\bibitem[Christmann and Steinwart(2007)]{christmann2007consistency}
A.~Christmann and I.~Steinwart.
\newblock Consistency and robustness of kernel-based regression in convex risk
  minimization.
\newblock \emph{Bernoulli}, pages 799--819, 2007.

\bibitem[{Degras} et~al.(2012){Degras}, {Xu}, {Zhang}, and
  {Wu}]{wu2012parallel}
D.~{Degras}, Z.~{Xu}, T.~{Zhang}, and W.~B. {Wu}.
\newblock Testing for parallelism among trends in multiple time series.
\newblock \emph{IEEE Transactions on Signal Processing}, 60:\penalty0
  1087--1097, 2012.

\bibitem[Dempster et~al.(1977)Dempster, Laird, and Rubin]{rubinEM1977}
A.~P. Dempster, N.~M. Laird, and D.~B. Rubin.
\newblock Maximum likelihood from incomplete data via the em algorithm.
\newblock \emph{Journal of the Royal Statistical Society. Series B
  (Methodological)}, 39:\penalty0 1--38, 1977.

\bibitem[Durrett(2019)]{durrett2019probability}
Rick Durrett.
\newblock \emph{Probability: theory and examples}, volume~49.
\newblock Cambridge university press, 2019.

\bibitem[Freedman and Peters(1984)]{Freedman1984}
D.~A. Freedman and S.~C. Peters.
\newblock Bootstrapping a regression equation: Some empirical results.
\newblock \emph{Journal of the American Statistical Association}, 79:\penalty0
  97--106, 1984.

\bibitem[Hahn et~al.(2001)Hahn, Todd, and der Klaauw]{Hahn2001}
J.~Hahn, P.~Todd, and W.~Van der Klaauw.
\newblock Identification and estimation of treatment effects with a
  regression-discontinuity design.
\newblock \emph{Econometrica}, 69:\penalty0 201--209, 2001.
\newblock ISSN 00129682, 14680262.
\newblock URL \url{http://www.jstor.org/stable/2692190}.

\bibitem[Hastie and Zhu(2006)]{penalty2006trevor}
T.~Hastie and J.~Zhu.
\newblock Comment.
\newblock \emph{Statistical Science}, 21:\penalty0 352--357, 2006.

\bibitem[Hinkley(1988)]{Hinkley1988}
D.~V. Hinkley.
\newblock Bootstrap methods.
\newblock \emph{Journal of the Royal Statistical Society: Series B
  (Methodological)}, 50:\penalty0 321--337, 1988.

\bibitem[Imbens and Lemieux(2008)]{Imbensrdd2008}
G.~W. Imbens and T.~Lemieux.
\newblock Regression discontinuity designs: A guide to practice.
\newblock \emph{Journal of Econometrics}, 142:\penalty0 615 -- 635, 2008.

\bibitem[K{\"u}hn(1987)]{kuhn1987eigenvalues}
T.~K{\"u}hn.
\newblock Eigenvalues of integral operators generated by positive definite
  h{\"o}lder continuous kernels on metric compacta.
\newblock In \emph{Indagationes Mathematicae (Proceedings)}, volume~90, pages
  51--61. Elsevier, 1987.

\bibitem[Liu and Cheng(2018)]{Liu2018early}
M.~Liu and G.~Cheng.
\newblock Early stopping for nonparametric testing.
\newblock In S.~Bengio, H.~Wallach, H.~Larochelle, K.~Grauman, N.~Cesa-Bianchi,
  and R.~Garnett, editors, \emph{Advances in Neural Information Processing
  Systems}, volume~31. Curran Associates, Inc., 2018.

\bibitem[Meng(1994)]{meng1994posterior}
X.-L. Meng.
\newblock Posterior predictive p-values.
\newblock \emph{The Annals of Statistics}, pages 1142--1160, 1994.

\bibitem[Micchelli et~al.(2006)Micchelli, Xu, and
  Zhang]{micchelli2006universal}
C.~A. Micchelli, Y.~Xu, and H.~Zhang.
\newblock Universal kernels.
\newblock \emph{The Journal of Machine Learning Research}, 7:\penalty0
  2651--2667, 2006.

\bibitem[Neumeyer and Dette(2003)]{neumeyer2003nonparametric}
N.~Neumeyer and H.~Dette.
\newblock Nonparametric comparison of regression curves: an empirical process
  approach.
\newblock \emph{The Annals of Statistics}, 31:\penalty0 880--920, 2003.

\bibitem[Pardo-Fern{\'a}ndez et~al.(2007)Pardo-Fern{\'a}ndez, Van~Keilegom, and
  Gonz{\'a}lez-Manteiga]{pardo2007testing}
J.~C. Pardo-Fern{\'a}ndez, I.~Van~Keilegom, and W.~Gonz{\'a}lez-Manteiga.
\newblock Testing for the equality of k regression curves.
\newblock \emph{Statistica Sinica}, 17\penalty0 (3):\penalty0 1115, 2007.

\bibitem[Raskutti et~al.(2014)Raskutti, Wainwright, and Yu]{raskutti2014early}
Garvesh Raskutti, Martin~J Wainwright, and Bin Yu.
\newblock Early stopping and non-parametric regression: an optimal
  data-dependent stopping rule.
\newblock \emph{The Journal of Machine Learning Research}, 15\penalty0
  (1):\penalty0 335--366, 2014.

\bibitem[Rasmussen and Williams(2006)]{rasmussen2006gaussian}
C.~E. Rasmussen and C.~K.~I. Williams.
\newblock \emph{Gaussian Processes for Machine Learning}.
\newblock The MIT Press, 2006.

\bibitem[Rischard et~al.(2018)Rischard, Branson, Miratrix, and
  Bornn]{rischard2018bayesian}
M.~Rischard, Z.~Branson, L.~Miratrix, and L.~Bornn.
\newblock A bayesian nonparametric approach to geographic regression
  discontinuity designs: Do school districts affect nyc house prices?
\newblock \emph{arXiv preprint arXiv:1807.04516}, 2018.

\bibitem[Shang and Cheng(2013)]{shang2013local}
Z.~Shang and G.~Cheng.
\newblock Local and global asymptotic inference in smoothing spline models.
\newblock \emph{The Annals of Statistics}, 41:\penalty0 2608--2638, 2013.

\bibitem[Shi and Choi(2011)]{shi2011gaussian}
J.~Q. Shi and T.~Choi.
\newblock \emph{Gaussian process regression analysis for functional data}.
\newblock CRC Press, 2011.

\bibitem[Thistlethwaite and Campbell(1960)]{thistlethwaite1960regression}
D.~L. Thistlethwaite and D.~T. Campbell.
\newblock Regression-discontinuity analysis: An alternative to the ex post
  facto experiment.
\newblock \emph{Journal of Educational psychology}, 51:\penalty0 309, 1960.

\bibitem[Wang(2005)]{wang2005support}
L.~Wang.
\newblock \emph{Support vector machines: theory and applications}, volume 177.
\newblock Springer Science \& Business Media, 2005.

\bibitem[Xing et~al.(2020)Xing, Liu, Ma, and Zhong]{xin2020}
X.~Xing, M.~Liu, P.~Ma, and W.~Zhong.
\newblock Minimax nonparametric parallelism test.
\newblock \emph{Journal of Machine Learning Research}, 21:\penalty0 1--47,
  2020.

\end{thebibliography}

\newpage

\setcounter{equation}{0}
\setcounter{section}{0}
\setcounter{figure}{0}
\setcounter{example}{0}
\setcounter{proposition}{0}
\setcounter{corollary}{0}
\setcounter{theorem}{0}
\setcounter{table}{0}

\renewcommand {\theproposition} {A\arabic{proposition}}
\renewcommand {\theexample} {A\arabic{example}}
\renewcommand {\thefigure} {A\arabic{figure}}
\renewcommand {\thetable} {A\arabic{table}}
\renewcommand {\theequation} {A\arabic{equation}}
\renewcommand {\thelemma} {A\arabic{lemma}}
\renewcommand {\thesection} {A\arabic{section}}
\renewcommand {\thetheorem} {A\arabic{theorem}}
\renewcommand {\thecorollary} {A\arabic{corollary}}

\setcounter{page}{1}
\renewcommand {\thepage} {A\arabic{page}}

\begin{center}
	\bf \LARGE 
	Supplementary Material 
\end{center}

Appendix \ref{sec:mle_gpr} investigates the computation of maximum likelihood estimates under GPR models. 

Appendix \ref{sec:simu_kerpara} conducts simulations to study the choice of kernel parameters.  

Appendix \ref{sec:simu_nonGauss} conducts simulations for non-Gaussian noises. 

Appendix \ref{sec:simu_corr} conducts simulations for correlated noises. 

Appendix \ref{sec:proof_fix_null} studies validity of the partial permutation test with a fixed unknown underlying function, including the proofs of Theorems \ref{linear_permutation_pval_valid}--\ref{thm:fixed_property_theorem} and 
Corollaries \ref{cor:kernel_finite_dim_feature_space}-- \ref{cor:fixed_property_theorem_balanced_design}. 

Appendix \ref{sec:proof_GPR_null} studies validity of the partial permutation test under the GPR model, including the proofs of Proposition \ref{thm:general_consistency} and Theorems \ref{thm:general_ppt_valid} and \ref{thm:gp_finite_property_theorem}. 

Appendix \ref{sec:proof_alter} studies the partial permutation test under alternative hypotheses, including the proofs of Theorems \ref{thm:general_F_test} and \ref{thm:power_diverg_kernel}. 

Appendix \ref{sec:proof_kernel} studies properties of kernels and proves Proposition \ref{prop:decay_eigenvalue}. 

Appendix \ref{sec:proof_corr} studies validity of the partial permutation test with a fixed unknown underlying function and correlated noises, including the proofs of Theorem \ref{thm:fixed_property_theorem_corr} and Corollaries \ref{cor:kernel_finite_dim_feature_space_corr}--\ref{cor:fixed_property_theorem_balanced_design_corr}. 

Appendix \ref{sec:proof_corr_GPR} studies validity of the partial permutation test under the GPR model with correlated noises, including the proofs of \ref{thm:general_ppt_valid_corr} and \ref{thm:gp_finite_property_theorem_corr}.

\section{Computation details for the maximum likelihood estimate under Gaussian process regression models}\label{sec:mle_gpr}

We consider here the computation of the maximum likelihood estimates under null, alternative, and "pseudo" alternative GPR models. Note that  conditional distribution of the response given the covariates and group indicators for each of these models can be written as:
\begin{align}\label{eq:vcm_all_model}
\bs{Y}  \mid \bXmatrix, \bs{Z}, \tilde{H}_0 & \sim \mathcal{N}\left(0,\frac{\delta_{0}^2}{n^{1-\gamma}} \bs{K}_{n}+\sigma_{0}^{2}\bs{I}_n\right),\\\nonumber
\bs{Y}  \mid  \bXmatrix, \bs{Z}, \tilde{H}_1 & \sim \mathcal{N}\left(0,\frac{\delta_{0}^2}{n^{1-\gamma}} \bs{K}_{n}+
\sum_{h=1}^H\frac{\delta_{h}^2}{n^{1-\gamma}} \bs{K}_{n}^{(h)} + 
\sigma_{0}^{2}\bs{I}_n\right),\\\nonumber
\bs{Y}  \mid  \bXmatrix, \bs{Z}, \tilde{H}_1' & \sim \mathcal{N}\left(0,\frac{\delta_{0}^2}{n^{1-\gamma}} \bs{K}_{n}+
\sum_{h=1}^H \frac{\delta_{h}^2}{n^{1-\gamma}} \bs{K}_{n}^{(h)} + 
\sum_{h=1}^H \sigma_{h}^{2}\bs{I}_n^{(h)}\right),  \\\nonumber
\bs{Y}  \mid \bXmatrix, \bs{Z},  \tilde{H}_{\text{pseudo}} & \sim \mathcal{N}\left(0,
\sum_{h=1}^H\frac{\delta_{h}^2}{n^{1-\gamma}} \bs{K}_{n}^{(h)} + 
\sigma_{0}^{2}\bs{I}_n\right), 
\end{align}
where $\bs{K}_n$ is the kernel matrix with $[\bs{K}_n]_{ij} = K(\bs{X}_i, \bs{X}_j)$, 
$\bs{K}_n^{(h)}$ is the kernel matrix for group $h$ with 
$[\bs{K}_n^{(h)}]_{ij} = 
\I(Z_i = h, Z_j = h) \cdot [\bs{K}_n]_{ij}$,  
$\bs{I}_n$ is an $n \times n$ identity matrix, 
and $\bs{I}_n^{(h)}$ is an $n \times n$ matrix with 
$[\bs{I}_n^{(h)}]_{ij} = \I(Z_i = h) \cdot \I(i=j)$. 
Therefore, we only need to consider maximizing likelihoods for the following two types of variance component models (VCM),
\begin{align*}
\text{VCM I:} \quad & \bs{Y}  \sim \mathcal{N}(0, \tau_1^2\bs{G}_1 + \tau_2^2 \bs{I}),\\\nonumber
\text{VCM II:} \quad & \bs{Y}  \sim \mathcal{N}(0, \sum_{j=1}^J \tau_j^2 \bs{G}_j).
\end{align*}
For VCM I, we can directly use the EM algorithm. However, for VCM II, the EM algorithm may suffer from slow convergence. Thus, we propose to use Newton's method %
with Fisher scoring by solving a quadratic programming problem at each iterative step. We will describe these algorithms in details in the following subsections.

\subsection{The EM Algorithm for VCM I \& II}
Consider the more general variance component model VCM II. Let $\bs{Y} = \sum_{j=1}^J \bs{\xi}^j$, where $\bs{\xi}^j \sim \mathcal{N}(0, \tau_j^2 G_j)$ for $1\leq j \leq J$. We can treat $\{ \bs{\xi}^j: j=1,2,\cdots, J-1\}$ as missing data and use EM algorithm, as summarized in Algorithm \ref{alg:em}.
\begin{algorithm}
\caption{EM algorithm for VCM II}\label{alg:em}
\begin{itemize}
\item[1)] Suppose at $t$-step, we have $\{\tau_{j,t}^2: j=1,2,\cdots,J\}$.
\item[2)] The iteration of EM algorithm is given by:
\begin{align*}
\tau_{j,t+1}^{2}=\frac{1}{n}\left[\text{tr}(\bs{G}_{j}^{-1}\Lambda_{j,t})+\mu_{j,t}^\top\bs{G}_{j}^{-1}\mu_{j,t}\right], \quad 1\leq j\leq J.
\end{align*}
where
\begin{align*}
\mu_{j,t} = & \tau_{j,t}^{2}\bs{G}_{j}\left(\sum_{j=1}^{J}\tau_{j,t}^{2}\bs{G}_{j}\right)^{-1} \bs{Y}, \quad 1\leq j \leq J,\\\nonumber
\Lambda_{j,t} = & \tau_{j,t}^{2}\bs{G}_{j}-\tau_{j,t}^{2}\bs{G}_{j}\left(\sum_{j=1}^{J}\tau_{j,t}^{2}\bs{G}_{j}\right)^{-1}\tau_{j,t}^{2}\bs{G}_{j}, \quad 1\leq j \leq J.
\end{align*}
\end{itemize}
\end{algorithm}
There are high computational costs since we need to invert %
an 
$n\times n$ matrix at each iterative step, and the convergence of the EM algorithm is slow under $\tilde{H}_1$ and $\tilde{H}'_1$. 
However, for  VCM I, which corresponds to model $H_0$, the EM iterative step can be greatly simplified and is computationally very efficient, as described in Algorithm \ref{alg:em_0}. 
\begin{algorithm}
\caption{EM algorithm for VCM I}\label{alg:em_0}
\begin{itemize}
\item[1)] Suppose at $t$-step, we have $\{\tau_{1,t}^2, \tau_{2,t}^2 \}$.
\item[2)] Perform eigen-decompostion on $\bs{G}_1$, i.e. $\bs{G}_1=\bs{V} \bs{D} \bs{V}^\top$, and let $\bs{U} = \bs{V}^\top \bs{Y}$. Then the iteration of the EM algorithm is
\begin{align*}
\tau^2_{1,t+1} &= \frac{1}{n}\left[ \text{tr}\left( \bs{G}_1^{-1}\bs{\Lambda}_{1,t}  \right) + \mu_{1,t}^\top\bs{G}_1^{-1} \mu_{1,t}  \right],\\\nonumber
\tau^2_{2,t+1} &= \frac{1}{n} \left[\text{tr}\left(\bs{\Lambda}_{2,t}\right)+ \mu_{2,t}^\top\mu_{2,t} \right],
\end{align*}
where
\begin{align*}
\mu_{1,t}^\top\bs{G}_1^{-1} \mu_{1,t} &= (\tau_{1,t}^2)^2\bs{U}^\top(\tau_{1,t}^2 \bs{D}+\tau_{2,t}^2\bs{I})^{-1}\bs{D}(\tau_{1,t}^2\bs{D}+\tau_{2,t}^2\bs{I})^{-1}\bs{U},\\\nonumber
\mu_{2,t}^\top\mu_{2,t} &= (\tau_{2,t}^2)^2\bs{U}^\top(\tau_{1,t}^2\bs{D}+\tau_{2,t}^2\bs{I})^{-2}\bs{U},\\\nonumber
\text{tr}\left( \bs{G}_1^{-1}\bs{\Lambda}_{1,t}  \right) &= \tau_{2,t}^2 \cdot \text{tr}\left\{ \bs{D}^{-1}\left[   \bs{I}-\tau_{2,t}^2(\tau_{1,t}^2\bs{D}+\tau_{2,t}^2\bs{I})^{-1} \right]  \right\},\\\nonumber
\text{tr}\left(\bs{\Lambda}_{2,t}\right) &= \tau_{2,t}^2 \cdot \text{tr}\left\{   \bs{I}-\tau_{2,t}^2(\tau_{1,t}^2\bs{D}+\tau_{2,t}^2\bs{I})^{-1}  \right\}.
\end{align*}
\end{itemize}
\end{algorithm}
Note that $\bs{D}$ is a diagonal matrix, under which the calculation for matrix multiplication and inversion can be greatly simplified. 
For the "pseudo" alternative model $\tilde{H}_{\text{pseudo}}$, the corresponding likelihood  can be decomposed into $H$ components corresponding to samples in the $H$ groups, each of which is of type VCM I.
Therefore, the maximum likelihood under $\tilde{H}_{\text{pseudo}}$ can be obtained by repeating Algorithm \ref{alg:em_0} %
$H$ times for the $H$ groups of samples.

\subsection{Newton's Method for VCM II}
We use $l(\tau^2_{\cdot})$ to denote the log-likelihood of $\tau^2_{\cdot} = (\tau^2_1, \tau^2_2, \cdots, \tau^2_J)$ under model VCM II. Given $\tau^2_{\cdot,t}$ at $t$-step and by Taylor expansion $l(\tau^2_{\cdot})$ at $\tau^2_{\cdot,t}$, we have
\begin{align*}
l(\tau^2_{\cdot})
= & l(\tau^2_{\cdot,t})+\left(\frac{\partial l(\tau^2_{\cdot,t})}{\partial\tau^2_\cdot}\right)^\top(\tau^2_{\cdot}-\tau^2_{\cdot,t})+\frac{1}{2}(\tau^2_{\cdot}-\tau^2_{\cdot,t})^\top\left(\frac{\partial^{2}l(\tau^2_{\cdot,t})}{\partial\tau^2_\cdot\partial{\tau^{2}_{\cdot}}^\top}\right)(\tau^2_{\cdot}-\tau^2_{\cdot,t}) +o(\|\tau^2_{\cdot}-\tau^2_{\cdot,t}\|^{2})\\\nonumber
\approx & l_2(\tau^2_\cdot | \tau^2_{\cdot,t}) + o(\|\tau^2_{\cdot}-\tau^2_{\cdot,t}\|^{2}).
\end{align*}
Newton's method solves the equation:
$
{\partial l_2(\tau^2_\cdot | \tau^2_{\cdot,t})}/{\partial \tau^2_\cdot}=0.
$ 
However, the solution may not satisfying the constraint that all the coordinates are nonnegative. Since $l_2$ is an approximation of $l$ and our goal is to find $\tau^2_\cdot$ that maximizes $l$, we want to find $\tau^2_\cdot$ that maximizes $l_2(\tau^2_\cdot | \tau^2_{\cdot,t})$ subject to that all the coordinates of $\tau^2_\cdot$ are nonnegative. 
Because the Hessian matrix $-{\partial^{2}l(\tau^2_{\cdot,t})}/{\partial\tau^2_\cdot\partial{\tau^{2}_{\cdot}}^\top}$ may not be positive semidefinite, 
simply maximizing $l_2(\tau^2_\cdot | \tau^2_{\cdot,t})$ may produce maximizer with infinite coordinates. 
Therefore, 
we instead use Fisher information
$\mathbb{E}\left(-\frac{\partial^{2}l(\tau^2_{\cdot,t})}{\partial\tau^2_\cdot\partial{\tau^{2}_{\cdot}}^\top}\right)$ 
to replace
$-\frac{\partial^{2}l(\tau^2_{\cdot,t})}{\partial\tau^2_\cdot\partial{\tau^{2}_{\cdot}}^\top}$, 
that is, Fisher scoring. 
Consequently, the optimization at $t$-step becomes 
\begin{align*}\nonumber
\tau^2_{\cdot,t+1} 
& =\arg\max_{\tau^2_{\cdot}\geq 0}
l(\tau^2_{\cdot,t})+\left(\frac{\partial l(\tau^2_{\cdot,t})}{\partial\tau^2_\cdot}\right)^\top(\tau^2_{\cdot}-\tau^2_{\cdot,t})+\frac{1}{2}(\tau^2_{\cdot}-\tau^2_{\cdot,t})^\top \mathbb{E}\left(\frac{\partial^{2}l(\tau^2_{\cdot,t})}{\partial\tau^2_\cdot\partial{\tau^{2}_{\cdot}}^\top}\right)(\tau^2_{\cdot}-\tau^2_{\cdot,t})\\\nonumber
& = \arg\min_{\tau^2_{\cdot}\succcurlyeq 0}
\left(-\frac{\partial l(\tau^2_{\cdot,t})}{\partial\tau^2_\cdot}\right)^\top(\tau^2_{\cdot}-\tau^2_{\cdot,t})+\frac{1}{2}(\tau^2_{\cdot}-\tau^2_{\cdot,t})^\top \mathbb{E}\left(-\frac{\partial^{2}l(\tau^2_{\cdot,t})}{\partial\tau^2_\cdot\partial{\tau^{2}_{\cdot}}^\top}\right)(\tau^2_{\cdot}-\tau^2_{\cdot,t}),
\end{align*}
where the constraint $\tau^2_{\cdot}\succcurlyeq 0$ means that each coordinate of $\tau^2_{\cdot}$ is nonnegative. This reduces to a quadratic programming problem, and the procedure is summarized in Algorithm \ref{alg:newton}.
\begin{algorithm}
\caption{Newton's method for VCM II}\label{alg:newton}
\begin{itemize}
\item[1)] Suppose at $t$-step we have $\tau_{\cdot,t}^2 = (\tau_{1,t}^2,\cdots, \tau_{J,t}^2)$.
\item[2)] The iteration of Newton's method with Fisher scoring is obtained by solving the following quadratic programming problem:
\begin{align*}
\tau^2_{\cdot,t+1} = & \arg\min_{\tau^2_{\cdot}\geq 0}
\left(-\frac{\partial l(\tau^2_{\cdot,t})}{\partial\tau^2_\cdot}\right)^\top(\tau^2_{\cdot}-\tau^2_{\cdot,t})+\frac{1}{2}(\tau^2_{\cdot}-\tau^2_{\cdot,t})^\top \mathbb{E}\left(-\frac{\partial^{2}l(\tau^2_{\cdot,t})}{\partial\tau^2_\cdot\partial{\tau^{2}_{\cdot}}^\top}\right)(\tau^2_{\cdot}-\tau^2_{\cdot,t}).
\end{align*}
\end{itemize}
\end{algorithm}

\subsection{Numerical Singularity Issue of Kernel Matrices} The kernel matrix $\bs{K}_n$ for, say, the Gaussian kernel $K_{\text{G}}$ in \eqref{eq:gaussian_kernel} with distinct covariates is theoretically positive definite. However, in practice, the kernel matrix can be numerically singular. 
Thus we propose to use $\bs{K}_{n,s} = \bs{K}_n + s \bs{I}_n$ instead of $\bs{K}_n$, with very small $s$ (say $s=10^{-5}$), to avoid the singularity issue. Note that for model $\tilde{H}_0$ in (\ref{eq:h0_gp}) and $\tilde{H}_1'$ in (\ref{eq:h1prime}), according to (\ref{eq:vcm_all_model}), we have
\begin{align*}
\bs{Y} \mid \bXmatrix, \bs{Z}, \tilde{H}_0 & \sim \mathcal{N}\left(\bs{0},\frac{\delta_{0}^2}{n^{1-\gamma}} \bs{K}_{n,s}+(\sigma_{0}^{2}-s\frac{\delta_{0}^2}{n^{1-\gamma}})\bs{I}_n\right),\\\nonumber
\bs{Y} \mid \bXmatrix, \bs{Z}, \tilde{H}_1' & \sim \mathcal{N}\left(\bs{0},\frac{\delta_{0}^2}{n^{1-\gamma}} \bs{K}_{n,s}+
\sum_{h=1}^H \frac{\delta_{h}^2}{n^{1-\gamma}} \bs{K}_{n,s}^{(h)} %
+ 
\sum_{h=1}^H (\sigma_{h}^{2}-s\frac{\delta_{0}^2}{n^{1-\gamma}} - s\frac{\delta_{h}^2}{n^{1-\gamma}} )\bs{I}_n^{(h)}\right).  
\end{align*}
If the maximum likelihood estimates under model $\tilde{H}_0$ or $\tilde{H}_1'$ satisfy
$
\hat{\sigma}_{0}^{2}-s{\hat{\delta}_{0}^2}/{n^{1-\gamma}} \geq 0
$ 
or
$
\hat{\sigma}_{h}^{2}-s{\hat{\delta}_{0}^2}/{n^{1-\gamma}} - s{\hat{\delta}_{h}^2}/{n^{1-\gamma}} \geq 0,
$ 
then using $\bs{K}_{n,s}$ instead of $\bs{K}_n$ will not change the maximum likelihood estimates under $\tilde{H}_0$ or $\tilde{H}_1'$. 
In contrast, 
under model $\tilde{H}_1$ or $\tilde{H}_{\text{pseudo}}$, using $\bs{K}_{n,s}$ instead of $\bs{K}_n$ may change the maximum likelihood estimates. 
However, 
because the partial permutation test allows for an arbitrary choice of test statistic, 
even if we compute maximum likelihood ratio of $\tilde{H}_1$, $\tilde{H}_1'$, or $\tilde{H}_{\text{pseudo}}$ against $\tilde{H}_0$ with $\bs{K}_n$ replaced by $\bs{K}_{n,s}$, the resulting test can still be asymptotically valid under model $\tilde{H}_0$. 
Moreover, we can view the likelihood ratio statistic using $\bs{K}_{n,s}$ %
as an approximation to the one using $\bs{K}_n$.

\section{Simulation to Study the Choice of Kernel Parameters}\label{sec:simu_kerpara}
In this section we conduct simulations to investigate the choice of kernel parameters for our partial permutation test. 
Specifically, we consider the following two choices of functional relationship between response and covariates within each of the two groups: 
\begin{align}\label{eq:kernel_choice_f_choice}
\begin{tabular}{lll}
$\text{(i)}~ f_1 = 4x^3/3-x/3,$ \qquad \quad
&
$f_2 = 4x^3/3 - x/3 + \delta x,$ 
\\
$\text{(ii)}~ f_1  = \sin (6x),$ \qquad \quad
&
$f_2 = \sin (6x) + \delta x,$ 
\end{tabular}
\end{align}
where $\delta$ characterizes the functional heterogeneity between the two groups. 
We generate the covariates, group indicators and noises in the same way as in \eqref{eq:scalar_null}, and consider both cases (a) and (c) in Table \ref{tab:balance} corresponding to balanced and unbalanced covariate distributions between the two groups. We fix the noise level $\sigma^2 = 0.1$, and vary $\delta$ from $0$ to $1$ by a step of $0.1$. Obviously, $\delta = 0$ corresponds to the case where the null hypothesis $H_0$ in \eqref{eq:H_0_fixed} holds. 
We consider partial permutation test based on Gaussian kernel $K_{\text{G}}(x, x') = \exp\{-\omega(x-x')^2\}$, and investigate the impact of the kernel parameter $\omega$.

We first consider the scenario where the underlying functions have form (i) in \eqref{eq:kernel_choice_f_choice}. 
Tables \ref{tab:f3case1_ker_para} and \ref{tab:f3case3_ker_para} show type-I errors (when $\delta=0$) and powers (when $\delta>0$) of the partial permutation test when the kernel parameter $\omega$ increases from $10^{-4}$ to 10. 
Both tables also show mean squared errors  between the estimated and the   true underlying function under the null hypothesis with $\delta=0$ and various choices of $\omega$, where the estimator is  the posterior mean from the GPR model, as discussed in Section \ref{sec:choice_size}. 
Tables \ref{tab:f3case1_ker_para} and \ref{tab:f3case3_ker_para} show that the mean squared errors are significantly larger when the kernel parameter $\omega$ is much smaller, 
whereas the estimation ls less sensitive to larger $\omega$'s.
In both cases, the type-I error  ($\delta=0$) is well controlled over a wide range of $\omega$, and the test is very conservative when $\omega=10^{-4}$ in case (c). 
Nevertheless, an overly small or large $\omega$ leads to a significant power loss under the alternative ($\delta>0$). 
The last rows of both tables show the average values of $\omega$ from the maximum likelihood approach as discussed in Section \ref{sec:kernel} under the null hypothesis,
as well as the corresponding power under the alternative hypotheses with $\delta>0$. 
From both tables, with the estimated $\omega$, the type-I error  is well controlled, 
and the power is among the best of the $\omega$ values under consideration.

We then consider the scenario where the underlying functions have form (ii) in \eqref{eq:kernel_choice_f_choice}. 
Tables \ref{tab:f6case1_ker_para} and \ref{tab:f6case3_ker_para} show the mean squared error for estimating the underlying function and type-I error when $\delta=0$ and the power of the partial permutation test when $\delta>0$, similar to that in Tables \ref{tab:f3case1_ker_para} and \ref{tab:f3case3_ker_para}. 
The observations are quite close to that from the previous scenario, except that in case (c) small $\omega$ can lead to inflated type-I errors. 
Again, with the estimated kernel parameter $\omega$, the test controls the type-I error well, and its power is among the best in both tables.

\begin{table}[htbp]
    \centering
    \caption{
    Partial permutation test based on Gaussian with varying kernel parameter $\omega$ when the underlying function has the form (i) in \eqref{eq:kernel_choice_f_choice} and the covariates and group indicators are balanced (case (a) in Table \ref{tab:balance}). 
    The 1st column shows the $\omega$ value, where the last row is the average value of $\omega$ when it is estimated based on maximum likelihood. 
    The 2nd column shows the mean squared error for estimating the underlying function under the null hypothesis with $\delta=0$. 
    The 3rd column shows the type-I error of the test under null. 
    The 4th to 13th column show the power of the test when $\delta$ increases from $0.1$ to $1$.}\label{tab:f3case1_ker_para}
    \resizebox{\textwidth}{!}{
    \begin{tabular}{cccccccccccccc}
    \toprule
    $\omega$ & MSE & 0 & 0.1 & 0.2 & 0.3 & 0.4 & 0.5 & 0.6 & 0.7 & 0.8 & 0.9 & 1\\
    \midrule
0.0001&0.0276&0.060&0.114&0.318&0.546&0.706&0.802&0.862&0.870&0.930&0.962&0.98\\
0.0010&0.0276&0.050&0.132&0.404&0.758&0.942&0.984&1.000&1.000&1.000&1.000&1.000\\
0.0020&0.0277&0.056&0.130&0.404&0.746&0.944&0.982&1.000&1.000&1.000&1.000&1.000\\
0.0100&0.0279&0.052&0.136&0.422&0.746&0.934&0.982&0.998&1.000&1.000&1.000&1.000\\
0.0200&0.0261&0.054&0.110&0.384&0.698&0.910&0.972&0.998&1.000&1.000&1.000&1.000\\
0.1000&0.0079&0.064&0.128&0.414&0.826&0.978&0.998&1.000&1.000&1.000&1.000&1.000\\
0.2000&0.0081&0.060&0.122&0.412&0.840&0.974&0.998&1.000&1.000&1.000&1.000&1.000\\
0.3000&0.0084&0.046&0.142&0.412&0.824&0.974&0.998&1.000&1.000&1.000&1.000&1.000\\
0.4000&0.0085&0.054&0.132&0.420&0.822&0.972&0.998&1.000&1.000&1.000&1.000&1.000\\
0.5000&0.0085&0.056&0.128&0.416&0.812&0.968&0.998&1.000&1.000&1.000&1.000&1.000\\
0.6000&0.0085&0.052&0.128&0.422&0.792&0.966&0.998&1.000&1.000&1.000&1.000&1.000\\
1.0000&0.0086&0.062&0.128&0.404&0.792&0.970&0.996&1.000&1.000&1.000&1.000&1.000\\
2.0000&0.0090&0.058&0.112&0.376&0.750&0.962&0.996&1.000&1.000&1.000&1.000&1.000\\
4.0000&0.0094&0.060&0.110&0.346&0.702&0.948&0.996&1.000&1.000&1.000&1.000&1.000\\
10.0000&0.0103&0.052&0.092&0.308&0.666&0.924&0.992&1.000&1.000&1.000&1.000&1.000\\
{\bf 0.7086}&{\bf 0.0089}&
{\bf 0.058}&{\bf 0.138}&{\bf 0.418} &{\bf 0.806} &{\bf 0.972}&{\bf 0.998}&{\bf 1.000 }&{\bf 1.000 }&{\bf 1.000 }&{\bf 1.000 }&{\bf 1.000 }\\
    \bottomrule
    \end{tabular}}%
\end{table}

\begin{table}[htbp]
    \centering
    \caption{
    Partial permutation test based on Gaussian with varying kernel parameter $\omega$ when the underlying function has the form (i) in \eqref{eq:kernel_choice_f_choice} and the covariates are unbalanced (case (c) in Table \ref{tab:balance}). 
    The description of the table is the same as that of Table \ref{tab:f3case1_ker_para}. 
    }\label{tab:f3case3_ker_para}
    \resizebox{\textwidth}{!}{
    \begin{tabular}{cccccccccccccc}
    \toprule
    $\omega$ & MSE & 0 & 0.1 & 0.2 & 0.3 & 0.4 & 0.5 & 0.6 & 0.7 & 0.8 & 0.9 & 1\\
    \midrule
0.0001&0.0276&0.002&0.002&0.018&0.066&0.154&0.266&0.410&0.582&0.700&0.792&0.834\\
0.0010&0.0276&0.056&0.080&0.216&0.428&0.664&0.874&0.960&0.988&0.998&1.000&1.000\\
0.0020&0.0277&0.068&0.098&0.202&0.398&0.636&0.854&0.958&0.990&1.000&1.000&1.000\\
0.0100&0.0278&0.050&0.082&0.182&0.336&0.560&0.776&0.912&0.984&0.996&1.000&1.000\\
0.0200&0.0265&0.048&0.064&0.140&0.274&0.512&0.736&0.902&0.980&1.000&1.000&1.000\\
0.1000&0.0080&0.044&0.096&0.276&0.592&0.850&0.974&0.998&1.000&1.000&1.000&1.000\\
0.2000&0.0083&0.038&0.098&0.292&0.594&0.866&0.976&1.000&1.000&1.000&1.000&1.000\\
0.3000&0.0085&0.044&0.098&0.296&0.608&0.860&0.974&1.000&1.000&1.000&1.000&1.000\\
0.4000&0.0087&0.046&0.092&0.296&0.586&0.856&0.970&1.000&1.000&1.000&1.000&1.000\\
0.5000&0.0087&0.046&0.104&0.280&0.586&0.854&0.972&1.000&1.000&1.000&1.000&1.000\\
0.6000&0.0087&0.050&0.094&0.300&0.596&0.854&0.970&0.998&1.000&1.000&1.000&1.000\\
1.0000&0.0088&0.056&0.088&0.278&0.582&0.844&0.956&0.998&1.000&1.000&1.000&1.000\\
2.0000&0.0091&0.054&0.090&0.254&0.540&0.822&0.944&0.998&1.000&1.000&1.000&1.000\\
4.0000&0.0095&0.062&0.078&0.234&0.516&0.806&0.934&0.998&1.000&1.000&1.000&1.000\\
10.0000&0.0104&0.058&0.072&0.198&0.492&0.748&0.908&0.996&1.000&1.000&1.000&1.000\\
{\bf 0.6453} & {\bf 0.0090} & {\bf 0.046} & {\bf 0.090} & {\bf 0.276} & {\bf 0.606} & {\bf 0.870} & {\bf 0.966} & {\bf 1.000} & {\bf 1.000} & {\bf 1.000} & {\bf 1.000} & {\bf 1.000}\\
    \bottomrule
    \end{tabular}}%
\end{table}

\begin{table}[htbp]
    \centering
    \caption{
    Partial permutation test based on Gaussian with varying kernel parameter $\omega$ when the underlying function has the form (ii) in \eqref{eq:kernel_choice_f_choice} and the covariates and group indicators are balanced (case (a) in Table \ref{tab:balance}). 
    The description of the table is the same as that of Table \ref{tab:f3case1_ker_para}.
    }\label{tab:f6case1_ker_para}
    \resizebox{\textwidth}{!}{
    \begin{tabular}{cccccccccccccc}
    \toprule
    $\omega$ & MSE & 0  & 0.1 &  0.2  & 0.3 &  0.4 &  0.5 &  0.6 &  0.7 &  0.8 &  0.9    & 1
    \\
    \midrule
0.01&0.4440&0.012&0.022&0.038&0.120&0.230&0.336&0.410&0.418&0.442&0.554&0.772\\
0.05&0.3763&0.054&0.082&0.240&0.496&0.598&0.560&0.512&0.520&0.636&0.816&0.948\\
0.10&0.0116&0.042&0.068&0.156&0.364&0.544&0.690&0.784&0.866&0.896&0.948&0.980\\
0.50&0.0039&0.060&0.120&0.382&0.780&0.962&0.996&1.000&1.000&1.000&1.000&1.000\\
1.00&0.0040&0.054&0.114&0.372&0.754&0.954&0.996&1.000&1.000&1.000&1.000&1.000\\
1.50&0.0044&0.056&0.112&0.360&0.724&0.958&0.996&1.000&1.000&1.000&1.000&1.000\\
2.00&0.0047&0.054&0.112&0.350&0.718&0.952&0.996&1.000&1.000&1.000&1.000&1.000\\
2.50&0.0049&0.050&0.112&0.326&0.718&0.952&0.996&1.000&1.000&1.000&1.000&1.000\\
3.00&0.0051&0.056&0.108&0.316&0.696&0.948&0.994&1.000&1.000&1.000&1.000&1.000\\
5.00&0.0058&0.044&0.100&0.300&0.672&0.928&0.996&1.000&1.000&1.000&1.000&1.000\\
10.00&0.0074&0.050&0.084&0.274&0.598&0.900&0.988&0.998&1.000&1.000&1.000&1.000\\
20.00&0.0097&0.046&0.080&0.214&0.530&0.846&0.982&0.996&1.000&1.000&1.000&1.000\\
{\bf 2.18} & {\bf 0.0049} & {\bf 0.052} & {\bf 0.120} & {\bf 0.340} & {\bf 0.716} & {\bf 0.956} & {\bf 0.996} & {\bf 1.000} & {\bf 1.000} & {\bf 1.000} & {\bf 1.000} & {\bf 1.000}\\
    \bottomrule
    \end{tabular}}%
\end{table}

\begin{table}[htbp]
    \centering
    \caption{
    Partial permutation test based on Gaussian with varying kernel parameter $\omega$ when the underlying function has the form (ii) in \eqref{eq:kernel_choice_f_choice} and the covariates are unbalanced  (case (c) in Table \ref{tab:balance}). 
    The description of the table is the same as that of Table \ref{tab:f3case1_ker_para}.
    }\label{tab:f6case3_ker_para}
    \resizebox{\textwidth}{!}{
    \begin{tabular}{cccccccccccccc}
    \toprule
    $\omega$ & MSE & 0 & 0.1 & 0.2 & 0.3 & 0.4 & 0.5 & 0.6 & 0.7 & 0.8 & 0.9 & 1\\
    \midrule
0.01&0.4437&0.708&0.678&0.662&0.700&0.756&0.828&0.872&0.904&0.924&0.944&0.954\\
0.05&0.3724&0.224&0.244&0.340&0.428&0.552&0.710&0.820&0.908&0.938&0.952&0.964\\
0.10&0.0119&0.068&0.074&0.130&0.242&0.380&0.566&0.662&0.736&0.818&0.884&0.934\\
0.50&0.0041&0.058&0.096&0.252&0.540&0.826&0.938&0.996&1.000&1.000&1.000&1.000\\
1.00&0.0041&0.044&0.088&0.248&0.534&0.826&0.942&0.996&1.000&1.000&1.000&1.000\\
1.50&0.0044&0.056&0.080&0.228&0.512&0.814&0.940&0.994&1.000&1.000&1.000&1.000\\
2.00&0.0047&0.058&0.072&0.226&0.524&0.804&0.942&0.998&1.000&1.000&1.000&1.000\\
2.50&0.0049&0.058&0.070&0.218&0.508&0.804&0.938&0.996&1.000&1.000&1.000&1.000\\
3.00&0.0051&0.056&0.086&0.214&0.514&0.800&0.934&0.994&1.000&1.000&1.000&1.000\\
5.00&0.0058&0.054&0.072&0.202&0.472&0.774&0.910&0.992&1.000&1.000&1.000&1.000\\
10.00&0.0074&0.050&0.074&0.178&0.428&0.710&0.894&0.982&1.000&1.000&1.000&1.000\\
20.00&0.0098&0.044&0.060&0.168&0.366&0.648&0.858&0.962&1.000&1.000&1.000&1.000\\
{\bf 2.17} & {\bf 0.0049} & {\bf 0.054} & {\bf 0.078} & {\bf 0.228} & {\bf 0.508} & {\bf 0.806} & {\bf 0.940} & {\bf 0.994} & {\bf 1.000} & {\bf 1.000} & {\bf 1.000} & {\bf 1.000} \\
    \bottomrule
    \end{tabular}}%
\end{table}

\section{Simulation under the null hypothesis with non-Gaussian noises}\label{sec:simu_nonGauss}

We here investigate the sensitivity of the partial permutation test  to the violation of the Gaussian noise assumption. 
Specifically, we generate data from model \eqref{eq:scalar_null} but with errors following either  Unif$(-\sqrt{3},\sqrt{3})$ (so that its variance is 1)
or the $t$-distribution with 5 degrees of freedom.
Moreover, we consider both linear and nonlinear underlying functional relations, i.e, (i) $f_0=x$ and (v) $f_0=\sin(4x)$ in \eqref{eq:f_choice_scalar}, 
and consider both balanced and unbalanced groups and covariate distributions, i.e., cases (a) and (d) in Table \ref{tab:balance}. 
\begin{figure}[h]
	\centering
		\includegraphics[width=0.7\linewidth]{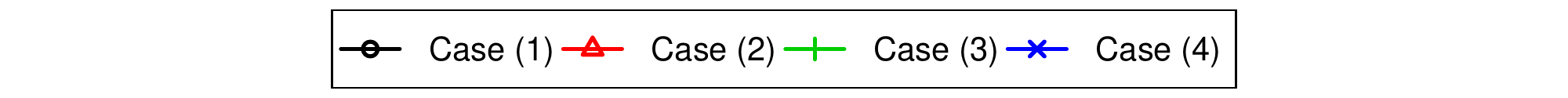}
	\begin{subfigure}{.3\textwidth}
		\centering
		\includegraphics[width=1\linewidth]{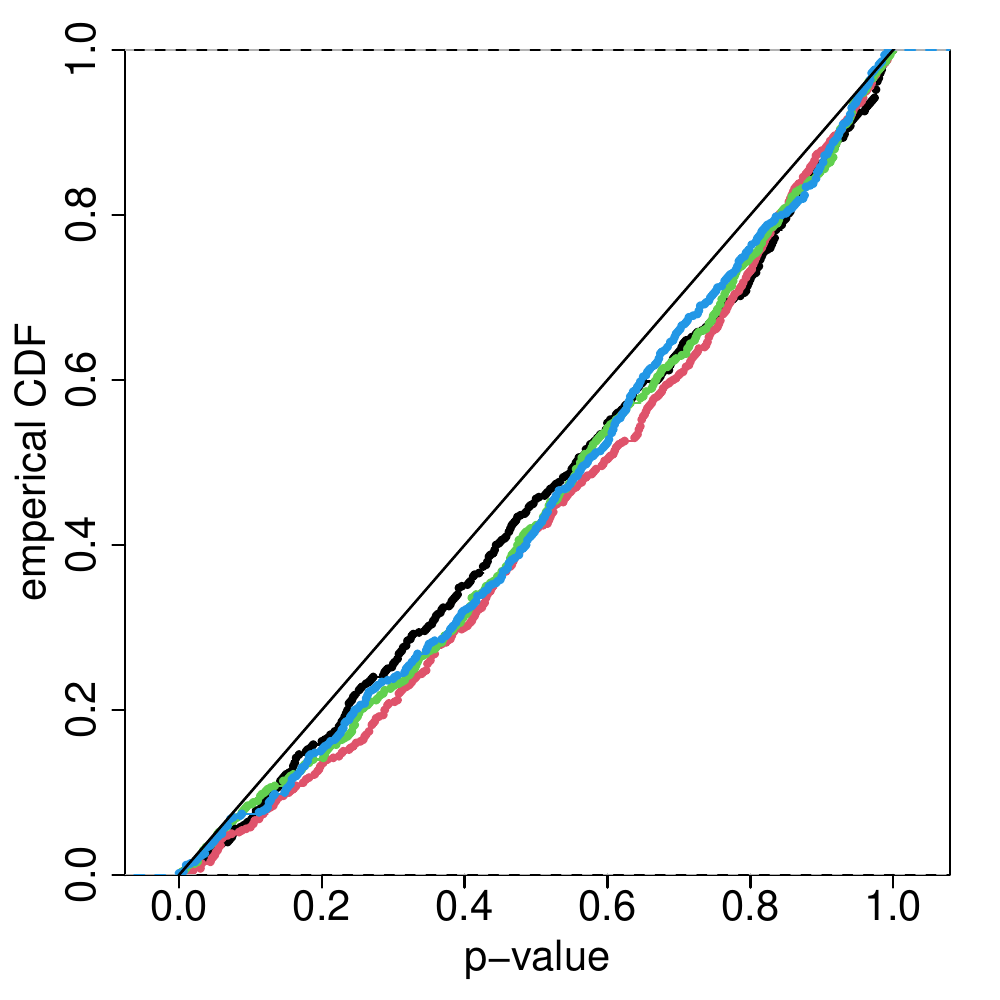}
		\caption{Unif($-\sqrt{3}, \sqrt{3}$)}
	\end{subfigure}%
	\begin{subfigure}{.3\textwidth}
		\centering
		\includegraphics[width=1\linewidth]{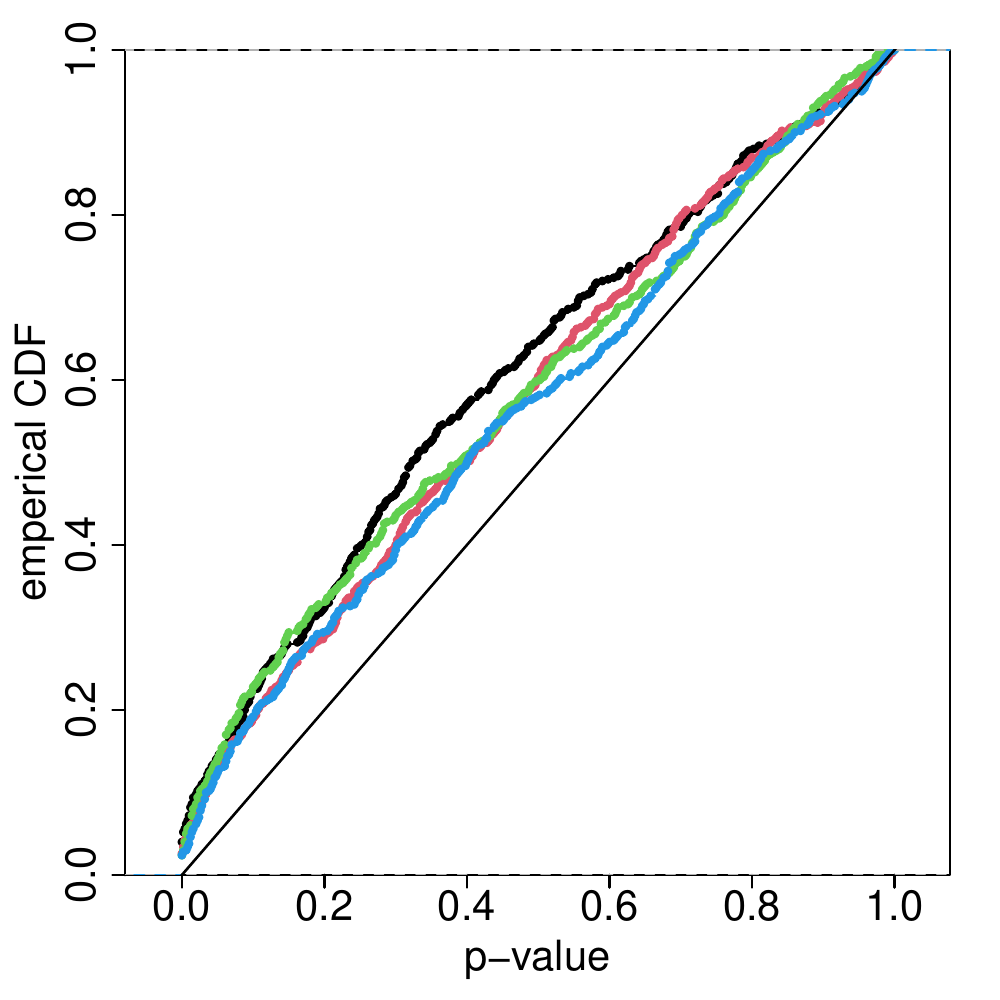}
		\caption{$t_5$}
	\end{subfigure}%
	\begin{subfigure}{.3\textwidth}
		\centering
		\includegraphics[width=1\linewidth]{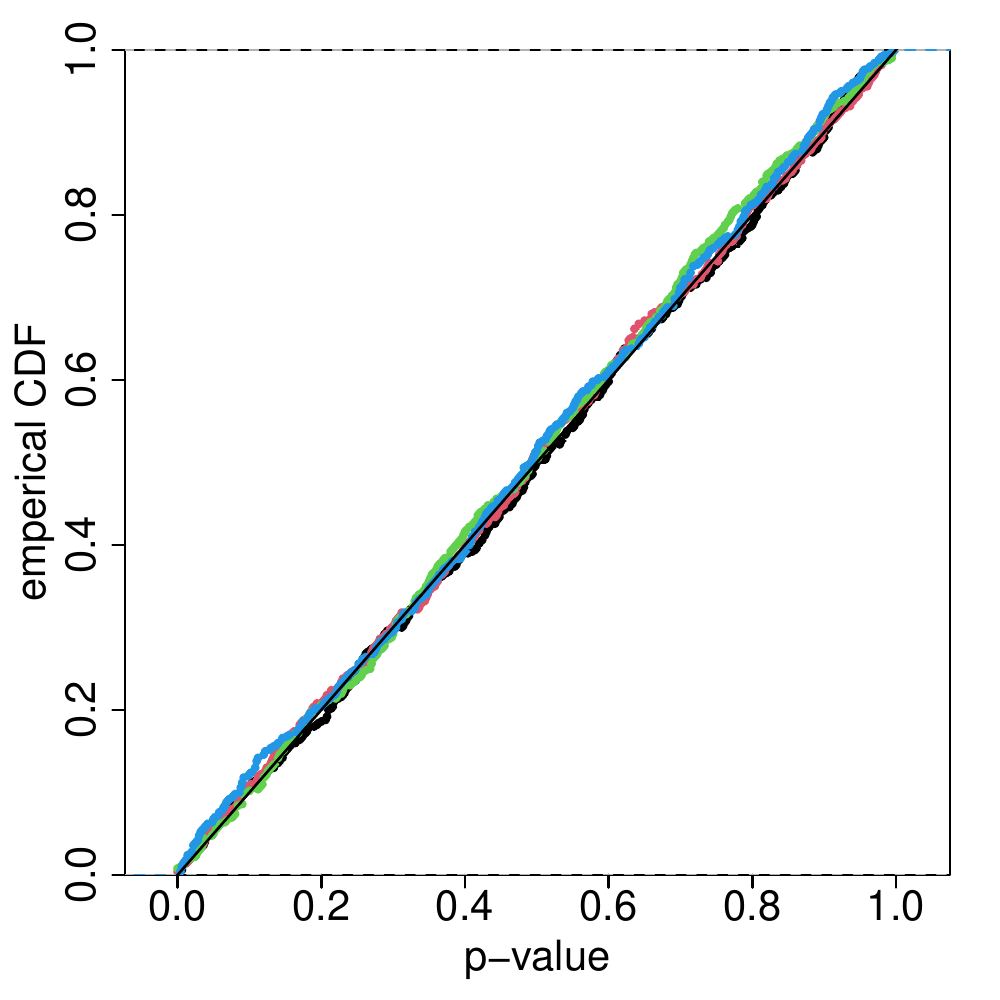}
		\caption{$t_5$ \& truncation}
	\end{subfigure}
	\caption{
	Empirical distributions of the partial permutation $p$-values using Gaussian kernel. Data are generated as in \eqref{eq:scalar_null} with non-Gaussian errors and sample size $n=200$. 
	Cases (1) and (2): $f_0=x$, and covariate distributions are balanced and unbalanced (cases (a) and (d) in Table \ref{tab:balance}), respectively;
	Cases (3) and (4): $f_0=\sin(4x)$, and covariate distributions are balanced and unbalanced, respectively. 
	Observation errors in Figures (a) and (b) are generated as indicated in their respective labels, In Figure (c), the errors are from $t_5$ distribution, but the extreme values of the fitted residuals are truncated. 
    }\label{fig:heavy_tail}
\end{figure}
Figure \ref{fig:heavy_tail} (a) and (b) show the empirical distributions of the partial permutation $p$-values using Gaussian kernel in these situations.

Figure \ref{fig:heavy_tail}(a) suggests that as long as the distribution of the observation noises are light-tailed (such as sub-Gaussian), the type-I error of the partial permutation test should be  well under control, and is often slightly conservative.
However, Figure \ref{fig:heavy_tail}(b) shows that the partial permutation test is sensitive to heavy-tailed observation noises, which result in inflated type-I errors. 
To overcome this vulnerability of the method in the presence of outliers, 
we suggest to pre-process the data by truncating extreme noise values. 
Specifically, we first get an estimate of the underlying function in the same way as in Section \ref{sec:choice_size} and obtain the fitted residuals.
We then truncate the fitted residuals on both the lower and upper $2\%$ of their sample quantiles. 
For example, for fitted residuals above the upper $2\%$ sample quantile, we set their values to be the upper $2\%$ quantile.
The empirical distributions of the resulting partial permutation $p$-values are shown in Figure \ref{fig:heavy_tail}(c), which demonstrates that truncating extreme residual values can get
the type-I error nearly perfectly calibrated at the nominal level (i.e., conforming to the uniform distribution).

\section{Simulation under the null hypothesis with correlated noises}\label{sec:simu_corr}

Following Section \ref{sec:corr}, we conduct a simulation study for the partial permutation test in the presence of correlated noises. 
Specifically, we generate data from the following model:
\begin{align}\label{eq:corr_model}
    \text{Scenario 6:} \quad & Y_i = 2.5 \cdot \sin (3\pi X_i) \cdot (1-X_i) + \varepsilon_i, 
    \quad 
    (\varepsilon_i, \varepsilon_{n/2+i})^\top \mid \mathrm{\bs{X}}, \bs{Z}
    \sim 
    \mathcal{N}(0, \sigma_0^2 \bs{R}_{\rho}), 
    \nonumber
    \\
    & 
    X_{i} \equiv X_{n/2+i} \ \overset{\iid}{\sim} \  \text{Unif}[0,1], \quad (1\le i \le n/2)
    \nonumber
    \\
    & Z_1 = \ldots = Z_{n/2} = 1, \quad 
    Z_{n/2+1} =\ldots = Z_{n} = 2, 
\end{align}
where $\bs{R}_{\rho}$ is a $2\times 2$ correlation matrix with off-diagonal elements $\rho$. 
For example, the $i$th and $(n/2+i)$th samples can correspond to two measurements of the same individual, either the same outcome measured at two time points or two different outcomes, with possibly correlated noises. 
We fix the sample size at $n=200$ and the noise level at $\sigma_0^2 = 0.1$, and consider two cases with correlations $\pm 0.5$. 
The partial permutation test can accommodate correlated measurement errors that are characterized by the $n\times n$ covariance matrix $\bs{\Sigma}$ by using the transformed response vector $\bs{Y}^{\text{C}} = \bs{\Sigma}^{-1/2} \bs{Y}$ and kernel matrix $\bs{K}_n^{\text{C}} = \bs{\Sigma}^{-1/2} \bs{K}_n \bs{\Sigma}^{-1/2}$ as discussed in Section \ref{sec:corr}. %
We here consider the partial permutation test with a pseudo likelihood ratio test statistic using: (i)   the true  $\bs{\Sigma}$;
(ii) an estimated  $\bs{\Sigma}$ from the data; (iii)  the identity  matrix (ignoring the correlations).
Specifically, we use the likelihood ratio of $\tilde{H}_{\text{pseudo}}$ versus $\tilde{H}_0$ in \eqref{eq:vcm_all_model} but with $\bs{Y}$ and $\bs{K}_n$ there replaced by $\bs{Y}^{\text{C}}$ and $\bs{K}_n^{\text{C}}$.

As shown in  Figure \ref{fig:corr_noise}, the $p$-values from the partial permutation test that uses either the true or estimated correlation matrix is approximately uniformly distributed. 
However, %
ignoring the correlation among the noises resulted in  either very conservative $p$-value as in Figure \ref{fig:corr_noise}(a) with positively correlated noises or invalid $p$-value with highly inflated type I error as in Figure \ref{fig:corr_noise}(b). 
In conclusion, when there exist possible correlations among measurement errors, we suggest to estimate the error covariance matrix $\bs{\Sigma}$ and take it into account when conducting the partial permutation test.
This procedure can control the type-I error well if  $\bs{\Sigma}$ can be estimated well. Since  $\bs{\Sigma}$ is of size $n\times n$, it is only estimable if there is a certain special structure to govern how errors are correlated, such as the case of this example.

\begin{figure}[h]
	\centering
		\includegraphics[width=0.7\linewidth]{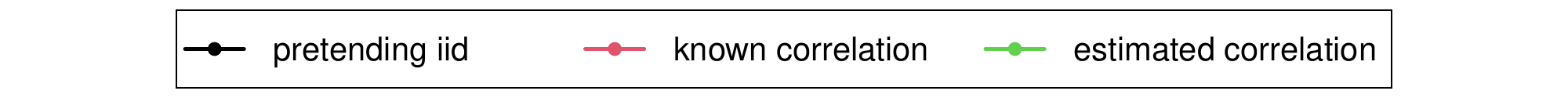}
	\begin{subfigure}{.3\textwidth}
		\centering
		\includegraphics[width=1\linewidth]{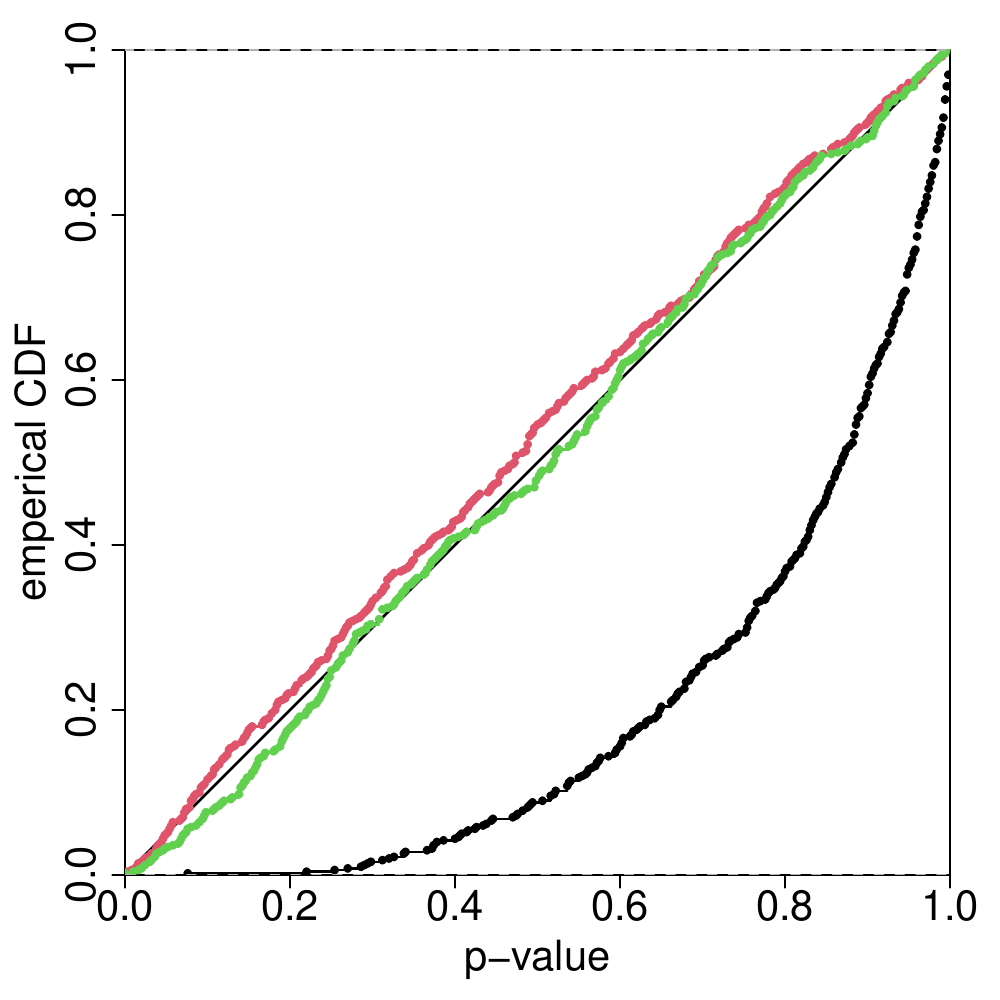}
		\caption{$\rho = 0.5$}
	\end{subfigure}%
	\begin{subfigure}{.3\textwidth}
		\centering
		\includegraphics[width=1\linewidth]{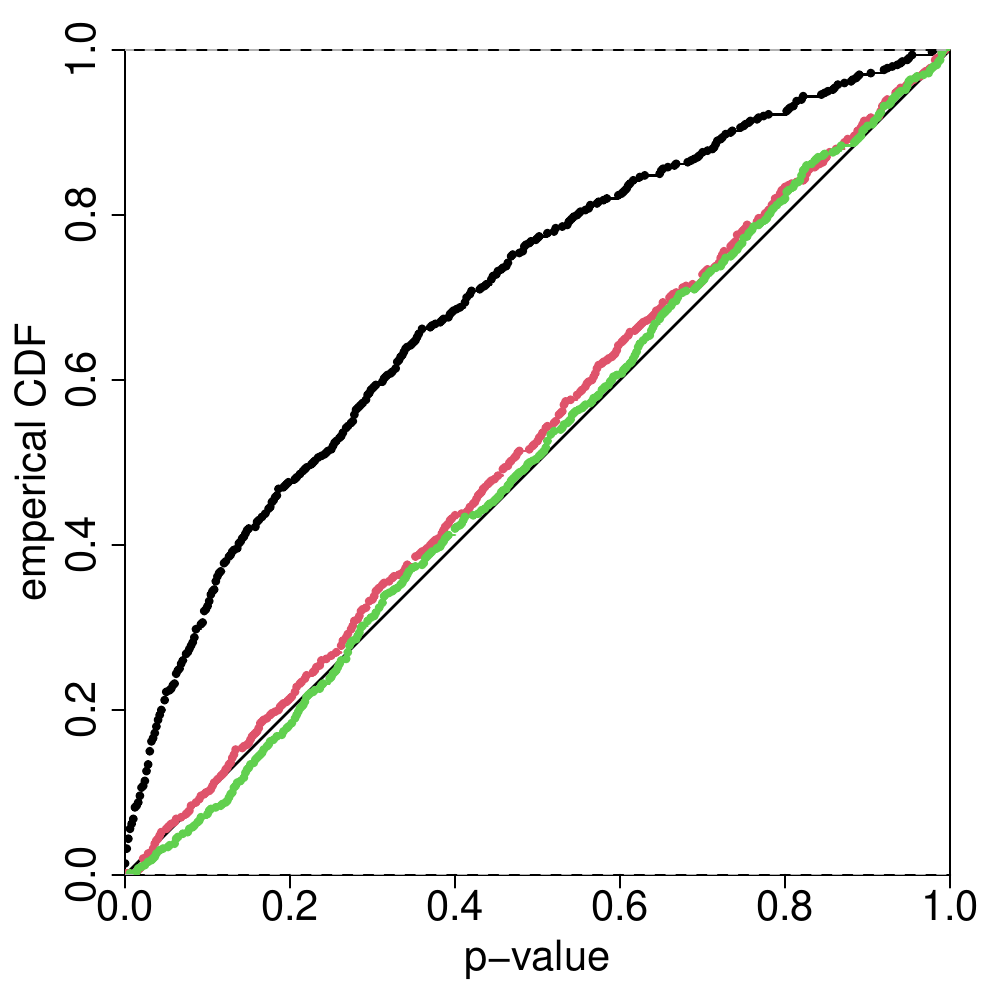}
		\caption{$\rho = -0.5$}
	\end{subfigure}
	\caption{
	Empirical CDFs of the partial permutation $p$-values using Gaussian kernel for the data of size $n=200$ generated from Scenario 6 in \eqref{eq:corr_model} with correlated Gaussian errors. 
    }\label{fig:corr_noise}
\end{figure}

\section{Partial Permutation Test with Fixed Functional Relationship}\label{sec:proof_fix_null}

Note that Theorems \ref{linear_permutation_pval_valid} and \ref{poly_permutation_pval_valid} are special cases of Corollary \ref{cor:kernel_finite_dim_feature_space}, we omit their proofs. 
To prove Theorem \ref{thm:fixed_property_theorem}, we need the following three lemmas. 

\begin{lemma}\label{lemma:survival_dominant_by_uniform}
	Let $Z$ be a univariate random variable with CDF $F(z)$. Let 
	\begin{align*}
	G(z) = \Pr(Z\geq z)=1-F(z-),
	\end{align*}
	then we have $\Pr(G(Z)\leq \alpha)\leq \alpha, \forall \alpha\in [0,1]$.
\end{lemma}
\begin{proof}[Proof of Lemma \ref{lemma:survival_dominant_by_uniform}]
By definition, $G(z)$ is nonincreasing in $z$. Define
\begin{align*}
G^{-1}(\alpha) = \inf \{x: G(x)\leq \alpha\}, \quad  (0 \le \alpha \le 1). 
\end{align*}
For any $\alpha\in [0,1]$, 
if $G(G^{-1}(\alpha))\leq \alpha$, then 
\begin{align*}
\Pr(G(Z)\leq \alpha) & \leq \Pr(Z\geq G^{-1}(\alpha))=G(G^{-1}(\alpha))\leq \alpha;
\end{align*}
otherwise, $G(G^{-1}(\alpha))>\alpha$, and 
\begin{align*}
	\Pr(G(Z)\leq \alpha) \leq& \Pr(Z> G^{-1}(\alpha))=\lim_{m\rightarrow \infty} \Pr\left(Z\geq G^{-1}(\alpha)+\frac{1}{m}\right)
	= \lim_{m\rightarrow \infty} G\left(G^{-1}(\alpha)+\frac{1}{m}\right)\leq \alpha.
\end{align*}
Therefore, Lemma \ref{lemma:survival_dominant_by_uniform} holds. 
\end{proof}

\begin{lemma}\label{lemma:basic_ineq_perm_p_val}
For any random elements $\bXmatrix \in \mathbb{R}^{n\times d}, \bs{Y} \in \mathbb{R}^n$ and $\bs{Z} \in \mathbb{R}^n$, 
let 
$\mathcal{S}_y$ be the discrete (or continuous) permutation set defined as in Algorithm \ref{alg:partial_permu}, $\bs{Y}^p$ be a random vector uniformly distributed on $\mathcal{S}_y$ given $\bXmatrix, \mathcal{S}_y$ and $\bs{Z}$, and $\nu$ and $\nu_0$ be two measures on $\mathbb{R}^n$ defined as 
\begin{align*}
\nu(\mathcal{A}) = \Pr( \bs{Y} \in \mathcal{A} \mid \bXmatrix, \mathcal{S}_y, \bs{Z} ),
\qquad
\nu_0(\mathcal{A}) = \Pr(\bs{Y}^p \in \mathcal{A} \mid \bXmatrix,\mathcal{S}_y, \bs{Z} ),
\end{align*}
for any measurable set $\mathcal{A}\subset \mathbb{R}^n$. 
For any test statistic $T$, 
the corresponding discrete (or continuous) permutation $p$-value from Algorithm \ref{alg:partial_permu} satisfies that, for any $\alpha,\delta\geq 0$,  
\begin{align*}
\Pr\{
p( \bXmatrix, \bs{Y}, \bs{Z} )\leq \alpha  \mid \bXmatrix,  \bs{Z} 
\} \leq \alpha + \delta + \Pr\{
\left\| \nu - \nu_0 \right\|_{\TV} > \delta \mid \bXmatrix,  \bs{Z} 
\}.
\end{align*}
\end{lemma}

\begin{proof}[Proof of Lemma \ref{lemma:basic_ineq_perm_p_val}]
Define $G$ and $G_u$ as 
\begin{align*}
G(z) & =  \Pr\left\{ T( \bXmatrix, \bs{Y}, \bs{Z} ) \geq z\mid \bXmatrix, \mathcal{S}_y, \bs{Z} \right\},
\quad 
G_u(z) =   \Pr\left\{ T( \bXmatrix, \bs{Y}^p, \bs{Z} ) \geq z\mid \bXmatrix,\mathcal{S}_y, \bs{Z} \right\}.
\end{align*}
By definition, for any $z\in \mathbb{R}$, 
$
|G(z)-G_u(z)| \leq \|\nu - \nu_0\|_{\TV},  
$
which implies that 
\begin{align*}
 |G(T( \bXmatrix, \bs{Y}, \bs{Z} ))-p( \bXmatrix, \bs{Y}, \bs{Z} )|
& = |G(T( \bXmatrix, \bs{Y}, \bs{Z} ))-G_u(T( \bXmatrix, \bs{Y}, \bs{Z} ))| \leq  \|\nu - \nu_0\|_{\TV}.
\end{align*} 
From Lemma \ref{lemma:survival_dominant_by_uniform}, for any $\alpha \in [0,1]$, 
$
\Pr\{ G( T( \bXmatrix, \bs{Y}, \bs{Z} ) ) \leq \alpha \mid \bXmatrix, \mathcal{S}_y, \bs{Z} \} \leq \alpha.
$
Thus, for any $\delta \geq 0$, 
\begin{eqnarray*}
& & \Pr\{  p( \bXmatrix, \bs{Y}, \bs{Z} ) \leq \alpha \mid \bXmatrix, \mathcal{S}_y, \bs{Z} \} \\
& = & \Pr\{  p( \bXmatrix, \bs{Y}, \bs{Z} ) \leq \alpha, \|\nu -\nu_0\|_{\TV} \leq \delta \mid \bXmatrix, \mathcal{S}_y, \bs{Z} \} + \Pr\{  p( \bXmatrix, \bs{Y}, \bs{Z} ) \leq \alpha, \|\nu -\nu_0\|_{\TV} > \delta \mid \bXmatrix, \mathcal{S}_y, \bs{Z} \} \\
& \leq & \Pr\{  G(T( \bXmatrix, \bs{Y}, \bs{Z} )) \leq \alpha+\delta \mid \bXmatrix, \mathcal{S}_y, \bs{Z} \} + \Pr\{\|\nu -\nu_0\|_{\TV} > \delta \mid \bXmatrix, \mathcal{S}_y, \bs{Z} \} \\
& \leq &\alpha+\delta+\Pr\{\| \nu -\nu_0 \|_{\TV} > \delta \mid  \bXmatrix, \mathcal{S}_y, \bs{Z} \}.
\end{eqnarray*}
Taking conditional expectations given $\bXmatrix$ and $\bs{Z}$ on both sides of the above inequality, we can then derive Lemma \ref{lemma:basic_ineq_perm_p_val}.
\end{proof}

\begin{lemma}\label{lemma:bound_TV}
Let $\{(\bs{X}_i,Y_i,Z_i)\}_{1\leq i \leq n}$ be samples from model ${H}_0$ in \eqref{eq:H_0_fixed},
$\mathcal{S}_y$ be the discrete (or continuous) permutation set defined as in Algorithm \ref{alg:partial_permu}, $\bs{Y}^p$ be a random vector uniformly distributed on $\mathcal{S}_y$ given $\bXmatrix, \mathcal{S}_y$ and $\bs{Z}$, and $\nu$ and $\nu_0$ be two measures on $\mathbb{R}^n$ defined as 
\begin{align*}
\nu(\mathcal{A}) = \Pr( \bs{Y} \in \mathcal{A} \mid \bXmatrix, \mathcal{S}_y, \bs{Z} ),
\qquad
\nu_0(\mathcal{A}) = \Pr(\bs{Y}^p \in \mathcal{A} \mid \bXmatrix,\mathcal{S}_y, \bs{Z} ),
\end{align*}
for any measurable set $\mathcal{A}\subset \mathbb{R}^n$. 
Then for any $1 \le b_n \le n$, 
$
   \|\nu - \nu_0\|_{\TV} \le  (e^{\Delta_n}-1)/2, 
$
where 
\begin{align*}
\Delta_n 
& =  2 \sqrt{2} \sqrt{\omega(b_n, \sigma_0^{-1}f_0) } \cdot \sqrt{
	\omega(b_n, \sigma_0^{-1}f_0) + \sigma_0^{-2}
	\sum_{i=n-b_n+1}^{n}(\bs{\gamma}_{i}^\top\bs{\varepsilon})^2 }
\end{align*}
and $\omega(b_n, \sigma_0^{-1}f_0) =  \sigma_0^{-2}\sum_{i=a_n}^n(\bs{\gamma}_i^\top\bs{f}_0)^2$ is defined the same as in Section \ref{sec:general}. 
\end{lemma}

\begin{proof}[Proof of Lemma \ref{lemma:bound_TV}]
We prove only the lemma for the discrete permutation set $\mathcal{S}_y$, since the proof for the continuous permutation set is very similar. 
Recall that $\bs{f}_0=(f_0(\bs{X}_1), \cdots, f_0(\bs{X}_n))^\top$ and $\bs{\varepsilon} = (\varepsilon_1, \varepsilon_2, \ldots, \varepsilon_n)^\top$. Then, 
$
\bs{Y} = \bs{f}_0+\bs{\varepsilon},$
and 
$
\bs{W} = \bs{\Gamma}^\top \bs{Y} = \bs{\Gamma}^\top \bs{f}_0+\bs{\Gamma}^\top \bs{\varepsilon}.
$
For any $\psi \in \mathcal{M}(n,b_n)$ and $\bs{Y}_\psi = \bs{\Gamma} \bs{W}_{\psi} = \bs{\Gamma}(W_{\psi(1)}, \ldots, W_{\psi(n)})^\top \in \mathcal{S}_y$, 
the density for the conditional distribution of $\bs{Y}$ given $(\bXmatrix, \bs{Z})$ under model $H_0$ evaluated at $\bs{Y}_\psi$ has the following equivalent forms:
\begin{align*}
g( \bs{Y}_\psi \mid \bXmatrix, \bs{Z} ) 
& = 
\left(2\pi\sigma_0^2\right)^{-n/2}\exp\left\{
-\frac{1}{2\sigma_0^{2}} \sum_{i=1}^{n-b_n}\left(W_{\psi(i)}-\bs{\gamma}_{i}^\top\bs{f}_0\right)^{2}
-\frac{1}{2\sigma_0^{2}} \sum_{i=n-b_n+1}^{n}\left(W_{\psi(i)}-\bs{\gamma}_{i}^\top\bs{f}_0\right)^{2}
\right\}\\
& = \left(2\pi\sigma_0^2\right)^{-n/2}\exp\left\{
-\frac{1}{2\sigma_0^{2}} \sum_{i=1}^{n-b_n}\left(W_{i}-\bs{\gamma}_{i}^\top\bs{f}_0\right)^{2}
\right\}
\exp\left\{
-\frac{1}{2\sigma_0^{2}} \sum_{i=n-b_n+1}^{n}\left(W_{\psi(i)}-\bs{\gamma}_{i}^\top\bs{f}_0\right)^{2}
\right\},
\end{align*}
where the last equality holds by the definition of $\mathcal{M}(n,b_n)$. 
Thus, for any $\psi, \phi \in \mathcal{M}_{n,b_n}$, 
\begin{align*}
	\left|
	\sigma_0^2\log\left(\frac{
		g( \bs{Y}_\psi \mid \bXmatrix, \bs{Z} ) 
	}{
		g( \bs{Y}_\phi \mid \bXmatrix, \bs{Z} ) 
	}\right) 
	\right|
	& =   
	\left| \frac{1}{2}\sum_{i=n-b_n+1}^{n}\left(W_{\phi(i)}-\bs{\gamma}_{i}^\top\bs{f}_0\right)^{2}- \frac{1}{2}\sum_{i=n-b_n+1}^{n}\left(W_{\psi(i)}-\bs{\gamma}_{i}^\top\bs{f}_0\right)^{2}
	\right|
	\\
	& =  \left| \sum_{i=n-b_n+1}^{n}W_{\psi(i)}\bs{\gamma}_{i}^\top\bs{f}_0
	- \sum_{i=n-b_n+1}^{n}W_{\phi(i)}\bs{\gamma}_{i}^\top\bs{f}_0\right|
	\\
	& \leq 
	\left| \sum_{i=n-b_n+1}^{n}W_{\psi(i)}\bs{\gamma}_{i}^\top\bs{f}_0
	\right| + 
	\left|\sum_{i=n-b_n+1}^{n}W_{\phi(i)}\bs{\gamma}_{i}^\top\bs{f}_0\right|. 
\end{align*}
By the Cauchy–Schwarz inequality, we have
\begin{align*}
\left|
\sigma_0^2\log\left(\frac{
g( \bs{Y}_\psi \mid \bXmatrix, \bs{Z} ) 
}{
g( \bs{Y}_\phi \mid \bXmatrix, \bs{Z} ) 
}\right) 
\right| 
&  \leq  
\sqrt{\sum_{i=n-b_n+1}^{n}W_{\psi(i)}^2 \cdot \sum_{i=n-b_n+1}^{n} (\bs{\gamma}_{i}^\top\bs{f}_0)^2} + \sqrt{\sum_{i=n-b_n+1}^{n}W_{\phi(i)}^2 \cdot \sum_{i=n-b_n+1}^{n} (\bs{\gamma}_{i}^\top\bs{f}_0)^2}
\\
& = 2 \sqrt{\sum_{i=n-b_n+1}^{n}W_{i}^2 \cdot \sum_{i=n-b_n+1}^{n} (\bs{\gamma}_{i}^\top\bs{f}_0)^2} \ .
\end{align*}
Note that  
$W_i^2 = (\bs{\gamma}_{i}^\top\bs{f}_0 + \bs{\gamma}_{i}^\top\bs{\varepsilon})^2 \le 2 (\bs{\gamma}_{i}^\top\bs{f}_0)^2 + 2 (\bs{\gamma}_{i}^\top\bs{\varepsilon})^2$. 
We then have  
\begin{align*}
\left|
\sigma_0^2\log\left(\frac{
	g( \bs{Y}_\psi \mid \bXmatrix, \bs{Z} ) 
}{
	g( \bs{Y}_\phi \mid \bXmatrix, \bs{Z} ) 
}\right) 
\right| 
& \le 2 \sqrt{\sum_{i=n-b_n+1}^{n} (\bs{\gamma}_{i}^\top\bs{f}_0)^2 \cdot \sum_{i=n-b_n+1}^{n}W_{i}^2}
\\
& \leq 2\sqrt{2} \sqrt{\sum_{i=n-b_n+1}^{n} (\bs{\gamma}_{i}^\top\bs{f}_0)^2} \cdot \sqrt{
	\sum_{i=n-b_n+1}^{n}(\bs{\gamma}_{i}^\top\bs{f}_0)^2 + 
	\sum_{i=n-b_n+1}^{n}(\bs{\gamma}_{i}^\top\bs{\varepsilon})^2 }
\end{align*}
By the definition of $\omega(b_n, \sigma_0^{-1}f_0)$ and $\Delta_n$, this further implies that 
\begin{align*}
\left|
\log\left(\frac{
	g( \bs{Y}_\psi \mid \bXmatrix, \bs{Z} ) 
}{
	g( \bs{Y}_\phi \mid \bXmatrix, \bs{Z} ) 
}\right) 
\right| 
\le 
2\sqrt{2} \sqrt{\omega(b_n, \sigma_0^{-1}f_0)} \sqrt{\omega(b_n, \sigma_0^{-1}f_0) + \sigma_0^{-2} \sum_{i=n-b_n+1}^{n}(\bs{\gamma}_{i}^\top\bs{\varepsilon})^2 }
= 
\Delta_n. 
\end{align*}
Consequently, for any $\psi, \phi \in \mathcal{M}_{n,b_n}$, 
$
\exp(-\Delta_n) \le g(Y_\psi \mid \bXmatrix, \bs{Z} ) / g(Y_\phi \mid \bXmatrix, \bs{Z} )  \le \exp(\Delta_n). 
$

Let $|\mathcal{M}_{n,b_n}|$ be the cardinality of set $\mathcal{M}_{n,b_n}$. From the discussion before, for any $\psi \in \mathcal{M}(n, b_n)$, the measure of $\nu$ on $\{ \bs{Y}_\psi\}$ satisfies
\begin{align*}
\nu(\{Y_\psi\}) = \frac{g(\bs{Y}_\psi \mid \bXmatrix, \bs{Z} )}{\sum_{\phi \in \mathcal{M}(n,b_n)}g( \bs{Y}_\phi \mid \bXmatrix, \bs{Z} )} \in \left[\frac{1}{|\mathcal{M}(n,b_n)|}e^{-\Delta_n},
\ \ 
\frac{1}{|\mathcal{M}(n,b_n)|}e^{\Delta_n}
\right],
\end{align*}
which further implies that
\begin{align*}
\nu(\{Y_\psi\}) - \nu_0(\{Y_\psi\}) \in \left[-\frac{1}{|\mathcal{M}(n,b_n)|}(1-e^{-\Delta_n}),
\ \ \ 
\frac{1}{|\mathcal{M}(n,b_n)|}(e^{\Delta_n}-1)
\right].
\end{align*}
Therefore, the total variation distance between $\nu$ and $\nu_0$ satisfies 
\begin{align*}
\|\nu - \nu_0\|_{\TV} & = \frac{1}{2} \sum_{\psi \in \mathcal{M}(n,b_n)} \left|\nu(\{Y_\psi\}) - \nu_0(\{Y_\psi\}) \right|
\leq  \frac{1}{2}\sum_{\psi \in \mathcal{M}(n,b_n)} \frac{1}{ |\mathcal{M}(n,b_n)| }\max \{
1-e^{-\Delta_n}, e^{\Delta_n}-1
\}\\
& 
= \frac{1}{2} \max \{
1-e^{-\Delta_n}, e^{\Delta_n}-1
\}
\leq \frac{1}{2} (e^{\Delta_n}-1),
\end{align*}
i.e., Lemma \ref{lemma:bound_TV} holds. 
\end{proof}

\begin{proof}[{\bf Proof of Theorem \ref{thm:fixed_property_theorem}}]
We prove only the theorem for the discrete permutation $p$-value, since the proof for the continuous permutation $p$-value is very similar.
Define 
$\nu$, $\nu_0$ and $\Delta_n$ the same as in Lemma \ref{lemma:bound_TV}. 
From Lemmas \ref{lemma:basic_ineq_perm_p_val} and \ref{lemma:bound_TV}, the permutation $p$-value satisfies that, for any $\alpha, \delta \ge 0$,  
\begin{align*}
\Pr \{  p( \bXmatrix, \bs{Y}, \bs{Z} ) \leq \alpha \mid \bXmatrix, \bs{Z} \}
\leq & \alpha+ \delta +\Pr \{\| \nu - \nu_0 \|_{\TV} > \delta \mid \bXmatrix, \bs{Z} \}
\leq  \alpha+ \delta + \Pr \left( \frac{e^{\Delta_n}-1}{2} > \delta\mid \bXmatrix, \bs{Z}  \right). 
\end{align*}
Given any $\alpha_0\in (0,1)$,
let $\delta = v(b_n, \sigma_0^{-1} f_0, \alpha_0)$ defined the same as in \eqref{eq:v_fixed_H0}. 
By definition, we then have 
\begin{align*}
 & \ \frac{e^{\Delta_n}-1}{2} > \delta 
 \\
 \Longleftrightarrow & \ 
\frac{1}{2}\exp\left\{
2 \sqrt{2} \sqrt{\omega(b_n, \sigma_0^{-1}f_0) } \cdot \sqrt{
	\omega(b_n, \sigma_0^{-1}f_0) + \sigma_0^{-2}
	\sum_{i=n-b_n+1}^{n}(\bs{\gamma}_{i}^\top\bs{\varepsilon})^2 }
\right\} - \frac{1}{2}\\
& \quad \ 
> 
\frac{1}{2}\exp\left\{
2\sqrt{2\omega(b_n,\sigma_0^{-1}f_0) } \sqrt{Q_{b_n}(1-\alpha_0)+ \omega(b_n,\sigma_0^{-1}f_0)} 
\right\}-\frac{1}{2}
\\
 \Longleftrightarrow & \ 
 \omega(b_n, \sigma_0^{-1}f_0) > 0  \ \ \text{and} \ \ 
 \sigma_0^{-2}
 \sum_{i=n-b_n+1}^{n}(\bs{\gamma}_{i}^\top\bs{\varepsilon})^2 > Q_{b_n}(1-\alpha_0),
\end{align*}
where $Q_{b_n}$ is the quantile function of the chi-square distribution with degrees of freedom $b_n$. 
This then implies that 
\begin{align*}
\Pr \left( \frac{e^{\Delta_n}-1}{2} > \delta\mid \bXmatrix, \bs{Z}  \right)
& \le 
\Pr \left(
 \sigma_0^{-2}
\sum_{i=n-b_n+1}^{n}(\bs{\gamma}_{i}^\top\bs{\varepsilon})^2 > Q_{b_n}(1-\alpha_0)
\mid \bXmatrix, \bs{Z}  \right)
= \alpha_0, 
\end{align*}
where the last equality holds because 
$ \sigma_0^{-2}
\sum_{i=n-b_n+1}^{n}(\bs{\gamma}_{i}^\top\bs{\varepsilon})^2 \sim \chi^2_{b_n}$ conditional on $\bXmatrix$ and $\bs{Z}$. 
Consequently, for any $\alpha \in [0,1]$,  
\begin{align*}
\Pr \{  p( \bXmatrix, \bs{Y}, \bs{Z} ) \leq \alpha \mid \bXmatrix, \bs{Z} \}
& 
\leq  \alpha+ \delta + \Pr \left( \frac{e^{\Delta_n}-1}{2} > \delta\mid \bXmatrix, \bs{Z}  \right)
\le \alpha + v(b_n, \sigma_0^{-1} f_0, \alpha_0) + \alpha_0. 
\end{align*}
By the definition of $p_c(\bXmatrix, \bs{Y}, \bs{Z})$, 
this immediately implies that, for any $\alpha \in (0,1)$, 
\begin{align*}
& \quad \ \Pr \{  p_c( \bXmatrix, \bs{Y}, \bs{Z} ) \leq \alpha \mid \bXmatrix, \bs{Z} \}
\\
& = 
\Pr \{  p( \bXmatrix, \bs{Y}, \bs{Z} ) \leq \alpha - v(b_n, \sigma_0^{-1} f_0, \alpha_0) - \alpha_0 \mid \bXmatrix, \bs{Z} \}
\\
& \le 
\I\left\{\alpha \ge v(b_n, \sigma_0^{-1} f_0, \alpha_0) + \alpha_0\right\} 
\Pr \{  p( \bXmatrix, \bs{Y}, \bs{Z} ) \leq \alpha - v(b_n, \sigma_0^{-1} f_0, \alpha_0) - \alpha_0 \mid \bXmatrix, \bs{Z} \}
\\
& \le 
\I\left\{\alpha \ge v(b_n, \sigma_0^{-1} f_0, \alpha_0) + \alpha_0\right\} \cdot \alpha \le \alpha. 
\end{align*}
Therefore, Theorem \ref{thm:fixed_property_theorem} holds. 
\end{proof}

\begin{proof}[\bf Proof of Corollary \ref{cor:kernel_finite_dim_feature_space}]
Let $\bs{\Phi} = (\phi(\bs{X}_1), \phi(\bs{X}_2), \ldots, \phi(\bs{X}_n))^\top \in \mathbb{R}^{n \times q}$ be the matrix consisting of all covariates mapped into the feature space. 
We can then verify that the kernel matrix $\bs{K}_n$ can be equivalently written as $\bs{K}_n = \bs{\Phi} \bs{\bs{\Phi}}^\top$, whose rank is at most $q$. 
Thus, the eigen-decomposition of $\bs{K}_n$ must satisfy that $\bs{\gamma}_i^\top \bs{K}_n \bs{\gamma} = 0$ for $q+1 \le i \le n$, 
or equivalently $\bs{\gamma}_i^\top \bs{\Phi} \bs{\bs{\Phi}}^\top \bs{\gamma} = 0$ for $q+1 \le i \le n$. 
This further implies that $\bs{\gamma}_i^\top \bs{\Phi} = \bs{0}_{1\times q}$ for $q+1 \le i \le n$.

Because the underlying function $f_0(\bs{x})$ is linear in $\phi(x)$, 
we can write $f_0(\bs{x})$ as $f_0(\bs{x}) = \phi(x)^\top \bs{\beta}$ for some $\bs{\beta}\in \mathbb{R}^q$. 
Consequently, the vector consisting of the function values evaluated at all the covariates has the equivalent form
$\bs{f}_0 = \bs{\Phi} \bs{\beta}$. 
From the discussion before, we can know that, for $q+1 \le i \le n$, $\bs{\gamma}_i^\top \bs{f}_0 = \bs{\gamma}_i^\top \bs{\Phi} \bs{\beta} = \bs{0}_{1\times q} \bs{\beta} =0$. 
Thus, for any $1 \le b_n \le n-q$, we have 
\begin{align}\label{eq:omega_0_finite_feature}
\omega(b_n, \sigma_0^{-1} f_0) = \sigma_0^{-2} \sum_{i=n-b_n+1} \left( \bs{\gamma}_i^\top \bs{f}_0 \right)^2 = 0. 
\end{align}
Using Lemmas \ref{lemma:basic_ineq_perm_p_val} and \ref{lemma:bound_TV}, the permutation $p$-value satisfies that, for any $\alpha, \delta \ge 0$,  
\begin{align*}
\Pr \{  p( \bXmatrix, \bs{Y}, \bs{Z} ) \leq \alpha \mid \bXmatrix, \bs{Z} \}
\leq & 
\alpha+ \delta + \Pr \left( \frac{e^{\Delta_n}-1}{2} > \delta\mid \bXmatrix, \bs{Z}  \right), 
\end{align*}
where $\Delta_n$ is defined the same as in Lemma \ref{lemma:bound_TV} and equals zero here due to \eqref{eq:omega_0_finite_feature}. 
Thus, we must have 
$\Pr \{  p( \bXmatrix, \bs{Y}, \bs{Z} ) \leq \alpha \mid \bXmatrix, \bs{Z} \}
\le 
\alpha,$ i.e., Corollary \ref{cor:kernel_finite_dim_feature_space} holds. 
\end{proof}

\begin{proof}[\bf Proof of Corollary \ref{cor:diverg_kernel}]
    First, we consider the limiting behavior of $\omega(b_n, \sigma_0^{-1} f_0)$ as $n \rightarrow \infty$. 
    For $q\ge 1$, we introduce $\bs{\beta}_q\in \mathbb{R}^q$ to denote the coefficient vector for the best linear approximation of $f_0$ using the first $q$ basis functions, 
    i.e., 
    $\remainder(f_0;q) = \int ( f - \bs{\beta}_q^\top \phi_q )^2 \text{d} \mu$. 
    Let $\bs{\Phi}_q = (\phi_q(\bs{X}_1), \ldots, \phi_q(\bs{X}_n))^\top \in \mathbb{R}^{n\times q}$ be the matrix consisting of transformed covariates using the first $q$ basis functions. 
    We can verify that the kernel matrix $\bs{K}_n$ can be written as $\bs{K}_n = \bs{\Phi}_{q_n} \bs{\Phi}_{q_n}^\top$, whose rank is at most $q_n < n$. 
    Thus, the eigenvectors of $\bs{K}_n$ must satisfy that $0 = \bs{\gamma}_i^\top \bs{K}_n \bs{\gamma}_i = \bs{\gamma}_i^\top \bs{\Phi}_{q_n} \bs{\Phi}_{q_n}^\top \bs{\gamma}_i$ for $q_n < i \le n$. 
    Consequently, 
    for $q_n < i \le n$, 
    $\bs{\gamma}_i^\top \bs{\Phi}_{q_n} = \bs{0}$, 
    and 
    $
        \bs{\gamma}_i^\top \bs{f}_0 
        = 
        \bs{\gamma}_i^\top 
        ( \bs{\Phi}_{q_n} \bs{\beta}_{q_n} + \bs{f}_0 - \bs{\Phi}_{q_n} \bs{\beta}_{q_n} )
        = \bs{\gamma}_i^\top (\bs{f}_0 - \bs{\Phi}_{q_n} \bs{\beta}_{q_n} ), 
    $
    where $\bs{f}_0 = (f_0(\bs{X}_1), \ldots, f_0(\bs{X}_n))$. 
    This then implies that 
    \begin{align*}
        \sum_{i=q_n+1}^n ( \bs{\gamma}_i^\top \bs{f}_0 )^2 
        & = 
        \sum_{i=q_n+1}^n 
        \big\{ \bs{\gamma}_i^\top (\bs{f}_0 - \bs{\Phi}_{q_n} \bs{\beta}_{q_n} ) \big\}^2 
        \le 
        \| \bs{f}_0 - \bs{\Phi}_{q_n} \bs{\beta}_{q_n} \|_2^2 
        = 
        \sum_{i=1}^n 
        \{ f_0(\bs{X}_i) - \bs{\beta}_{q_n}^\top \phi_{q_n}(\bs{X}_i) \}^2 
        \\
        & = 
        \sum_{i=1}^n \E [ \{ f_0(\bs{X}_i) - \bs{\beta}_{q_n}^\top \phi_{q_n}(\bs{X}_i) \}^2 ]
        \cdot
        O_{\Pr}(1)
        = 
        n \remainder(f_0; q_n) \cdot O_{\Pr}(1)
    \end{align*}
    Consequently, we have 
    $
       \omega(b_n,\sigma_0^{-1}f_0)
       = 
       \sigma_0^{-2}
       \sum_{i=n-b_n+1}^n ( \bs{\gamma}_{i}^\top \bs{f}_0)^2
       =  n \remainder(f_0; q_n) \cdot O_{\Pr}(1). 
    $
    
    Second, we prove the asymptotic validity of the partial permutation test. 
    From Lemmas \ref{lemma:basic_ineq_perm_p_val} and \ref{lemma:bound_TV} and by the law of the iterated expectation, for any $\alpha, \delta > 0$, 
    the $p$-value from the partial permutation test with kernel $K_{q_n}$, permutation size $b_n \le n-q_n$ and any test statistic $T$ satisfies that 
    \begin{align}\label{eq:pval_bound_diverging}
    \Pr \{  p( \bXmatrix, \bs{Y}, \bs{Z} ) \leq \alpha \}
    \leq & 
    \alpha+ \delta + \Pr \left( e^{\Delta_n}-1 > 2\delta \right)
    = 
    \alpha+ \delta + \Pr \left\{ \Delta_n > \log(1+2\delta) \right\}
    \end{align}
    with 
    $
        \Delta_n 
        =  2 \sqrt{2} \sqrt{\omega(b_n, \sigma_0^{-1}f_0) } \cdot \sqrt{
	\omega(b_n, \sigma_0^{-1}f_0) + \sigma_0^{-2}
	\sum_{i=n-b_n+1}^{n}(\bs{\gamma}_{i}^\top\bs{\varepsilon})^2 }.
    $
    Note that $(\bs{\gamma}_1^\top \bs{\varepsilon}, \ldots, \bs{\gamma}_n^\top \bs{\varepsilon})$ are i.i.d.\ standard Gaussian. We must have 
    $
    \sum_{i=n-b_n+1}^{n}(\bs{\gamma}_{i}^\top\bs{\varepsilon})^2 
    = O_{\Pr}(b_n ). 
    $
    Consequently, 
    \begin{align*}
    \Delta_n 
    & =  \sqrt{n \remainder (f_0; q_n) } \cdot \sqrt{
	n \remainder (f_0; q_n)   + b_n } \cdot O_{\Pr}(1)
	=
	\sqrt{
	\{ n (n-q_n) \remainder (f_0; q_n)\}^2 + n (n-q_n) \remainder (f_0; q_n)
	}\cdot O_{\Pr}(1),
    \end{align*}
    where the last equality holds because $\max\{1, b_n\}\le n-q_n$. 
    From the condition in Corollary \ref{cor:diverg_kernel}, we must have $\Delta_n = o_{\Pr}(1)$. 
    Letting $n$ go to infinity in \eqref{eq:pval_bound_diverging}, we can know that, for any $\alpha, \delta>0$, 
    \begin{align*}
    \limsup_{n\rightarrow \infty}\Pr \{  p( \bXmatrix, \bs{Y}, \bs{Z} ) \leq \alpha \}
    \leq & 
    \alpha+ \delta + 
    \limsup_{n\rightarrow \infty} \Pr \left\{ \Delta_n > \log(1+2\delta) \right\}
    = 
    \alpha+\delta. 
    \end{align*}
    Because the above inequality holds for any $\delta> 0$, 
    we must have 
    $
    \limsup_{n\rightarrow \infty}\Pr \{  p( \bXmatrix, \bs{Y}, \bs{Z} ) \leq \alpha  \}
    \le 
    \alpha. 
    $
    Therefore, $p( \bXmatrix, \bs{Y}, \bs{Z} )$ is an asymptotically valid $p$-value. 
    
    From the above, Corollary \ref{cor:diverg_kernel} holds. 
\end{proof}

\begin{proof}[\bf Proof of Corollary \ref{cor:fixed_property_theorem_balanced_design}]
    From Lemmas \ref{lemma:basic_ineq_perm_p_val} and \ref{lemma:bound_TV} and by the same logic as the proof of Corollary \ref{cor:kernel_finite_dim_feature_space}, it suffices to prove that $\omega(b_n, \sigma_0^{-1}f_0) = 0$ for $b_n \le n - r$. 
    For descriptive convenience, we introduce $m=n/H$. 
    
    Without loss of generality, we assume the units are ordered according to their group indicators in the sense that 
    $Z_{(k-1)m + 1} = Z_{(k-1)m + 1} = \ldots = Z_{km} = k$ for $1\le k \le H$, 
    and the covariates in these groups satisfy 
    \begin{align}\label{eq:cov_balance_groups}
        \begin{pmatrix}
            \bs{X}_1^\top\\
            \bs{X}_2^\top\\
            \vdots \\
            \bs{X}_m^\top
        \end{pmatrix}
        = 
        \begin{pmatrix}
            \bs{X}_{m+1}^\top\\
            \bs{X}_{m+2}^\top\\
            \vdots \\
            \bs{X}_{2m}^\top
        \end{pmatrix}
        = 
        \ldots
        = 
        \begin{pmatrix}
            \bs{X}_{(H-1)m+1}^\top\\
            \bs{X}_{(H-1)m+2}^\top\\
            \vdots \\
            \bs{X}_{Hm}^\top
        \end{pmatrix}. 
    \end{align}
    We further assume that the covariates within each group are ordered such that samples with the same covariate values are ordered consecutively, and use $m_1, m_2, \ldots, m_r$ to denote the number of samples with distinct covariate values within each group. Obviously, $\sum_{l=1}^r m_l = n/H$. 
    Let $(\tilde{\bs{X}}_1, \ldots, \tilde{\bs{X}}_r)$ denote these distinct values of covariates. 
    
    First, we consider simplifying the kernel matrix $\bs{K}_n$. 
    Let $\bs{G} \in \mathbb{R}^{r \times r}$ be the kernel matrix for the $r$ distinct covariate values in each group, with $G_{ij} = K(\tilde{\bs{X}}_i, \tilde{\bs{X}}_j)$ being its $(i,j)$ element. 
    We can then verify that the kernel matrix $\bs{J}_m$ for samples within each group has the following equivalent forms:
    \begin{align*}
        \bs{J}_m & = 
        \begin{pmatrix}
            K(\bs{X}_1, \bs{X}_1) & K(\bs{X}_1, \bs{X}_2) & \cdots & K(\bs{X}_1, \bs{X}_m)\\
            K(\bs{X}_2, \bs{X}_1) & K(\bs{X}_2, \bs{X}_2) & \cdots & K(\bs{X}_2, \bs{X}_m)\\
            \vdots & \vdots & \ddots & \vdots\\
            K(\bs{X}_m, \bs{X}_1) & K(\bs{X}_m, \bs{X}_2) & \cdots & K(\bs{X}_m, \bs{X}_m)
        \end{pmatrix}
        = 
        \begin{pmatrix}
        G_{11} \bs{1}_{m_1}\bs{1}_{m_1}^\top & G_{12} \bs{1}_{m_1}\bs{1}_{m_2}^\top & \vdots & G_{1r} \bs{1}_{m_1}\bs{1}_{m_r}^\top\\
        G_{21} \bs{1}_{m_2}\bs{1}_{m_1}^\top & G_{22} \bs{1}_{m_2}\bs{1}_{m_2}^\top & \vdots & G_{2r} \bs{1}_{m_2}\bs{1}_{m_r}^\top\\
        \vdots & \vdots & \ddots & \vdots\\
        G_{r1} \bs{1}_{m_r}\bs{1}_{m_1}^\top & G_{r2} \bs{1}_{m_r}\bs{1}_{m_2}^\top & \vdots & G_{rr} \bs{1}_{m_r}\bs{1}_{m_r}^\top
        \end{pmatrix}\\
        & = 
        \begin{pmatrix}
            \bs{1}_{m_1} & \bs{0} & \cdots & \bs{0}\\
            \bs{0} & \bs{1}_{m_2} & \cdots & \bs{0}\\
            \vdots & \vdots & \ddots & \vdots\\
            \bs{0} & \bs{0} & \cdots & \bs{1}_{m_r}
        \end{pmatrix}
        \begin{pmatrix}
            G_{11} & G_{12} & \cdots & G_{1r}\\
            G_{21} & G_{22} & \cdots & G_{2r}\\
            \vdots & \vdots & \ddots & \vdots\\
            G_{r1} & G_{r2} & \cdots & G_{rr}
        \end{pmatrix}
        \begin{pmatrix}
            \bs{1}_{m_1}^\top & \bs{0} & \cdots & \bs{0}\\
            \bs{0} & \bs{1}_{m_2}^\top & \cdots & \bs{0}\\
            \vdots & \vdots & \ddots & \vdots\\
            \bs{0} & \bs{0} & \cdots & \bs{1}_{m_r}^\top
        \end{pmatrix}, 
    \end{align*}
    where $\bs{1}_{m_l} \in \mathbb{R}^{m_l}$ denotes a column vector with all elements being 1. 
    Consequently, the kernel matrix $\bs{K}_n$ for all samples 
    has the following block structure and simplifies to 
    \begin{align*}
        \bs{K}_n & = 
        \begin{pmatrix}
            \bs{J}_m & \bs{J}_m & \cdots & \bs{J}_m\\
            \bs{J}_m & \bs{J}_m & \cdots & \bs{J}_m\\
            \vdots & \vdots & \ddots & \vdots\\
            \bs{J}_m & \bs{J}_m & \cdots & \bs{J}_m
        \end{pmatrix}
        = 
        \begin{pmatrix}
            \bs{I}_m\\
            \bs{I}_m\\
            \vdots\\
            \bs{I}_m
        \end{pmatrix}
        \bs{J}_m
        \begin{pmatrix}
            \bs{I}_m & \bs{I}_m & \ldots & \bs{I}_m
        \end{pmatrix}
        \\
        & = 
        \begin{pmatrix}
            \bs{I}_m\\
            \bs{I}_m\\
            \vdots\\
            \bs{I}_m
        \end{pmatrix}
        \begin{pmatrix}
            \bs{1}_{m_1} & \bs{0} & \cdots & \bs{0}\\
            \bs{0} & \bs{1}_{m_2} & \cdots & \bs{0}\\
            \vdots & \vdots & \ddots & \vdots\\
            \bs{0} & \bs{0} & \cdots & \bs{1}_{m_r}
        \end{pmatrix}
        \bs{G}
        \begin{pmatrix}
            \bs{1}_{m_1}^\top & \bs{0} & \cdots & \bs{0}\\
            \bs{0} & \bs{1}_{m_2}^\top & \cdots & \bs{0}\\
            \vdots & \vdots & \ddots & \vdots\\
            \bs{0} & \bs{0} & \cdots & \bs{1}_{m_r}^\top
        \end{pmatrix}
        \begin{pmatrix}
            \bs{I}_m & \bs{I}_m & \ldots & \bs{I}_m
        \end{pmatrix}\\
        & = 
        \bs{\Pi} 
        \bs{G} 
        \bs{\Pi}^\top, 
    \end{align*}
    where $\bs{I}_m\in \mathbb{R}^{m\times m}$ denotes an identity matrix, and 
    \begin{align}\label{eq:Pi}
        \bs{\Pi} & = 
        \begin{pmatrix}
            \bs{I}_m\\
            \bs{I}_m\\
            \vdots\\
            \bs{I}_m
        \end{pmatrix}
        \begin{pmatrix}
            \bs{1}_{m_1} & \bs{0} & \cdots & \bs{0}\\
            \bs{0} & \bs{1}_{m_2} & \cdots & \bs{0}\\
            \vdots & \vdots & \ddots & \vdots\\
            \bs{0} & \bs{0} & \cdots & \bs{1}_{m_r}
        \end{pmatrix}. 
    \end{align}
    From the condition in Corollary \ref{cor:fixed_property_theorem_balanced_design}, 
    $\bs{G}$ is positive define. 
    Thus, we can verify that the kernel matrix $\bs{K}_n$ is of rank $r$, and the smallest $n-r$ eigenvalues of $\bs{K}_n$ must all be zero. 
    Thus, for $r+1 \le i \le n$, 
    the eigenvector corresponding to the $i$th largest eigenvalue must satisfy that 
    $
    \bs{\gamma}_i^\top \bs{\Pi} \bs{G} \bs{\Pi}^\top\bs{\gamma}_i  = \bs{\gamma}_i^\top \bs{K}_n \bs{\gamma}_i = 0, 
    $
    which further implies that $\bs{\gamma}_i^\top \bs{\Pi} = \bs{0}$.
    
    Second, we consider the function value vector $\bs{f}_0 = (f_0(\bs{X}_1), f_0(\bs{X}_2), \ldots, f_0(\bs{X}_n))^\top$. 
    By the property of the covariate matrix $\bXmatrix$, 
    the vector $\bs{f}_0 \in \mathbb{R}^n$ 
    has the following block structure and simplifies to  
    \begin{align*}
        \bs{f}_0 & = 
        \begin{pmatrix}
            f_0(\bs{X}_1)\\
            f_0(\bs{X}_2)\\
            \vdots\\
            f_0(\bs{X}_n)
        \end{pmatrix}
        = 
        \begin{pmatrix}
            \bs{I}_m\\
            \bs{I}_m\\
            \vdots\\
            \bs{I}_m
        \end{pmatrix}  
        \begin{pmatrix}
            f_0(\bs{X}_1)\\
            f_0(\bs{X}_2)\\
            \vdots\\
            f_0(\bs{X}_m)
        \end{pmatrix}
        = 
        \begin{pmatrix}
            \bs{I}_m\\
            \bs{I}_m\\
            \vdots\\
            \bs{I}_m
        \end{pmatrix}  
        \begin{pmatrix}
            f_0(\tilde{\bs{X}}_1) \bs{1}_{m_1}\\
            f_0(\tilde{\bs{X}}_2) \bs{1}_{m_2}\\
            \vdots\\
            f_0(\tilde{\bs{X}}_r) \bs{1}_{m_r}
        \end{pmatrix} 
        \\
        & = 
        \begin{pmatrix}
            \bs{I}_m\\
            \bs{I}_m\\
            \vdots\\
            \bs{I}_m
        \end{pmatrix}
        \begin{pmatrix}
            \bs{1}_{m_1} & \bs{0} & \cdots & \bs{0}\\
            \bs{0} & \bs{1}_{m_2} & \cdots & \bs{0}\\
            \vdots & \vdots & \ddots & \vdots\\
            \bs{0} & \bs{0} & \cdots & \bs{1}_{m_r}
        \end{pmatrix}
        \begin{pmatrix}
            f_0(\tilde{\bs{X}}_1) \\
            f_0(\tilde{\bs{X}}_2) \\
            \vdots\\
            f_0(\tilde{\bs{X}}_r) 
        \end{pmatrix} 
        = 
        \bs{\Pi} \tilde{\bs{f}}_0, 
    \end{align*}
    where $\tilde{\bs{f}}_0 = (f_0(\tilde{\bs{X}}_1), f_0(\tilde{\bs{X}}_2), \ldots, f_0(\tilde{\bs{X}}_r))^\top$. 
    From the discussion before,  $
    \bs{\gamma}^\top_i \bs{f}_0 = \bs{\gamma}_i^\top \bs{\Pi} \tilde{\bs{f}}_0 = 0
    $ for $r+1 \le i \le n$.  
    Therefore, for $1 \le b_n \le n - r$, we have
    $
        \omega(b_n,\sigma_0^{-1}f_0) =  \sigma_0^{-2} \sum_{i=n-b_n+1}^n (\bs{\gamma}_{i}^\top \bs{f}_0)^2 = 0.
    $
    
    From the above, Corollary \ref{cor:fixed_property_theorem_balanced_design} holds. 
\end{proof}

\section{Partial Permutation Test under Gaussian Process Regression}\label{sec:proof_GPR_null}

To prove Proposition \ref{thm:general_consistency}, we need the following two lemmas. 
\begin{lemma}\label{bayesian_penal}
The posterior mean induced by the GPR model $\tilde{H}_0$ in \eqref{eq:h0_gp} is the same as $\hat{f}_{n,\tau_n}$ defined in formula \eqref{eq:kernel_penal} when $\tau_n = \sigma_0^2/(n^{\gamma}\delta_0^2)$.
\end{lemma}
\begin{proof}[Proof of Lemma \ref{bayesian_penal}]
Note that the solution from kernel based regression in  \eqref{eq:kernel_penal} is
\begin{align*}
\hat{f}_{n,\tau_n}(\bs{x}) =& \left(
K(\bs{x},\bs{X}_1), K(\bs{x},\bs{X}_2), \cdots, K(\bs{x},\bs{X}_n)
\right)(\bs{K}_n + n\tau_n \bs{I})^{-1} \bs{Y},
\end{align*}
and 
the posterior mean induced by model $\tilde{H}_0$ in \eqref{eq:h0_gp} is
\begin{align*}
\tilde{f}_{n}(\bs{x}) =& \left(
K(\bs{x},\bs{X}_1), K(\bs{x},\bs{X}_2), \cdots, K(\bs{x},\bs{X}_n)
\right)
 \left( 
 \bs{K}_n + n\frac{\sigma_0^2}{n^{\gamma}\delta_0^2}\bs{I}
 \right)^{-1} \bs{Y}.
\end{align*}
Therefore, if $\tau_n = {\sigma_0^2}/(n^{\gamma}\delta_0^2)$, then  $\hat{f}_{n,\tau_n}=\tilde{f}_{n}$.
\end{proof}

\begin{lemma}%
\label{consistency_penalty}
Under model $H_0$ in \eqref{eq:H_0_fixed}, if 
$\mathcal{X} \in \mathcal{R}^d$ is compact, and $\mathcal{H}_K$ is a RKHS of a universal kernel on $\mathcal{X}$, then for any sequence $\{\tau_n\} \subset (0,\infty)$ with $\tau_n \rightarrow 0$ and $\tau_n^{4}n \rightarrow \infty$,  $\hat{f}_{n,\tau_n}$ in \eqref{eq:kernel_penal} is consistent for the true $f$, that is,
$
\mathbb{E} | \hat{f}_{n,\tau_n}(\bs{X}) -  f(\bs{X}) |^2   \stackrel{\Pr}{\longrightarrow} 0. 
$
\end{lemma}
\begin{proof}[Proof of Lemma \ref{consistency_penalty}]
Lemma \ref{consistency_penalty} follows directly from  \citet[][Theorem 12]{christmann2007consistency} with squared loss. 
\end{proof}

\begin{proof}[{\bf Proof of Proposition \ref{thm:general_consistency}}]
Let $\tau_n = \sigma_0^2/(n^{\gamma}\delta_0^2).$ Because $0<\gamma<1/4$, we have
$
\tau_n = {\sigma_0^2}/( {n^{\gamma}\delta_0^2} ) \rightarrow 0,$
and 
$\tau_n^{4}n = {\sigma_0^8}/{\delta_0^8} \cdot n^{1-4\gamma}\rightarrow \infty.
$
From Lemmas \ref{bayesian_penal} and \ref{consistency_penalty}, Proposition \ref{thm:general_consistency} holds. 
\end{proof}

To prove Theorem \ref{thm:gp_finite_property_theorem}, we need the following two lemmas.

\begin{lemma}\label{total_var_eta}
Let $\bs{W} =(W_1,\ldots, W_n)$ be a random vector in $\mathbb{R}^n$, and $\bs{W} \sim \mathcal{N}(\bs{0}, \sigma^2\bs{D} + \sigma_0^2 \bs{I}_n)$, where $\bs{D}$ is a diagonal matrix with diagonal elements $d_1\geq d_2 \geq \cdots \geq d_n$. Let $\bs{w}=(w_1,\cdots, w_n)$ be a constant vector in $\mathbb{R}^n$, and $\mathcal{S}$ be a discrete permutation set of $\bs{w}$, 
\begin{align*}
	\mathcal{S} =& \left\{ \bs{w}_{\psi}=(w_{\psi(1)},\cdots, w_{\psi(n)}): \ \ \ \psi \in \mathcal{M}(b,b_n)
	\right\},
\end{align*}
or a continuous permutation set of $\bs{w}$,
\begin{align*}
\mathcal{S} = \{\bs{w}^*: 
\bs{w}^*_i = \bs{w}_i, 1\leq i \leq n-b_n, \sum_{j=n-b_n+1}^n {\bs{w}^*_j}^2
= \sum_{j=n-b_n+1}^n {\bs{w}_j}^2
\}.
\end{align*}
Let $\bs{W}^p$ be a random vector uniformly distributed on $\mathcal{S}$, and define measures $\nu$ and $\nu_0$ on $\mathbb{R}^n$ as
\begin{align*}
\nu(\mathcal{A}) & = \Pr(\bs{W} \in \mathcal{A} \mid \bs{W} \in \mathcal{S}), \quad 
\nu_0(\mathcal{A}) = \Pr(\bs{W}^p \in \mathcal{A} \mid \mathcal{S}), 
\qquad \text{for any measurable $\mathcal{A} \subset \mathbb{R}^n$}. 
\end{align*}
Then we have
$
\left\| \nu - \nu_0 \right\|_{\TV}\leq (e^{\overline{\Delta}_{n}}-1)/2,
$
with 
$$
\overline{\Delta}_{n} = 
\frac{d_{a_n} \sigma^{2}}{2\sigma_{0}^{2}(\sigma^2 d_{a_n} + \sigma_0^2)} \sum_{j=a_{n}}^{n} w_{j}^{2}
$$ 
and $a_n = n-b_n+1$.
\end{lemma}

\begin{proof}[Proof of Lemma \ref{total_var_eta}]
We prove only the case with a discrete permutation set, since the proof for the case with a continuous permutation set is very similar.
For any $\psi \in \mathcal{M}(n,b_n)$,  measure $\nu$ satisfies 
\begin{align*}
\nu(\{\bs{w}_\psi\}) & \propto   (2\pi)^{-n/2}\left\{\prod_{j=1}^{n}
(\sigma^{2}d_{j}+\sigma_{0}^{2})\right\}^{-1/2}\exp\left\{ -\sum_{j=1}^{a_{n}-1}\frac{w_{j}^{2}}{2(\sigma^{2}d_{j}+\sigma_{0}^{2})}\right\} \exp\left\{ -\sum_{j=a_{n}}^{n}\frac{w_{\psi(j)}^{2}}{2(\sigma^{2}d_{j}+\sigma_{0}^{2})}\right\} \\
& \propto \exp\left\{ -\sum_{j=a_{n}}^{n}\frac{w_{\psi(j)}^{2}}{2(\sigma^{2}d_{j}+\sigma_{0}^{2})}\right\}. 
\end{align*}
This implies that for  any $\psi \in \mathcal{M}(n,b_n)$, 
\begin{align*}
\nu(\{\bs{w}_\psi\}) = \frac{1}{C}\exp\left\{ -\sum_{j=a_{n}}^{n}\frac{w_{\psi(j)}^{2}}{2(\sigma^{2}d_{j}+\sigma_{0}^{2})}\right\},\quad 
\text{with }
C=\sum_{\phi\in \mathcal{M}(n,b_n)} \exp\left\{ -\sum_{j=a_{n}}^{n}\frac{w_{\phi(j)}^{2}}{2(\sigma^{2}d_{j}+\sigma_{0}^{2})}\right\}. 
\end{align*}
Note that for any $\psi \in \mathcal{M}(n, b_n)$, 
\begin{align*}
0
& \le  
\sum_{j=a_{n}}^{n}\frac{w_{\psi(j)}^{2}}{2\sigma_{0}^{2}} -\sum_{j=a_{n}}^{n}\frac{w_{\psi(j)}^{2}}{2(\sigma^{2}d_{j}+\sigma_{0}^{2})} \\
& \leq 
\sum_{j=a_{n}}^{n}\frac{w_{\psi(j)}^{2}}{2\sigma_{0}^{2}} -\sum_{j=a_{n}}^{n}\frac{w_{\psi(j)}^{2}}{2(\sigma^{2}d_{a_n}+\sigma_{0}^{2})}  = 
\frac{1}{2}\left( \frac{1}{\sigma_0^2} - \frac{1}{\sigma^2d_{a_n} + \sigma_0^2} \right)\sum_{j=a_{n}}^{n}w_{\psi(j)}^{2} 
\\
& = 
\frac{\sigma^{2}d_{a_n}}{2\sigma_0^2(\sigma^{2}d_{a_n}+\sigma_{0}^{2})}\sum_{j=a_{n}}^{n}w_{j}^{2} 
=\overline{\Delta}_n.
\end{align*}
Thus, for any $\psi \in \mathcal{M}(n,b_n)$,
\begin{align*}
1 \leq \frac{\exp\left\{ -\sum_{j=a_{n}}^{n}\frac{w_{\psi(j)}^{2}}{2(\sigma^{2}d_{j}+\sigma_{0}^{2})}\right\}}{\exp\left(-\sum_{j=a_{n}}^{n}\frac{w_{j}^{2}}{2\sigma_{0}^{2}}\right)} 
= 
\exp\left\{
\sum_{j=a_{n}}^{n}\frac{w_{\psi(j)}^{2}}{2\sigma_{0}^{2}} -\sum_{j=a_{n}}^{n}\frac{w_{\psi(j)}^{2}}{2(\sigma^{2}d_{j}+\sigma_{0}^{2})}
\right\}
\leq e^{\overline{\Delta}_n}.
\end{align*}
Let $C_0 = \exp\{-\sum_{j=a_{n}}^{n}w_{j}^{2}/({2\sigma_{0}^{2}})\}$. We then have
\begin{align*}
C_0 \le \exp\left\{ -\sum_{j=a_{n}}^{n}\frac{w_{\psi(j)}^{2}}{2(\sigma^{2}d_{j}+\sigma_{0}^{2})}\right\} \le C_0 \cdot e^{\overline{\Delta}_n},
\end{align*}
and 
\begin{align*}
|\mathcal{M}(n,b_n)| \cdot C_0 
\le 
C=\sum_{\phi\in \mathcal{M}(n,b_n)} \exp\left\{ -\sum_{j=a_{n}}^{n}\frac{w_{\phi(j)}^{2}}{2(\sigma^{2}d_{j}+\sigma_{0}^{2})}\right\} \le |\mathcal{M}(n,b_n)|\cdot C_0 \cdot e^{\overline{\Delta}_n}, 
\end{align*}
where $|\mathcal{M}(n,b_n)|$ denotes the cardinality of the set $\mathcal{M}(n,b_n)$. 
These imply that
\begin{align*}
\nu(\{\bs{w}_\psi\}) = \frac{1}{C}\exp\left\{ -\sum_{j=a_{n}}^{n}\frac{w_{\psi(j)}^{2}}{2(\sigma^{2}d_{j}+\sigma_{0}^{2})}\right\} \in \left[\frac{1}{|\mathcal{M}(n,b_n)|}e^{-\overline{\Delta}_n} ,\  \frac{1}{|\mathcal{M}(n,b_n)|}e^{\overline{\Delta}_n}\right].
\end{align*}
Consequently, 
\begin{align*}
\nu(\{\bs{w}_\psi\}) - \nu_0(\{\bs{w}_\psi\}) \in \left[ - \frac{1}{|\mathcal{M}(n,b_n)|}(1 - e^{-\overline{\Delta}_n}) ,\  \frac{1}{|\mathcal{M}(n,b_n)|}(e^{\overline{\Delta}_n}-1)\right],
\end{align*}
which immediately implies that
\begin{align*}
\left| \nu(\{\bs{w}_\psi\}) - \nu_0(\{\bs{w}_\psi\}) \right| \leq \frac{1}{|\mathcal{M}(n,b_n)|}\max\{1-e^{-\overline{\Delta}_n}, e^{\overline{\Delta}_n}-1\}\leq \frac{1}{|\mathcal{M}(n,b_n)|}(e^{\overline{\Delta}_{n}}-1).
\end{align*}
Therefore,
\begin{align*}
\|\nu - \nu_0 \|_{\TV} = & \frac{1}{2}\sum_{\psi\in \mathcal{M}(n,b_n)}\left| \nu(\{\bs{w}_\psi\}) - \nu_0(\{\bs{w}_\psi\}) \right|
\leq \frac{1}{2}(e^{\overline{\Delta}_{n}}-1), 
\end{align*}
i.e., Lemma \ref{total_var_eta} holds. 
\end{proof}

\begin{lemma}\label{lemma:bound_TV_GP}
Let $\{(\bs{X}_i,Y_i,Z_i)\}_{1\leq i \leq n}$ be samples from model $\tilde{H}_0$ in \eqref{eq:h0_gp},
$\mathcal{S}_y$ be the discrete (or continuous) permutation set defined as in Algorithm \ref{alg:partial_permu}, $\bs{Y}^p$ be a random vector uniformly distributed on $\mathcal{S}_y$ given $\bXmatrix, \mathcal{S}_y$ and $\bs{Z}$, and $\nu$ and $\nu_0$ be two measures on $\mathbb{R}^n$ defined as 
\begin{align*}
\nu(\mathcal{A}) = \Pr( \bs{Y} \in \mathcal{A} \mid \bXmatrix, \mathcal{S}_y, \bs{Z} ),
\qquad
\nu_0(\mathcal{A}) = \Pr(\bs{Y}^p \in \mathcal{A} \mid \bXmatrix,\mathcal{S}_y, \bs{Z} ),
\end{align*}
for any measurable set $\mathcal{A}\subset \mathbb{R}^n$. 
Then for any $1 \le b_n \le n$, 
$
   \|\nu - \nu_0\|_{\TV} \le  (e^{\tilde{\Delta}_n}-1)/2, 
$
where 
\begin{align*}
\tilde{\Delta}_n 
&
\equiv 
\frac{c_{n-b_n+1} \delta_0^2/n^{1-\gamma}}{ 2 \sigma_{0}^{2}} \sum_{j=n-b_n+1}^{n} \frac{\left(\bs{\gamma}_j^\top \bs{Y}\right)^{2}}{c_{j} \delta_0^2/n^{1-\gamma} + \sigma_0^2}
= 
\frac{c_{n-b_n+1} \delta_0^2/n^{1-\gamma}}{ 2 \sigma_{0}^{2}} \sum_{j=n-b_n+1}^{n} \frac{W_{j}^{2}}{c_{j} \delta_0^2/n^{1-\gamma} + \sigma_0^2}, 
\end{align*}
where $\bs{\gamma}_j$ is the eigenvector of the kernel matrix $\bs{K}_n$ corresponding to the $j$th largest eigenvalue, and $c_{n-b_n+1}$ is the $(n-b_n+1)$th largest eigenvalue of $\bs{K}_n$. 
\end{lemma}

\begin{proof}[Proof of Lemma \ref{lemma:bound_TV_GP}]
We prove only the case with the discrete permutation set $\mathcal{S}_y$, since the proof for the case with continuous permutation set is very similar. 
Under model $\tilde{H}_0$ in \eqref{eq:h0_gp}, we have 
\begin{align*}
\bs{Y} \mid \bXmatrix, \bs{Z},\tilde{H}_{0} \sim \mathcal{N}\left(\bs{0}, \frac{\delta_0^{2}}{n^{1-\gamma}}\bs{K}_n+\sigma_{0}^{2}\bs{I}_n\right).
\end{align*}
Recall that $\bs{K}_n=\bs{\Gamma} \bs{C} \bs{\Gamma}^\top$ is the eigen-decomposition of $\bs{K}_n$, where $\bs{\Gamma}$ is an orthogonal matrix and $\bs{C}=\text{diag}\{c_1,\cdots, c_n\}$ with $c_1 \geq c_2 \geq \cdots \geq c_n$, $\bs{W} = \bs{\Gamma}^\top \bs{Y}$, and 
$\mathcal{S}_{w} = \bs{\Gamma}^\top \mathcal{S}_y \equiv \{\bs{\Gamma}^\top \bs{y}: \bs{y} \in \mathcal{S}_y \}$. 
We then have
\begin{align*}
\bs{W} \mid \bXmatrix, \bs{Z}, \tilde{H}_{0}\sim \mathcal{N}\left(\bs{0},\frac{\delta_0^{2}}{n^{1-\gamma}}\bs{C} +\sigma_{0}^{2}\bs{I}_n\right).
\end{align*}
Moreover, for any measurable set $\mathcal{A} \in \mathbb{R}^n$, 
\begin{align*}
    \nu(\mathcal{A}) & = \Pr( \bs{Y} \in \mathcal{A} \mid \bXmatrix, \mathcal{S}_y, \bs{Z} ) 
    = \Pr( \bs{W} \in \bs{\Gamma}^\top \mathcal{A} \mid \bs{W} \in  \mathcal{S}_w,  \bXmatrix, \bs{Z}), \\
    \nu_0(\mathcal{A}) & = \Pr(\bs{Y}^p \in \mathcal{A} \mid \bXmatrix,\mathcal{S}_y, \bs{Z} )
    = 
    \Pr(\bs{W}^p \in \bs{\Gamma}^\top \mathcal{A} \mid  \mathcal{S}_w,  \bXmatrix, \bs{Z}), 
\end{align*}
where $\bs{\Gamma}^\top \mathcal{A} = \{\bs{\Gamma}^\top \bs{y}: \bs{y} \in \mathcal{A}\}$ and $\bs{W}^p$ is uniformly distributed in $\mathcal{S}_{w}$. 
From Lemma \ref{total_var_eta}, the total variation distance between the two measures $\nu$ and $\nu_0$ can be bounded by 
\begin{align*}
    \| \nu - \nu_0 \|_{\TV} & \le 
    \frac{1}{2}
    \exp\left\{
    \frac{c_{n-b_n+1} \delta_0^2/n^{1-\gamma}}{2\sigma_{0}^{2}(c_{n-b_n+1} \delta_0^2/n^{1-\gamma}  + \sigma_0^2)} \sum_{j=n - b_n + 1 }^{n} W_{j}^{2} \right\} - \frac{1}{2}
    \\
    & = 
    \frac{1}{2}
    \exp\left\{
    \frac{c_{n-b_n+1} \delta_0^2/n^{1-\gamma}}{2\sigma_{0}^{2}} \sum_{j=n-b_{n}+1}^{n} 
    \frac{W_{j}^{2}}{c_{n-b_n+1} \delta_0^2/n^{1-\gamma}  + \sigma_0^2} \right\} - \frac{1}{2}
    \\
    & \le 
    \frac{1}{2}
    \exp\left\{
    \frac{c_{n-b_n+1} \delta_0^2/n^{1-\gamma}}{2\sigma_{0}^{2}} \sum_{j=n-b_n+1}^{n} 
    \frac{W_{j}^{2}}{c_{j} \delta_0^2/n^{1-\gamma}  + \sigma_0^2} \right\} - \frac{1}{2}
    = 
    \frac{1}{2} \left( e^{\tilde{\Delta}} - 1 \right). 
\end{align*}
Therefore, Lemma \ref{lemma:bound_TV_GP} holds. 
\end{proof}

\begin{proof}[{\bf Proof of Theorem \ref{thm:gp_finite_property_theorem}}]
    Define $\nu, \nu_0$ and $\tilde{\Delta}_n$ the same as in Lemma \ref{lemma:bound_TV_GP}. 
	From Lemma \ref{lemma:basic_ineq_perm_p_val}, for any $\alpha, \eta \ge 0$, 
	\begin{align*}
	\Pr \{  p(\bXmatrix, \bs{Y}, \bs{Z} ) \leq \alpha  \mid \bXmatrix, \bs{Z} \} 
	& \leq  \alpha+ \eta +\Pr \{\| \nu - \nu_0 \|_{\TV} > \eta \mid \bXmatrix, \bs{Z}  \}, 
	\end{align*}
	Given any $\alpha \in (0,1)$, let $\eta = \tilde{v} (b_n, \xi_n, \alpha_0)$. 
	By definition and from Lemma \ref{lemma:bound_TV_GP}, we have 
	\begin{align*}
	   \| \nu - \nu_0 \|_{\TV} > \tilde{v} (b_n, \xi_n, \alpha_0)
	   & \Longrightarrow 
	   \frac{1}{2} \left( e^{\tilde{\Delta}} - 1 \right) > 
	   \frac{1}{2} \exp\left\{ \frac{1}{2} c_{n-b_n+1} \frac{\delta_0^2/n^{1-\gamma}}{\sigma_0^2} Q_{b_n}(1-\alpha_0) \right\} - \frac{1}{2}
	   \\
	   & \Longrightarrow  \tilde{\Delta} > \frac{1}{2} c_{n-b_n+1} \frac{\delta_0^2/n^{1-\gamma}}{\sigma_0^2} Q_{b_n}(1-\alpha_0)
	   \\
	   & \Longrightarrow 
	   \sum_{j=n-b_n+1}^{n} 
    \frac{W_{j}^{2}}{c_{j} \delta_0^2/n^{1-\gamma}  + \sigma_0^2} > Q_{b_n}(1-\alpha_0). 
	\end{align*}
	Note that conditional on $(\bXmatrix, \bs{Z})$, 
	$
	\sum_{j=a_{n}}^{n} W_j^2 / (c_{j} \delta_0^2/n^{1-\gamma}  + \sigma_0^2)
	$
	follows $\chi^2$ distribution with degrees of freedom $b_n$. 
	Thus, 
	\begin{align*}
	   \Pr \{\| \nu - \nu_0 \|_{\TV} > \tilde{v} (b_n, \xi_n, \alpha_0) \mid \bXmatrix, \bs{Z}  \} 
	   \le 
	   \Pr\left\{
	   \sum_{j=n-b_n+1}^{n} 
    \frac{W_{j}^{2}}{c_{j} \delta_0^2/n^{1-\gamma}  + \sigma_0^2} > Q_{b_n}(1-\alpha_0) 
    \mid  \bXmatrix, \bs{Z} 
	   \right\}
	   = \alpha_0, 
	\end{align*}
	and consequently, 
	\begin{align*}
	\Pr \{  p( \bXmatrix, \bs{Y}, \bs{Z} ) \leq \alpha  \mid \bXmatrix, \bs{Z} \} 
	& \leq  \alpha+ \tilde{v}(b_n, \xi_n, \alpha_0) +\Pr \{\| \nu - \nu_0 \|_{\TV} > \tilde{v}(b_n, \xi_n, \alpha_0) \mid \bXmatrix, \bs{Z}  \}
	\nonumber
	\\
	& = \alpha+ \tilde{v}(b_n, \xi_n, \alpha_0) + \alpha_0. 
	\end{align*}
    By definition, this immediately implies that, for any $\alpha\in (0,1)$, 
	\begin{align*}
	& \quad \ \Pr(\tilde{p}_c( \bXmatrix, \bs{Y}, \bs{Z} )\leq\alpha \mid \bXmatrix,  \bs{Z} )
	= \Pr\left\{
	p( \bXmatrix, \bs{Y}, \bs{Z} )  \leq \alpha-\tilde{v}(b_n, \xi_n, \alpha_0)- \alpha_0 \mid \bXmatrix,  \bs{Z}
	\right\}\\
	& \le 
	\I(\alpha \ge \tilde{v}(b_n, \xi_n, \alpha_0) + \alpha_0) \cdot
	\Pr\left\{
	p( \bXmatrix, \bs{Y}, \bs{Z} )  \leq \alpha-\tilde{v}(b_n, \xi_n, \alpha_0)- \alpha_0 \mid \bXmatrix,  \bs{Z}
	\right\} \\
	& \le \I(\alpha \ge \tilde{v}(b_n, \xi_n, \alpha_0) + \alpha_0) \cdot \alpha 
	\le \alpha. 
	\end{align*}
	Therefore, 
	Theorem \ref{thm:gp_finite_property_theorem} holds. 
\end{proof}

To prove Theorem \ref{thm:general_ppt_valid}, we need the following three lemmas. 

\begin{lemma}\label{lemma:accurate_bound}
Let $K$ be a Mercer kernel on a probability space $\mathcal{X}$
with probability measure $\mu$, which can be written as
\[
K(\bs{x}_1,\bs{x}_2)=\sum_{i=1}^{\infty}\lambda_{i}\psi_{i}(\bs{x}_1)\psi_{i}(\bs{x}_2),
\]
where $\{\lambda_i\}_{i\geq 1}$ is a sequence of summable non-negative, non-increasing numbers, and $\{\psi_i\}_{i\geq 1}$ is a family of mutually orthogonal unit norm functions with respect to the scalar product 
$(f,g)\mapsto \int _{\mathcal{X}}fg\text{d}\mu$. 
We have
\begin{itemize}
\item[1)]for $1\leq r\leq n$, $1\leq i \le  n$, 
	\begin{align*}
	|l_{i}-\lambda_{i}|\leq \lambda_{i}C(r,n)+E(r,n), 
	\end{align*}
where $l_{1},l_{2},\cdots,l_{n}$ are the eigenvalues of ${\bs{K}}_{n}/n$
	satisfying $l_{1}\geq l_{2}\geq\cdots\geq l_{n}\geq0$,
$\left[\bs{K}_{n}\right]_{ij}=K(\bs{X}_{i},\bs{X}_{j})$, $\bs{X}_i$'s are i.i.d.\ samples from $(\mathcal{X}, \mu)$, and 	
	 $C(r,n)$ and $E(r,n)$ are defined in \citet*{braun2006accurate};
\item[2)]if $K$ also satisfies $K(\bs{x},\bs{x})\leq K_0<\infty$
	for all $x\in\mathcal{X}$, then, for $1\leq r\leq n$, with probability
	larger than $1-\delta$, 
	\begin{align*}
	C(r,n) & < \tilde{C}(r, n, \delta) \equiv  r\sqrt{\frac{2K_0}{n\lambda_{r}}\log\frac{2r(r+1)}{\delta}}+\frac{4K_0r}{3n\lambda_{r}}\log\frac{2r(r+1)}{\delta},\\
	E(r,n) & < \tilde{E}(r, n, \delta) \equiv \lambda_{r}+\Lambda_{>r}+\sqrt{\frac{2K_0\Lambda_{>r}}{n}\log\frac{2}{\delta}}+\frac{2K_0}{3n}\log\frac{2}{\delta},
	\end{align*}
where $\Lambda_{>r}=\sum_{i=r+1}^{\infty}\lambda_i$.
\end{itemize}
\end{lemma}

\begin{proof}[Proof of Lemma \ref{lemma:accurate_bound}]
Lemma \ref{lemma:accurate_bound} follows directly from \citet[][Theorems 1 and 5]{braun2006accurate}. 
\end{proof}

\begin{lemma}\label{ineq_lambda_n}
	Let $K$ be a Mercer kernel on $\mathcal{X}$ with probability measure $\mu$ satisfying $K(x,x)\leq K_0<\infty$
	for all $x\in\mathcal{X}$. If the eigenvalues of $K$ satisfying
	$\lambda_{k}=O(k^{-\rho})$, $\rho>1$, i.e. there exists $M$
	such that $\lambda_{k}\leq Mk^{-\rho}$. Then, for any fixed $v>0$, any $n\geq2,$ and any $i\geq[n^{1/\rho}]$,
	\[
	P\left(|l_{i}-\lambda_{i}|\leq  A_{0}\frac{\log n}{n^{1-1/\rho}}\right)\geq1-\frac{1}{n^{v}},
	\]
	where $A_{0}$ only depends on $K_0,M,v,\rho$ and does not depends
	on $n$, $[x]$ denotes the smaller integer larger than or equal to $x$, $l_{1}\geq l_{2}\geq\cdots\geq l_{n}\geq0$
	are eigenvalues of $\bs{K}_{n}/n$ with $\left[\bs{K}_{n}\right]_{ij}=K(\bs{X}_{i},\bs{X}_{j})$, and $\bs{X}_i$'s are i.i.d.\ samples from $(\mathcal{X}, \mu)$.
\end{lemma}

\begin{proof}[Proof of Lemma \ref{ineq_lambda_n}]
Because 
$\lambda_{k}\leq Mk^{-\rho}$ for all $k\geq 1$, we have 
$$
\Lambda_{>r}=\sum_{i=r+1}^{\infty}\lambda_i \leq M \sum_{i=r+1}^\infty i^{-\rho} \leq M \int_{r}^{\infty} x^{-\rho} \text{d}x = \frac{M}{\rho-1} \frac{1}{r^{\rho-1}}. 
$$ 
For any $\delta \in (0,1)$ and $i \ge r \ge 1$,  by the definition in Lemma \ref{lemma:accurate_bound}, we have $\lambda_i \le \lambda_r$, 
\begin{align}\label{eq:C_tilde_bound}
\lambda_i \tilde{C}(r, n, \delta) 
& \le 
\lambda_r \tilde{C}(r, n, \delta) 
= 
\lambda_r
\left\{ r\sqrt{\frac{2K_0}{n\lambda_{r}}\log\frac{2r(r+1)}{\delta}}+\frac{4K_0r}{3n\lambda_{r}}\log\frac{2r(r+1)}{\delta}
\right\}
\nonumber
\\
& = 
r\sqrt{ \lambda_{r} \frac{2K_0}{n}\log\frac{2r(r+1)}{\delta}}+\frac{4K_0r}{3n}\log\frac{2r(r+1)}{\delta}
\le 
r\sqrt{ \lambda_{r} \frac{2K_0}{n}\log\frac{4r^2}{\delta}}+\frac{4K_0r}{3n}\log\frac{4r^2}{\delta}
\nonumber
\\
& \le 
r\sqrt{ \frac{M}{r^\rho} \frac{2K_0}{n}\log\frac{4r^2}{\delta}}+\frac{4K_0r}{3n}\log\frac{4r^2}{\delta},
\end{align}
and 
\begin{align}\label{eq:E_tilde_bound}
	\tilde{E}(r, n, \delta) & = \lambda_{r}+\Lambda_{>r}+\sqrt{\frac{2K_0\Lambda_{>r}}{n}\log\frac{2}{\delta}}+\frac{2K_0}{3n}\log\frac{2}{\delta}
	\nonumber
	\\
	& \le 
	\frac{M}{r^\rho} + \frac{M}{\rho-1} \frac{1}{r^{\rho-1}} + 
	\sqrt{\frac{2K_0}{n}\frac{M}{\rho-1}\frac{1}{r^{\rho-1}} \log\frac{2}{\delta}
	}+\frac{2K_0}{3n}\log\frac{2}{\delta}. 
\end{align}
Let $i\geq r=[n^{1/\rho}],\delta=n^{-v}$. We have $r\leq 2n^{1/\rho}$, $r^\rho\geq n$, $r^{\rho-1}\geq n^{1-1/\rho}$, 
and 
\begin{align*}
	\log\frac{4r^2}{\delta} & \le \log\left\{
	4 \cdot (2n^{1/\rho})^2 \cdot n^v
	\right\}
	= 
	\log 16 + \left(\frac{2}{\rho} + v\right) \log n, \\
	\log\frac{2}{\delta} & = \log\left(2\cdot n^v\right) = \log 2 + v \log n. 
\end{align*}
From \eqref{eq:C_tilde_bound} and \eqref{eq:E_tilde_bound}, 
these imply that when $i\geq r=[n^{1/\rho}]$ and $\delta=n^{-v}$, 
\begin{align*}
\lambda_i \tilde{C}(r, n, \delta) 
& \le 
r\sqrt{ \frac{M}{r^\rho} \frac{2K_0}{n}\log\frac{4r^2}{\delta}}+\frac{4K_0r}{3n}\log\frac{4r^2}{\delta}\\
& \le 
2 n^{1/\rho}\sqrt{ \frac{M}{n} \frac{2K_0}{n}
\left\{	\log 16 + \left(\frac{2}{\rho} + v\right) \log n\right\}
}+\frac{4K_0 \cdot 2 n^{1/\rho}}{3n}
\left\{
\log 16 + \left(\frac{2}{\rho} + v\right) \log n
\right\}\\
& \le A_1 \frac{\sqrt{\log n}}{n^{1-1/\rho}} + A_2 \frac{\log n}{n^{1-1/\rho}}, 
\end{align*}
and 
\begin{align*}
\tilde{E}(r, n, \delta) 
& \le 
\frac{M}{r^\rho} + \frac{M}{\rho-1} \frac{1}{r^{\rho-1}} + 
\sqrt{\frac{2K_0}{n}\frac{M}{\rho-1}\frac{1}{r^{\rho-1}} \log\frac{2}{\delta}
}+\frac{2K_0}{3n}\log\frac{2}{\delta}\\
& \le 
\left( \frac{M}{n} + \frac{M}{\rho-1} \cdot \frac{1}{n^{1-1/\rho}} \right) + 
\sqrt{\frac{2K_0}{n}\frac{M}{\rho-1}\frac{1}{n^{1-1/\rho}} \left(\log 2 + v \log n\right)
}+\frac{2K_0}{3n}\left(\log 2 + v \log n\right)\\
& \le 
A_3 \frac{1}{n^{1-1/\rho}} + A_{4}\frac{1}{n^{1/(2\rho)}}\frac{\sqrt{\log n}}{n^{1-1/\rho}}+A_{5}\frac{1}{n^{1/\rho}}\frac{\log n}{n^{1-1/\rho}},
\end{align*}
where $(A_1, \ldots, A_5)$ depend on $K_0,M,v,\rho$ but do not depend on $n$. 
Thus, 
for $i\geq [n^{1/\rho}]$ and $\delta=n^{-v}$, 
there must exist $A_0$, which depends on $K_0,M,v,\rho$ but does not depends on $n$, such that 
$$
	\lambda_i \tilde{C}(r, n, \delta)  + \tilde{E}(r, n, \delta) 
	\le A_0 \frac{\log n}{n^{1-1/\rho}}.  
$$
From Lemma \ref{lemma:accurate_bound}, for any $i\geq [n^{1/\rho}]$ and $\delta=n^{-v}$, we have 
\begin{align*}
	1 - \frac{1}{n^v} & = 1 - \delta \le \Pr\left\{
	C(r,n) < \tilde{C}(r, n, \delta), 
	E(r,n) < \tilde{E}(r, n, \delta)
	\right\}\\
	& \le 
	\Pr\left\{
	\lambda_{i}C(r,n)+E(r,n) < \lambda_{i}\tilde{C}(r, n, \delta) + \tilde{E}(r, n, \delta) 
	\right\}\\
	& \le \Pr\left\{
	|l_i - \lambda_i |< \lambda_{i}\tilde{C}(r, n, \delta) + \tilde{E}(r, n, \delta) 
	\right\}\\
	& \le \Pr\left\{
	|l_i - \lambda_i | \le A_0 \frac{\log n}{n^{1-1/\rho}}
	\right\}.
\end{align*}
Therefore, Lemma \ref{ineq_lambda_n} holds. 
\end{proof}

\begin{lemma}\label{dn_times_bn}
	Let $K$ be a Mercer kernel on a compact set $\mathcal{X}$ with prbability measure $\mu$,
if the eigenvalues of $K$ satisfying
	$\lambda_{k}=O(k^{-\rho})$, $\rho>1$, i.e. there exists $M$
	such that $\lambda_{k}\leq Mk^{-\rho}$. Let $\gamma$ be a constant in $(0, 1-\rho^{-1})$, and let $\kappa$ be a constant in $(0, 1-\rho^{-1}-\gamma)$. If $b_n$ satisfying $b_n = O(n^\kappa)$, then we have
	\begin{align*}
	 \frac{c_{n-b_n+1}}{n^{1-\gamma}} \cdot b_n \stackrel{\Pr}{\longrightarrow} 0,
	\end{align*}
	where $c_{1}\geq c_{2}\geq\cdots\geq c_{n}\geq0$
	are the eigenvalues of $\bs{K}_{n}$ with $\left[\bs{K}_{n}\right]_{ij}=K(\bs{X}_{i},\bs{X}_{j})$, and $\bs{X}_i$'s are i.i.d samples from $(\mathcal{X}, \mu)$.
\end{lemma}

\begin{proof}[Proof of Lemma \ref{dn_times_bn}]
Because $K$ is a Mercer kernel, $K$ is a continuous function on $\mathcal{X}\times \mathcal{X}$. Since $\mathcal{X}$ is compact, we know that there exists $K_0<\infty$ such that $K(\bs{x},\bs{x})\leq K_0$ for all $\bs{x} \in \mathcal{X}$.
Let $l_{1}\geq l_{2}\geq\cdots\geq l_{n}\geq0$
be eigenvalues of $\bs{K}_{n}/n$. 
Obviously, $l_i = c_i/n$, $1\le i \le n$. 
From
Lemma \ref{ineq_lambda_n}, 
for any $v>0$, there exists a constant $A_0$ depending only on $K_0, M, v, \rho$ such that, for any $i\geq[n^{1/\rho}]$, 
\[
\Pr\left(|l_{i}-\lambda_{i}|\leq A_{0}\frac{\log n}{n^{1-1/\rho}}\right)\geq1-\frac{1}{n^{v}}.
\]
Note that for $n\geq2$ and $i\geq[n^{1/\rho}]$, 
\begin{align*}
l_{i}\leq & |l_{i}-\lambda_{i}|+\lambda_{i}
\leq  |l_{i}-\lambda_{i}|+\frac{M}{i^\rho}
\leq  |l_{i}-\lambda_{i}|+\frac{M}{n}
=  |l_{i}-\lambda_{i}|+  \frac{M}{n^{1/\rho}\log n} \cdot \frac{\log n}{n^{1-1/\rho}}\\
\leq & |l_{i}-\lambda_{i}|+\frac{M}{2^{1/\rho}\log2} \cdot \frac{\log n}{n^{1-1/\rho}}.
\end{align*}
Let $A_0'=A_{0}+M/(2^{1/\rho}\log2)$, which depends only on $K_0, M, v, \rho$. 
Then for any $n\geq2$ and $i\geq[n^{1/\rho}]$, we have 
\begin{align}\label{eq:prob_bound_l}
	\Pr \left(l_{i}> A_0' \frac{\log n}{n^{1-1/\rho}}\right)
	& \le  
	\Pr\left(
	|l_{i}-\lambda_{i}|+\frac{M}{2^{1/\rho}\log2} \cdot \frac{\log n}{n^{1-1/\rho}} >  A_0'  \frac{\log n}{n^{1-1/\rho}}
	\right)
	\nonumber
	\\
	& = 
	\Pr\left(
	|l_{i}-\lambda_{i}|  >  A_{0} \frac{\log n}{n^{1-1/\rho}}
	\right)
	= 
	1 - \Pr\left(|l_{i}-\lambda_{i}|\leq A_{0}\frac{\log n}{n^{1-1/\rho}}\right)
	& \le \frac{1}{n^v}. 
\end{align}
Note that  
$$
n-b_n +  1 - [n^{1/\rho}]\geq n-b_n - n^{1/\rho} = n -O(n^\kappa) - n^{1/\rho},
$$
and for any $\varepsilon > 0$, 
\begin{align*}
\lim_{n\rightarrow \infty}
\frac{ A_0' b_n n^\gamma  \log n}{\varepsilon n^{1-1/\rho}}
= \lim_{n\rightarrow \infty}\frac{A_0'}{\varepsilon}\frac{\log n}{n^{1-1/\rho-\gamma-\kappa}}  \frac{b_n}{n^{\kappa}} = 0.
\end{align*}
Thus, for any $\varepsilon>0$, there must exist $N_\varepsilon$ such that when $n > N_\varepsilon$, 
$n-b_n +  1 \ge [n^{1/\rho}]$
and 
$
A_0' b_n n^\gamma  \log n < \varepsilon n^{1-1/\rho}. 
$
From \eqref{eq:prob_bound_l}, 
for any $\varepsilon>0$, when $n>N_{\varepsilon}$, 
we then have
\begin{align*}
\Pr(n^\gamma l_{n-b_n+1} \cdot b_n>\varepsilon) & = \Pr\left(l_{n-b_n+1}>\frac{\varepsilon}{n^\gamma  b_n}\right)
=  \Pr\left(l_{n-b_n+1}>A_0'\frac{\log n}{n^{1-1/\rho}} \cdot  \frac{\varepsilon}{A_0'} \frac{n^{1-1/\rho}}{ b_n n^\gamma\log n}\right)\\
& \leq \Pr\left(l_{n-b_n+1}> A_0 ' \frac{\log n}{n^{1-1/\rho}} \right)
\leq  \frac{1}{n^{v}}. 
\end{align*}
which further implies that 
$
\lim_{n\rightarrow\infty}P(n^\gamma l_{n-b_n+1} \cdot b_n>\varepsilon)=0. 
$
Consequently, 
$
n^\gamma l_{n-b_n+1} \cdot b_n \stackrel{\Pr}{\longrightarrow} 0.
$
By definition, 
we then have 
$
c_{n-b_n+1}/n^{1-\gamma} \cdot b_n \stackrel{\Pr}{\longrightarrow} 0.
$
Therefore, Lemma \ref{dn_times_bn} holds. 
\end{proof}

\begin{proof}[{\bf Proof of Theorem \ref{thm:general_ppt_valid}}]
Define $\nu, \nu_0$ and $\tilde{\Delta}_n$ the same as in Lemma \ref{lemma:bound_TV_GP}. 
From Lemmas \ref{lemma:basic_ineq_perm_p_val} and \ref{lemma:bound_TV_GP}, for any $\alpha, \eta \ge 0$,
\begin{align*}
	\Pr \{  p(\bXmatrix, \bs{Y}, \bs{Z} ) \leq \alpha  \mid \bXmatrix, \bs{Z} \} 
	& \leq  \alpha+ \eta +\Pr \left( \| \nu - \nu_0 \|_{\TV} > \eta \mid \bXmatrix, \bs{Z}  \right)
	\leq  \alpha+ \eta +\Pr \left( \frac{e^{\tilde{\Delta}_n} - 1}{2} > \eta \mid \bXmatrix, \bs{Z}  \right)
	\\
	& = \alpha+ \eta +\Pr \left\{ \tilde{\Delta}_n > \log(2\eta+1) \mid \bXmatrix, \bs{Z}  \right\}. 
\end{align*}
By the definition in Lemma \ref{lemma:bound_TV_GP}, we have 
\begin{align*}
\tilde{\Delta}_n 
&
= 
\frac{c_{n-b_n+1} \delta_0^2/n^{1-\gamma}}{ 2 \sigma_{0}^{2}} \sum_{j=n-b_n+1}^{n} \frac{W_{j}^{2}}{c_{j} \delta_0^2/n^{1-\gamma} + \sigma_0^2}
= \frac{c_{n-b_n+1} \delta_0^2/n^{1-\gamma}}{ 2 \sigma_{0}^{2}} O_{\Pr}(b_n)
= 
\frac{\delta_0^2}{2\sigma_{0}^{2}}
\frac{c_{n-b_n+1}}{n^{1-\gamma} } b_n \cdot O_{\Pr}(1),
\end{align*}
where the second last equality holds because $\sum_{j=n-b_n+1}^{n} W_{j}^{2}/(c_{j} \delta_0^2/n^{1-\gamma} + \sigma_0^2)$ follows $\chi^2$ distribution with degrees of freedom $b_n$ conditional on $\bXmatrix$ and $\bs{Z}$. 
From Lemma \ref{dn_times_bn}, 
as $n\rightarrow \infty$, 
we have $\tilde{\Delta}_n \stackrel{\Pr}{\longrightarrow} 0$ and thus 
$\Pr \{ \tilde{\Delta}_n > \log(2\eta+1) \mid \bXmatrix, \bs{Z}  \} \rightarrow 0$ for any $\eta >0$. 
Thus, for any $\alpha\ge 0$ and $\eta >0$, 
\begin{align*}
\limsup_{n\rightarrow \infty} \Pr \{  p(\bXmatrix, \bs{Y}, \bs{Z} ) \leq \alpha  \mid \bXmatrix, \bs{Z} \} 
\le 
\alpha+ \eta + \limsup_{n\rightarrow \infty}\Pr \left\{ \tilde{\Delta}_n > \log(2\eta+1) \mid \bXmatrix, \bs{Z}  \right\} = \alpha+\eta. 
\end{align*}
Because the above inequality holds for any $\eta > 0$, we then have, for any $\alpha \in (0,1)$, 
\begin{align*}
    \limsup_{n\rightarrow \infty} \Pr \{  p(\bXmatrix, \bs{Y}, \bs{Z} ) \leq \alpha  \mid \bXmatrix, \bs{Z} \}  \le \alpha. 
\end{align*}
Therefore, Theorem \ref{thm:general_ppt_valid} holds. 
\end{proof}

\section{Partial Permutation Test under Alternative Hypotheses}\label{sec:proof_alter}

\begin{proof}[\bf Proof of Theorem \ref{thm:general_F_test}]
Without loss of generality, 
we assume that the samples are ordered based on their group indicators in an increasing way. Thus, the covariates mapped into the feature space have the following block structure:
\begin{align*}
    \begin{pmatrix}
        \phi(\bs{X}_1) & \phi(\bs{X}_2) & \cdots & \phi(\bs{X}_n)
    \end{pmatrix}
    = 
    \begin{pmatrix}
        \bs{\Phi}_1^\top & \bs{\Phi}_2^\top & \cdots & \bs{\Phi}_H^\top
    \end{pmatrix}, 
\end{align*}
where $\bs{\Phi}_h \in \mathbb{R}^{n_h \times q}$ consists of samples in group $h$, $1\le h \le H$. 
Define 
\begin{align}\label{eq:design_mat_null_full}
    \bs{\Phi} = 
    \begin{pmatrix}
        \bs{\Phi}_1\\ 
        \bs{\Phi}_2\\ 
        \vdots \\
        \bs{\Phi}_n
    \end{pmatrix}
    \quad 
    \text{and} 
    \quad 
    \bs{B} 
    = 
    \begin{pmatrix}
        \bs{\Phi}_1 & \bs{0} & \cdots & \bs{0}\\
        \bs{0} & \bs{\Phi}_2 & \cdots & \bs{0}\\
        \vdots & \vdots & \ddots & \vdots \\
        \bs{0} & \bs{0} & \cdots & \bs{\Phi}_H
    \end{pmatrix}
\end{align}
to denote the design matrices for the hypothesized null model (i.e., model \eqref{eq:linear_feature_alter} with $\bs{\beta}_1 = \ldots = \bs{\beta}_H$) and full model (i.e., model \eqref{eq:linear_feature_alter} without any constraint). 
By definition, we know that 
$\bs{\Phi}$ and $\bs{B} $ have ranks $p_0$ and $p_1$, respectively. 
Let $\bs{K}_n = \bs{\Gamma} \bs{C} \bs{\Gamma}^\top$ be the eigen-decomposition
of $\bs{K}_n = \bs{\Phi}\bs{\Phi}^\top$, 
where $\bs{\Gamma}\in \mathbb{R}^{n \times n}$ is a diagonal matrix and $\bs{C}$ is a diagonal matrix with diagonal elements ordered in a decreasing way.   
Let $\bs{\Gamma}_1$ and $\bs{\Gamma}_2$ be the submatrices of $\bs{\Gamma}$ consisting of the first $p_0$ and last $n-p_0$ columns.

First, we consider model \eqref{eq:linear_feature_alter} with the constraint that $\bs{\beta}_1 = \ldots = \bs{\beta}_H$. 
Recall that 
$\bs{P}_{0}$ is the projection matrix onto the column space of $\bs{\Phi}$, 
and note that the column space of $\bs{\Phi}$ is the same as that of $\bs{\Gamma}_1$ and is orthogonal to that of $\bs{\Gamma}_2$.
We can then simplify the residual sum of squares under the null model as 
\begin{align*}
    \text{RSS}_0 
    & = 
    \bs{Y}^\top \big( \bs{I}_n - \bs{P}_{0} \big) \bs{Y}
    =
    \bs{Y}^\top \bs{\Gamma} \bs{\Gamma}^\top \big( \bs{I}_n - \bs{P}_{0} \big) \bs{\Gamma} \bs{\Gamma}^\top \bs{Y}
    = 
    \bs{Y}^\top \bs{\Gamma}
    \begin{pmatrix}
        \bs{\Gamma}_1^\top \\
        \bs{\Gamma}_2^\top 
    \end{pmatrix}
    \big( \bs{I}_n - \bs{P}_{0} \big) (\bs{\Gamma}_1, \bs{\Gamma}_2) \bs{\Gamma}^\top \bs{Y}
    \\
    & = 
    \bs{Y}^\top 
    \bs{\Gamma} 
    \begin{pmatrix}
        \bs{0} & \bs{0} \\
        \bs{0} & \bs{\Gamma}_2^\top \bs{\Gamma}_2
    \end{pmatrix}
    \bs{\Gamma}^\top \bs{Y}
    = 
    \bs{Y}^\top 
    \bs{\Gamma} 
    \begin{pmatrix}
        \bs{0} & \bs{0} \\
        \bs{0} & \bs{I}_{n-p_0}
    \end{pmatrix}
    \bs{\Gamma}^\top \bs{Y},  
\end{align*}
where the last equality holds because $(\bs{\Gamma}_1, \bs{\Gamma}_2)$ is an orthogonal matrix. 

Second, we consider the model \eqref{eq:linear_feature_alter} without any constraint on the coefficients. 
The design matrix under the full model is $\bs{B}$ in \eqref{eq:design_mat_null_full}, 
whose column space obviously covers that of $\bs{\Phi}$ or equivalently $\bs{\Gamma}_1$. 
Recall that $\bs{P}_{1}$ is the projection matrix onto the column space of $\bs{B}$. 
The residual sum of squares under the full model then has the following equivalent forms: 
\begin{align}\label{eq:rss1}
    \text{RSS}_1 & = 
    \bs{Y}^\top (\bs{I}_n - \bs{P}_{1}) \bs{Y}
    = 
    \bs{Y}^\top 
    \bs{\Gamma} 
    \begin{pmatrix}
        \bs{\Gamma}_1^\top \\
        \bs{\Gamma}_2^\top 
    \end{pmatrix}
    (\bs{I}_n - \bs{P}_{1}) 
    (\bs{\Gamma}_1, \bs{\Gamma}_2) \bs{\Gamma}^\top
    \bs{Y}
    = 
    \bs{Y}^\top 
    \bs{\Gamma} 
    \begin{pmatrix}
        \bs{0} & \bs{0} \\
        \bs{0} & 
        \bs{I}_{n-p_0} - \bs{\Gamma}_2^\top \bs{P}_{1}  \bs{\Gamma}_2
    \end{pmatrix}
    \bs{\Gamma}^\top
    \bs{Y}
    \nonumber
    \\
    & = 
    \text{RSS}_0 - \bs{Y}^\top 
    \bs{\Gamma} 
    \begin{pmatrix}
        \bs{0} & \bs{0} \\
        \bs{0} & 
        \bs{\Gamma}_2^\top \bs{P}_{1}  \bs{\Gamma}_2
    \end{pmatrix}
    \bs{\Gamma}^\top
    \bs{Y}. 
\end{align}

Third, we consider the permutation distribution of the statistic $( \text{RSS}_0-\text{RSS}_1 )/\text{RSS}_1$ under our continuous partial permutation test. 
From the description of Algorithm 1, the permutation distribution of this test statistic with permutation size $b_n = n-p_0$ is 
\begin{align*}
    \frac{\widetilde{\text{RSS}}_0 - \widetilde{\text{RSS}}_1}{\widetilde{\text{RSS}}_1}
    & = 
    \big( \bs{\Gamma}^\top \bs{Y}^{\perm} \big)^\top 
    \begin{pmatrix}
        \bs{0} & \bs{0} \\
        \bs{0} & 
        \bs{\Gamma}_2^\top \bs{P}_{1}  \bs{\Gamma}_2
    \end{pmatrix}
    \bs{\Gamma}^\top
    \bs{Y}^{\perm} 
    /
    \big( \bs{\Gamma}^\top \bs{Y}^{\perm} \big)^\top 
    \begin{pmatrix}
        \bs{0} & \bs{0} \\
        \bs{0} & \bs{I}_{n-p_0} - \bs{\Gamma}_2^\top \bs{P}_{1}  \bs{\Gamma}_2
    \end{pmatrix}
    \bs{\Gamma}^\top \bs{Y}^{\perm}\\
    & = 
    \big( \bs{W}^{\perm} \big)^\top 
    \begin{pmatrix}
        \bs{0} & \bs{0} \\
        \bs{0} & 
        \bs{\Gamma}_2^\top \bs{P}_{1}  \bs{\Gamma}_2
    \end{pmatrix}
    \bs{W}^{\perm}
    /
    \big( \bs{W}^{\perm} \big)^\top 
    \begin{pmatrix}
        \bs{0} & \bs{0} \\
        \bs{0} & \bs{I}_{n-p_0} - \bs{\Gamma}_2^\top \bs{P}_{1}  \bs{\Gamma}_2
    \end{pmatrix}
    \bs{W}^{\perm}
    \\
    & = 
    \big( \tilde{\bs{W}}^{\perm} \big)^\top 
       \big( \bs{\Gamma}_2^\top \bs{P}_{1}  \bs{\Gamma}_2  \big)
    \tilde{\bs{W}}^{\perm}
    /
    \big( \tilde{\bs{W}}^{\perm} \big)^\top 
    \big( \bs{I}_{n-p_0} - \bs{\Gamma}_2^\top \bs{P}_{1}  \bs{\Gamma}_2 \big)
    \tilde{\bs{W}}^{\perm}, 
\end{align*}
where $\bs{Y}^{\perm} = \bs{\Gamma} \bs{W}^{\perm}$, 
$\bs{W}^{\perm}$ is uniformly distributed on the set $\mathcal{S}_w$ in Algorithm \ref{alg:partial_permu}, 
and $\tilde{\bs{W}}^{\perm}$ is the subvector of $\bs{W}^{\perm}$ consisting of the last $n-p_0$ coordinates. 
By definition, $\tilde{\bs{W}}^{\perm}/\|\tilde{\bs{W}}^{\perm} \|_2$ is uniformly distributed on the $n-p_0-1$ dimensional unit sphere. 
Thus, $\tilde{\bs{W}}^{\perm}/\|\tilde{\bs{W}}^{\perm} \|_2 \sim \bs{\eta}/\|\bs{\eta}\|_2$, where $\bs{\eta} \sim \mathcal{N}(\bs{0}, \bs{I}_{n-p_0})$, 
and the permutation distribution of the statistic $( \text{RSS}_0-\text{RSS}_1 )/\text{RSS}_1$ then has the following equivalent forms:
\begin{align*}
    \frac{\widetilde{\text{RSS}}_0 - \widetilde{\text{RSS}}_1}{\widetilde{\text{RSS}}_1}
    & = 
    \frac{
    \big( \tilde{\bs{W}}^{\perm} \big)^\top 
       \big( \bs{\Gamma}_2^\top \bs{P}_{1}  \bs{\Gamma}_2  \big)
    \tilde{\bs{W}}^{\perm}
    }{
    \big( \tilde{\bs{W}}^{\perm} \big)^\top 
    \big( \bs{I}_{n-p_0} - \bs{\Gamma}_2^\top \bs{P}_{1}  \bs{\Gamma}_2 \big)
    \tilde{\bs{W}}^{\perm}
    }
    \sim 
    \frac{
    \bs{\eta}^\top 
       \big( \bs{\Gamma}_2^\top \bs{P}_{1}  \bs{\Gamma}_2  \big)
    \bs{\eta}
    }{
    \bs{\eta}^\top 
    \big( \bs{I}_{n-p_0} - \bs{\Gamma}_2^\top \bs{P}_{1}  \bs{\Gamma}_2 \big)
    \bs{\eta}
    }. 
\end{align*}

Fourth, we consider the matrix $\bs{V} \equiv \bs{I}_{n-p_0} -  \bs{\Gamma}_2^\top \bs{P}_{1}  \bs{\Gamma}_2 =  \bs{\Gamma}_2^\top (\bs{I}_n - \bs{P}_{\bs{1}})  \bs{\Gamma}_2$. 
From \eqref{eq:rss1}, we can know that $\text{rank}(\bs{V}) = \text{rank}(\bs{I}_n - \bs{P}_1) = n-p_1$. 
Moreover, 
$\bs{V}^2$ 
has the following equivalent forms: 
\begin{align*}
    \bs{V}^2 = \bs{\Gamma}_2^\top (\bs{I}_n - \bs{P}_{1})  \bs{\Gamma}_2 \bs{\Gamma}_2^\top (\bs{I}_n - \bs{P}_{1})  \bs{\Gamma}_2
    = \bs{\Gamma}_2^\top (\bs{I}_n - \bs{P}_{1})  \bs{\Gamma}_2
    \big( \bs{\Gamma}_2^\top \bs{\Gamma}_2 \big)^{-1}
    \bs{\Gamma}_2^\top (\bs{I}_n - \bs{P}_{1})  \bs{\Gamma}_2. 
\end{align*}
Note that (i) $\bs{\Gamma}_2\big( \bs{\Gamma}_2^\top \bs{\Gamma}_2 \big)^{-1}\bs{\Gamma}_2^\top$ is the projection matrix onto the column space of $\bs{\Gamma}_2$, which is the orthogonal complement of the column space of $\bs{\Gamma}_1$, 
(ii) $\bs{I}_n - \bs{P}_{\bs{1}}$ is the projection matrix onto the  space orthogonal to the column space of $\bs{B}$, 
and 
(iii) the column space of $\bs{\Gamma}_1$ is in the column space of $\bs{B}$. 
Thus, the column space of $\bs{I}_n - \bs{P}_{1}$ must be in the column space of $\bs{\Gamma}_2$, 
and $\bs{V}^2$ simplifies to 
\begin{align*}
    \bs{V}^2 
    = \bs{\Gamma}_2^\top (\bs{I}_n - \bs{P}_{\bs{B}})  \bs{\Gamma}_2
    \big( \bs{\Gamma}_2^\top \bs{\Gamma}_2 \big)^{-1}
    \bs{\Gamma}_2^\top (\bs{I}_n - \bs{P}_{\bs{B}})  \bs{\Gamma}_2
    = 
    \bs{\Gamma}_2^\top (\bs{I}_n - \bs{P}_{\bs{B}}) (\bs{I}_n - \bs{P}_{\bs{B}})  \bs{\Gamma}_2
    = 
    \bs{\Gamma}_2^\top (\bs{I}_n - \bs{P}_{\bs{B}}) \bs{\Gamma}_2
    = 
    \bs{V}. 
\end{align*}
Therefore, $\bs{V}$ is a projection matrix with $n-p_1$ nonzero eigenvalues.

From the above,
the permutation distribution of the F statistic has the following equivalent forms:
\begin{align*}
    \tilde{F}
    & = 
    \frac{(\widetilde{\text{RSS}}_0 -\widetilde{\text{RSS}}_1)/(p_1 - p_0) }{\widetilde{\text{RSS}}_1/(n-p_1)}
    \sim 
    \frac{
    \bs{\eta}^\top (\bs{I}_{n-p_0} - \bs{V}) \bs{\eta}/(p_1 - p_0)
    }{
    \bs{\eta}^\top \bs{V} \bs{\eta}/(n-p_1)
    }.
\end{align*}
By the properties of Gaussian distribution and projection matrix, 
we can verify that 
$\bs{\eta}^\top \bs{V} \bs{\eta}$ and $\bs{\eta}^\top (\bs{I}_{n-p_0} - \bs{V}) \bs{\eta}$ follow $\chi^2$ distributions with degrees of freedom $n-p_1$ and $p_1-p_0$, 
and they are mutually independent. 
Therefore, the permutation distribution of the $F$ statistic follows an $F$ distribution with degrees of freedom $p_1 - p_0$ and $n-p_1$. 
Consequently, Theorem \ref{thm:general_F_test} holds. 
\end{proof}

To prove Theorem \ref{thm:power_diverg_kernel}, we need the following three lemmas. 

\begin{lemma}\label{lemma:ratio_chi_sq_CLT}
Let $\{\eta_{n1}\}, \{\eta_{n2}\}$, $\{U_{n1}\}$ and $\{ U_{n2} \}$ be four sequences of random variables, 
where 
$U_{n1}$ and $U_{n2}$ are independent and follow chi-squared distributions with degrees of freedom $d_{n1}$ and $d_{n2}$. 
If, as $n \rightarrow \infty$, 
$d_{n1} \rightarrow \infty$,
$d_{n1}/d_{n2} \rightarrow 0$, 
$\eta_{n1}/\sqrt{d_{n1}} \convergep \theta$ and $ \eta_{n2} \sqrt{d_{n1}} / d_{n2}  \convergep 0$, 
then 
\begin{align}\label{eq:chi_sq_ratio_std}
    \frac{d_{n2}}{\sqrt{2d_{n1}}} 
    \frac{\eta_{n1} + U_{n1}}{\eta_{n2} + U_{n2}}
    -\sqrt{d_{n1}/2}
    & \converged 
    \mathcal{N}(\theta/\sqrt{2}, \  1).
\end{align}
\end{lemma}
\begin{proof}[Proof of Lemma \ref{lemma:ratio_chi_sq_CLT}]
Let $\tilde{U}_{n1} = (U_{n1} - d_{n1})/\sqrt{2d_{n1}}$. 
By the central limit theorem, 
$\tilde{U}_{n1} \converged \mathcal{N}(0, 1)$ as $n\rightarrow \infty$. 
By some algebra, the left hand side of \eqref{eq:chi_sq_ratio_std} has the following equivalent forms: 
\begin{align*}
    & \quad \ \frac{d_{n2}}{\sqrt{2d_{n1}}} 
    \frac{\eta_{n1} + U_{n1}}{\eta_{n2} + U_{n2}}
    -\sqrt{d_{n1}/2}
    \\
    & = 
    \frac{d_{n2}}{\sqrt{2d_{n1}}} 
    \frac{\eta_{n1} + (U_{n1}-d_{n1}) + d_{n1}}{\eta_{n2} + U_{n2}}
    -\sqrt{d_{n1}/2}
    = 
    \frac{\eta_{n1}/\sqrt{2d_{n1}} + \tilde{U}_{n1}+\sqrt{d_{n1}/2}}{
    \eta_{n2}/d_{n2} + U_{n2}/d_{n2}
    }
    - \sqrt{d_{n1}/2}
    \\
    & = 
    \frac{\eta_{n1}/\sqrt{2d_{n1}} + \tilde{U}_{n1}}{
    \eta_{n2}/d_{n2} + U_{n2}/d_{n2}
    }
    - 
    \sqrt{d_{n1}/2} \cdot
    \frac{\eta_{n2}/d_{n2} + U_{n2}/d_{n2}-1}{\eta_{n2}/d_{n2} + U_{n2}/d_{n2}}. 
\end{align*}
By Chebyshev's inequality, 
$
U_{n2}/d_{n2} - 1 = O_{\Pr}(d_{n2}^{-1/2}). 
$
From the conditions in Lemma \ref{lemma:ratio_chi_sq_CLT}, as $n\rightarrow \infty$, 
$\eta_{n1}/\sqrt{2d_{n1}}\convergep \theta/\sqrt{2}$, 
$
\eta_{n2}/d_{n2} + U_{n2}/d_{n2} \convergep 1, 
$
and 
\begin{align*}
   & \quad \ \sqrt{d_{n1}/2} \cdot \left( \eta_{n2}/d_{n2} + U_{n2}/d_{n2}-1 \right) 
   \\ & = \sqrt{d_{n1}/2} \cdot \left\{ \eta_{n2}/d_{n2} + O_{\Pr}(d_{n2}^{-1/2}) \right\}
   = \left( \eta_{n2} \sqrt{d_{n1}} / d_{n2} + 
   \sqrt{d_{n1}/d_{n2}} \right) \cdot O_{\Pr}(1)
   \\
   & = o_{\Pr}(1). 
\end{align*}
By Slutsky's theorem, we can then derive that 
\begin{align*}
    \frac{d_{n2}}{\sqrt{2d_{n1}}} 
    \frac{\eta_{n1} + U_{n1}}{\eta_{n2} + U_{n2}}
    -\sqrt{d_{n1}/2}
    & = 
     \frac{\eta_{n1}/\sqrt{2d_{n1}} + \tilde{U}_{n1}}{
    \eta_{n2}/d_{n2} + U_{n2}/d_{n2}
    }
    - 
    \sqrt{d_{n1}/2} \cdot
    \frac{\eta_{n2}/d_{n2} + U_{n2}/d_{n2}-1}{\eta_{n2}/d_{n2} + U_{n2}/d_{n2}}
    \\
    & \converged 
    \mathcal{N}(\theta/\sqrt{2}, \  1). 
\end{align*}
Therefore, Lemma \ref{lemma:ratio_chi_sq_CLT} holds. \end{proof}

\begin{lemma}\label{lemma:ratio_chi_sq_survival}
Let $\{\eta_{n1}\}, \{\eta_{n2}\}$, $\{U_{n1}\}$ and $\{ U_{n2} \}$ be four sequences of random variables, 
where
$U_{n1}$ and $U_{n2}$ are independent and follow chi-squared distributions with degrees of freedom $d_{n1}$ and $d_{n2}$. 
For $\alpha\in (0,1)$, 
let $\nu_{n,\alpha}$ be the $\alpha$th quantile of the distribution of $U_{n1}/U_{n2}$, 
and $z_{\alpha}$ be the $\alpha$th quantile of the standard Gaussian distribution. 
If, as $n \rightarrow \infty$, 
$d_{n1} \rightarrow \infty$,
$d_{n1}/d_{n2} \rightarrow 0$, 
$\eta_{n1}/\sqrt{d_{n1}} \convergep \theta$ and $ \eta_{n2} \sqrt{d_{n1}} / d_{n2}  \convergep 0$, 
then 
for any $\alpha\in (0,1)$, 
\begin{align*}
    \Pr\left( 
    \frac{\eta_{n1}+U_{n1}}{\eta_{n2}+U_{n2}} 
    \ge \nu_{n,\alpha}
    \right)
    \rightarrow
    1 - \Phi(z_{\alpha} - \theta/\sqrt{2}), 
\end{align*}
where $\Phi(\cdot)$ denotes the distribution function of the standard Gaussian distribution. 
\end{lemma}

\begin{proof}[Proof of Lemma \ref{lemma:ratio_chi_sq_survival}]
From Lemma \ref{lemma:ratio_chi_sq_survival}, we can know that 
\begin{align*}
    \frac{d_{n2}}{\sqrt{2d_{n1}}} 
    \frac{U_{n1}}{U_{n2}}
    -\sqrt{d_{n1}/2}
    \converged 
    \mathcal{N}(0, \  1), 
    \quad
    \frac{d_{n2}}{\sqrt{2d_{n1}}} 
    \frac{\eta_{n1} + U_{n1}}{\eta_{n2} + U_{n2}}
    -\sqrt{d_{n1}/2}
    \converged 
    \mathcal{N}(\theta/\sqrt{2}, \  1).
\end{align*}
These immediately imply that, for any $\alpha\in (0,1)$, 
$
    d_{n2}/\sqrt{2d_{n1}} \cdot
    \nu_{n, \alpha}
    -\sqrt{d_{n1}/2} \rightarrow  z_{\alpha}, 
$
and 
\begin{align*}
    \Pr\left( 
    \frac{\eta_{n1}+U_{n1}}{\eta_{n2}+U_{n2}} 
    \ge \nu_{n,\alpha}
    \right)
    & = 
    \Pr\left\{
    \frac{d_{n2}}{\sqrt{2d_{n1}}} 
    \frac{\eta_{n1} + U_{n1}}{\eta_{n2} + U_{n2}}
    -\sqrt{d_{n1}/2}
    - 
    \left( 
    \frac{d_{n2}}{\sqrt{2d_{n1}}} \nu_{n, \alpha}
    -\sqrt{d_{n1}/2} \right) \ge 0
    \right\}\\
    & \rightarrow
    \Pr(\varepsilon_0 + \theta/\sqrt{2} - z_{\alpha} \ge 0)
    = 1 - \Phi(z_{\alpha} - \theta/\sqrt{2}), 
\end{align*}
where $\varepsilon_0$ denotes a standard Gaussian random variable. 
Therefore, Lemma \ref{lemma:ratio_chi_sq_survival} holds. 
\end{proof}

\begin{lemma}\label{lemma:power_given_X}
Let $\{(\bs{X}_i,Y_i,Z_i)\}_{1\leq i \leq n}$ denote samples from model $H_1$ in \eqref{eq:H_1_equal}, where we allow the underlying functions for the $H$ groups to vary with sample size and will write them explicitly as $f_{n1}, \ldots, f_{nH}$. 
If, as $n\rightarrow \infty$, 
\begin{align*}
    p_{n1} - p_{n0} \rightarrow \infty, \ \  
    \frac{p_{n1} - p_{n0}}{n-p_{n1}} \rightarrow 0, \ \ 
    \frac{\bs{f}^\top (\bs{I}_n - \bs{P}_{n1}) \bs{f}}{\sqrt{p_{n1}-p_{n0}}} \rightarrow 0, \ \ 
    \frac{\bs{f}^\top (\bs{I}_n - \bs{P}_{n0}) \bs{f}}{\sqrt{p_{n1}-p_{n0}}}  = \text{(or $\ge$) } \theta + o(1), 
\end{align*}
then 
$\Pr\{ p(\bXmatrix, \bs{Y}, \bs{Z}) \ge \alpha \mid \bXmatrix\}= \text{(or $\ge$) } \Phi(z_{\alpha} + \theta/\sqrt{2}) + o(1)$, 
where $p_{n1}, p_{n0}, \bs{f}, \bs{P}_{n1}, \bs{P}_{n0}, \Phi$ and $z_\alpha$ are defined the same as that in Theorem \ref{thm:power_diverg_kernel}, and $p(\bXmatrix, \bs{Y}, \bs{Z})$ is the continuous partial permutation $p$-value with kernel $K_{q_n}$, permutation size $n-p_{n0}$ and F test statistic. 
\end{lemma}

\begin{proof}[Proof of Lemma \ref{lemma:power_given_X}]
First, we give some equivalent forms of the residual sums of squares from the null and full models in \eqref{eq:linear_feature_alter} with $\phi$ replaced by $\phi_{q_n}$.  
By definition, 
\begin{align*}
    \text{RSS}_1 & = \bs{Y}^\top (\bs{I}_n - \bs{P}_{n1}) \bs{Y}
    = (\bs{f} + \bs{\varepsilon})^\top (\bs{I}_n - \bs{P}_{n1}) (\bs{f} + \bs{\varepsilon})
    \\
    & = 
    \big\{ \bs{f}^\top (\bs{I}_n - \bs{P}_{n1}) \bs{f} 
    + 
    2 \bs{f}^\top (\bs{I}_n - \bs{P}_{n1}) \bs{\varepsilon} \big\}
    + \bs{\varepsilon}^\top (\bs{I}_n - \bs{P}_{n1}) \bs{\varepsilon}
    \equiv \eta_{n2} + U_{n2}, 
\end{align*}
where $\eta_{n2} \equiv \bs{f}^\top (\bs{I}_n - \bs{P}_{n1}) \bs{f} + 2 \bs{f}^\top (\bs{I}_n - \bs{P}_{n1}) \bs{\varepsilon}$ and 
$U_{n2} = \bs{\varepsilon}^\top (\bs{I}_n - \bs{P}_{n1}) \bs{\varepsilon}$, 
and 
\begin{align*}
    \text{RSS}_0 - 
    \text{RSS}_1
    & 
    = 
    \bs{Y}^\top (\bs{I}_n - \bs{P}_{n0}) \bs{Y} - \bs{Y}^\top (\bs{I}_n - \bs{P}_{n1}) \bs{Y}
    = \bs{Y}^\top (\bs{P}_{n1}-\bs{P}_{n0}) \bs{Y}
    = 
    (\bs{f} + \bs{\varepsilon})^\top (\bs{P}_{n1}-\bs{P}_{n0}) (\bs{f} + \bs{\varepsilon})
    \\
    & = 
    \big\{ \bs{f}^\top (\bs{P}_{n1}-\bs{P}_{n0}) \bs{f} 
    + 
    2 \bs{f}^\top (\bs{P}_{n1}-\bs{P}_{n0}) \bs{\varepsilon} \big\}
    + \bs{\varepsilon}^\top (\bs{P}_{n1}-\bs{P}_{n0}) \bs{\varepsilon}
    \equiv \eta_{n1} + U_{n1}, 
\end{align*}
where $\eta_{n1} \equiv \bs{f}^\top (\bs{P}_{n1}-\bs{P}_{n0}) \bs{f} + 2 \bs{f}^\top (\bs{P}_{n1}-\bs{P}_{n0}) \bs{\varepsilon}$ 
and $U_{n1} \equiv \bs{\varepsilon}^\top (\bs{P}_{n1}-\bs{P}_{n0}) \bs{\varepsilon}$. 

Second, 
by the properties of Gaussian distributions and projection matrices, conditional on $\bXmatrix$, 
$U_{n1}$ and $U_{n2}$ follow chi-squared distributions with degrees of freedom $d_{n1} \equiv p_{n1} - p_{n0}$ and $d_{n2} \equiv n-p_{n1}$, 
and they are mutually independent. 

Third, we consider the limiting behavior of $\eta_{n1}$ and $\eta_{n2}$ as $n\rightarrow \infty$.
Define $\xi_{n1} \equiv \bs{f}^\top (\bs{P}_1^{(n)}-\bs{P}_0^{(n)}) \bs{f}$ and $\xi_{n2} \equiv \bs{f}^\top (\bs{I}_n - \bs{P}_1^{(n)}) \bs{f}$. 
From the conditions in Lemma \ref{lemma:power_given_X}, 
$\xi_{n2}/\sqrt{d_{n1}} = o(1)$, 
and $(\xi_{n1}+\xi_{n2})/\sqrt{d_{n1}} = \text{(or $\ge$) } \theta + o(1)$. 
These then imply that $\xi_{n1}/\sqrt{d_{n1}} = \text{(or $\ge$) } \theta + o(1)$.
Note that 
conditional on $\bXmatrix$, 
$\eta_{n1} \mid \bXmatrix \sim \mathcal{N}(\xi_{n1}, 4 \xi_{n1})$ 
and 
$\eta_{n2} \mid \bXmatrix \sim \mathcal{N}(\xi_{n2}, 4 \xi_{n2})$.  
By Chebyshev's inequality, 
we must have, 
conditional on $\bXmatrix$, 
\begin{align*}
    \frac{\eta_{n1}}{\sqrt{d_{n1}}}
    & = 
    \frac{\xi_{n1}}{\sqrt{d_{n1}}}
    + \frac{\sqrt{\xi_{n1}}}{\sqrt{d_{n1}}} \cdot O_{\Pr}(1)
    =
    \left( \frac{\xi_{n1}}{\sqrt{d_{n1}}} \right)^{1/2} 
    \left\{ 
    \left( \frac{\xi_{n1}}{\sqrt{d_{n1}}} \right)^{1/2} + \frac{O_{\Pr}(1)}{d_{n1}^{1/4}} 
    \right\}
    = \text{(or $\ge$) } \theta + o_{\Pr}(1), 
\end{align*}
and 
\begin{align*}
    \frac{\eta_{n2}}{\sqrt{d_{n1}}}
    & = 
    \frac{\xi_{n2}}{\sqrt{d_{n1}}}
    + \frac{\sqrt{\xi_{n2}}}{\sqrt{d_{n1}}} \cdot O_{\Pr}(1)
    = 
    \frac{\xi_{n2}}{\sqrt{d_{n1}}}
    + 
    \left( \frac{\xi_{n2}}{\sqrt{d_{n1}}} \right)^{1/2} \frac{1}{d_{n1}^{1/4}} \cdot O_{\Pr}(1)
    = o_{\Pr}(1), 
\end{align*}
where the latter immediately implies that 
$
    \eta_{n2}\sqrt{d_{n1}}/d_{n2}
    = 
    \eta_{n2}/\sqrt{d_{n1}} \cdot 
    d_{n1}/d_{n2} = o_{\Pr}(1). 
$

Fourth, we consider the conditional distribution of the partial permutation $p$-value given $\bXmatrix$. From Theorem \ref{thm:general_F_test}, we have 
$
p(\bXmatrix, \bs{Y}, \bs{Z}) = 1 - F_{d_{n1}, d_{n2}} (F(\bXmatrix, \bs{Y}, \bs{Z})), 
$
where $F_{d_{n1}, d_{n2}}$ denotes the distribution function of the F distribution with degrees of freedom $d_{n1}$ and $d_{n2}$. 
From the conditions in \ref{lemma:power_given_X} and the discussion before, using Lemma \ref{lemma:ratio_chi_sq_survival}, 
we have 
\begin{align*}
    & \quad \ \Pr\left( p(\bXmatrix, \bs{Y}, \bs{Z}) \le \alpha \mid \bXmatrix \right)
    \\ 
    & = 
    \Pr\left( F(\bXmatrix, \bs{Y}, \bs{Z}) \ge F^{-1}_{d_{n1}, d_{n2}}(1-\alpha) \mid \bXmatrix \right)
    = 
    \Pr\left( \frac{\eta_{n1} + U_{n1}}{\eta_{n2} + U_{n2}} \ge \frac{d_{n1}}{d_{n2}} F^{-1}_{d_{n1}, d_{n2}}(1-\alpha) \mid \bXmatrix \right)
    \\
    & = \text{(or $\ge$) } 
    1 - 
    \Phi(z_{1-\alpha} - \theta/\sqrt{2}) + o(1)
    = \Phi(z_{\alpha} + \theta/\sqrt{2}) + o(1). 
\end{align*}

From the above, Lemma \ref{lemma:power_given_X} holds. 
\end{proof}

\begin{proof}[\bf Proof of Theorem \ref{thm:power_diverg_kernel}]
First, we consider the limiting behavior of $\bs{f}^\top (\bs{I}_n - \bs{P}_{n1}) \bs{f}$ as $n\rightarrow \infty$. 
Let $\bs{\beta}_{q_n, h}\in \mathbb{R}^{q_n}$ be the coefficient for the best linear approximation of $f_{nh}$ using the basis functions $\{e_j\}_{j=1}^{q_n}$, i.e., 
$\remainder(f_{nh}; q_n) = \int (f_{nh} - \bs{\beta}_{q_n, h}^\top \phi_{q_n})^2 \text{d} \mu$, for all $n$ and $1\le h \le H$. 
Let $\check{\bs{f}} = (\check{f}_1, \ldots, \check{f}_n)$ with $\check{f}_i = \sum_{i=1}^H \I(Z_i=h) \bs{\beta}_{q_n, h}^\top \phi_{q_n}(\bs{X}_i)$. 
Note that $\bs{P}_{n1}$ is the projection matrix onto the column space of the transformed covariates under the full model. 
We must have 
\begin{align*}
    \bs{f}^\top (\bs{I}_n - \bs{P}_{n1}) \bs{f}
    & = 
    ( \bs{f} - \check{\bs{f}} )^\top (\bs{I}_n - \bs{P}_{n1}) ( \bs{f} - \check{\bs{f}} ) 
    \le 
    \| \bs{f} - \check{\bs{f}} \|_2^2 
    = 
    \sum_{h=1}^H \sum_{i:Z_i=h} \{f_{nh}(\bs{X}_i) - \bs{\beta}_{q_n, h}^\top \phi_{q_n}(\bs{X}_i)\}^2
    \\
    & = \sum_{h=1}^H \sum_{i:Z_i=h} \E [ \{f_{nh}(\bs{X}_i) - \bs{\beta}_{q_n, h}^\top \phi_{q_n}(\bs{X}_i)\}^2 ]
    \cdot O_{\Pr}(1)
    \le n \sum_{h=1}^H \remainder(f_{nh}; q_n) \cdot  O_{\Pr}(1). 
\end{align*}
From the conditions in Theorem \ref{thm:power_diverg_kernel}, we then have 
\begin{align}\label{eq:rss1_prob_converge}
    \frac{\bs{f}^\top (\bs{I}_n - \bs{P}_{n1}) \bs{f}}{\sqrt{p_{n1}-p_{n0}}}
    = \frac{n \sum_{h=1}^H \remainder(f_{nh}; q_n)}{\sqrt{p_{n1}-p_{n0}}} \cdot  O_{\Pr}(1)
    = o_{\Pr}(1). 
\end{align}

Second, we consider the conditional distribution of the partial permutation $p$-value given $\bs{X}$. 
From the conditions in Theorem \ref{thm:power_diverg_kernel} and \eqref{eq:rss1_prob_converge}, 
and using the property of convergence in probability \citep[][Theorem 2.3.2]{durrett2019probability}, 
we can then derive that 
$\Pr\{ p(\bXmatrix, \bs{Y}, \bs{Z}) \ge \alpha \mid \bXmatrix\}= \text{(or $\ge$) } \Phi(z_{\alpha} + \theta/\sqrt{2}) + o_{\Pr}(1)$. 

Third, because $\Pr\{ p(\bXmatrix, \bs{Y}, \bs{Z}) \ge \alpha \mid \bXmatrix\}$ is bounded between 0 and 1, we must have 
\begin{align*}
    \Pr\{ p(\bXmatrix, \bs{Y}, \bs{Z}) \ge \alpha \}
    = 
    \E[ \Pr\{ p(\bXmatrix, \bs{Y}, \bs{Z}) \ge \alpha \mid \bXmatrix\} ] 
    = \text{(or $\ge$) } \Phi(z_{\alpha} + \theta/\sqrt{2}) + o(1). 
\end{align*}
Therefore, Theorem \ref{thm:power_diverg_kernel} holds. 
\end{proof}

\begin{proof}[\bf Comments on the implication from Theorem \ref{thm:power_diverg_kernel} when $f_{n1} = \ldots = f_{nH} = f_0$]\ \\
Let $\bs{\beta}_{q_n}\in \mathbb{R}^{q_n}$ be the coefficient for the best linear approximation of $f_{0}$ using the basis functions $\{e_j\}_{j=1}^{q_n}$, i.e., 
$\remainder(f_{0}; q_n) = \int (f_{0} - \bs{\beta}_{q_n}^\top \phi_{q_n})^2 \text{d} \mu$. 
From the definition of $\bs{P}_{n0}$ and by the same logic as that in the proof of Theorem \ref{thm:power_diverg_kernel}, 
\begin{align*}
    \bs{f}^\top (\bs{I}_n - \bs{P}_{n0}) \bs{f}
    & \le 
    \sum_{i=1}^n (f_0(\bs{X}_i) -  \bs{\beta}_{q_n}^\top \phi_{q_n}(\bs{X}_i))^2 
    = 
    \sum_{i=1}^n \E\{ (f_0(\bs{X}_i) -  \bs{\beta}_{q_n}^\top \phi_{q_n}(\bs{X}_i))^2 \} \cdot O_{\Pr}(1)
    \\
    & = 
    n \remainder(f_{0}; q_n) \cdot O_{\Pr}(1). 
\end{align*}
Thus, when  $p_{n1} - p_{n0} \asymp q_n$, a sufficient condition for 
$\bs{f}^\top (\bs{I}_n - \bs{P}_{n0}) \bs{f}/\sqrt{p_{n1}-p_{n0}} = o_{\Pr}(1)$ is 
$n \remainder(f_{0}; q_n) = o_{\Pr}(q_n^{1/2})$. 
From Theorem \ref{thm:power_diverg_kernel}, 
if $p_{n0} \asymp q_n$, $p_{n1} - p_{n0} \asymp q_n$, $q_n \rightarrow \infty$, $q_n/n\rightarrow 0$ and $n \remainder(f_{0}; q_n) = o_{\Pr}(q_n^{1/2})$, 
then 
the partial permutation test must be asymptotically valid. 
\end{proof}

\begin{proof}[\bf Comments on the implication from Theorem \ref{thm:power_diverg_kernel} when $f_{nh} = f_0 + \delta_n \zeta_h$ for $1\le h \le H$]\  \\
First, we consider the limiting behavior of $n \sum_{h=1}^H \remainder(f_{nh}; q)$. 
By definition, we can verify that 
$
  \remainder(f_{nh}; q)  =    \remainder(f_0 + \delta_n \zeta_h; q) 
\le 
2 \remainder(f_0; q) + 2 \delta_n^2 \remainder(\zeta_h; q)
$
for all $n, q$ and $1\le h\le H$. 
Consequently, 
\begin{align*}
    n \sum_{h=1}^H \remainder(f_{nh}; q) \le 2 n H \remainder(f_0; q) + 2 n \delta_n^2 \sum_{h=1}^H \remainder(\zeta_h; q)
    = 
    n \big\{\remainder(f_0; q) + \sum_{h=1}^H \remainder(\zeta_h; q) \big\} \cdot O(1). 
\end{align*}

Second, we consider the limiting behavior of $\bs{f}^\top (\bs{I}_n - \bs{P}_{n0}) \bs{f}$. 
Assume that the covariates are exactly the same across all groups, and that they are ordered in the same way. 
For samples in group $h$, 
let $\bs{f}_h \in \mathbb{R}^{n/H}$ be the vector of function values, and $\bs{\Phi}_{nh}\in \mathbb{R}^{(n/H)\times q_n}$ be the matrix of transformed covariates. 
Obviously, $\bs{\Phi}_{n1} = \bs{\Phi}_{n2} = \ldots = \bs{\Phi}_{nH}$. 
Moreover, by the property of projection matrix, we know $\bs{f}^\top (\bs{I}_n - \bs{P}_{n0}) \bs{f} = \min_{\bs{b}} \sum_{h=1}^H \left\|\bs{f}_h - \bs{\Phi}_{nh} \bs{b} \right\|_2^2$. 
Note that for any $1\le h,h'\le H$ and any $\bs{b}\in \mathbb{R}^q$, we have 
\begin{align*}
    & \quad \ \left\|\bs{f}_h - \bs{\Phi}_{nh} \bs{b} \right\|_2^2 + \left\|\bs{f}_{h'} - \bs{\Phi}_{nh'} \bs{b} \right\|_2^2
    = \left\|\bs{f}_h - \bs{\Phi}_{nh} \bs{b} \right\|_2^2 + \left\|\bs{f}_{h'} - \bs{\Phi}_{nh} \bs{b} \right\|_2^2\\
    & = 
    \frac{1}{2} 
    \left\{ \left\| (\bs{f}_h - \bs{\Phi}_{nh} \bs{b}) - (\bs{f}_{h'} - \bs{\Phi}_{nh} \bs{b}) \right\|_2^2 + 
    \left\| (\bs{f}_h - \bs{\Phi}_{nh} \bs{b}) + (\bs{f}_{h'} - \bs{\Phi}_{nh} \bs{b}) \right\|_2^2
    \right\}
    \\
    & \ge \frac{1}{2 }\left\| \bs{f}_h - \bs{f}_{h'} \right\|_2^2. 
\end{align*}
Consequently, we must have 
\begin{align*}
    \bs{f}^\top (\bs{I}_n - \bs{P}_{n0}) \bs{f}
    & 
    \ge 
    \frac{1}{2} \max_{h, h'} \|\bs{f}_h - \bs{f}_{h'}\|_2^2
    = 
    \frac{n}{2H} \max_{h, h'} \frac{H}{n} \sum_{i:Z_i=1} \left\{f_{nh}(\bs{X}_i) - f_{nh'}(\bs{X}_i)\right\}^2
    \\
    & = \frac{n\delta_n^2}{2H}  \max_{h, h'} \frac{H}{n} \sum_{i:Z_i=1} \left\{\zeta_{h}(\bs{X}_i) - \zeta_{h'}(\bs{X}_i)\right\}^2
    = 
    \frac{n\delta_n^2}{2H} \big(   \max_{h, h'} \tau_{hh'} +  o_{\Pr}(1) \big), 
\end{align*}
where the last equality holds by the law of large numbers.

Third, we consider the case in which 
$n\remainder(f_0; q) = o(q_n^{1/2})$, $n \remainder(\zeta_h; q) = o(q_n^{1/2})$ for all $h$, and $n\delta_n^2 \ge \theta \sqrt{8H^3q_n}$ for some $\theta>0$. 
Suppose that 
$p_{n0} \asymp q_n$, $p_{n1} - p_{n0} \asymp q_n$, $q_n \rightarrow \infty$ and $q_n/n\rightarrow 0$. 
From the discussion before, we have 
\begin{align*}
    \frac{n \sum_{h=1}^H \remainder(f_{nh}; q)}{\sqrt{p_{n1}-p_{n0}}} 
    = 
    \frac{n \{\remainder(f_0; q) + \sum_{h=1}^H \remainder(\zeta_h; q) \}}{\sqrt{q_n}} \cdot O(1) = o(1), 
\end{align*}
and 
\begin{align*}
    \frac{\bs{f}^\top (\bs{I}_n - \bs{P}_{n0}) \bs{f}}{\sqrt{p_{n1}-p_{n0}}}
    & \ge 
    \frac{\bs{f}^\top (\bs{I}_n - \bs{P}_{n0}) \bs{f}}{\sqrt{Hq_n}}
    \ge 
    \frac{1}{\sqrt{Hq_n}}
    \frac{n\delta_n^2}{2H}  \max_{h, h'} \frac{H}{n} \sum_{i:Z_i=1} \left\{\zeta_{h}(\bs{X}_i) - \zeta_{h'}(\bs{X}_i)\right\}^2
    \\
    & = 
    \frac{\sqrt{2}\  n\delta_n^2}{\sqrt{8 H^3q_n}}
    \max_{h, h'} \frac{H}{n} \sum_{i:Z_i=1} \left\{\zeta_{h}(\bs{X}_i) - \zeta_{h'}(\bs{X}_i)\right\}^2
    \ge 
    \sqrt{2} \ \theta \max_{h, h'} \frac{H}{n} \sum_{i:Z_i=1} \left\{\zeta_{h}(\bs{X}_i) - \zeta_{h'}(\bs{X}_i)\right\}^2
    \\
    & =\sqrt{2} \ \theta \max_{h, h'} \tau_{hh'} +  o_{\Pr}(1). 
\end{align*}
From Theorem \ref{thm:power_diverg_kernel}, the power of the partial permutation test must satisfy that 
$
    \Pr\{ p(\bXmatrix, \bs{Y}, \bs{Z}) \le \alpha\} \ge \Phi(z_{\alpha}+\theta \max_{h, h'} \tau_{hh'}) + o(1). 
$
\end{proof}

\section{Properties of Kernel Functions}\label{sec:proof_kernel}

To prove Proposition \ref{prop:decay_eigenvalue},  we need the following lemma.

\begin{lemma}%
\label{kuhn_1987}
	Let $\mathcal{X}$ be a metric compactum with $\varepsilon_n(\mathcal{X})\asymp \varepsilon_{2n}(\mathcal{X})$, and let $0<s\leq 1$. Then for every positive definite kernel $K \in C^{s,0}(\mathcal{X},\mathcal{X})$ and every finite Borel measure $\mu$ on $\mathcal{X}$ one has
	\begin{align*}
	\lambda_n(T_{K,\mu}) = O(n^{-1}\varepsilon_n(\mathcal{X})^s), n\geq 1,
	\end{align*}
where $\varepsilon_n(\mathcal{X})$ is the $n$th entropy number of $\mathcal{X}$ defined as
\begin{align*}
\varepsilon_n(\mathcal{X}) = \text{inf}\left\{
\varepsilon>0: \mathcal{X} \text{ can be covered by $n$ balls of radius }\varepsilon
\right\}.
\end{align*}
\end{lemma}

\begin{proof}[Proof of Lemma \ref{kuhn_1987}]
	Lemma \ref{kuhn_1987} follows directly from Theorem 4 in \citet{kuhn1987eigenvalues}. 
\end{proof}

\begin{proof}[\bf Proof of Proposition \ref{prop:decay_eigenvalue}]
Define a measure $\tilde{\mu}$ in $\tilde{\mathcal{X}} = [-b, b]^d$ as 
$
\tilde{\mu}(A) = \mu(A\cap \mathcal{X})
$
for all measurable set $A \subseteq \tilde{\mathcal{X}}$.
It is not difficult to see that 
$K$ defined on $(\mathcal{X}\times \mathcal{X}, \mu\times \mu)$ has the same eigenvalues as $\tilde{K}$ defined on $(\tilde{\mathcal{X}} \times \tilde{\mathcal{X}}, \tilde{\mu}\times\tilde{\mu})$. Because 
$\varepsilon_n([-b,b]^d) \asymp n^{-1/d}$, from Lemma \ref{kuhn_1987}, we have
$
\lambda_n = O(n^{-1}\varepsilon_n(\tilde{\mathcal{X}})^s) = O(n^{-1-s/d}).
$
Therefore, Proposition \ref{prop:decay_eigenvalue} holds. 
\end{proof}

\section{Partial Permutation Test with Fixed Functional Relationship and Correlated Noises}\label{sec:proof_corr}

\begin{proof}[\bf Proof of Theorem \ref{thm:fixed_property_theorem_corr}]
Note that under model $H_{0\text{C}}$ in \eqref{eq:v_fixed_H0_corr}, 
$\bs{\Sigma}^{-1/2} \bs{Y} = \bs{\Sigma}^{-1/2} \bs{f} + \bs{\Sigma}^{-1/2} \bs{\varepsilon}$, 
where $\bs{\Sigma}^{-1/2} \bs{\varepsilon} \sim \mathcal{N}(0, \sigma_0^2 \bs{I})$. 
Moreover, we are essentially conducting the partial permutation test using the response vector $\bs{\Sigma}^{-1/2} \bs{Y}$ and the kernel matrix $\bs{K}_n^{\text{C}} = \bs{\Sigma}^{-1/2} \bs{K}_n \bs{\Sigma}^{-1/2}$. 
Therefore, 
by the same logic as Theorem \ref{thm:fixed_property_theorem},
we can derive Theorem \ref{thm:fixed_property_theorem_corr}. 
For conciseness, we omit the details here. 
\end{proof}

\begin{proof}[\bf Proof of Corollary \ref{cor:kernel_finite_dim_feature_space_corr}]

By the same logic as Corollary \ref{cor:kernel_finite_dim_feature_space}, 
it suffices to prove that 
$\omega_{\text{C}}(b_n,\sigma_0^{-1}f_0, \bs{\Sigma}) =  \sigma_0^{-2} \sum_{i=n-b_n+1}^n (\bs{\gamma}_{i}^\top \bs{\Sigma}^{-1/2} \bs{f}_0)^2 = 0$ for permutation size $b_n \le n-q$. 
Let $\bs{\Phi} = (\phi(\bs{X}_1), \ldots, \phi(\bs{X}_n))^\top \in \mathbb{R}^{n \times q}$ be the matrix consisting of all covariates mapped into the feature space. 
We can then write 
$\bs{K}_n^{\text{C}}$ equivalently as 
$
\bs{K}_n^{\text{C}} = \bs{\Sigma}^{-1/2} \bs{K}_n \bs{\Sigma}^{-1/2}
= \bs{\Sigma}^{-1/2} \bs{\Phi} \bs{\Phi}^{\top} \bs{\Sigma}^{-1/2}, 
$
whose rank is at most $q$. 
Consequently, 
for the eigen-decomposition of $\bs{K}_n^{\text{C}}$, 
we have $\bs{\gamma}_i^\top \bs{K}_n^{\text{C}} \bs{\gamma}_i = \bs{\gamma}_i^\top \bs{\Sigma}^{-1/2} \bs{\Phi} \bs{\Phi}^{\top} \bs{\Sigma}^{-1/2} \bs{\gamma}_i = 0$ for $i>q$, 
which immediately implies that $\bs{\gamma}_i^\top \bs{\Sigma}^{-1/2} \bs{\Phi} = \bs{0}_{1\times q}$ for $i > q$. 

Because the function $f(\bs{x})$ is linear in $\phi(x)$, we can write $f(\bs{x})$ as $f(\bs{x}) = \phi(x)^\top \bs{\beta}$ for some $\bs{\beta}\in \mathbb{R}^q$. 
This immediately implies that $\bs{f} = \bs{\Phi} \bs{\beta}$. 
Thus, for any $i>q$, we have 
$
\bs{\gamma}_{i}^\top \bs{\Sigma}^{-1/2} \bs{f}
= \bs{\gamma}_{i}^\top \bs{\Sigma}^{-1/2} \bs{\Phi} \bs{\beta}
= \bs{0}_{1\times q} \bs{\beta} = 0. 
$
Therefore, 
when $b_n\le n-q$ or equivalently $n-b_n+1\ge q+1$, 
we must have 
$\omega_{\text{C}}(b_n,\sigma_0^{-1}f_0, \bs{\Sigma}) =  \sigma_0^{-2} \sum_{i=n-b_n+1}^n (\bs{\gamma}_{i}^\top \bs{\Sigma}^{-1/2} \bs{f}_0)^2 = 0$. 

From the above, Corollary \ref{cor:kernel_finite_dim_feature_space_corr} holds. 
\end{proof}

\begin{proof}[\bf Proof of Corollary \ref{cor:corr_diverg_kernel}]
First, we consider the limiting behavior of $\omega(b_n, \sigma_0^{-1} f_0, \bs{\Sigma})$ as $n \rightarrow \infty$. 
For $q\ge 1$, let $\bs{\beta}_q\in \mathbb{R}^q$ be the coefficient vector for the best linear approximation of $f_0$ using the first $q$ basis functions, 
    i.e., 
    $\remainder(f_0;q) = \int ( f - \bs{\beta}_q^\top \phi_q )^2 \text{d} \mu$. 
    Let $\bs{\Phi}_q = (\phi_q(\bs{X}_1), \ldots, \phi_q(\bs{X}_n))^\top \in \mathbb{R}^{n\times q}$ be the matrix consisting of transformed covariates using the first $q$ basis functions. 
    Then the kernel matrix $\bs{K}_n^{\text{C}}$ can be written as $\bs{K}_n^{\text{C}} = \bs{\Sigma}^{-1/2}\bs{\Phi}_{q_n} \bs{\Phi}_{q_n}^\top \bs{\Sigma}^{-1/2}$, whose rank is at most $q_n < n$. 
    Thus, the eigenvectors of $\bs{K}_n^{\text{C}}$ must satisfy that $0 = \bs{\gamma}_i^\top \bs{K}_n^{\text{C}} \bs{\gamma}_i = \bs{\gamma}_i^\top \bs{\Sigma}^{-1/2} \bs{\Phi}_{q_n} \bs{\Phi}_{q_n}^\top \bs{\Sigma}^{-1/2} \bs{\gamma}_i$ for $q_n < i \le n$. 
    Consequently, 
    for $q_n < i \le n$, 
    $\bs{\gamma}_i^\top \bs{\Sigma}^{-1/2} \bs{\Phi}_{q_n} = \bs{0}$, 
    and 
    $
        \bs{\gamma}_i^\top \bs{\Sigma}^{-1/2} \bs{f}_0 
        = 
        \bs{\gamma}_i^\top \bs{\Sigma}^{-1/2}
        ( \bs{\Phi}_{q_n} \bs{\beta}_{q_n} + \bs{f}_0 - \bs{\Phi}_{q_n} \bs{\beta}_{q_n} )
        = \bs{\gamma}_i^\top \bs{\Sigma}^{-1/2} (\bs{f}_0 - \bs{\Phi}_{q_n} \bs{\beta}_{q_n} ), 
    $
    where $\bs{f}_0 = (f_0(\bs{X}_1), \ldots, f_0(\bs{X}_n))$. 
    This then implies that 
    \begin{align*}
        \sum_{i=q_n+1}^n ( \bs{\gamma}_i^\top \bs{\Sigma}^{-1/2} \bs{f}_0 )^2 
        & = 
        \sum_{i=q_n+1}^n 
        \big\{ \bs{\gamma}_i^\top \bs{\Sigma}^{-1/2} (\bs{f}_0 - \bs{\Phi}_{q_n} \bs{\beta}_{q_n} ) \big\}^2 
        \le 
        \| \bs{\Sigma}^{-1/2} ( \bs{f}_0 - \bs{\Phi}_{q_n} \bs{\beta}_{q_n} ) \|_2^2 
        \\
        & 
        \le \frac{1}{\lambda_{\min}(\bs{\Sigma})} \|  \bs{f}_0 - \bs{\Phi}_{q_n} \bs{\beta}_{q_n} \|_2^2
        = 
        \frac{1}{\lambda_{\min}(\bs{\Sigma})} \sum_{i=1}^n 
        \{ f_0(\bs{X}_i) - \bs{\beta}_{q_n}^\top \phi_{q_n}(\bs{X}_i) \}^2 
        \\
        & = 
        \frac{1}{\lambda_{\min}(\bs{\Sigma})}
        \sum_{i=1}^n \E [ \{ f_0(\bs{X}_i) - \bs{\beta}_{q_n}^\top \phi_{q_n}(\bs{X}_i) \}^2 ]
        \cdot
        O_{\Pr}(1)
        = 
        \frac{n \remainder(f_0; q_n)}{\lambda_{\min}(\bs{\Sigma})} \cdot O_{\Pr}(1)
    \end{align*}
    Consequently, we have 
    $
       \omega_{\text{C}}(b_n,\sigma_0^{-1}f_0, \bs{\Sigma})
       = 
       \sigma_0^{-2}
       \sum_{i=n-b_n+1}^n ( \bs{\gamma}_{i}^\top \bs{\Sigma}^{-1/2} \bs{f}_0)^2
       =  n \remainder(f_0; q_n)/\lambda_{\min}(\bs{\Sigma}) \cdot O_{\Pr}(1). 
    $
    
    Second, we prove the asymptotic validity of the partial permutation test. 
    By the same logic as the proof of Theorem \ref{thm:fixed_property_theorem_corr}, and using Lemmas \ref{lemma:basic_ineq_perm_p_val} and \ref{lemma:bound_TV}, we can know that, for any $\alpha, \delta > 0$, 
    the $p$-value from the partial permutation test with kernel $K_{q_n}$, permutation size $b_n \le n-q_n$, any test statistic $T$ and covariance matrix $\bs{\Sigma}$ satisfies that 
    \begin{align}\label{eq:pval_bound_diverging_corr}
    \Pr \{  p( \bXmatrix, \bs{Y}, \bs{Z} ) \leq \alpha \}
    \leq & 
    \alpha+ \delta + \Pr \left( e^{\Delta_n}-1 > 2\delta \right)
    = 
    \alpha+ \delta + \Pr \left\{ \Delta_n > \log(1+2\delta) \right\}
    \end{align}
    with 
    $
    \Delta_n 
    =  2 \sqrt{2} \sqrt{\omega_{\text{C}}(b_n, \sigma_0^{-1}f_0, \bs{\Sigma}) } \cdot \sqrt{
	\omega_{\text{C}}(b_n, \sigma_0^{-1}f_0, \bs{\Sigma}) + \sigma_0^{-2}
	\sum_{i=n-b_n+1}^{n}(\bs{\gamma}_{i}^\top\bs{\varepsilon})^2 }.
    $
    Note that $(\bs{\gamma}_1^\top \bs{\varepsilon}, \ldots, \bs{\gamma}_n^\top \bs{\varepsilon})$ are i.i.d.\ standard Gaussian. We must have 
    $
    \sum_{i=n-b_n+1}^{n}(\bs{\gamma}_{i}^\top\bs{\varepsilon})^2 
    = O_{\Pr}(b_n ). 
    $
    Consequently, 
    \begin{align*}
    \Delta_n 
    & =  \sqrt{ \frac{n \remainder (f_0; q_n)}{\lambda_{\min}(\bs{\Sigma})} } \cdot \sqrt{
	\frac{n \remainder (f_0; q_n)}{\lambda_{\min}(\bs{\Sigma})}   + b_n } \cdot O_{\Pr}(1)
	=
	\sqrt{
	\left\{ \frac{n (n-q_n) \remainder (f_0; q_n)}{\lambda_{\min}(\bs{\Sigma})}\right\}^2 + 
	\frac{n (n-q_n) \remainder (f_0; q_n)}{\lambda_{\min}(\bs{\Sigma})}
	}\cdot O_{\Pr}(1),
    \end{align*}
    where the last equality holds because $\max\{1, b_n\}\le n-q_n$. 
    From the condition in Corollary \ref{cor:corr_diverg_kernel}, we must have $\Delta_n = o_{\Pr}(1)$. 
    Letting $n$ go to infinity in \eqref{eq:pval_bound_diverging_corr}, we can know that, for any $\alpha, \delta>0$, 
    \begin{align*}
    \limsup_{n\rightarrow \infty}\Pr \{  p( \bXmatrix, \bs{Y}, \bs{Z}, \bs{\Sigma} ) \leq \alpha \}
    \leq & 
    \alpha+ \delta + 
    \limsup_{n\rightarrow \infty} \Pr \left\{ \Delta_n > \log(1+2\delta) \right\}
    = 
    \alpha+\delta. 
    \end{align*}
    Because the above inequality holds for any $\delta> 0$, 
    we must have 
    $
    \limsup_{n\rightarrow \infty}\Pr \{  p( \bXmatrix, \bs{Y}, \bs{Z}, \bs{\Sigma} ) \leq \alpha  \}
    \le 
    \alpha. 
    $
    Therefore, $p( \bXmatrix, \bs{Y}, \bs{Z}, \bs{\Sigma} )$ is an asymptotically valid $p$-value. 
    
    From the above, Corollary \ref{cor:corr_diverg_kernel} holds. 
\end{proof}

\begin{proof}[\bf Proof of Corollary \ref{cor:fixed_property_theorem_balanced_design_corr}]
By the same logic as Corollary \ref{cor:kernel_finite_dim_feature_space_corr}, it suffices to prove that 
$\omega_{\text{C}}(b_n,\sigma_0^{-1}f_0, \bs{\Sigma}) =  \sigma_0^{-2} \sum_{i=n-b_n+1}^n (\bs{\gamma}_{i}^\top \bs{\Sigma}^{-1/2} \bs{f}_0)^2 = 0$ for permutation size $b_n \le n-r$. 
 Without loss of generality, we assume the units are ordered according to their group indicators in the sense that 
$Z_{(k-1)m + 1} = Z_{(k-1)m + 1} = \ldots = Z_{km} = k$ for $1\le k \le H$,
and the covariates in these groups satisfy \eqref{eq:cov_balance_groups}. 
Following the proof of Corollary \ref{cor:fixed_property_theorem_balanced_design}, 
we can write the kernel matrix $\bs{K}_n$ and the function vector $\bs{f}$ equivalently as 
$\bs{K}_n = \bs{\Pi} \bs{G} \bs{\Pi}^\top$ and $\bs{f} = \bs{\Pi} \tilde{\bs{f}}$, where $\bs{\Pi} \in \mathbb{R}^{n\times r}, \bs{G}\in \mathbb{R}^{r\times r}$ and $\tilde{\bs{f}}\in \mathbb{R}^r$ is defined the same as that in the proof of Corollary \ref{cor:fixed_property_theorem_balanced_design}. 
Moreover, $\bs{G}$ is positive definite, and $\bs{K}_n$ is of rank $r$. 
Thus, $\bs{K}_n^{\text{C}}$ is also of rank $r$, and its eigenvectors satisfies that $\bs{\gamma}_i^\top  \bs{K}_n^{\text{C}} \bs{\gamma}_i = \bs{\gamma}_i^\top \bs{\Sigma}^{-1/2} \bs{K}_n \bs{\Sigma}^{-1/2} \bs{\gamma}_i = \bs{\gamma}_i^\top \bs{\Sigma}^{-1/2} \bs{\Pi} \bs{G} \bs{\Pi}^\top \bs{\Sigma}^{-1/2} \bs{\gamma}_i = 0$ for $i > r$. 
This immediately implies that, for $i>r$, 
$
\bs{\gamma}_i^\top \bs{\Sigma}^{-1/2} \bs{\Pi} = \bs{0}_{1\times r}
$
and 
thus 
$
\bs{\gamma}_i^\top \bs{\Sigma}^{-1/2} \bs{f} = \bs{\gamma}_i^\top \bs{\Sigma}^{-1/2} \bs{\Pi} \tilde{\bs{f}} = 0. 
$
Therefore, when $b_n\le n-r$, we have 
$n-b_n+1 > r$, and consequently 
$\omega_{\text{C}}(b_n,\sigma_0^{-1}f_0, \bs{\Sigma}) =  \sigma_0^{-2} \sum_{i=n-b_n+1}^n (\bs{\gamma}_{i}^\top \bs{\Sigma}^{-1/2} \bs{f}_0)^2 = 0$. From the above, Corollary \ref{cor:fixed_property_theorem_balanced_design_corr} holds. 
\end{proof}

\begin{proof}[\bf Comment on the special case with a particular covariance structure]
Here we consider a special case under the setting of Corollary \ref{cor:fixed_property_theorem_balanced_design_corr}. 
Specifically, all covariates within each group are distinct, and, assuming the covariates are ordered such that \eqref{eq:cov_balance_groups} holds,  the covariate matrix $\bs{\Sigma}$ up to a certain positive scale has the form 
\begin{align*}
\bs{\Sigma} = 
\begin{pmatrix}
    \bs{I}_m & \rho \bs{I}_m & \cdots & \rho \bs{I}_m\\
    \rho \bs{I}_m & \bs{I}_m & \cdots & \rho \bs{I}_m\\
    \vdots & \vdots & \ddots & \vdots\\ 
    \rho \bs{I}_m &  \rho \bs{I}_m & \cdots &\bs{I}_m
\end{pmatrix}
= 
\rho 
\begin{pmatrix}
\bs{I}_m\\
\bs{I}_m\\
\vdots\\
\bs{I}_m
\end{pmatrix}
\begin{pmatrix}
\bs{I}_m & 
\bs{I}_m & 
\cdots & 
\bs{I}_m
\end{pmatrix} + (1-\rho ) \bs{I}_n
= 
\rho \bs{\Pi} \bs{\Pi}^\top + (1-\rho ) \bs{I}_n,
\end{align*}
where $m = n/H$ and $\bs{\Pi}$ is defined the same as in \eqref{eq:Pi} noting that all covariates within each group are all distinct. 
From the proof of Corollary \ref{cor:fixed_property_theorem_balanced_design_corr}, 
we have $\bs{K}_n^{\text{C}} = \bs{\Sigma}^{-1/2} \bs{K}_n \bs{\Sigma}^{-1/2} = \bs{\Sigma}^{-1/2} \bs{\Pi} \bs{G} \bs{\Pi}^\top \bs{\Sigma}^{-1/2}$, 
where $\bs{G}$ is defined the same as in Corollary \ref{cor:fixed_property_theorem_balanced_design_corr} and is positive definite. 
Moreover, for $i>m$, 
the eigenvector of $\bs{K}_n^{\text{C}}$ corresponding to the $i$th largest eigenvalue must satisfy that $\bs{\gamma}_i^\top \bs{\Sigma}^{-1/2} \bs{\Pi} = \bs{0}_{1\times m}$. 

Below we first study the eigenvectors of $\bs{K}_n^{\text{C}}$. 
Let $\bs{\Pi} = \bs{A} \bs{D} \bs{B}^\top$ be the singular decomposition of the matrix $\bs{\Pi}$, 
where $\bs{A}\in \mathbb{R}^{n\times n}$ and $\bs{B}\in \mathbb{R}^{m\times m}$ are orthogonal matrices,  $\bs{D} = (\bs{D}_1^\top, \bs{0}_{m \times (n-m)}) ^\top \in \mathbb{R}^{n\times m}$, $\bs{D}_1\in \mathbb{R}^{m\times m}$ is a diagonal matrix with positive diagonal elements and $\bs{0}_{m \times (n-m)} \in \mathbb{R}^{n\times (n-m)}$ is a matrix with all elements being zero. 
We then have 
\begin{align*}
\bs{\Sigma} & = \rho \bs{\Pi} \bs{\Pi}^\top + (1-\rho ) \bs{I}_n
= \rho \bs{A} \bs{D} \bs{B}^\top \bs{B} \bs{D}^\top \bs{A}^\top + (1-\rho ) \bs{I}_n
= 
\rho \bs{A} \bs{D} \bs{D}^\top \bs{A}^\top + (1-\rho ) \bs{I}_n\\
& = 
\bs{A} \big\{ \rho  \bs{D} \bs{D}^\top + (1-\rho ) \bs{I}_n \big\} \bs{A}^\top. 
\end{align*}
Thus, $\bs{\Sigma}^{-1/2} \bs{\Pi}$ simplifies to 
\begin{align*}
\bs{\Sigma}^{-1/2} \bs{\Pi} & = 
\bs{A} \big\{ \rho  \bs{D} \bs{D}^\top + (1-\rho ) \bs{I}_n \big\}^{-1/2} \bs{A}^\top \bs{A} \bs{D} \bs{B}^\top
= 
\bs{A} \big\{ \rho  \bs{D} \bs{D}^\top + (1-\rho ) \bs{I}_n \big\}^{-1/2} \bs{D} \bs{B}^\top\\
& = 
\bs{A}
\begin{pmatrix}
\big\{ \rho \bs{D}_1^2  + (1-\rho) \bs{I}_m \big\}^{-1/2} & \bs{0} \\
\bs{0} & (1-\rho)^{-1/2} \bs{I}_{n-m}
\end{pmatrix}
\begin{pmatrix}
\bs{D}_1\\
\bs{0}
\end{pmatrix} \bs{B}^\top
\\
& = 
\bs{A}
\begin{pmatrix}
\big\{ \rho \bs{D}_1^2  + (1-\rho) \bs{I}_m \big\}^{-1/2} \bs{D}_1 \\
\bs{0} 
\end{pmatrix}
\bs{B}^\top
= 
\begin{pmatrix}
\bs{A}_1 & \bs{A}_2
\end{pmatrix}
\begin{pmatrix}
\big\{ \rho \bs{D}_1^2  + (1-\rho) \bs{I}_m \big\}^{-1/2} \bs{D}_1 \\
\bs{0} 
\end{pmatrix}
\bs{B}^\top
\\
& = \bs{A}_1 \big\{ \rho \bs{D}_1^2  + (1-\rho) \bs{I}_m \big\}^{-1/2} \bs{D}_1 \bs{B}^\top, 
\end{align*}
where $\bs{A}_1$ and $\bs{A}_2$ are the submatrices of $\bs{A}$ consisting of the first $m$ and the last $n-m$ columns. 
Because $\bs{A}$ is an orthogonal matrix, this then implies that 
$\bs{A}_2^\top \bs{\Sigma}^{-1/2} \bs{\Pi} = \bs{0}_{(n-m)\times m}$. 
Therefore, the space spanned by the eigenvectors $\{\bs{\gamma}_{m+1}, \ldots, \bs{\gamma}_n\}$ must be the same as the column space of $\bs{A}_2$, which is equivalently the space orthogonal to the column space of $\bs{A}_1$ or $\bs{\Pi}$. 

From the above discussion, we can know that the space spanned by the last $n-m$ eigenvectors of $\bs{K}_{n}^{\text{C}}$ is always the same as the space orthogonal to the column space of $\bs{\Pi}$, regardless of the value of $\rho$ for the covariate matrix (up to a positive scale) $\bs{\Sigma}$. 
Therefore, even if we use an incorrect covariate matrix $\tilde{\bs{\Sigma}}$ with correlation $\tilde{\rho} \ne \rho$, the corresponding eigenvectors $\tilde{\bs{\gamma}}_i$'s for $\tilde{\bs{\Sigma}}^{-1/2}\bs{K}_n \tilde{\bs{\Sigma}}^{-1/2}$ will also satisfy that $\tilde{\bs{\gamma}}_i^\top \bs{\Sigma}^{-1/2} \bs{\Pi} = \bs{0}_{1\times m}$ for $i > m$. 
By the same logic as Corollary \ref{cor:fixed_property_theorem_balanced_design_corr}, the resulting partial permutation test is still valid. 
\end{proof}

\section{Partial Permutation Test under Gaussian Process Regression with Correlated Noises}\label{sec:proof_corr_GPR}

\begin{proof}[\bf Proof of Theorem \ref{thm:gp_finite_property_theorem_corr}]
Note that under $\tilde{H}_{0\text{C}}$ in \eqref{eq:h0_gp_corr}, 
\begin{align*}
    \bs{\Sigma}^{-1/2} \bs{Y} \mid \bXmatrix, \bs{Z} \sim 
    \mathcal{N}
    \left( 
    \bs{0}, \ \frac{\delta_0^2}{n^{1-\gamma}} \bs{K}_n^{\text{C}} + \sigma_0^2 \bs{I}_n
    \right).
\end{align*}
Moreover, we are essentially conducting the partial permutation test using the response vector $\bs{\Sigma}^{-1/2} \bs{Y}$ and the kernel matrix $\bs{K}_n^{\text{C}} = \bs{\Sigma}^{-1/2} \bs{K}_n \bs{\Sigma}^{-1/2}$. 
Therefore, 
by the same logic as Theorem \ref{thm:gp_finite_property_theorem},
we can derive Theorem \ref{thm:gp_finite_property_theorem_corr}. 
For conciseness, we omit the details here. 
\end{proof}

\begin{proof}[\bf Proof of Theorem \ref{thm:general_ppt_valid_corr}]
By the same logic as the proof of Theorems \ref{thm:general_ppt_valid_corr} and \ref{thm:general_ppt_valid}, it suffices to show that as $n \rightarrow \infty$, for $b_n = O(n^\kappa)$ with $0 < \kappa < 1 - \rho^{-1} - \zeta - \gamma$, 
$
    {\xi_{n-b_n+1}}/{n^{1-\gamma}} \cdot b_n \stackrel{\Pr}{\longrightarrow} 0, 
$
where $\xi_{n-b_n+1}$ is the $(n-b_n+1)$th largest eigenvalue of $\bs{K}_n^{\bs{C}} = \bs{\Sigma}^{-1/2} \bs{K}_n \bs{\Sigma}^{-1/2}$. 
For any matrix $\bs{A}$, let $\lambda_i(\bs{A})$ denote the $i$th largest eigenvalue of $\bs{A}$, and $\lambda_{\max}(\bs{A})$ and $\lambda_{\min}(\bs{A})$ denote the largest and smallest eigenvalues of $\bs{A}$. 
By the property of eigenvalues\footnote{See, e.g., \url{https://math.stackexchange.com/q/326944}}, we have, for $1\le i \le n$, 
\begin{align*}
    \lambda_i(\bs{K}_n^{\bs{C}})
    & = 
    \lambda_i(\bs{\Sigma}^{-1/2} \bs{K}_n \bs{\Sigma}^{-1/2})
    \le \lambda_i(\bs{K}_n) \lambda_{\max} (\bs{\Sigma}^{-1}) 
    = \lambda_i(\bs{K}_n) / \lambda_{\min} (\bs{\Sigma}). 
\end{align*}
This then implies that $
\xi_{n-b_n+1} \le c_{n-b_n+1} / \lambda_{\min} (\bs{\Sigma}), 
$
recalling that $c_i = \lambda_i(\bs{K}_n)$. 
From the conditions in Theorem \ref{thm:general_ppt_valid_corr}, 
we then have 
\begin{align*}
\frac{\xi_{n-b_n+1}}{n^{1-\gamma}} \cdot b_n \le 
\frac{c_{n-b_n+1} / \lambda_{\min} (\bs{\Sigma}) }{n^{1-\gamma}}  \cdot b_n
= 
\frac{c_{n-b_n+1}}{n^{1-\gamma}} \cdot O(n^{\zeta}) \cdot O(n^{\kappa})
= 
\frac{c_{n-b_n+1}}{n^{1-\gamma}} n^{\zeta+\kappa} \cdot O(1). 
\end{align*}
Note that $\zeta+\kappa$ is less than $1-\rho^{-1} - \gamma$. 
From Lemma \ref{dn_times_bn}, we must have 
${c_{n-b_n+1}}/{n^{1-\gamma}} \cdot n^{\zeta+\kappa} \stackrel{\Pr}{\longrightarrow} 0$ as $n\rightarrow \infty$, 
which immediately implies that ${\xi_{n-b_n+1}}/{n^{1-\gamma}} \cdot b_n \stackrel{\Pr}{\longrightarrow} 0$. 
From the above, Theorem \ref{thm:general_ppt_valid_corr} holds. 
\end{proof}

\end{document}